\documentclass{cqtthesis}

\title{Bandits roaming Hilbert space}
\author{JOSEP LUMBRERAS ZARAPICO}
\previousdegree{M.Sc., University of Barcelona}
\subyear{2025}
\advisor{Professor Marco Tomamichel}
\examinera{Professor Vincent Tan}
\examinerb{Associate Professor Patrick Rebentrost}
\examinerc{Professor Jens Eisert, Freie Universität Berlin}

\begin{document}

\maketitle

\begin{abstract}
This thesis studies exploration/exploitation trade-offs in online learning of properties of quantum states using multi-armed bandits. Given streaming access to an unknown quantum state, in each round we select an observable from a set of actions to maximize its expectation value. Using past information, we refine actions to minimize regret: the cumulative gap between current reward and maximum possible. We derive information-theoretic lower bounds and optimal strategies with matching upper bounds, showing regret typically scales as the square root of rounds. As an application, we reframe quantum state tomography to both learn the state efficiently and minimize measurement disturbance. For pure states and continuous actions, we achieve polylogarithmic regret using a sample-optimal algorithm based on a weighted online least squares estimator. The algorithm relies on the optimistic principle and controls the eigenvalues of the design matrix. We also apply our framework to quantum recommender systems and thermodynamic work extraction from unknown states. In this last setting, our results demonstrate an exponential advantage in work dissipation over tomography-based protocols.
\end{abstract}

%Acknowledgments

\begin{acknowledgements}

  I would like to begin by thanking my advisor, Marco Tomamichel. I am grateful for all the blackboard discussions we had over the years, which helped me grow as a researcher. More than that, I appreciate how he showed me what it means to work as a theoretical scientist, from turning ideas into solid results to carefully going through every step of a proof, and paying attention to all the details when writing.

I would like to thank Roberto Rubboli and Erkka Haapasalo, with whom I shared an office from the beginning of my PhD. We frequently exchanged ideas and engaged in discussions; with Roberto, I often discussed research problems when facing difficulties, and from Erkka I learned the level of precision required in theoretical work, as he was always willing to address my questions in the early stages of my PhD.

Since research is never done alone, I want to thank all my co-authors: Shrigyan Bramachari, Niklas Galke, Mikhail Terekhov, Erkka Haapasalo, Ruo Cheng, Yanglin Hu, Mile Gu, Marco Fannizza, Andreas Winter and Aaliya Vij. I’m also grateful to Jan Seyfried, Ruo Cheng and Sayantan Sen for their feedback and help in revising this thesis.

\end{acknowledgements}

% Table of Contents
\tableofcontents
\clearpage

% Summary
\cleardoublepage
\phantomsection
\addcontentsline{toc}{chapter}{Summary}
\chapter*{Summary}

In this thesis, we initiate the study of exploration/exploitation trade-offs in online learning of properties of quantum states using the multi-armed bandit framework. Given streaming access to copies of an unknown quantum state, in each round, we are tasked with choosing an observable from a set of actions, aiming to maximize its expectation value on the unknown state (the reward). Information gained from previous rounds is used to gradually improve the choice of action, thereby reducing the regret—the gap between the reward and the maximal reward attainable with the given action set. We provide various information-theoretic lower bounds and optimal strategies with matching upper bounds, showing that for most action sets and unknown states, regret scales as the square root of the number of rounds.

As a first application, we take a different perspective on quantum state tomography, aiming not only to efficiently learn the state but also to use measurements that minimally disturb it. This task is connected to our previous setting when the unknown state is pure and the action set is continuous. Surprisingly, in this setting, we show a polylogarithmic regret scaling using a new, sample-optimal tomography algorithm based on a median-of-means weighted online least squares estimator. To design and analyze this algorithm, we study a more general classical linear bandit problem where the variance of the noise parameter vanishes linearly as actions on the unit sphere are chosen closer to the unknown vector. Our strategy relies on the optimistic principle and introduces a novel approach that controls the eigenvalues of the design matrix at each time step through geometrical arguments, which we believe to be of independent interest. This provides a new approach to study regret bounds in linear stochastic bandits that require to have guarantees on the instantaneous regret at any time.

In the final part of this thesis, we explore other quantum applications where the exploration/exploitation trade-off is relevant and our framework applies. In particular, we investigate work extraction protocols designed to transfer the maximum possible energy to a battery using sequential access to finite copies of an unknown pure qubit state. The challenge lies in designing interactions between the qubit and battery in each round that not only efficiently charge the battery but also extract information about the unknown system to enhance performance in subsequent rounds. We present adaptive protocols for different battery models that incur energy dissipation growing only polylogarithmically on the number of copies, an exponential improvement over naïve protocols that first perform full state tomography. This performance is made possible by leveraging a connection to the exploration-exploitation trade-off and relating the cumulative dissipation to the regret of our framework.

The last application we study is recommender systems for quantum data, where the learner, at each round, is given a context in the form of an observable and must choose from a set of unknown quantum processes performing a measurement. Using this model, we formulate the low-energy quantum state recommendation problem, where the context is a Hamiltonian and the goal is to recommend the state with the lowest energy. For this task, we study two families of contexts: the Ising model and a generalized cluster model.

\pagenumbering{arabic}

% Introduction Chapter
\chapter{Introduction}

%-- put all papers from the author

%Some reinforcement learning? decision making in general?

Humans have always built tools to make work easier. First came machines to replace physical labor, then computers to automate tasks like calculation. Now, the goal is to create machines that can think and learn, machines with intelligence. But what does intelligence mean?

At its core, intelligence is the ability to learn from experience and adjust behavior based on what has been learned. A bird that escapes a predator learns to be more cautious next time. A child playing with blocks figures out how to stack them higher after several failed attempts. These are examples of learning through interaction—by trying, failing, and gradually improving. This is the kind of learning we aim to bring to machines: the ability to improve through experience and adapt to new situations.

To achieve this, researchers have developed machine learning, where algorithms improve by analyzing data. While traditional programming requires rules to be explicitly defined by the programmer, machine learning systems instead discover patterns and strategies by learning from examples or feedback. Rather than being told exactly what to do~\cite{samuel1959some}, these systems adapt their behavior through experience, improving as they process more data. Within this discipline, \textit{reinforcement learning} (RL) stands out by capturing the dynamic, trial-and-error nature of decision-making over time. Unlike supervised learning, which relies on labeled datasets, or unsupervised learning, which looks for patterns without feedback, RL involves an agent that learns by interacting with its environment. At each step, the agent takes actions, receives feedback in the form of rewards, and gradually learns a policy to maximize long-term outcomes. This setup makes RL particularly powerful in scenarios where feedback is sparse or delayed, such as in robotics~\cite{kober2013reinforcement}, game-playing~\cite{mnih2015human, silver2016mastering}, or autonomous systems~\cite{lillicrap2015continuous}, where decisions influence not just the present, but future outcomes as well. At the heart of reinforcement learning lies the challenge of balancing \textit{exploration}—trying out new actions to gather information via feedback—and \textit{exploitation}—using current knowledge to achieve the best outcome.

To study and design reinforcement learning algorithms, researchers often rely on simplified mathematical models. One of the most widely used is the \textit{Markov Decision Process} (MDP). An MDP provides a formal framework for sequential decision-making, where an agent interacts with an environment that is fully observable and evolves in a probabilistic way. At each time step, the agent observes the current state, selects an action, and receives a reward while the environment transitions to a new state. The key assumption is the Markov property: the future depends only on the present state and action, not on the full history. In many real-world situations, however, the agent cannot directly observe the true state of the environment. These problems are modeled by \textit{Partially Observable Markov Decision Processes} (POMDPs), where the agent must infer hidden states from noisy or incomplete observations.  However, algorithms for both MDP and POMDPs often come with high computational complexity and require handling large or even infinite state spaces. Motivated by these limitations, we ask: what is the simplest framework that still captures the core challenge of reinforcement learning—the tradeoff between exploration and exploitation? To address this question, we turn to the \textit{multi-armed bandit} problem, one of the simplest models of decision making with uncertainty.

\section{Multi-armed bandits}

The multi-armed bandit problem was introduced long before the term ``reinforcement learning" was coined. It dates back to 1933, when Thompson~\cite{thompson33} proposed it as a way to model decision-making in medical trials. The goal was to decide which treatment to give to the next patient, given several available options, to increase the chances of success. The problem was later formalized and popularized by Robbins~\cite{robbins52}, who described it as a scenario where each action, called an arm, gives a reward drawn from an unknown and independent probability distribution. This version is known as the stochastic multi-armed bandit model. The term ``bandit" comes from old slot machines, often called ``one-armed bandits" because they would take your money with the pull of a lever. 

Given a set of arms, the bandit problem models a learner who interacts with these arms sequentially, observing a reward after each choice. The goal is to identify the arm with the highest expected reward while also maximizing the total reward collected over time. This problem captures the fundamental dilemma in reinforcement learning of \emph{exploration-exploitation}. On the one hand, the learner needs to explore different arms to discover which ones are best; on the other, the learner wants to exploit the arms that have already given high rewards. Bandit algorithms are online, meaning that their decisions are updated after each round based on past observations.

More formally, we can model a multi-armed bandit problem using a set of \textit{actions} $\mathcal{A}$ and a set of possible \textit{environments} $\Gamma$. The interaction takes place between a learner and an environment $\gamma \in \Gamma$, which is unknown to the learner. The interaction happens over a finite number of rounds $T \in \mathbb{N}$, called the \textit{time horizon}. At each time step $t \in [T]$, the learner selects an action $A_t \in \mathcal{A}$ and then observes a reward $X_t$. Typically, the reward is modeled as a real number, that is, $X_t \in \mathbb{R}$, representing a quantifiable outcome the learner aims to maximize. Thus, the goal of the learner is to maximize the cumulative reward $\sum_{t=1}^T X_t$. 

The challenge lies in how to design strategies that, at any time step $t \in [T]$, use the past information of actions and rewards $(A_s, X_s)_{s=1}^{t-1}$ to decide the next action $A_t$. We can model the learner by a \emph{policy} $\pi = (\pi_t)_{t=1}^T$, where each $\pi_t$ is a probability distribution over $\mathcal{A}$ conditioned on the past history $(A_s, X_s)_{s=1}^{t-1}$. 

The difficulty of the problem depends on how the rewards are generated by the environment, given the selected actions. In general, the environment can be seen as a function that maps the past history and the current action $A_t$ to a reward $X_t$, which can be deterministic or randomized. In this thesis, we focus on \textit{stochastic bandits}, where the learner is interacting with an environment that is modeled as a fixed set of probability distributions $(P_a : a \in \mathcal{A})$ such that when the learner selects an action $A_t \in \mathcal{A}$, the reward $X_t$ is sampled according to $P_{A_t}$. Another important class of problems is \textit{adversarial bandits}, where the environment has more flexibility and can adapt the rewards based on the past history of actions and observations. For a more complete treatment of stochastic bandits, we refer the reader to~\cite[Chapter 4]{lattimore_banditalgorithm_book}, and for adversarial bandits, to~\cite[Chapter 11]{lattimore_banditalgorithm_book}.

How do we quantify the performance of the learner for stochastic bandits? Usually, we want to quantify the performance relative to the best policy. In general, we assume that there exists an optimal policy that, for any environment in $\Gamma$, always selects the action with the highest expected reward. Then we define the \textit{regret} $R_T$ as the cumulative difference between the expected reward obtained by always playing the best action and the expected cumulative reward obtained by the learner:
\begin{align}\label{eq:regret_mab_intro}
    R_T = T \max_{\substack{a \in \mathcal{A} \\ X \sim P_a}} \EX[X] - \EX\left[\sum_{t=1}^T X_t\right].
\end{align}
The expectation is taken over the randomness of the policy of the learner and the unknown environment and the first factor of $T$ arises from assuming that at every time step the maximal expected reward is the same.
The regret measures how much reward the learner loses by not acting optimally from the beginning. A learner is considered successful if the regret grows sublinearly with the time horizon $T$, meaning that the average regret $R_T/T$ vanishes as $T$ grows to infinity. As an example one can consider a two-armed Bernoulli bandit where the action set contains only two elements $\mathcal{A} = \lbrace 1, 2 \rbrace$ and the unknown environment is given by the set $\nu = \lbrace \text{Bern}(p_1), \text{Bern}(p_2) \rbrace  $ with $p_1 > p_2$, where $ \text{Bern}(p)$ is a Bernoulli distribution with mean $p\in[0,1]$. In this scenario, if the learner chooses $A_t\in\mathcal{A}$ it receives reward $X_t = 1$ with probability $p_{A_t}$ and $X_t = 0$ with probability $1-p_{A_t}$. Then the regret then can be written as $R_T = Tp_1 - \EX [ \sum_{t=1}^T X_{A_t} ]$.

There are two main techniques to address the problem of regret minimization in stochastic bandits. One is based on \textit{upper confidence bounds} (UCB), introduced by Lai and Robbins~\cite{lairobbins85} in 1985, and the other is based on posterior sampling, introduced earlier by Thompson~\cite{thompson33} in 1933. The algorithm proposed by Thompson has the honour of being the first bandit algorithm and is usually referred to as \textit{Thompson sampling}. The original algorithm did not have theoretical guarantees, but it presented a simple and intuitive idea: update a prior over the environment using the observed data and act according to the best action suggested by the posterior. Although theoretical guarantees for Thompson sampling were only established recently~\cite{agrawal2012analysis,agrawal2017near}, and are generally slightly weaker in scaling compared to UCB methods, Thompson sampling is often found to perform better empirically and is much more computationally efficient. On the other hand, the UCB-based approach had asymptotic theoretical guarantees already in the original work~\cite{lairobbins85}, and later finite-time guarantees were established in the seminal work by Auer et al.~\cite{auer2002finite}. In this thesis, we will mainly consider the UCB approach.

The family of UCB algorithms is based on one of the main principles to solve reinforcement learning problems: the \textit{optimism in the face of uncertainty principle} (OFU). In short, this principle tells us to choose the best actions within the limits of the observed data. More precisely, in our previous two-armed Bernoulli bandit example, at each round, the learner constructs a confidence interval for the mean reward of each arm based on past observations and selects the arm with the highest plausible reward. This strategy naturally balances the exploration of arms with high uncertainty and the exploitation of arms that have shown good performance with high empirical means.

As we described, multi-armed stochastic bandits have a simple yet powerful theoretical formulation, and nowadays they form a very active research area, with growing interest in both theory and practice~\cite{lattimore_banditalgorithm_book,bandittheory1,bandittheory2,banditapplications1}. Their adaptive nature makes them well-suited to modern applications~\cite{bouneffouf2020survey}, including clinical trials~\cite{durand2018contextual}, dynamic pricing~\cite{dynamical}, advertisement recommendation~\cite{advertisement}, online recommender systems~\cite{contextual_new_article_recommendation,mcinerney2018explore}, or identifying chemical reaction conditions~\cite{wang2024identifying}. As an example, in~\cite{contextual_new_article_recommendation}, a news article recommender system was considered, where the actions correspond to the articles to recommend, and based on user features, a binary reward model indicates whether the user clicks on the recommended article.

\section{Reinforcement learning for quantum technologies}

Over recent years, quantum technologies have emerged with applications ranging from cryptography and sensing to the holy grail of universal quantum computing. What distinguishes quantum technologies from classical ones is the hardware, which is based on the manipulation of quantum systems, allowing the exploitation of uniquely quantum properties such as entanglement and superposition. These properties open the door to new computational and communication paradigms that are fundamentally different from classical systems. However, some of the main weaknesses of these systems are their instability, which motivates the need for quantum error-correcting codes, and the difficulty of learning a classical description of quantum experiments or systems, as this typically requires quantum state tomography, a process that is highly resource-intensive in terms of both sample complexity and time. In particular, real-time control of these systems is an important task. Consider a scenario where a quantum computer experiences small drifts in the qubit frequencies due to changes in the surrounding environment. If we do not detect and correct these drifts quickly, the quantum gates that rely on precise control of the qubit states will accumulate errors. To prevent this, we need real-time strategies that monitor the system’s behavior and update control parameters on the fly, ensuring that the device maintains good performance during computations.

In order to accomplish these tasks, researchers have widely employed classical computing techniques based on reinforcement learning. Some applications of reinforcement learning to quantum technologies include quantum noise modeling~\cite{bordoni2024quantum}, measurement of quantum devices~\cite{nguyen2021deep}, quantum gate control~\cite{an2019deep, niu2019universal, youssry2020modeling}, quantum multiparameter estimation~\cite{cimini2023deep}, and the design of error correction codes~\cite{su2025discovery}, among many others. Although these techniques inherently involve the fundamental dilemma of exploration versus exploitation, they generally apply well-established classical methods to complex models and tasks with large state spaces, often with limited theoretical study in the quantum context. In this thesis, we aim to fill the gap in the literature where the basic exploration-exploitation problem has not been explored within a quantum setting. This leads us to the first question of the thesis:

\begin{itemize} 
\item \textit{\textbf{Question 1.} Can we model a simple exploration-exploitation setting motivated by a quantum mechanical problem?} 
\end{itemize}

\section{Thesis framework}

As explained before, the multi-armed bandit problem is the simplest reinforcement learning setting that captures the exploration-exploitation tradeoff. Thus, in order to study this tradeoff in the quantum context, we propose the \textit{multi-armed quantum bandit} model. In this model, the arms correspond to different observables or measurements, and the corresponding reward is distributed according to Born's rule for these measurements on an unknown quantum state. We are interested in the optimal tradeoff between \emph{exploration} (i.e., learning more about the unknown quantum state) and \emph{exploitation} (i.e., using the acquired information about the unknown state to choose the most rewarding, but not necessarily the most informative, measurement). More precisely, we consider a learner that at each round has access to a copy of an unknown quantum state $\rho$ (the \emph{environment}) and must choose an observable from a given set $\mathcal{A}$ (the \emph{action set}) in order to perform a measurement and receive a \emph{reward}. The {reward} is sampled from the probability distribution associated with the measurement of $\rho$ using the chosen observable.

The figure of merit that we study is the \emph{cumulative expected regret}, defined as the sum over all rounds of the difference between the maximal expected reward over $\mathcal{A}$ and the expected actual reward associated with the chosen observable at each round. The regret for this model can be written as 
\begin{align}\label{eq:regret_maqb_intro} 
R_T := T \max_{O\in\mathcal{A}} \Tr(\rho O) - \EX \left[ \sum_{t=1}^T \Tr(\rho O_{t}) \right], 
\end{align} 
where $O_t\in\mathcal{A}$ is the observable selected by the learner at time step $t\in[T]$. We note that this regret has the same form as that of the multi-armed stochastic bandit~\eqref{eq:regret_mab_intro} (where we have already taken the expectation over the environment). Indeed, the multi-armed quantum bandit problem falls into the class of multi-armed stochastic bandits. The precise formulation of this framework will be discussed in Chapter~\ref{ch:maqb}.

This notion of regret~\eqref{eq:regret_maqb_intro} has a natural interpretation in certain physical settings. Consider, for example, a sparse source of single photons with a fixed but unknown polarization (i.e., the reference frame is unknown). In order to learn the unknown reference frame, we can perform a phase shift and apply a polarization filter, adjusting the phase (the action) until the photons consistently pass through the filter (the reward). The regret would then be proportional to the energy absorbed by the filter during the learning process.

The regret can also be interpreted in terms of the disturbance caused by applying measurements. In the previous photon example, we can consider that the photons are prepared in some unknown state $\psi$, and at each time step $t$, the polarization filter corresponds to a known direction $\psi_t$. According to Born's rule, after the measurement, with probability $p_t = |\langle \psi | \psi_t \rangle|^2$, the state will collapse to $\psi$, and with probability $1 - p_t$, it will collapse to its orthogonal complement $\psi^c$. We can then define the cumulative disturbance over a sequence of $T$ steps as the cumulative infidelity between the unknown state and the post-measurement state $\rho_t = p_t \psi + (1-p_t) \psi^c$: % 
\begin{align} 
\text{Disturbance}(T) := \sum_{t=1}^T \EX\left[ 1 - \langle \psi | \rho_t | \psi \rangle \right]. 
\end{align} 
We note that $\text{Disturbance}(T) \leq 2 R_T$ since $\max_{O\in \mathcal{A}} \Tr(\rho O) = 1$ when the environment $\rho$ is a pure state and the observables correspond to measurement directions (i.e., rank-1 projectors). Thus, minimizing the regret also implies minimizing the disturbance.

As mentioned above, classical bandit algorithms have found important applications in online recommendation systems, and it is a natural question to ask whether our quantum bandit models could similarly be used for recommendation systems for quantum devices or if the regret can have other physical interpretations. This leads us to the following question:

\begin{itemize} 
\item \textit{\textbf{Question 2.} Can we find natural applications of the multi-armed quantum bandit model?} 
\end{itemize}

The applications of the multi-armed quantum bandit model will be discussed in Chapter~\ref{ch:applications}, which will include a more detailed treatment of the notion of disturbance. We will also discuss applications to quantum state-work extraction protocols, where the goal is to transfer the maximum possible amount of energy from an unknown quantum system into a battery system. In this setting, the regret will quantify the dissipation (i.e., the loss of energy) resulting from not applying the optimal work extraction protocol. Lastly, we will explore the application of our model to recommender systems for quantum data. An example of this application could model a set of noisy quantum computers, where at each round, the learner is given a quantum algorithm to run and must recommend the best quantum computer for performing the task.

\section{Related works}

As we mentioned earlier, beyond bandits, there are other theoretical reinforcement learning frameworks that model more complex environments where actions have long-term consequences, such as (hidden) Markov decision processes (MDPs). MDPs have been generalized to the quantum setting (see~\cite{barry11,ying2021}), where the underlying states are quantum states and the evolution and rewards follow quantum processes. Our model can be seen as a restricted version of the framework in~\cite{ying2021}, since in our setting, actions do not induce state transformations. The model in~\cite{ying2021} also incorporates measurements that provide additional partial information about the environment. It is worth emphasizing that these models do not focus on regret minimization, but rather on a different problem known as ``state reachability", which involves deciding under what conditions certain states in the state space can be reached. By contrast, our more specific model enables us to obtain finer results, including tight lower and upper bounds on the cumulative regret.

Another approach to a quantum version of the multi-armed bandit problem involves quantum algorithms. This approach considers the classical problem encoded into a quantum circuit or quantum computer and seeks speedups by using quantum algorithmic techniques. Some works in this direction include quantum algorithms for the best-arm identification problem~\cite{quantumbandits,quantumbandits2} (finding the best arm), regret minimization in stochastic environments~\cite{wan2022quantum,wu2023quantumheavytailedbandits} (both uncorrelated and linearly correlated actions), and adversarial environments~\cite{cho2022quantum}.  A quantum version of the Hedging algorithm, which is related to the adversarial bandit model, has also been studied~\cite{hedging}. A quantum neural network approach for a simple best-arm identification problem was also explored in~\cite{hu2019training}. Additionally, a quantum algorithm for a classical recommender system was proposed in~\cite{kerenidis2017quantum}, claiming an exponential speedup over known classical algorithms; however,~\cite{tang2019quantum} later demonstrated that the speedup relies heavily on the assumptions regarding quantum state preparation, and argued that under analogous classical assumptions, a classical algorithm can achieve a similar speedup.

\section{Technical challenges and contributions}

From the technical side, in Chapter~\ref{ch:lowerbounds} and Chapter~\ref{ch:upperbounds}, we will characterize the regret~\eqref{eq:regret_maqb_intro} for different settings. More specifically, we will study the minimax version of the regret, which is a figure of merit that quantifies the difficulty of a multi-armed quantum bandit problem. The minimax regret is defined as
\begin{align}
    R^*_T := \min_{\pi} \max_{\rho\in\Gamma} R_T ,
\end{align}
where the minimization is taken over the set of policies $\pi$, and the maximization is over a set of possible environments $\Gamma$.

Our goal is to establish tight upper and lower bounds on the minimax regret for different sets of actions and environments, as functions of the time horizon $T$, the number of observables (arms) in the action set, and the dimension $d$ of the underlying Hilbert space. For the lower bounds, we refer to Chapter~\ref{ch:lowerbounds}, where we use information-theoretic techniques to argue about the inherent difficulty of the bandit problem for any policy. For the upper bounds, we refer to Chapter~\ref{ch:upperbounds}, where we design explicit policies for specific problem instances and bound the regret achieved by these policies. Since the minimax regret involves a minimization over all possible policies, establishing an upper bound for a particular policy is sufficient. In general, we will prove the following result:
\begin{itemize}
    \item \textit{\textbf{Result 1.} For general $d$-dimensional environments and $k$ observables (arms), we have }
        \begin{align}
            R^*_T = \widetilde{\Theta} \left( \sqrt{\min\lbrace k , d^2 \rbrace T} \right).
        \end{align}
\end{itemize}
The square root behavior is very typical in bandit algorithms. Indeed, it is worth stressing here that our model falls within the class of (classical) multi-armed stochastic linear bandits. This is not surprising: most problems in quantum tomography, metrology, or learning can be converted to a classical problem with additional structure, since, in the end, we are interested in learning a matrix of complex numbers constituting the density operator. The multi-armed linear stochastic bandit problem was first considered in~\cite{abe1999associative}. In this model, the arms can be viewed as vectors, and the expected reward of each arm is given by the inner product of the vector associated with the arm and an unknown vector that is the same for all arms. As we will see in Section~\ref{sec:connection_linearbandits}, the multi-armed quantum bandit problem is a specific class of linear bandits.

The question of interest then is whether the structure imposed by the quantum problem is helpful to find more efficient algorithms. Indeed, the fact that our model is a subset of a more general class of problems for which classical algorithms are known does not imply that these algorithms yield the smallest possible regret for the subset of problems we consider. In fact, our results show that classical algorithms are optimal in most cases. However, in at least one setting—learning pure states using the set of all rank-1 projectors (see Theorem~\ref{th:regret_PSMAQB})—we prove that bespoke algorithms achieve a strictly lower regret than any previously studied classical algorithm. Moreover, while inspired by the lower bounds on the classical multi-armed stochastic bandit problem, our bounds require novel constructions that are specific to the quantum state space. On the one hand, due to correlations in the reward distribution associated to each arm (observable), the standard multi-armed stochastic bandit lower bounds proofs do not apply to our case. On the other hand, lower bounds for linear stochastic bandits have been studied only for specific action sets like the hypercube or unit sphere (see \cite{hypercube} and \cite{lin2}) and there is no combination of action sets and environments in the classical proofs that can be mapped to our particular class of problems. For that reason, known classical regret lower bounds do not apply to our case.

Returning to the previously discussed setting of pure states, we note that the regret can be expressed as 
\begin{align}\label{eq:intro_purestateregret} 
R_T = \sum_{t=1}^T \left( 1 - |\langle \psi | \psi_t \rangle|^2 \right), \end{align} 
where $\psi$ is the unknown pure state and $\psi_t$ is the selected action or measurement direction at time $t$. Minimizing the regret in this setting corresponds to finding measurement directions $\psi_t$ that are close to $\psi$ in terms of infidelity distance. Thus, the problem can be viewed as a task in online quantum state tomography, with the regret serving as a figure of merit. Leveraging this observation, we can interpret the actions taken by a strategy up to time $t \in [T]$ as performing quantum state tomography using $t$ copies of the state. From known results, the optimal achievable fidelity scales as $1 - 1/t$~\cite{hayashi2005reexamination,haah2016sample,ODonnell_2016,Yuen_2023,chen_adaptive}. Adapting this tomography result to our bandit setting, we can establish the following lower bound.
\begin{itemize}
    \item \textit{\textbf{Result 2.} For pure states environments and rank-1 projectors observables, we have}
    \begin{align}
        R^*_T = \Omega ( d\log T ).
    \end{align}
\end{itemize}
This is the only setting that we can not prove a $\Omega (\sqrt{T})$ lower bound. Surprisingly, if the unknown environment is a mixed state, we can prove the square root lower bound. This motivates the following question whether, in fact, the square root barrier can be surpassed.

\begin{itemize}
\item \textit{ \textbf{Question 3.} Can we identify a multi-armed quantum bandit setting where the square root barrier of regret is broken?}
\end{itemize}

Achieving a better scaling for an instance of the multi-armed quantum bandit problem would provide a physically motivated linear bandit setting where the square root barrier can be surpassed. Finding such a setting would provide the first example of a nontrivial linear stochastic bandit problem with continuous action sets that breaks the square root barrier in the regret. Achieving this requires new algorithms and techniques that take advantage of the additional structure provided by the multi-armed quantum bandit problem as a linear stochastic bandit.

Since our problem is related to quantum state tomography, we also study the following question at the intersection of linear bandits and quantum state tomography.

\begin{itemize}
\item \textit{ \textbf{Question 4.} Can we perform single-copy sample-optimal state tomography in infidelity while achieving sublinear regret in the number of copies $T$? How much adaptiveness is needed for this task? }
\end{itemize}

Adaptiveness plays a crucial role in algorithms that aim to minimize the regret of the multi-armed quantum bandit setting while performing sample-optimal state tomography. A natural approach is to modify existing sample-optimal algorithms from the incoherent setting (single copy measurements), such as those in~\cite{haah2016sample,kueng2017low,guctua2020fast}, for the bandit problem. However, these algorithms rely on fixed bases or randomized measurements, which inevitably lead to a linear regret scaling, i.e, $\text{Regret}(T) = O(T)$ (ignoring dimensional dependencies). A natural next step is to consider a simple strategy with one round of adaptiveness, where we use a fraction $\alpha \in [0,1] $ of the copies for state tomography to produce a good estimate \(\hat{\psi}\) of the unknown $\psi$, and use the remaining copies to measure along the estimated direction. Using sample-optimal state tomography algorithms, this leads to a regret scaling
\begin{align}
   R_T = O\left(\alpha T + (T - \alpha T)\frac{1}{\alpha T} \right),
\end{align}
which, optimized over $ \alpha$, gives $\text{Regret}(T) = O(\sqrt{T})$, but results in a sub-optimal error scaling $1- |\langle \psi | \psi_T \rangle|^2= O(1/\sqrt{T})$.

In Chapter~\ref{ch:upperbounds}, we provide affirmative answers at the same time to Questions 3 and 4 through the following result.

\begin{itemize}
\item \textit{\textbf{Result 3.}    For any unknown pure qubit state $\psi$ environment, our algorithm achieves
    \begin{align}
    R_T = O \big( \log^2 (T) \big).
    \end{align}
    Moreover, at each time step $t\in [T]$, our algorithm outputs an online estimate $\hat{\psi}_t$  with infidelity scaling as
    \begin{align}
    \mathbb{E} \left[ 1 -  |\langle \psi | \hat{\psi}_t \rangle|^2\right] = \widetilde{O}\left( \frac{1}{t} \right).
    \end{align}
    Both statements also holds with high probability. }
\end{itemize}

To prove this result, we will provide an almost fully adaptive quantum state tomography algorithm that uses $O(T/\log(T))$ rounds of adaptiveness. The exact algorithm and Theorem can be found in Section~\ref{sec:linucb_vvn}. We say that our algorithm is ``online'' because it is able to output at each time step $t\in [T]$ an estimator with the almost optimal infidelity scaling $O( \frac{1}{t})$ up to logarithmic factors. Now we sketch the main idea of how our algorithm updates the actions (measurements).

\begin{enumerate}
    \item \textbf{Estimation.} At each time step $t\in[T]$, we use the past information of measurements on the direction of $\ket{\psi_{A_1}},...,\ket{\psi_{A_{t-1}}}$ and associated outcomes $X_1,...,X_{t-1}\in\lbrace 0 , 1 \rbrace^{\otimes t-1}$ to build a high probability confidence region $\mathcal{C}_t$ for the unknown environment $|\psi \rangle $.
    \item \textbf{Exploration--exploitation.} A batch of measurements is performed, given by the directions of maximum uncertainty of $\mathcal{C}_t$ such that they give enough information to construct $\mathcal{C}_{t+1}$ (exploration) and also minimise the regret~\eqref{eq:intro_purestateregret} (exploitation).
\end{enumerate}

For the estimation part, we work with the Bloch sphere representation of the unknown state $\Pi = |\psi \rangle \! \langle \psi | = \frac{I+\theta\cdot\sigma}{2}$ where $\theta\in\mathbb{R}^3$ is normalzied $\|\theta\|_2 = 1$ and for $\sigma$ we can take the standard Pauli Basis, i.e, $\sigma = (\sigma_x,\sigma_y,\sigma_z )$. For each measurement direction $\Pi_{A_t} = |\psi_{A_t}\rangle \! \langle \psi_{A_t} |$, our algorithm performs $k$ independent measurements using the same direction, and it builds the following $k$ online weighted least squares estimators of $\theta$,
\begin{align}\label{eq:abstract_weighted_lse}
    \widetilde{\theta}_{t,i} = V_t^{-1} \sum_{s=1}^t \frac{1}{\hat{\sigma}^2_s (A_s)} X_{s,i} A_s \quad \text{for }i\in[k],
\end{align}
where $X_{s,i}\in\lbrace 0,1\rbrace$ is the outcome of the measurement (up to some renormalization) using the projector $\Pi_{A_s}$ with Bloch vector $A_s\in\mathbb{R}^3$, $V_t =  \mathbb{I} + \sum_{s=1}^t \frac{1}{\hat{\sigma}^2_s (A_s)} A_s A_s^{\mathsf{T}}$ is the \textit{design matrix} and $\hat{\sigma}^2_s (A_s)$ is a variance estimator of the real variance associated to the outcome $X_s$. The key point where we take advantage of the structure of the quantum problem is that the variance of the outcome $X_a$ associated with the action $\Pi_A$ can be bounded as $\VX [X_a] \leq 1 - |\langle \psi | \psi_A \rangle|^2$.

The idea is that, through a careful choice of actions, we can make the terms \( 1/\hat{\sigma}^2_s (A_s) \) arbitrarily large and thereby ``boost" the confidence on the directions \( A_s \) in the estimators~\eqref{eq:abstract_weighted_lse} that are close to \(\theta\). However, this comes at a cost: in order to obtain good concentration bounds for our estimator, we must address the challenge of dealing with unbounded random variables with only finite variance. We tackle this issue by using recent advances on median-of-means (MoM) techniques for online least squares estimators, introduced in~\cite{bandits_heavytail,heavy_tail_linear_noptimal,heavy_tail_linear_optimal}. The construction is inspired by the classical \emph{median-of-means} method for real-valued random variables with unbounded support and bounded variance~\cite[Chapter 3]{lerasle2019lecture}, but requires nontrivial adaptations to the online linear least squares setting. Analogously to the real-valued case, we use \(k\) independent weighted estimators of the form~\eqref{eq:abstract_weighted_lse} to construct the MoM estimator \(\widetilde{\theta}^{\text{\tiny wMoM}}_{t}\), enabling the construction of confidence regions with concentration bounds scaling as \(1-\exp(-k)\). We present the exact construction in Section~\ref{sec:MoMLSE}.

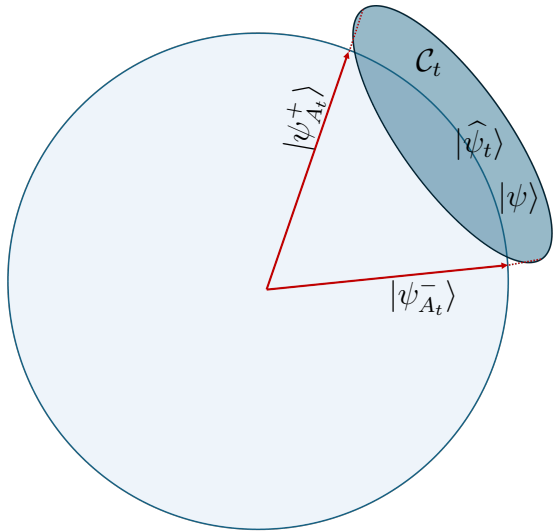
\begin{figure}
    \centering
    \begin{overpic}[percent,width=0.45\textwidth]{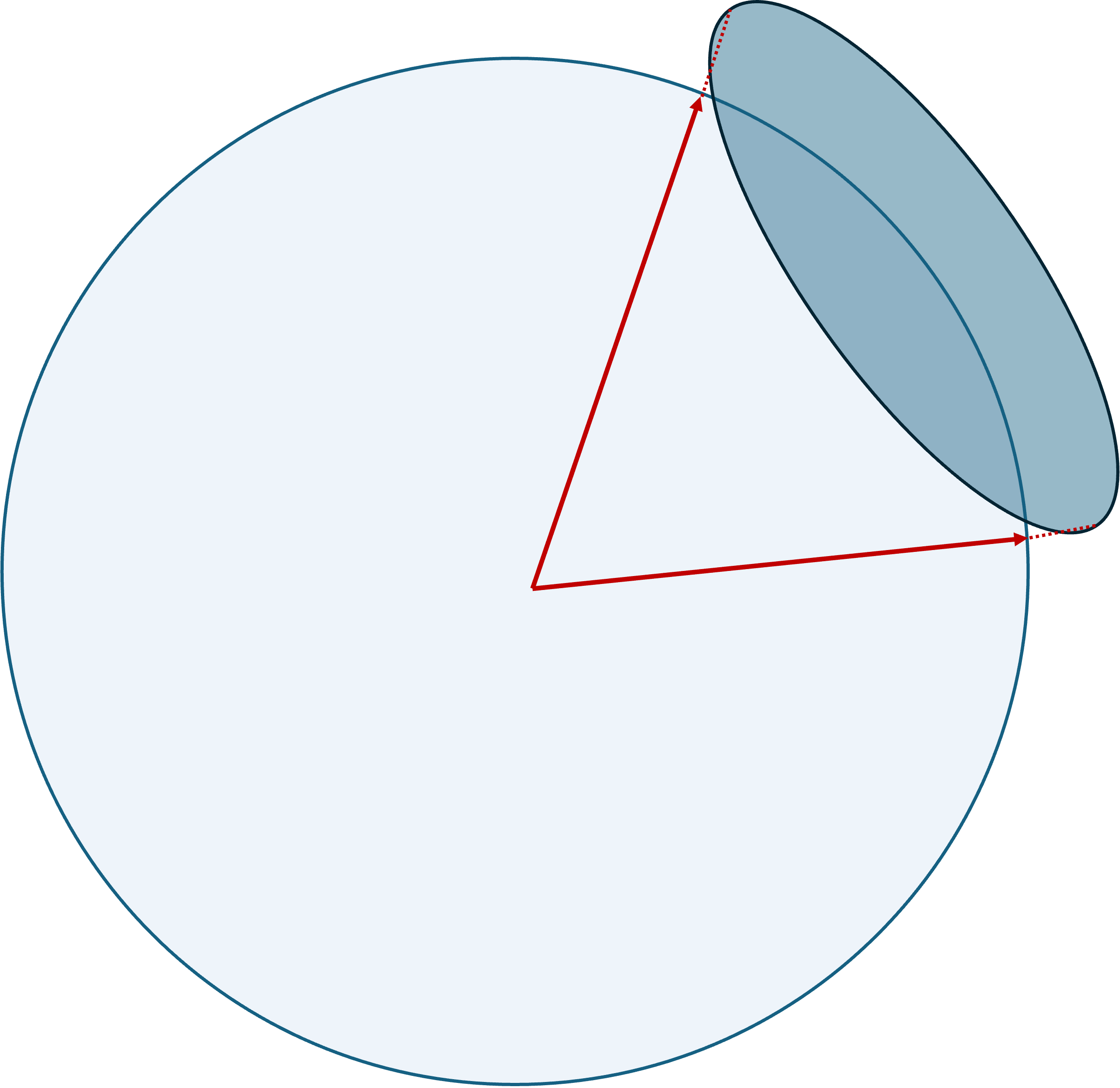}
    \put(50,70){\rotatebox{80}{$|\psi^+_{A_t}\rangle$}}
    \put(70,42){\rotatebox{0}{$|\psi^-_{A_t}\rangle$}}
    \put(75,85){\rotatebox{0}{$\mathcal{C}_t$}}
    \put(82,70){{$|\widehat{\psi}_t\rangle$}}
    \put(90,60){{$| \psi \rangle$}}
    
    \end{overpic}
    \caption{The algorithm at each time step outputs an estimator $|\widehat{\psi}_t\rangle$ and builds a high-probability confidence region $\mathcal{C}_t$ (shaded region) around the unknown state $ | \psi \rangle $ on the Bloch sphere representation. Then uses the optimistic principle to output measurement directions $|\psi^\pm_{A_t} \rangle$ that are close to the unknown state $\ket{\psi} $ projecting into the Bloch sphere, the extreme points of the largest principal axis of $\mathcal{C}_t$. This particular choice allows optimal learning of $\ket{\psi}$ (exploration) and simultaneously minimizes the regret (exploitation).}
    \label{fig:psmaqb_exploration_exploitation}
\end{figure}

For the exploration-exploitation part, we need to develop new ideas in order to update the actions (see Figure~\ref{fig:psmaqb_exploration_exploitation}). We give the precise action choice through Chapter~\ref{ch:upperbounds}, and here we sketch the main points. We take inspiration from the \emph{optimistic principle} for bandit algorithms that tells us to choose the most rewarding actions with the available information. In order to use this idea, we use the confidence region that we built in the estimation part, and we select measurements that align with the (unknown) direction of $\ket{\psi}$. See Figure~\ref{fig:psmaqb_exploration_exploitation}.
Our algorithm also achieves the relation $1 - |\langle \psi | \psi_{A_t} \rangle |^2 = O\left( 1/\lambda_{\min} (V_t )\right)$, where the minimum eigenvalue $\lambda_{\min}(V_t )$ quantifies the direction of maximum uncertainty (exploration) of our estimator. The maximum eigenvalue $\lambda_{\max}(V_t)$ quantifies the amount of exploitation (direction that has been selected the most). We can relate these two concepts through Theorem~\ref{th:main_eigenvalues} that we formally state and prove in Section~\ref{sec:designmatrix}, which states that for our particular measurement choice we have $\lambda_{\min}(V_t) = \Omega (\sqrt{\lambda_{\max}(V_t)} )$. Using this relation and a careful analysis, we can show that $\lambda_{\max}(V_t) = \Omega (t^2)$ which gives $\lambda_{\min}(V_t) = \Omega (t)$ and the scaling $1 -  |\langle \psi | \psi_{a_t} \rangle |^2 = O(1/t)$. We emphasize that the key point that allows us to achieve the rate $\lambda_{\min}(V_t) = \Omega (t)$ is the fact that the variance estimators $\hat{\sigma}^2_s$ can get as close as possible to zero since the variance of the rewards goes to zero if we select measurements close to $|\psi \rangle$.

\section{Outline}

The rest of the thesis is organized as follows.

\begin{itemize}
    \item In Chapter~\ref{ch:maqb}, we introduce the main framework of this thesis: the multi-armed quantum bandit problem and the figure of merit we will study, the minimax regret. We discuss both discrete and continuous versions of the problem and explain how this framework fits into the class of linear stochastic bandits.

    \item In Chapter~\ref{ch:lowerbounds}, we introduce several information-theoretic tools, including a modified version of the Bretagnolle–Huber inequality. We use these tools to establish regret lower bounds for different combinations of environments and action sets.

    \item In Chapter~\ref{ch:upperbounds}, we begin by reviewing the linear version of the UCB algorithm and its discrete variants. We then present a new alternative to UCB based on tight control of the eigenvalues of the design matrix. We introduce and study linear bandits with vanishing subgaussian parameters and variance, and show how our new algorithm achieves polylogarithmic regret. This algorithm is also applicable to the multi-armed quantum bandit setting with pure state environments and a continuous action set.

    \item In Chapter~\ref{ch:applications}, we explore applications of the multi-armed quantum bandit model, including learning with minimal quantum state disturbance, quantum state-agnostic work extraction protocols, and recommender systems for quantum data.

    \item In Chapter~\ref{ch:open_problems}, we summarize the results of the previous chapters and present several open problems that arise from our results, highlighting potential directions for future research.
\end{itemize}

\chapter{The multi-armed quantum bandit framework}\label{ch:maqb}

This chapter is based on the author's work~\cite{Lumbreras2022multiarmedquantum} and provides an introduction to the multi-armed quantum bandit framework. The notation is partly inspired by~\cite[Chapter 4]{lattimore_banditalgorithm_book}, which offers a rigorous treatment of stochastic bandits. We also present the linear stochastic bandit model and explore its connection to the multi-armed quantum bandit framework, emphasizing both similarities and differences. The linear stochastic bandit model plays a significant role in this thesis, particularly in Chapter~\ref{ch:upperbounds}, where we develop and analyze algorithms for linear bandits with applications to the quantum setting.

\section{Notation and conventions}

We present some standard notations and conventions that we will use throughout the thesis.

\textbf{Real matrices and vectors.} For real vectors $x,y\in\mathbb{R}^d$, we denote their inner product as $\langle x, y \rangle = x_1y_1+...+x_dy_d$. The trace of a square matrix $A\in\mathbb{R}^{d\times d}$ is denoted as $\Tr (A) $ and the transpose of a matrix $B\in\mathbb{R}^{m\times n}$ is denoted $B^{\mathsf{T}}$. The set of positive
semi-definite matrices as $\mathsf{P}^d_+ = \lbrace X\in\mathbb{R}^{d\times d}: X\geq 0 \rbrace$. Given a real vector $x\in\mathbb{R}^d$, we denote the 2-norm as $\| x \|_2 = \sqrt{\langle x , x \rangle}$ and for a real semi-positive definite matrix $A\in\mathbb{R}^{d\times d}$, $A \geq 0$, the weighted norm with $A$ as $\|x \|_A = \sqrt{\langle x , A x \rangle }$. For a real symmetric matrix $A\in\mathbb{R}^{d\times d}$, we denote by $\lambda_{\max} (A)$, $\lambda_{\min} (A)$ its maximum and minimum eigenvalues, respectively. We use the ordering $\lambda_{\min} (A) \leq \lambda_2 (A) \leq .... \leq \lambda_{d-1} (A)\leq \lambda_{\max}(A)$ for the $i$-th $\lambda_i (A)$ eigenvalue in non-decreasing order. The $d$-dimensional identity matrix is denoted as $\mathbb{I}_{d\times d}$ and we will also denote it as $\mathbb{I}$ if $d$ when is clear from the context.

 \textbf{Sets.} Let $[t]$ denote the set $[t] = \lbrace 1,2,...,t \rbrace$ for $t\in\mathbb{N}$. The $d$-dimensional unit sphere is denoted as $\mathbb{S}^{d-1} = \lbrace x\in\mathbb{R}^d : \| x \|_2 = 1 \rbrace$ and the ball of radius $r>0$ and center $c\in\mathbb{R}^d$ is denoted as $\mathbb{B}^d_r (c ) = \lbrace x\in\mathbb{R}^d: \| x - c \|_2 \leq r \rbrace$.

\textbf{Probability}. For a random variable $X$ (discrete or continuous, we denote $\EX [X]$ and $\mathrm{Var} [X] $ its expectation value and variance, respectively. Given an event $A$, we denote by $\mathbbm{1}(A)$ the indicator function of $A$, that is,  $\mathbbm{1}(A) =1$ if $A$ is true and $0$ otherwise.

 For two probability measures $P, Q$ defined in the same probability space $(\Omega, \Sigma)$, we may introduce a (probability) measure $\mu$ that dominates $P$ and $Q$ and define the Radon--Nikodym derivatives $p = \frac{dP}{d\mu}$ and $q=\frac{dQ}{d\mu}$. The choice of $\mu$ is arbitrary for the following definitions.
We define the \emph{Kullback--Leibler} divergence as
\begin{align}
	D( P \| Q ) := \int p(\log p - \log q)\,d\mu ,
\end{align} 
where we use the convention that $0\log 0 = 0$ and $D(P\|Q) = \infty$ whenever $Q$ does not dominate $P$. We also define the (squared) \emph{Bhattacharyya coefficient} as 
\begin{align}
	F(P, Q) := \left( \int \sqrt{p q} \, d\mu \right)^2 \,.
\end{align}

\textbf{Quantum information.} Let $\mathcal{S}_d = \lbrace \rho\in\mathbb{C}^{d\times d}: \Tr(\rho) = 1 , \rho \geq 0\rbrace$ denote the set of \textit{quantum states} in a $d$-dimensional Hilbert space $\hil = \mathbb{C}^d$, and let $\mathcal{S}^*_d = \lbrace \rho\in\mathcal{S}_d : \rho^2 = \rho \rbrace$ be the subset of \textit{pure states}, which are rank-1 projectors. To compare quantum states, several distance and divergence measures are commonly used. The \textit{quantum fidelity} between two states $\rho, \sigma \in \mathcal{S}_d$ is defined as
  \begin{align}
  F(\rho,\sigma) := \left( \Tr\left( \sqrt{ \sqrt{\sigma} \rho \sqrt{\sigma} } \right) \right)^2,
  \end{align}
  and the \textit{infidelity} is given by $1 - F(\rho, \sigma)$.  The \textit{trace distance} between $\rho$ and $\sigma$ is
  \begin{align}
  \|\rho - \sigma\|_1 := \Tr\left( |\rho - \sigma| \right),
  \end{align}
  where $|A| := \sqrt{A^\dagger A}$. It is a metric on the space of quantum states and has an operational meaning in terms of the optimal probability of distinguishing the states. The \textit{quantum relative entropy} is a divergence defined as
  \begin{align}
  D(\rho \| \sigma) := 
  \begin{cases}
    \Tr(\rho \log \rho) - \Tr(\rho \log \sigma), & \text{if } \text{supp}(\rho) \subseteq \text{supp}(\sigma), \\
    \infty, & \text{otherwise}.
  \end{cases}
  \end{align}
  While not a metric, it plays a central role in quantum information theory due to its connections with hypothesis testing and information measures.

We use \textit{Dirac’s bra-ket notation} throughout: vectors in $\hil$ are denoted by kets $|\psi\rangle$, and their duals by bras $\langle\psi|$. The inner product between two vectors $|\psi\rangle, |\phi\rangle \in \hil$ is written as $\langle\psi|\phi\rangle$, and the outer product $|\psi\rangle\langle\psi|$ corresponds to the rank-1 projector onto $|\psi\rangle$. In particular, a pure state $\Pi \in \mathcal{S}^*_d$ can be written as $\Pi = |\psi\rangle \langle\psi|$ for some normalized vector $|\psi\rangle \in \hil$.

 For a Hilbert space $\hil$, the set of linear operators on it will be denoted by $\End(\hil)$. The joint state of a system consisting of $n$ copies of a pure state $\Pi\in \mathcal{S}_d^*$ is given by the $n$-th tensor power $\Pi^{\otimes n}\in \End(\hil^{\otimes n})$.  Then, the span of all $n$-copy states of the form $|\psi \>^{\otimes n}$ is called the symmetric subspace of $\hil^{\otimes n}$, denoted by $\hil^{\otimes n}_+$. Its dimension is $D_n=\binom{n+d-1}{d}$. The symmetrization operator $\Pi^+_n\in\End(\hil^{\otimes n})$ is the projector onto $\hil^{\otimes n}_+$.

Measurements in quantum theory are modeled by positive operator-valued measures (POVMs). A POVM is a collection $\lbrace M_a \rbrace_{a \in \mathcal{A}}$ of positive semidefinite operators acting on $\hil$ such that $\sum_{a \in \mathcal{A}} M_a = \mathbb{I}$. The probability of obtaining outcome $a$ when measuring a quantum state $\rho$ is given by the \emph{Born rule}: $\Pr[a] = \Tr(\rho M_a)$. Moreover, \textit{observables} are Hermitian operators acting on $\mathbb{C}^d$, collected in the set $\mathcal{O}_d = \lbrace O\in\mathbb{C}^{d\times d}: O^\dagger = O\rbrace.$ An observable $O$ is called $\textit{traceless}$ if $\Tr (O) =0$ and called \textit{sub-normalized} if $\| O\|\leq 1$, where $\| \cdot \|$ denotes the operator norm.  Measuring a quantum state $\rho$ with an observable $O$ means performing a projective measurement in the orthonormal basis of eigenvectors of $O$: if $O = \sum_{i} \lambda_i \Pi_i$ is the spectral decomposition of $O$, where $\lambda_i$ are the eigenvalues and $\Pi_i$ are the orthogonal projectors onto the corresponding eigenspaces, then the outcome $\lambda_i$ is obtained with probability $\Tr(\rho \Pi_i)$.

In order to parametrize quantum states and observables, we will use the Pauli matrices defined as follows 
\begin{align}
\sigma_x = \begin{pmatrix} 0 & 1 \\ 1 & 0 \end{pmatrix}, \quad
\sigma_y = \begin{pmatrix} 0 & -i \\ i & 0 \end{pmatrix}, \quad
\sigma_z = \begin{pmatrix} 1 & 0 \\ 0 & -1 \end{pmatrix}.
\end{align}

% We will use the subscript $\theta$ in the quantum state $\rho_\theta$ in order to denote the vector of the parametrization~\eqref{eq:parametrization}. In particular the normalization is taken such that $\|\theta\|_2^2\leq 1$ with equality if  $\rho_\theta$ is pure. Note that the parametrization enforces $\rho^\dagger_\theta = \rho_\theta$ and $\Tr(\rho_\theta) = 1$. Also there are some extra conditions on the vector $\theta$ regarding the positivity of the density matrix $\rho_\theta$ but we will not use them. 

\section{Multi-armed quantum bandits}

\subsection{Discrete bandits}

In the classical setting, a multi-armed bandit is defined by a set of actions (arms) and an unknown environment. At each step, a learner selects an action, and a corresponding reward is drawn from a probability distribution that depends on the chosen action and the underlying environment. The quantum version we consider is a specialized subset of the multi-armed bandit framework, where both the set of actions and the environment exhibit quantum-specific structure. Specifically, the actions correspond to a set of quantum observables, and the environment is characterized by a quantum state.  To formalize this, we define the notion of a multi-armed quantum bandit as follows:
\begin{definition}[Multi-armed quantum bandit]
	Let $d \in \mathbb{N}$. 
	A $d$-dimensional \emph{discrete multi-armed quantum bandit} (MAQB) is given by a finite set $\mathcal{A} \subseteq \mathcal{O}_d$ of observables that we call the \emph{action set}. The bandit is in an \emph{environment}, a quantum state $\rho \in \Gamma$, that is unknown but taken from a set of \emph{potential environments} $\Gamma \subseteq \mathcal{S}_d$. The bandit problem is fully characterized by the tuple $(\mathcal{A}, \Gamma)$.
\end{definition}
The interaction between the learner and the multi-armed quantum bandit is governed by the principles of quantum mechanics, particularly in how the reward is sampled. The stochastic reward follows the probabilistic nature described by Born's rule. Informally, the reward corresponds to the outcome of a measurement performed on the quantum state, using the observable associated with the chosen action from the action set.

Let us fix a  multi-armed quantum bandit ($\mathcal{A},\Gamma $) with environment $\rho\in \Gamma$ and action set $\mathcal{A} = \lbrace O_1, .... ,O_k \rbrace$ of cardinality $k = |\mathcal{A |}$. For each observable $O_a$ we introduce its spectral decomposition,
\begin{align}
    O_a = \sum_{i=1}^{d_a} \lambda_{a,i}\Pi_{a,i} ,
\end{align}
where $\lambda_{a,i} \in \mathbb{R}$ denote the $d_a \leq d$ distinct eigenvalues of $O_a$ and $\Pi_{a,i}$ are the orthogonal projectors on the respective eigenspaces. The learning process between the learner and the multi-armed quantum bandit occurs sequentially over $T$ rounds, where $T\in\mathbb{N}$ represents a finite time horizon. At each time step $t\in [T ]$ the learner interacts with the bandit as follows:
\begin{enumerate}
    \item The learner receives a copy of the unknown environment $\rho \in \Gamma$.
    \item The learner selects an action $A_t\in [k ]$ and performs a measurement on $\rho$ with the observable $O_{A_t}\in\mathcal{A}$.
    \item The learner samples a reward $X_t\in\mathbb{R}$ that is the outcome of the measurement and is distributed according to Born's rule as
    \begin{align}\label{eq:reward_distribution_discrete}
   P_{\rho} (x | a ) :=   \mathrm{Pr}\left( X_t = x | A_t = a  \right) = \begin{cases}
        \Tr \left(\rho \Pi_{a,i} \right) \quad \text{if } x = \lambda_{a,i} \\
        0 \quad \text{else}.
    \end{cases}
\end{align}
\end{enumerate}
The choice of action by the learner is governed by a policy, which is a conditional probability distribution over the set of actions based on the history of past actions and rewards observed by the learner. To formalize this, we define a policy as follows:
\begin{definition}
	\label{def:policy}
	A \emph{policy} (or algorithm) for a multi-armed quantum bandit is a set of (conditional) probability distributions $\pi = \{ \pi_t \}_{t \in \mathbb{N}}$ on the action index set $[k]$ of the form
	\begin{align}
		\pi_t ( \cdot \vert a_1 , x_1 , ... , a_{t-1}, x_{t-1} ) ,
	\end{align}
	defined for all valid combinations of actions and rewards $(a_1, x_1), \ldots, (a_{t-1}, x_{t-1})$ up to time $t-1$.
\end{definition}
For the conditional probability of reward $X_t$ given $A_1 , X_1 , \cdots , A_{t-1}, X_{t-1} , A_t$ we will use simply the notation $P_{\rho}(\cdot | A_t )$ as defined in~\eqref{eq:reward_distribution_discrete}. Then, if we run the policy $\pi$ on the state $\rho$ over $ T \in \mathbb{N}$ rounds, we can define a joint probability distribution over the set of actions and rewards as
\begin{align}\label{eq:probdens_discrete}
	P_{\rho ,\pi }(a_1,x_1,...,a_T,x_T) := \prod_{t=1}^n \pi_t ( a_t \vert a_1 , x_1 , ... , a_{t-1}, x_{t-1} ) P_{\rho}(x_t|a_t).
\end{align}
This distribution captures the stochastic nature of the learning process resulting from the interaction between the learner and the environment.

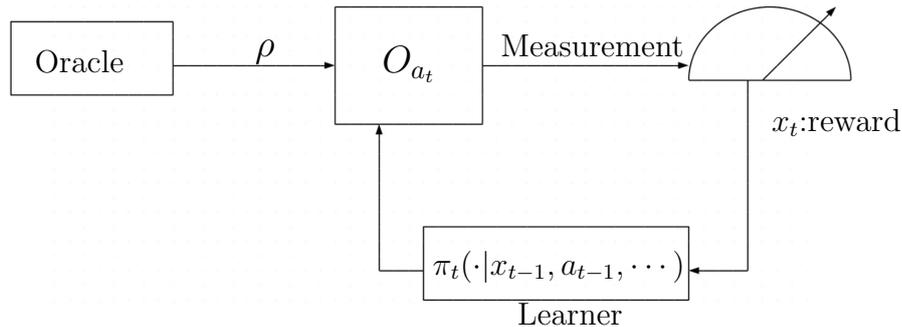
\begin{figure}
\centering
\begin{overpic}[percent,width=0.7\textwidth]{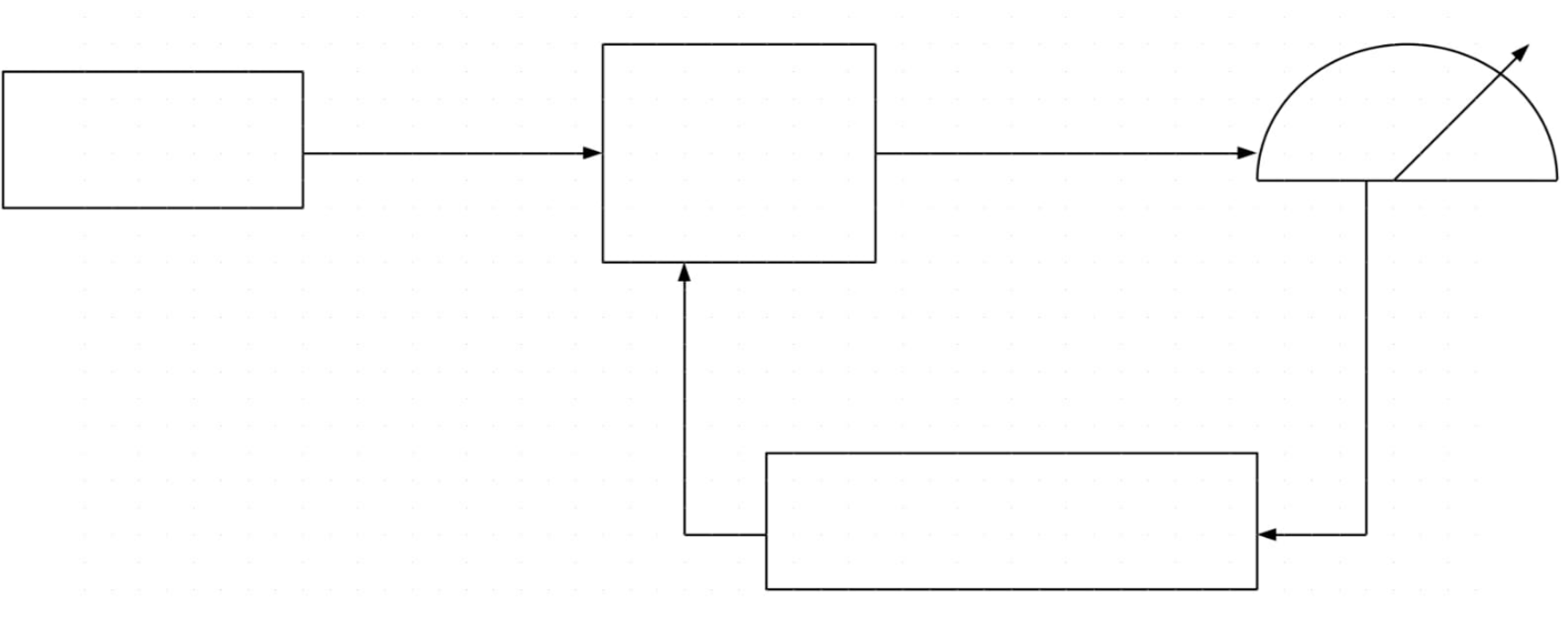}
\put(3,30){\color{black} Oracle}
\put(29,32){\large$\rho$}
\put(44,30){\color{black}\large$O_{a_t}$}
\put(50,6){$\pi_t (\cdot \vert x_{t-1},a_{t-1},\cdots )$}
\put(90,23){$x_t$:reward}
\put(60,0){\color{black} Learner}
\put(58,32){\color{black} Measurement}
\end{overpic}
\caption{Scheme for the multi-armed quantum bandit setting. At each time step $t$, the learner receives an unknown quantum state $\rho$, selects an action $a_t$ according to a policy $\pi_t$ based on past rewards and actions, and chooses the corresponding observable $O_{a_t}$ from the action set $\mathcal{A}$. A measurement of $\rho$ with $O_{a_t}$ yields an outcome $x_t$, which determines the reward.
}
\label{fig:scheme}
\end{figure}

\subsection{General bandits}\label{sec:general_bandits}

A natural extension of the previous model is to allow the action set to include a continuous set of observables. A particularly important subset, which will play a central role in this thesis, consists of all rank-1 projectors. While we have thus far focused on describing the model rather than specifying the learner's objective, the primary goal—similar to the classical multi-armed bandit problem—is to maximize the cumulative reward. As we will explore in later sections, the transition from discrete to continuous action sets introduces significant differences in how a learner can maximize the cummulative reward (or minimize the regret that we will introduce later). 

\begin{definition}\label{def:general_maqb}
	Let $d \in \mathbb{N}$.
	A general $d$-dimensional \emph{multi-armed quantum bandit} is given by a measurable space $(\mathcal{A} ,\Sigma)$, where $\Sigma$ is a $\sigma$-algebra of subsets of $\mathcal{A}$ and $\mathcal{A} \subseteq \mathcal{O}_d$ is a subset of observables. The bandit is in an \emph{environment}, a quantum state $\rho \in \Gamma$, that is unknown but taken from a set of \emph{potential environments} $\Gamma \subseteq \mathcal{S}_d$. The bandit problem is fully characterized by the tuple $(\mathcal{A}, \Gamma)$.
\end{definition}

The learning process remains the same as described in the previous section, with the reward distribution $P_{\rho}( x | a )$~\eqref{eq:reward_distribution_discrete} still determined by Born's rule and remaining discrete. The learner's strategy is defined by a policy, which in this case is represented by a set of conditional probability measures:
\begin{align}
    \pi_t(\cdot|a_1,x_1,\ldots,a_{t-1},x_{t-1}) : \Sigma\to[0,1] .
\end{align}
Here, the policy determines the probability distribution over the continuous action set at each time step based on the learner's past interactions. To simplify our presentation, we in the following also assume that the spectra of the observables satisfy $\Lambda_a \subset \mathcal{X}$ for some finite set $\mathcal{X}$, i.e., we only allow for a discrete set of possible rewards independently of the choice of action.  We also need to assume that the function $a \mapsto P_{\rho}(x | a) $ is $\Sigma$-measurable for all environment states $\rho$ and $x \in \mathcal{X}$. This is, for example, trivially satisfied for the case of rank-1 projectors, where $\mathcal{X} = \{0,1\}$ is trivial. In fact, this action set will be the main one that we will study in the next sections. Moreover, a specific setting of general bandits that will be important during this thesis is the setting of action set rank-1 projectors and the environment is an unknown pure state. Given its importance, we give a definition of this particular problem below.
\begin{definition}\label{def:PSMAQB}
 Let $d\in\mathbb{N}$. A  $d$-dimensional \textit{pure state multi-armed quantum bandit} (PSMAQB) is given
by a measurable space $(\mathcal{A}, \Sigma)$, where $\mathcal{A} = \mathcal{S}_d^*$ is the \textit{action set} and $\Sigma$ is a $\sigma$-algebra of subsets of $\mathcal{A}$. The bandit is in an environment, a quantum state $\Pi_\theta\in \mathcal{S}_d^*$, that is unknown.
\end{definition}
Thus, in a PSMAQB the learner will interact at each time step $t\in[T]$ with a pure state environment $|\psi\rangle \! \langle \psi |\in\mathcal{S}^*_d$ selecting an action $\Pi_{A_t}\in\mathcal{S}^*_d$ and sampling a Bernoulli reward distributed accordingly to
\begin{align}\label{eq:prob_quantum_reward_PSMAQB}
    \mathrm{Pr}_{\psi} \left(X_t | \Pi_{A_t} \right) = 
    \begin{cases}
    \langle \psi | \Pi_{A_t} | \psi \rangle \quad \text{if } x=1,\\
    1-\langle \psi | \Pi_{A_t} | \psi \rangle \quad \text{if }  x=0, \\
    0 \quad \text{else}.
    \end{cases}
\end{align}
Recall that we are measuring the unknown $\psi$ with the 2-outcome POVM $\lbrace \Pi_{A_t}, \mathbb{I} - \Pi_{A_t} \rbrace$.

Then for general bandits the environment state $\rho$ and the policy $\pi$ define the probability measure $P_{\rho,\pi}: (\Sigma \times \mathcal{X})^{\times T} \to[0,1]$ through
\begin{align}\label{eq:probdens_continous}
	&P_{\rho,\pi}(A_1,x_1,\ldots, A_T, x_T)\nonumber\\
	=&  \int_{A_1} \cdots \int_{A_T} P_{\rho}(x_T|a_T) \pi_T(\text{d} a_T|a_1,x_1,\ldots,a_{T-1},x_{T-1}) \cdots 
 	P_{\rho}(x_1|a_1) \pi_1(\mathrm{d} a_1),
\end{align}
which describes the full stochastic process analogous to the discrete case~\eqref{eq:probdens_discrete}.

\subsection{The regret}

So far, we have established the definitions and the interaction between the learner and the multi-armed quantum bandit. However, we have not yet defined the learner's objective. In the bandit literature, two primary objectives are typically considered: best-arm identification~\cite{audibert2010best,garivier2016optimal}, which involves finding the action (arm) with the highest expected reward (subject to certain constraints), and reward maximization~\cite{auer2002finite,lattimore_banditalgorithm_book}. In this work, we focus on the latter over a finite time horizon $T\in\mathbb{N}$. Specifically, this involves designing policies that maximize the cumulative reward $\sum_{t=1}^T X_t$. Instead of directly maximizing this quantity, we adopt a widely used metric in reinforcement learning called \emph{regret}, which quantifies the performance loss compared to the “best strategy". Formally, we define the regret in our setting as follows:

\begin{definition}\label{def:regret_maqb}
Given a multi-armed quantum bandit problem $(\mathcal{A}, \Gamma)$, a state $\rho \in \Gamma$, a policy $\pi$ and $T \in \mathbb{N}$, we define the \emph{cumulative regret} as
\begin{align}\label{regret}
  R_T ( \mathcal{A} ,\rho,\pi ) := \sum_{t=1}^T \max_{O \in \mathcal{A}} \Tr ( \rho O ) - \Tr ( \rho O_{A_t} ) .
\end{align}
The \emph{expected cumulative regret} is $\EX_{\rho , \pi} [  R_T ( \mathcal{A} ,\rho,\pi ) ]$ where the expectation value is taken with respect to the probability density~\eqref{eq:probdens_discrete} for the discrete case and~\eqref{eq:probdens_continous} for the general case.
\end{definition}

\textbf{Note.} If the environment is pure i.e $\rho = \ket{\psi}\! \bra{\psi}$ with $\ket{\psi}\in\mathbb{C}^d$ we will simply denote the regret as $R_T ( \mathcal{A},\psi , \pi )$ and the expectation over that environment as $\EX_{\psi , \pi} [ \cdot ] $. For the PSMAQB the regret for this problem has the following equivalent expressions
\begin{align}\label{eq:regret_psmaqb}
    R_T ( \mathcal{S}_d^*, \psi , \pi ) &= \sum_{t=1}^T 1 - \langle \psi | \Pi_{A_t} | \psi \rangle \\
    &= \sum_{t=1}^T 1 - F(\psi , \Pi_{A_t}) \\
    & = \frac{1}{4} \sum_{t=1}^T \||\psi\rangle\! \langle \psi | - \Pi_{A_t}  \|^2_1 ,
\end{align}
where we have used that for the optimal action 
\begin{align}
\max_{\Pi\in\mathcal{S}_d^*} \Tr (|\psi\rangle\! \langle \psi |\Pi) = \max_{\Pi\in\mathcal{S}_d^*} \langle \psi | \Pi | \psi \rangle = 1 .
\end{align}

We note that the regret can be interpreted as the sum, over all rounds, of the difference between the maximum expected reward achievable over the action set $\mathcal{A}$ and the expected reward associated with the observable chosen at each round. The learner's objective is to minimize the regret, which is equivalent to maximizing the cumulative reward. Although there is no universal consensus in the community, this definition of regret is often referred to as pseudo-regret~\cite{audibert2009exploration}. In the original work on multi-armed quantum bandits~\cite{Lumbreras2022multiarmedquantum}, the expected regret was introduced using the term $\EX_{\rho , \pi} [ X_t ]$. However, inconsistencies arose when analyzing the regret of certain policies, where the Definition~\ref{def:regret_maqb} was implicitly used instead. 

In this thesis, we will adhere to the above definition of regret, as it aligns with the connection to linear stochastic bandits discussed in Section~\ref{sec:connection_linearbandits}. It is worth noting that both regret and pseudo-regret have the same expected value, with the primary distinction being that pseudo-regret exhibits lower variance~\cite{abbasi2013online}. Despite this difference, both metrics serve to evaluate the performance of a policy against the best action with the highest expected value.

For a fixed policy $\pi$, the cumulative expected regret varies with $\rho$. Therefore, to evaluate the performance of a policy across a given set $\Gamma$ of environments, we need to establish state-independent metrics.

\begin{definition}
	Given a multi-armed quantum bandit problem $(\mathcal{A}, \Gamma)$, a policy $\pi$ and $T \in \mathbb{N}$, we define the \emph{worst-case regret} as
	\begin{align}
		R_T (\mathcal{A}, \Gamma , \pi) : = \sup_{\rho \in \Gamma} \EX_{\rho , \pi} [  R_T ( \mathcal{A}, \rho, \pi) ].
	\end{align}
	Moreover, the \textit{minimax regret} is defined as, 
	\begin{align}
		R_T ( \mathcal{A}, \Gamma) := \inf_{\pi} R_T ( \mathcal{A}, \Gamma , \pi) ,
	\end{align}
	where the infimum goes over all possible policies of the form in Definition~\ref{def:policy}.
\end{definition}

The worst-case regret quantifies how well a policy performs for a given set of density matrices $\Gamma$, while the minimax regret provides a measure of the inherent difficulty of the multi-armed quantum bandit problem—the smaller the minimax regret, the easier the problem.

\subsubsection{Expected regret decomposition}

It will be convenient for the lower bound analysis an equivalent definition of the expected regret in terms of the  \emph{sub-optimality gap} as,
\begin{align}\label{suboptimality}
\Delta_a := \max_{O \in \mathcal{A}} \Tr ( \rho O) - \Tr( \rho O_a),
\end{align}
for $a \in [k]$. The sub-optimality gap represents the relative loss between the optimal measurement and the measurement induced by $O_a \in \mathcal{A}$. Note that $\Delta_a \geq 0$ and the equality is only achieved by the optimal observables (it does not need to be unique). 
The equivalent expression for the discret case is given by
\begin{align}\label{regretarms}
\EX_{\rho,\pi}[ R_T  ( \mathcal{A}, \rho , \pi)] &= \sum_{a=1}^k \sum_{t=1}^T \EX_{\rho,\pi}[(\max_{O \in \mathcal{A}} \Tr ( \rho O )- \Tr (\rho O_{a} ) ) \mathbbm{1}\lbrace A_t = a \rbrace ] \\
							  &=	 \sum_{a=1}^k \sum_{t=1}^T \Delta_a \EX_{\rho,\pi}[\mathbbm{1}\lbrace A_t = a \rbrace]	\\			
							  &=	\sum_{a=1}^k \Delta_a \EX_{\rho,\pi}[T_a(T)] \,,
\end{align}
where $T_a(T) = \sum_{t=1}^T \mathbbm{1}\lbrace A_t = a \rbrace$ is a random variable that tells us how many times we have picked the action~$a$ over $T$ rounds.

In order to present a formula for regret for the general case, we need the continuous counterpart for the function $a\mapsto\mathbb{E}_{\rho,\pi}\big(T_a(T)\big)$ of the discrete case. To this end, we define, for all $t=1,\ldots,\,T$, the margins $P_{\rho,\pi}^{(t)}:\Sigma\times\mc X\to[0,1]$ through
\begin{align}
P_{\rho,\pi}^{(t)}(A,x)=\sum_{x_1\in\mc X}\cdots\sum_{x_{t-1}\in\mc X}\sum_{x_{t+1}\in\mc X}\cdots\sum_{x_T\in\mc X}P_{\rho,\pi}(\mc A,x_1,\ldots,\mc A,x_{t-1},A,x,\mc A,x_{t+1},\ldots,\mc A,x_T).
\end{align}
We now define the measure $\gamma_{\rho,\pi}:\Sigma\to[0,T]$ through
\begin{align}\label{eq:gammameasure}
\gamma_{\rho,\pi}(A)=\sum_{t=1}^T\sum_{x\in\mc X}P_{\rho,\pi}^{(t)}(A,x).
\end{align}
We denote by $E_{\rho|a}$ the expectation value of the conditional distribution $x\mapsto P_\rho(x|a)$ and by $\Delta_a$ the %$\gamma_{\rho,\pi}$-essential
supremum of the function $b\mapsto E_{\rho|b}-E_{\rho|a}$ for any action $a\in\mc A$; this is naturally the continuous analogue of the sub-optimality gap. Note that, in all the cases we study, $\mc A$ is a compact topological space, $\Sigma$ is the associated Borel $\sigma$-algebra, and $a\mapsto E_{\rho|a}$ is continuous for every environment $\rho$, so this supremum (which is also a maximum) technically makes sense. Using the measure $\gamma_{\rho,\pi}$ and the above sub-optimality gap, we may define the expected cumulative regret analogously to the discrete case: 
\begin{align}
\EX_{\rho,\pi}[ R_T  ( \mathcal{A}, \rho , \pi)])&:=\sum_{t=1}^T\mb E_{\rho,\pi}(\Delta_{a_t})\label{eq:contRegr2}\\
&=\sum_{t=1}^n\sum_{x_1\in\mc X}\cdots\sum_{x_n\in\mc X}\int_{\mc A^n}\Delta_{a_t}\,P_{\rho,\pi}(da_1,x_1,\ldots,da_n,x_n)\\
&=\int_{\mc A}\Delta_a\,d\gamma_{\rho,\pi}(a).\label{eq:contRegr1}
\end{align}
Note that~\eqref{eq:contRegr2} implies that, for the optimal strategy which always plays the best arm, the regret vanishes.

\subsubsection{Differences from quantum state tomography and online learning} 

Note that one can try to apply techniques from quantum state tomography~\cite{haah2016sample,guctua2020fast}, shadow tomography~\cite{shadow_aaronson} or classical shadows~\cite{huang2020predicting} in the multi-armed quantum bandit setting using the observables of the action set in order to estimate the unknown quantum state. If we learn the unknown quantum state (or some approximation thereof) we can choose the best action in order to minimize the regret. However, this strategy is not optimal: quantum state tomography algorithms can be thought of as pure exploration strategies since the algorithm only cares about choosing the action that helps us to learn most about the unknown quantum state, which is not necessarily the action that minimizes the regret. In particular, our setting is thus different from the setting of online learning of quantum states~\cite{online_learning_aaronson}. There, the learner is tasked to produce an estimate $\omega_t$ of $\rho$ in each round and the regret is related to the quality of this estimate. The main difference with our model is that we do not require our policy to produce an estimate of the unknown quantum state, just to choose an observables with large expectation value on the unknown quantum state. Also, the online quantum state learning problem is adversarial in the sense that the learner does not control the observables used for measurements. Instead, these observables are given by an unknown environment or an external agent (not necessarily stochastically), restricting the learner’s ability to choose the best measurements.

\section{Connection with linear stochastic bandits}\label{sec:connection_linearbandits}

Linear stochastic bandits extend the classical multi-armed bandit framework by incorporating a structured relationship between actions and rewards. In this model, each action is represented by a feature vector and belongs to a known decision set $\mathcal{A} \subseteq \mathbb{R}^d$ (for simplicity we assume to be fixed at each time $t$). After the learner selects an action $A_t \in \mathcal{A}_t$, the reward is sampled according to 
\begin{align}\label{eq:linar_bandit_reward}
    X_t = \langle A_t , \theta \rangle + \epsilon_t ,
\end{align}
where $\theta \in \mathbb{R}^d$ is an unknown parameter vector, and $\epsilon_t$ is \emph{conditionally $\eta$-subgaussian} which means
\begin{align}\label{eq:subgaussian_def}
    \EX \left[ \exp\left( \alpha \epsilon_t \right) | A_1, X_1, \ldots,  A_{t-1}, X_{t-1}, A_t \right] \leq \exp \left( \frac{\eta^2 \alpha^2}{2}\right) , \quad \text{for all } \alpha\in\mathbb{R} \text{ and } t , 
\end{align}
where the expectation is conditioned over all past actions and outcomes including $A_t$ before reward $X_t$ is observed. For simplicity we define the past history as $\mathcal{H}_{t-1} = (A_s,X_s)_{s=1}^{t-1}$. The term $\epsilon_t$ will be referred to as the \textit{noise} of reward or \textit{statistical noise}. The definition of $\eta$-subgaussian immediatly imples that $\EX [\epsilon_t | \mathcal{H}_{t-1},A_t ] = 0 $ and $\textrm{Var} (\epsilon_t | \mathcal{H}_{t-1},A_t) \leq \eta^2$. This can be seen immediately expanding the terms in~\eqref{eq:subgaussian_def}\footnote{
    $1 + \alpha \EX [ \epsilon_t | \mathcal{F}_{t-1}] + \frac{\alpha^2}{2}\EX (\epsilon^2_t|\mathcal{F}_{t-1}) + O( \alpha^3) \leq  1 + \frac{\alpha^2}{2}\sigma^2 + O(\alpha^4) $}. Then this implies
\begin{align}
    \EX [X_t | A_1, X_1, \ldots,  A_{t-1}, X_{t-1}, A_t] &= \langle \theta , A_t \rangle , \\ 
    \mathrm{Var}(X_t| A_1, X_1, \ldots,  A_{t-1}, X_{t-1}, A_t ) &\leq \eta^2 .
\end{align}
Thus, the expected reward is linear in both the action and the unknown parameter, which is what gives rise to the name linear bandits. The subgaussian parameter $\eta$ can be thought of as a quantity that controls the reward variance. Throughout, we will refer to the reward as $\eta$-subgaussian, meaning that the noise term $\epsilon_t$ satisfies the condition in~\eqref{eq:subgaussian_def}.

This structured model allows the learner to generalize knowledge across actions, making it especially effective in settings with large or continuous action sets. The objective remains to maximize the cumulative reward over a finite time horizon, but the main challenge lies in balancing exploration of the action space to estimate the unknown parameter $\theta$ and exploitation of the learned information to make optimal decisions. The regret in this model takes the form
\begin{align}\label{eq:regret_linearbandits}
    R_T (\mathcal{A} , \theta , \pi ) = \sum_{t=1}^T \max_{A\in\mathcal{A}} \langle \theta ,  A \rangle - \langle \theta , A_t \rangle,
\end{align}
where $\pi = \lbrace \pi_t \rbrace_{t\in\mathbb{N}}$ is the learner's policy as defined analogously as in~\ref{def:policy}.

The multi-armed quantum bandit model can be viewed as a specific instance of the linear stochastic bandit framework. To see this, consider a fixed set of orthonormal Hermitian matrices $\lbrace {\sigma_i \rbrace}_{i=1}^{d^2}$, which serve as a basis for parameterizing both the environment $\rho \in \Gamma$ and the action $O_{A_t} \in \mathcal{A}$. Using this basis, we express the state and observable as
\begin{align}\label{dparametrization}
\rho = \sum_{i=1}^{d^2} \theta_i \sigma_i, \quad O_{A_t} = \sum_{i=1}^{d^2} A_{t,i} \sigma_i,
\end{align}
where, with a slight abuse of notation, $A_t$ indexes the observable and also refers to the parameterized vector $A_t \in \mathbb{R}^d$. With this parameterization, the reward can be written in the same form as in Eq.~\eqref{eq:linar_bandit_reward}, since $\EX_{\rho,\pi} [X_t] = \Tr (\rho O_{A_t} ) = \langle \theta , A_t \rangle$. The additive noise $\epsilon_t$ can be seen to be subgaussian, with parameter $\eta_t$ determined by the distribution of eigenvalues of $O_{A_t}$ and applying Hoeffding Lemma~\cite[Equation 4.16]{Hoeffding}\footnote{The Hoeffding Lemma states that if is $X$ a random variable such that $X \in [a, b]$ almost surely. Then, for any $\lambda \in \mathbb{R}$, the following inequality $\mathbb{E}\left[e^{\lambda (X - \mathbb{E}[X])}\right] \leq \exp\left(\frac{\lambda^2 (b-a)^2}{8}\right) $.}

This is not surprising: most problems in quantum tomography, metrology or learning (in the single copy regime) can be converted to a classical problem with additional structure, since, in the end, we are interested in learning a matrix of complex numbers constituting the density operator. The question of interest then is whether the structure imposed by the quantum problem is helpful to find more efficient algorithms. Indeed, the fact that our model is a subset of a more general class of problems for which classical algorithms are known does not imply that these algorithms yield the smallest possible regret for the subset of problems we consider.

\chapter{Regret lower bounds}\label{ch:lowerbounds}

The content of this chapter is devoted to the derivation of minimax lower bounds for the multi-armed quantum bandit problem, separating the discrete and general settings. The results are mainly extracted from the works~\cite{Lumbreras2022multiarmedquantum,lumbreras24pure} with added extended discussions.

\section{Bretagnolle-Huber inequality and divergence decomposition lemma}

In this section, we present some important technical results that we need in the subsequent proofs in order to prove the lower bounds for the minimax regret.
One of the technical tools that we will use in order to bound the regret is the following lemma due to Bretagnolle and Huber \cite{divergenceineq}.
\begin{lemma}[Bretagnolle--Huber inequality]\label{pinsker}
Let $P$ and $Q$ be probability measures on the same measurable space $(\Omega , \Sigma) $, and let $A\in \Sigma$ be an arbitrary event. Then,
\begin{align}
 P(A) + Q(A^c)  \geq \frac{1}{2} \exp ( -D(P \| Q) ) ,   
\end{align} 
where $A^c = \Omega \backslash A $ is the complement of $A$ and $D(P\|Q)$ is the Kullback–Leibler divergence.
\end{lemma}
This inequality was first used for bandits in~\cite{Bubeck13}; however, we will also use a stronger bound that replaces $D(P \|Q)$ with the R\'enyi divergence of order $\alpha = \frac12$, defined as
\begin{align}
D_{\frac{1}{2}}(P \| Q) = - \log F(P, Q). 
\end{align}
We use the proof given in~\cite{lattimore_banditalgorithm_book}[Chapter 14] for the Bretagnolle-Huber inequality in order to prove the alternative version.
\begin{lemma}\label{pinsker2}
Let $P$ and $Q$ be probability measures on the same measurable space $(\Omega , \Sigma) $, and let $A\in \Sigma$ be an arbitrary event with $A^c$ its complement. Then,
\begin{align}
 P(A) + Q(A^c)  \geq \frac{1}{2} \exp ( -D_{\frac{1}{2}}(P\|Q) ) = \frac12 F(P, Q).
\end{align} 
\end{lemma}
\begin{proof}
Let us define the probability measure $\mu:=\frac{1}{2}(P+Q)$ which dominates both $P$ and $Q$ and denote $dP/d\mu=:p$ and $dQ/d\mu=:q$. Note that,
\begin{align}
P(A) + Q(A^c) =&\int_A p\,\dd\mu + \int_{A^C} q\,d\mu \geq \int_A \min \lbrace p , q \rbrace\,\dd\mu + \int_{A^C} \min \lbrace p , q \rbrace\,\dd\mu\\
=& \int \min \lbrace p , q \rbrace\,\dd\mu.
\end{align}
In order to find a lower bound on this expression, we use the Cauchy–Schwarz inequality to show that
\begin{align}
\frac{1}{2} \left( \int \sqrt{pq}\,d\mu \right)^2 
= \frac{1}{2} \left( \int \sqrt{\max \lbrace p , q \rbrace \min \lbrace p , q \rbrace}\,\dd\mu \right)^2 
& \\ \leq \frac{1}{2} \left( \int \max \lbrace p , q \rbrace\,\dd\mu  \right)\left( \int \min\lbrace p , q \rbrace\,\dd\mu \right) 
&\leq \int \min\lbrace p , q \rbrace\,\dd\mu ,
\end{align}
where in the final step we used that $ \int \max \lbrace p , q \rbrace \dd \mu \leq  \int (p + q) \dd \mu = 2$.
\end{proof}

The other main result that we will need allows us to decompose the divergence computed between two joint distributions that result from the same policy applied to two different quantum states.

\begin{lemma}[Divergence decomposition lemma]\label{divergence}
Let $\mathcal{A}=\left\lbrace O_1,...,O_k \right\rbrace$ be a set of actions and $\rho$ and $\rho'$ two quantum states defining two multi-armed quantum bandits with action set $\mathcal{A}$. Fix some policy $\pi$, time horizon $T\in\mathbb{N}$
and let $P_{\rho, \pi }$ and $ P_{\rho' \pi }$ be the probability distributions induced by the policy $\pi $ and  environment $\rho $ as described in~\eqref{eq:probdens_discrete}. Then,
\begin{align}
	D(P_{\rho, \pi}  \| P_{\rho' , \pi}) = \sum_{a=1}^k \EX_{\rho,\pi} [T_a ( T ) ] D \big( P_{\rho}(\cdot | a) \big\| P_{\rho'}(\cdot | a) \big) .
\end{align}
\end{lemma}

The above lemma and the proof can be found in \cite[Chapter 15]{lattimore_banditalgorithm_book} for the classical model of multi-armed stochastic bandits. We have restated the lemma for our quantum case, but this statement and the proof follow trivially from the original one. The proof is a consequence of the chain rule for the KL divergence. Unfortunately, there is no such decomposition for R\'enyi divergences but we will see in this Chapter that applying the data processing inequality and bounding it with the sandwiched quantum R\'enyi divergences (see, e.g.,~\cite{lennert13,wilde13}) between $\rho$ and $\rho'$ is sufficient to bound the regret.

For this purpose, we present a result on the divergences between the probability distributions $P_{\rho,\pi}$. Note that this result also holds in the general bandit case where the set of arms can be continuous. What we mean below by `quantum extension' of a classical relative entropy $D$ is that, if quantum states $\rho$ and $\sigma$ commute, then $\tilde{D}(\rho\|\sigma)=D(p\|q)$ where $p$ and, respectively, $q$ are the vectors of eigenvalues of $\rho$ and, respectively, $\sigma$, in a common eigenbasis.

\begin{lemma}\label{lemma:D1/2ineq}
Let $\alpha\in\R$ be such that the classical R\'{e}nyi relative entropy $D_\alpha$ can be given a quantum extension (which we denote with the same symbol) that is additive and satisfies the data processing inequality. For any policy $\pi$, time horizon $T\in\mathbb{N}$ and environment states $\rho,\,\rho'$, we have
\begin{align}\label{eq:D1/2ineq}
D_\alpha(P_{\rho,\pi}\|  P_{\rho',\pi})\leq T D_\alpha(\rho\|\rho').
\end{align}
\end{lemma}
In particular, for $\alpha=1/2$, we can choose $D_{\frac12}(\rho\|\rho') = - \log{F(\rho,\rho')}$ and, for $\alpha=1$, we can let the quantum extension $D \equiv D_1$ of the Kullback-Leibler relative entropy be the quantum relative entropy.

\begin{proof}
Let us fix the policy $\pi$ and use the notations and definitions of Section \ref{sec:general_bandits}. We prove the claim by constructing a POVM $\ms E$ over the value space $\mc C:=(\Sigma\times\mc X)^T$ and operating in $\hil^{\otimes T}$ such that
\begin{align}\label{eq:policyPOVM}
P_{\sigma,\pi}(A_1,x_1,\ldots,A_T,x_T)=\tr{\sigma^{\otimes T}\ms E(A_1,x_1,\ldots,A_T,x_T)}
\end{align}
for all states $\sigma$, and $A_t\in\Sigma$ and $x\in\mc X$ for $t=1,\ldots,\,T$. Note that, as $\ms E$ is not a set function, it is, strictly speaking, not a POVM in the same sense as $P_{\sigma,\pi}$ is not a measure. However, this should cause no confusion. Using the data-processing inequality and the additivity of $D_\alpha$, we have, for all states $\rho$ and $\rho'$,
\begin{align}
D_\alpha(P_{\rho,\pi}\|P_{\rho',\pi})\leq D_\alpha(\rho^{\otimes T}\|\rho'^{\otimes T})= TD_\alpha(\rho\|\rho'),
\end{align}
implying inequality \eqref{eq:D1/2ineq}. (Note that the POVM $\ms E$ corresponds to a quantum-to-classical channel which maps $\sigma^{\otimes T}$ into $ P_{\sigma,\pi}$). Thus, all that remains is to write down $\ms E$.

Recall the reward distributions $P_\sigma(x|a)$ in state $\sigma$ conditioned by the action $a\in\mc A$. Through linear extension, we may define these distributions $P_\tau(x|a)$ for any operators $\tau$; these are, in general, complex distributions. We may define the POVM $\ms E$ through
\begin{align}
&\tr{(\tau_1\otimes\cdots\otimes\tau_T)\ms E(A_1,x_1,\ldots,A_T,x_T)} \nonumber \\
=&\int_{A_1}\cdots\int_{A_T}\prod_{t=1}^TP_{\tau_t}(x_t|a_t)\pi_t(da_t|a_1,x_1,\ldots,a_{t-1},x_{t-1})\label{eq:EPOVM}
\end{align}
for all operators $\tau_t$ and $A_t\in\mc B$ and $x_t\in\mc X$, $t=1,\ldots,T$. As $\tau\mapsto P_\tau(x|a)$ is a linear functional and $x\mapsto P_\sigma(x|a)$ is a probability distribution whenever $\sigma$ is a state, it easily follows that there are POVMs $\ms P_a$ on $\mc X$ operating in $\hil$ such that $P_\tau(x|a)=\tr{\tau\ms P_a(x)}$. Equivalently (although slightly less formally), we may now write $\ms E$ in the differential form
\begin{align}
\ms E(da_1,x_1,\ldots,da_T,x_T)=\bigotimes_{t=1}^T\pi_t(da_t|a_1,x_1,\ldots,a_{t-1},x_{t-1})\ms P_{a_t}(x_t).
\end{align}
By substituting $\tau_t=\sigma$ for all $t=1,\ldots,T$ in~\eqref{eq:EPOVM}, we immediately see that \newline $P_{\sigma,\pi}(A_1,x_1\ldots,A_T,x_T)=\tr{\sigma^{\otimes n}\ms E(A_1,x_1\ldots,A_T,x_T)}$ for all $A_t\in\Sigma$ and $x_t\in\mc X$, $t=1,\ldots,T$, and, thus,~\eqref{eq:policyPOVM} holds.
\end{proof}

\section{Discrete lower bounds}

As an illustrative case to introduce our lower bound proof techniques, we consider multi-armed quantum bandit where the environment can be any state in a Hilbert space of dimension $d$, and the action set is any set of observables. We will derive a lower bound for the minimax regret by constructing a quantum state that ensures a non-trivial lower bound for every policy. Notably, if no conditions are imposed on the action set, the lower bound vanishes. The reason is that if there is an operator in the action set that dominates over the others, the policy that always chooses this operator at each round achieves $0$ regret. More specifically, suppose that $\mathcal{A} = \left\lbrace O_1 , O_2 \right\rbrace$ with $O_1\geq O_2$. Then we know that independently of the environment $\rho$, $\Tr(\rho O_1 ) \geq \Tr(\rho O_2 )$ and the policy that always chooses $O_1$ will always pick the optimal action.

So, we will impose a condition on the action set that ensures that there is no such dominant action/operator. The condition is the following:
in the action set $\mathcal{A}$ there exist at least two operators $O_a,O_b\in \mathcal{A}$ with maximal eigenvectors $\ket{\psi_A},\ket{\psi_B}$ such that for any $i\neq a$ and $j\neq b$ 
\begin{align}\label{condition}
\bra{\psi_A } O_a \ket{\psi_A} > \bra{\psi_A } O_i \ket{\psi_A}\quad \text{and}\quad
\bra{\psi_B } O_b \ket{\psi_B} > \bra{\psi_B } O_j \ket{\psi_B}.
\end{align}

\begin{theorem}\label{generalower}
Let $T\in\mathbb{N}$. For any policy $\pi$ and action set of traceless observables $\mathcal{A}$ that obeys condition \eqref{condition} there exists an environment $\rho\in\mathcal{S}_d $ such that,
\begin{align}
\EX_{\rho,\pi}  \left[  R_T(\mathcal{A},\rho,\pi) \right] \geq C_{\mathcal{A}}\sqrt{T} ,
 \end{align}
where $C_{\mathcal{A}}>0$ is a constant that depends on the action set.
\end{theorem}

\begin{proof}
Choose two operators $O_a,O_b\in\mathcal{A}$ that obey the condition in~\eqref{condition} with maximal eigenvectors $\ket{\psi_A}$ and $\ket{\psi_B}$, respectively. Define the following environments:
\begin{align}
\rho := \frac{1-\Delta}{d}\mathbb{I} + \Delta \ket{\psi_A}\!\bra{\psi_A},\quad \rho' := \frac{1-\Delta}{d}\mathbb{I} + \Delta \ket{\psi_B}\!\bra{\psi_B},
\end{align}
for some constant $0\leq \Delta \leq \frac{1}{2}$ to be defined later. Note that $\rho$ and $\rho'\geq 0$. Define
\begin{align}
c = \min \left\lbrace \bra{\psi_A } O_a \ket{\psi_A} - \bra{\psi_A } O_i \ket{\psi_A},
\bra{\psi_B } O_b \ket{\psi_B} - \bra{\psi_B } O_j \ket{\psi_B} \right\rbrace,
\end{align}
for $i\neq a$ and $j\neq b$. Using the expression for the regret in~\eqref{regretarms}, we can compute the regret for $\rho$ as,
\begin{align}
\EX_{\rho,\pi}  \left[ R_T(\mathcal{A}, \rho , \pi) \right] = \sum_{i=1}^k \EX_{\rho,\pi}(T_i(T))\Delta \left( \bra{\psi_A } O_a \ket{\psi_A} - \bra{\psi_A } O_i \ket{\psi_A} \right)
\end{align}
where we have used that the observables are traceless and 
\begin{align}
\max_{O_i\in \mathcal{A}} \Tr (\rho O_i ) = \Tr(\rho O_a ) = \Delta \bra{\psi_A } O_a \ket{\psi_A}.
\end{align}
Note that for $i=a$ the sub-optimality gaps in~\eqref{suboptimality} are zero. Thus, using Condition \eqref{condition} we can bound the regret as,
\begin{align}
\EX_{\rho,\pi}  \left[  R_T(\mathcal{A},\rho , \pi ) \right] = \sum_{i\neq a} \EX_{\rho,\pi}[T_i(T)]\Delta \left( \bra{\psi_A } O_a \ket{\psi_A} - \bra{\psi_A } O_i \ket{\psi_A} \right)\geq \Delta c \sum_{i\neq a} \EX_{\rho,\pi} [T_i(T)],
\end{align}
where $c>0$. Using that $T =\sum_{i=1}^k \EX_{\rho,\pi}[T_i(T)]$ and Markov inequality we have,
\begin{align}\label{rega}
\EX_{\rho,\pi}  \left[ R_T(\mathcal{A},\rho , \pi) \right] \geq \Delta c \EX_{\rho,\pi}  \left[ T - T_a(T) \right] \geq \frac{cn\Delta}{2}P_{\rho,\pi}\left( T_a(T)\leq \frac{T}{2} \right).
\end{align}
Similarly, the regret for $\rho'$ can be bounded as,
\begin{align}
\EX_{\rho',\pi}  \left[ R_T(\mathcal{A},\rho',\pi) \right] \geq \Delta \EX_{\rho',\pi}[T_a(T)]\left(  \bra{\psi_B } O_b \ket{\psi_B} - \bra{\psi_B } O_i \ket{\psi_B} \right),
\end{align}
where we have taken into account just the term with $i=a$. Using Condition \eqref{condition} and Markov inequality we have,
\begin{align}\label{regb}
\EX_{\rho',\pi}  \left[ R_T(\mathcal{A},\rho',\pi ) \right] \geq \frac{cT\Delta}{2}P_{\rho' , \pi} \left( T_a(T) > \frac{T}{2} \right).
\end{align}
Thus, combining~\eqref{rega} and~\eqref{regb},
\begin{align}
\EX_{\rho,\pi}  \left[ R_T(\mathcal{A},\rho ,\pi ) \right] + \EX_{\rho',\pi}  \left[ R_T(\mathcal{A},\rho' , \pi ) \right] \geq \frac{cT\Delta}{2}\left( P_{\rho,\pi}\left( T_a(T)\leq \frac{n}{2} \right) + P_{\rho' , \pi} \left( T_a(T) > \frac{T}{2} \right) \right).
\end{align}
Using Lemma \ref{pinsker} we can bound the above expression as,
\begin{align}\label{regb2}
\EX_{\rho,\pi}  \left[ R_T(\mathcal{A},\rho ,\pi ) \right] + \EX_{\rho',\pi}  \left[ R_T(\mathcal{A},\rho' , \pi ) \right] \geq \frac{cT\Delta}{4}\exp\left( -D(P_{\rho , \pi}  \| P_{\rho' , \pi}  ) \right).
\end{align}
Using Lemma \ref{divergence} combined with the data-processing inequality, we can bound the Kullback-divergence as,
\begin{align}
D(P_{\rho , \pi}  \| P_{\rho' , \pi}  ) =& \sum_{a=1}^k \EX_{\rho , \pi } (T_a(T) )  D \big( P_{\rho}(\cdot | a) \big\| P_{\rho'}(\cdot | a) \big) )\leq D(\rho \| \rho' )\sum_{a=1}^k \EX_{\rho , \pi } [T_a(T) ]\nonumber\\
=& T D(\rho \| \rho' ).\label{brelent}
\end{align}
where $D(\rho \| \rho' )= \Tr\rho\log\rho-\Tr\rho\log\rho'$ is the relative entropy between $\rho$ and $\rho'$.
Now define $f(\Delta) = D(\rho \| \rho' )$ for the corresponding expressions of $\rho$ and $\rho'$ and note that $f(0) = 0$ since for $\Delta=0$, $\rho = \rho'$ and $f'(0)=0$ using the convexity of the quantum relative entropy. Thus, using Taylor's theorem, we can express $f(\Delta )$ as,
\begin{align}
f(\Delta ) = \frac{\Delta^2}{2}f''(\chi ),
\end{align}
for some $\chi \in [0,\Delta ] $. Define, 
\begin{align}
c_f = \max f''(\chi ) \quad \text{for}\quad \chi\in \left[0,\frac{1}{2} \right].
\end{align}
Note that $f''(\chi)$ is well defined since the states $\rho,\rho'$ have full support, and then $f(\chi )$ is a smooth function. Thus using that $f''(\chi)\leq c_f$ for $\Delta\in [0,\frac{1}{2} ]$ we can plug the above expression into~\eqref{brelent} and using~\eqref{regb2} we have,
\begin{align}
\EX_{\rho,\pi}  \left[ R_T(\mathcal{A},\rho ,\pi ) \right] + \EX_{\rho',\pi}  \left[ R_T(\mathcal{A},\rho' , \pi ) \right]  \geq \frac{cT\Delta}{4}\exp\left( -\frac{T\Delta^2}{2} c_f\right).
\end{align}
Finally, if we choose $\Delta = \frac{1}{2\sqrt{T}}$, 
\begin{align}
\EX_{\rho,\pi}  \left[ R_T(\mathcal{A},\rho ,\pi ) \right] + \EX_{\rho',\pi}  \left[ R_T(\mathcal{A},\rho' , \pi ) \right]  \geq \frac{c}{8}\exp\left( -\frac{1}{8} c_f \right)\sqrt{T},
\end{align}
and the result follows using 
\begin{align}
2\max \left\lbrace \EX_{\rho,\pi}  \left[ R_T(\mathcal{A},\rho ,\pi ) \right] , \EX_{\rho',\pi}  \left[  R_T(\mathcal{A},\rho' , \pi) \right] \right\rbrace \geq \EX_{\rho,\pi}  \left[ R_T(\mathcal{A},\rho ,\pi ) \right] + \EX_{\rho',\pi}  \left[ R_T(\mathcal{A},\rho' , \pi ) \right].
\end{align}
\end{proof}

\subsection{Pauli observables}

In this section, we analyze the case of Pauli observables to show how the number of available actions affects the regret bound. If we consider a $d=2^m-$dimensional ($m$ qubits) Hilbert space, there are $d^2$ different Pauli observables and they can be expressed as the $m$-fold tensor product of the $2\times 2$ Pauli matrices. Let $\sigma_1,...,\sigma_{d^2}$ denote all the possible Pauli observables. Each $\sigma_i$ can be expressed as $\sigma_i = \Pi^+_i - \Pi^-_i$, where $\Pi^+_i ,\Pi^-_i$ are projectors associated to the $+1$ and $-1$ subspaces and they describe the 2 possible outcomes when we perform measurements using Pauli observables.

\begin{theorem}\label{th:pauliobservables}
Let $T\in\mathbb{N}$. For any policy $\pi$ and action set of observables $\mathcal{A}$ comprised of $k$ distinct length-$m$ strings of single-qubit Pauli observables for $d=2^m$ with $m\in\mathbb{N}$, there exists an environment $\rho \in \mathcal{S}_d$ such that
\begin{align}
	\EX_{\rho,\pi} [ R_T(\mathcal{A},\rho , \pi ) ] \geq \frac{3}{100}	\sqrt{(k-1)T},
\end{align}
for $T \geq 2( k-1)$.
\end{theorem}

\begin{proof}
The case $k = 1$ is trivial since $R_T = 0$ always. Suppose $k>1$ and let $0\leq \Delta \leq \frac{1}{3}$ be a constant to be chosen later. Pick $\sigma_1 \in \mathcal{A}$ and define the following environment,
\begin{align}
\rho := \frac{\mathbb{I}}{d} + \frac{\Delta}{d}\sigma_1.
\end{align}
Define
\begin{align}\label{lessplayed}
l := \argmin_{j>1} \EX_{\rho,\pi} \left[ T_j (n)  \right]
\end{align}
as the index for the least expected picked observable different from $\sigma_1$. Define a second environment as,
\begin{align}
\rho' := \frac{\mathbb{I}}{d} + \frac{\Delta}{d}\sigma_1 + \frac{2\Delta}{d}\sigma_l.
\end{align}
Note that $\rho \geq 0$ and $\rho' \geq 0$ since $0\leq \Delta \leq \frac{1}{3}$.
In order to compute the sub-optimality gaps \eqref{suboptimality} we compute the following quantities for $\sigma_i\in\mathcal{A}$,
\begin{align}\label{mean1pauli}
\Tr(\rho \sigma_i) = \frac{\Delta}{d}\Tr(\sigma_1\sigma_i) = \begin{cases}
\Delta \quad \text{if}\quad  i=1.\\
0 \quad \text{otherwise.}
\end{cases}
\end{align}
\begin{align}\label{mean2pauli}
\Tr(\rho' \sigma_i) = \frac{\Delta}{d}\left( \Tr(\sigma_1\sigma_i)+2\Tr(\sigma_l\sigma_i) \right) = \begin{cases}
2\Delta \quad i=l. \\ \Delta \quad i=1 . \\ 0 \quad \text{otherwise.}
\end{cases}
\end{align}
For $\rho$ the sub-optimality \eqref{suboptimality} gaps are $\Delta_i = \Delta$ for $i\neq 1$ and $\Delta_i = 0 $ for $i=1$. Thus we can compute the regret as,
\begin{align}
\EX_{\rho,\pi}[ R_T(\mathcal{A},\rho , \pi) ] = \sum_{a=1}^k \EX_{\rho,\pi}[ T_a(T) ]\Delta_a = \Delta \sum_{a\neq 1} \EX_{\rho,\pi}[ T_a(T) ]
\end{align}
Using Markov's inequality and $ T = \sum_{i=a}^k \EX_{\rho,\pi}[ T_a(T) ]$,
\begin{align}\label{pauli1}
\EX_{\rho,\pi} [R_T(\mathcal{A},\rho , \pi ) ] = \Delta \EX_{\rho,\pi}\left[ T - T_1(T) \right] \geq \frac{n\Delta}{2}P_{\rho , \pi}\left( T_1(T) \leq \frac{n}{2} \right).
\end{align}
For $\rho'$ we bound the regret for the term $i=1$, using that the sub-optimality gap is $\Delta_1 = \Delta$ and Markov's inequality we have,
\begin{align}\label{paulil}
\EX_{\rho',\pi} [ R_T(\mathcal{A},\rho' , \pi ) ] = \sum_{a=1}^k \EX_{\rho',\pi}[ T_a(T) ]\Delta_a \geq \Delta \EX_{\rho',\pi}[ T_1(T) ]\geq \frac{T\Delta}{2} P_{\rho' , \pi}\left( T_1(T) > \frac{T}{2} \right) .
\end{align}
%
%\begin{align}
%R_n(\rho_a , \pi ) + R_n(\rho_b , \pi ) \geq \frac{n\Delta}{2}\left( \mathbb{P}_{\rho_a , \pi}\left( T_a(n) \leq \frac{n}{2} \right) +  \mathbb{P}_{\rho_b , \pi}\left( T_a(n) > \frac{n}{2} \right)  \right)
%\end{align}
Combining~\eqref{pauli1},\eqref{paulil} and the Bretagnolle-Huber inequality \eqref{pinsker} we have,
\begin{align}\label{pauliregsum}
\EX_{\rho,\pi} [R_T(\mathcal{A},\rho , \pi )] + \EX_{\rho',\pi} [ R_T(\mathcal{A},\rho' , \pi ) ] \geq \frac{T\Delta}{4}\exp\left(-D(P_{\rho , \pi} \| P_{\rho' , \pi})  \right).
\end{align}
Note that the probabilities of the rewards are just Bernoulli distributions since the observables $\sigma_i\in\mathcal{A}$ have two outcomes, +1 and -1. Using Lemma \ref{divergence} we can decompose the relative entropy as
\begin{align}
D (P_{\rho , \pi} \| P_{\rho' , \pi }) = \sum_{a=1}^k \EX_{\rho ,\pi} [T_a ( T) ] D \big( P_{\rho}(\cdot | a) \big\| P_{\rho'}(\cdot | a) \big) .
\end{align}
Note that $P_{\rho}(\cdot | a) , P_{\rho'}(\cdot | a) $ are not equal only when $a=l$, so $D\big( P_{\rho}(\cdot | a) \big\| P_{\rho'}(\cdot | a) \big)  = 0$ for $a\neq l$ and,
\begin{align}\label{paulidiv1}
D (P_{\rho , \pi} ,P_{\rho' , \pi }) = \EX_{\rho ,\pi} [T_l (T ) ] D\big( P_{\rho}(\cdot | l) \big\| P_{\rho'}(\cdot | l) \big) .
\end{align}
Using~\eqref{mean1pauli} and~\eqref{mean2pauli} we have 
\begin{align}
P_{\rho}(1 | l) = \frac{1}{2}, \quad P_{\rho'}(1 | l) = \frac{1}{2}+\Delta.
\end{align}
Then we can compute the Kullback–Leibler divergence as,
\begin{align}\label{paulidivergence}
 D\big( P_{\rho}(\cdot | l) \big\| P_{\rho'}(\cdot | l) \big)  = \frac{1}{2}\log \frac{\frac{1}{2}}{\frac{1}{2}+\Delta}+\frac{1}{2}\log \frac{\frac{1}{2}}{\frac{1}{2}-\Delta} = \frac{1}{2}\log\frac{1}{1-4\Delta^2}.
\end{align}
Note that using the definition of $l$ $\eqref{lessplayed}$ it follows that 
\begin{align}
T = \sum_{i=1}^k E_{\rho,\pi}[T_i( T ) ] \geq \sum_{i=2}^k E_{\rho,\pi}[T_i( T ) ] \newline \geq (k-1)E_{\rho,\pi}[T_l( T ) ].
\end{align}
Thus, it holds that, 
 \begin{align}\label{pauliless}
 E_{\rho , \pi} [T_l ( T ) ] \leq \frac{T}{k-1}.
 \end{align}
Combining~\eqref{pauliregsum},\eqref{paulidiv1},\eqref{paulidivergence} and \eqref{pauliless} we have,
\begin{align}
\EX_{\rho,\pi} [R_T(\mathcal{A},\rho , \pi )] + \EX_{\rho',\pi} [ R_T(\mathcal{A},\rho' , \pi ) ] \geq \frac{T\Delta}{4}\exp\left(-\frac{T}{2(k-1)}\log\frac{1}{1-4\Delta^2}  \right).
\end{align}
Finally choosing $\Delta = \frac{1}{2}\sqrt{1-e^{-\frac{k-1}{T}}}$ and using that $e^{-x}\leq (e^{-1}-1)x+1$ for $0\leq x\leq 1$ we arrive to the result,
\begin{align}
\EX_{\rho,\pi} [R_T(\mathcal{A},\rho , \pi )] + \EX_{\rho',\pi} [ R_T(\mathcal{A},\rho' , \pi ) ]  \geq \frac{e^{-1/2}\sqrt{1-e^{-1}}}{8}\sqrt{(k-1)n}.
\end{align}
The condition $T\geq 2(k-1)$ suffices to have $0 \leq \Delta\leq \frac{1}{3}$. The theorem follows using $ \frac{e^{-1/2}\sqrt{1-e^{-1}}}{8} > \frac{3}{50}$.
\end{proof}

\subsection{Pure state environment}

In the previous setting, we considered mixed-state environments to construct worst-case instances that led to the $\sqrt{T}$ lower bound on regret. A natural question that arises is whether a similar behavior persists in the case of pure-state environments. That is, does the fundamental exploration-exploitation tradeoff still impose the same scaling on regret, or do pure states with discrete action sets allow for a more efficient learning strategy? To investigate this, we focus on a simplified yet representative scenario: 1-qubit pure-state environments, where the action set consists of the Pauli observables.

\begin{theorem}\label{th:pauliobservablesdiscrete}
Let $T\in \mathbb{N}$. For any policy $\pi$ and action set $\mathcal{A}$ containing the Pauli observables for 1-qubit, there exists an environment $|\psi \rangle \! \langle \psi | \in \mathcal{S}_2^*$ such that
\[  \EX_{\psi , \pi } [ R_T(\mathcal{A},\psi , \pi  ) ] \geq \frac{3}{200}\sqrt{T}. \]
\end{theorem}

\begin{proof}
Let $\Delta \in [0,1]$ and define the following the following two pure-states environments,
\begin{align}
| \psi_{\pm} \rangle := c_{\pm} \left( 1+\frac{1\pm\Delta}{\sqrt{2}} |0 \rangle +  \frac{1+\Delta}{\sqrt{2}} |1\rangle \right),
\end{align}
where $c_{\pm} = \frac{1}{\sqrt{\left( 1+\frac{1\pm\Delta}{\sqrt{2}} \right)^2 + \left(\frac{1\pm\Delta}{\sqrt{2}}  \right)^2 }}$.
Let $|0\rangle, |+\rangle , |\psi_y \rangle$ be the eigenvector of the Pauli observables $\sigma_z , \sigma_x , \sigma_y $ respectively and define $x_{\pm} = \frac{1\pm \Delta}{\sqrt{2}}$. Then, compute the following quantities
\begin{align}\label{probpure}
p^x_\pm &= | \langle + | \psi_{\pm} \rangle |^2 = \frac{(1+2x_\pm )^2}{2\left( (1+x_{\pm})^2+x_{\pm}^2 \right)} , \quad
p^z_\pm = | \langle 0 | \psi_{\pm} \rangle |^2 = \frac{(1+x_{\pm})^2}{(1+x_{\pm})^2+x_{\pm}^2} , \\
p^y_\pm &= | \langle \psi_y | \psi_\pm \rangle |^2 = \frac{1}{2}. \nonumber
\end{align}
Note that the expectation values for each action on the environments  can be computed as,
\begin{align}\label{gappure}
\Tr \left( \sigma_i |\psi_\pm \rangle \! \langle \psi_\pm | \right) = 2p^i_\pm -1,
\end{align}
for $i = x,y,z$. Note that $\Tr \left( \sigma_y |\psi_\pm \rangle \! \langle \psi_\pm | \right) = 0$, 
$\Tr \left( \sigma_x |\psi_+ \rangle \! \langle \psi_+ | \right) \geq \Tr \left( \sigma_z |\psi_+ \rangle \! \langle \psi_+ | \right)  \geq $ 0 and 
$\Tr \left( \sigma_z |\psi_- \rangle \! \langle \psi_- | \right) \geq \Tr \left( \sigma_x |\psi_- \rangle \! \langle \psi_- | \right) \geq 0 $  for $\Delta \in [0,1]$. Thus, for the environment $|\psi_+ \rangle$ the optimal action is given by $\sigma_x$ and for $|\psi_- \rangle$ by $\sigma_z$. Let $\Delta^{\pm}_i$ denote the sub-optimality gap for the $ i$-th action in the environment $\psi_\pm$ respectively. Using~\eqref{probpure} and \eqref{gappure} we can compute the sub-optimality gaps as,
\begin{align}\label{puregaps}
\Delta^+_x = 0, \quad 
\Delta^+_z = \frac{2+\Delta}{(1+\Delta)^2 + \sqrt{2}(1+\Delta)+1}\Delta, \quad 
\Delta^+_y = \frac{(\Delta+1)(1+\sqrt{2}+\Delta )}{\Delta^2 + (2+\sqrt{2})\Delta +2 +\sqrt{2}} \\
\Delta^-_x = \frac{2-\Delta}{(1-\Delta)^2 + \sqrt{2}(1-\Delta ) + 1}\Delta,\quad 
\Delta^-_z = 0 , \quad 
\Delta^-_y = \frac{-\sqrt{2}\Delta +1+\sqrt{2}}{\Delta^2-(\sqrt{2}+2)\Delta+\sqrt{2}+2}.
\end{align}

\begin{figure}[h]
\centering
\begin{overpic}[percent,width=0.75\textwidth]{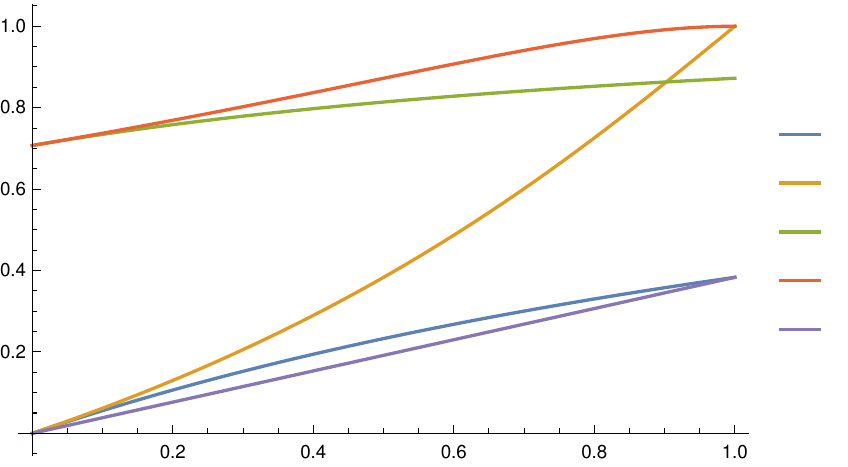}
\put(98,38){$\Delta^+_z$}
\put(98,33){$\Delta^-_x$}
\put(98,27){$\Delta^+_y$}
\put(98,21){$\Delta^-_y$}
\put(98,15){$\frac{3\Delta}{5+2\sqrt{2}}$}
\end{overpic}
\centering
\caption{Sub-optimality gaps for the environments $\psi_+,\psi_-$ and the lower bound $\frac{3}{5+2\sqrt(2)}\Delta $ in the range $\Delta \in [0,1]$.}
\end{figure}
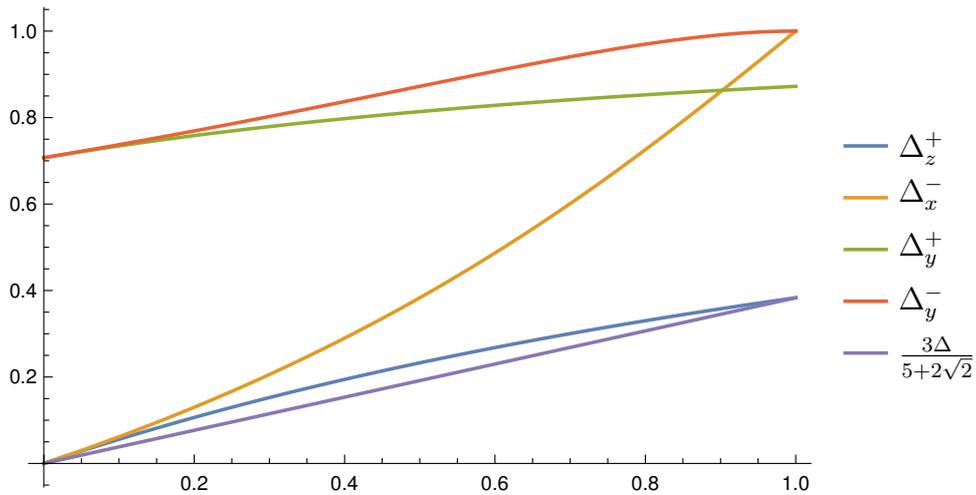

Analysing the behaviour of $\Delta^+_z,\Delta^+_y,\Delta^-_x,\Delta^-_y$ for $\Delta \in [0,1]$ we see that all these quantities are lower bounded by $\frac{3}{5+2\sqrt(2)}\Delta$. In order to simplify the expression we will use that $\frac{3}{5+2\sqrt(2)}\Delta > \frac{3}{10}\Delta  $ and take $\frac{3}{10}\Delta$ as the lower bound.

We start analysing the regret for $\psi_+$. Using the Equations~\eqref{puregaps} for the sub-optimality gaps and the bound $\Delta^+_z,\Delta^+_y\geq \frac{3}{10}\Delta$ we have
\begin{align}
\EX_{\psi_+,\pi} [ R_T(\mathcal{A},\psi_+ , \pi  ) ] = \sum_{i=x,y,z} \EX_{\psi_+ , \pi}[ T_i(T)]\Delta^+_i \geq \frac{3}{10}\Delta \sum_{i=y,z}\EX_{\psi_+ , \pi}(T_i(T)).
\end{align}
Now using $T = \sum_{i=x,y,z} \EX_{\psi_+ , \pi } [T_i (T) ]$ and applying Markov's inequality we have
\begin{align}\label{regret+}
\EX_{\psi_+,\pi} [ R_n(\mathcal{A},\psi_+ , \pi  ) ]\geq \frac{3}{10}\Delta \EX_{\psi_+,\pi} [ n - T_x (n)  ] \geq \frac{3T\Delta}{20}P_{\psi_+ , \pi} \left( T_x (T) \leq \frac{T}{2} \right).
\end{align}
In order to bound the regret for the environment $\psi_-$, we bound just the first term, use Markov inequality, and the bound $\Delta^-_x$ in order to obtain
\begin{align}\label{regret-}
\EX_{\psi_- , \pi} [R_T(\mathcal{A},\psi_-, \pi  ) ] = \sum_{i=x,y,z} \EX_{\psi_- , \pi}[T_i(T)]\Delta^-_i \geq \Delta^-_x \EX_{\psi_- , \pi} [ T_i(T) ] \geq \frac{3T\Delta}{20} P_{\psi_- , \pi} \left( T_x (T) > \frac{T}{2} \right).
\end{align}
Combining~\eqref{regret+} and \eqref{regret-}
\begin{align}
\EX_{\psi_+ , \pi} [ R_T(\mathcal{A},\psi_+ , \pi  )] &+ \EX_{\psi_- , \pi} [R_T(\mathcal{A},\psi_-, \pi  ) ] \\
&\geq \frac{3T\Delta}{20} \left( P_{\psi_+ , \pi} \left( T_x ( T ) \leq \frac{T}{2} \right) + P_{\psi_- , \pi} \left( T_x (T) > \frac{T}{2} \right) \right).
\end{align}
Applying Lemma \ref{pinsker2} together with Lemma \ref{lemma:D1/2ineq} we obtain
\begin{align}\label{discretesumreg2}
\EX_{\psi_+ , \pi} [ R_T(\mathcal{A},\psi_+ , \pi  )] + \EX_{\psi_- , \pi} [R_T(\mathcal{A},\psi_-, \pi  ) ] \geq \frac{3n\Delta}{40}\exp \left( -T D_{\frac{1}{2}} (\psi_+ \| \psi_- )  \right),
\end{align}
where $D_{\frac{1}{2}} (\psi_+ \| \psi_- ) = -\log | \langle \psi_+ | \psi_- \rangle |^2$. The overlap between the two environments can be computed as
\begin{align}
\langle \psi_+ | \psi_- \rangle = 2+\sqrt{2} - \Delta^2,
\end{align}
and we use it to give the following upper bound
\begin{align}\label{logoverlap}
-\log | \langle \psi_+ | \psi_- \rangle |= \log \frac{1}{2+\sqrt{2}-\Delta^2} \leq \frac{\Delta^2 -1 -\sqrt{2}}{2+\sqrt{2}-\Delta^2}\leq \frac{\Delta^2}{2+\sqrt{2}-\Delta^2} \leq \frac{\Delta^2}{1+\sqrt{2}},
\end{align}
where the last inequality follows from $\frac{1}{2+\sqrt{2}-\Delta^2} \leq \frac{1}{1+\sqrt{2}}$ for $\Delta \in [0,1]$. Thus, plugging~\eqref{logoverlap} into~\eqref{discretesumreg2} we have
\begin{align}
\EX_{\psi_+ , \pi} [ R_T(\mathcal{A},\psi_+ , \pi  )] + \EX_{\psi_- , \pi} [R_T(\mathcal{A},\psi_-, \pi  ) ] \geq \frac{3T\Delta}{40}\exp \left( -\frac{2}{1+\sqrt{2}} T\Delta^2 \right).
\end{align}
Finally, if we choose $\Delta = \frac{1}{\sqrt{T}}$,
\begin{align}
\EX_{\psi_+ , \pi} [ R_T(\mathcal{A},\psi_+ , \pi  )] + \EX_{\psi_- , \pi} [R_T(\mathcal{A},\psi_-, \pi  ) ]  \geq \frac{3}{40} \exp \left( -\frac{2}{1+\sqrt{2}}  \right) \sqrt{T},
\end{align}
and the result follows using $ \frac{3}{40} \exp \left( -\frac{2}{1+\sqrt{2}}  \right) \geq \frac{3}{100}$.
\end{proof}

\subsection{Problem dependent lower bound: one qubit environment and rank-1 projectors}\label{sec:discrete1qubit}

The final lower bound we study for discrete action sets has a different flavor from the previous ones. Until now, we have focused on worst-case bounds to establish minimax regret guarantees. In contrast, this last bound aims to provide an explicit expression for the constant in Theorem~\ref{generalower}, which depends on the action set. This results in a \textit{problem-dependent bound} that quantifies the difficulty of the problem as a function of the elements in the action set. To make the analysis more concrete, we restrict our study to environments consisting of \textit{one-qubit quantum states}, with the action set given by a collection of rank-1 projectors
\begin{align}
\mathcal{A} = \left\lbrace \Pi_1, \dots, \Pi_k \right\rbrace,
\end{align}
where each \(\Pi_i\) is a rank-1 matrix satisfying \(\Pi_i^2 = \Pi_i\), or $\Pi\in\mathcal{S}_2^*$.

In order to prove the main result we will need the computation of the relative entropy between two quantum states given by the following Lemma.
\begin{lemma}\label{relativeentropy}
Let $\rho = \frac{\mathbb{I}}{2} + \frac{\Delta}{ 2} \sigma_x$ and $\rho '= \frac{\mathbb{I}}{2} + \frac{\Delta}{ 2} \sigma_z$ be two one-qubit density matrices. Then, their quantum relative entropy can be computed as,
\begin{align}
	D(\rho \| \rho' ) = \frac{\Delta}{2} \log \left( \frac{1+\Delta}{1-\Delta} \right).
\end{align}
\end{lemma}

\begin{proof}
Let $\Pi^+_i , \Pi^-_i$ be the projectors for the $i=x,z$ Pauli matrix into the subspaces of eigenvalue $+1$ and $-1$ respectively. Then we can express the density matrices as,
\begin{align}\label{roproj}
\rho = \frac{1+\Delta}{2}\Pi^+_x + \frac{1-\Delta}{2}\Pi^-_x, \quad \rho'  = \frac{1+\Delta}{2}\Pi^+_z + \frac{1-\Delta}{2}\Pi^-_z,
\end{align}
where we have used that $ \sigma_i = \Pi^+_i - \Pi^-_i   $. Then using that $ \mathbb{I} =  \Pi^+_i + \Pi^-_i$, $\Tr(\sigma_i)=0$ and $\Tr ( \sigma_x \sigma_z ) = 0$ we can compute the following identity
\begin{align}\label{trproj}
\Tr (\Pi^+_i \Pi^-_j) = \left\lbrace\begin{array}{c} \frac{1}{2} \quad \text{if } i\neq j \\ 0 \quad \text{if } i=  j. \end{array}\right.
\end{align}
Then, using~\eqref{roproj} and \eqref{trproj},
\begin{align}
\Tr \rho \log \rho = \frac{1+\Delta}{2}\log\left( \frac{1+\Delta}{2} \right) + \frac{1-\Delta}{2}\log\left( \frac{1-\Delta}{2} \right),
\end{align}
\begin{align}
\Tr \rho \log \rho' = \frac{1}{2}\log\left( \frac{1+\Delta}{2} \right) + \frac{1}{2}\log\left( \frac{1-\Delta}{2} \right).
\end{align}
The result follows from the definition of the relative entropy and rearranging the above terms.

\end{proof}

\begin{theorem}\label{teo:discrete1qubit}
Let $T\in\mathbb{N}$. For any policy $\pi $ and finite action set of observables $\mathcal{A}$ of rank-1 projectors, there exists a one-qubit environment $\rho\in\mathcal{S}_2 $ such that,
\begin{align}
	\EX_{\rho,\pi}[ R_T (\mathcal{A},\rho,\pi ) ]\geq\frac{ \sqrt{1-c} - (1-c) }{30} \sqrt{T}, 
\end{align}
where $c = \Tr ( \Pi_a \Pi_b) $ and $(\Pi_a,\Pi_b) = \argmin_{\Pi_i,\Pi_j\in\mathcal{A}} \max \big\{ \Tr ( \Pi_i\Pi_j ), \Tr \big( \Pi_i(\mathbb{I}-\Pi_j ) \big) \big\} $.
\end{theorem}
Note that that $\sqrt{1-c} \geq 1-c$ since $c \in [0,1]$, and the constant is thus positive and only vanishes for $c \in \{0,1\}$. The value $c=1$ corresponds to the case where the action set has all elements equal, thus the regret always vanishes. The case $c=1$ corresponds to the case where the action set contains the projector $\Pi$ and the orthogonal $\mathbb{I} - \Pi$, which correspond to the same POVM but with the reverse outcomes, thus they are the same measurement. The intermediate cases $c\in(0,1)$ recover the $\sqrt{T}$ behaviour but the constant on the regret depends on the overlap of all projectors contained in the action set.  

\begin{proof}
Let $k = |\mathcal{A}| <\infty$ denote the number of elements in $\mathcal{A}$. We study the case $k\geq 2$ since the case $k=1$ is trivial. We will use the Bloch sphere representation of the actions, i.e, for each $\Pi_i\in\mathcal{A}$ we decompose it as
\begin{align}
\Pi_i = \frac{\mathbb{I}}{2} + \frac{1}{2} \sum_{j=x,y,z} A_{i,j} \sigma_j, \quad i\in[k],
\end{align}
where $A_i\in\mathbb{R}^3$ is a normalized unit vector ($\| A_i \|_2 = 1$) since $\Pi_i\in\mathcal{S}_2^*$ and $\sigma_x,\sigma_y$ and $\sigma_z$ are the qubit Pauli matrices. Let $\mathcal{A}_r=\left\lbrace A_1,...,A_k \right\rbrace$ be the set of Bloch vectors associated to the $k$ elements of $\mathcal{A}$.  Let $\Delta \in [0,\frac{1}{2}]$ be a constant to be chosen later. Pick $A_a,A_b\in \mathcal{A}_r$ such that they are the two directions with the smallest inner product, that is
\begin{align}
( A_a , A_b ) &= \argmin_{A_i,A_j \in \mathcal{A}_r} \vert A_i\cdot A_j \vert 
= \argmin_{\Pi_i,\Pi_j\in\mathcal{A}} \left| \Tr ( \Pi_i\Pi_j ) - \Tr \big( \Pi_i(I-\Pi_j ) \big) \right| \\
&= \argmin_{\Pi_i,\Pi_j\in\mathcal{A}} \max \big\{ \Tr ( \Pi_i\Pi_j ) - \Tr \big( \Pi_i(I-\Pi_j ) \big) , \Tr \big( \Pi_i(I-\Pi_j ) \big) - \Tr ( \Pi_i\Pi_j ) \big\} \\
&= \argmin_{\Pi_i,\Pi_j\in\mathcal{A}} \max \big\{ \Tr ( \Pi_i\Pi_j ), \Tr \big( \Pi_i(I-\Pi_j ) \big) \big\} \,,
\end{align}
where in the last step we used that $ \Tr ( \Pi_i\Pi_j ) + \Tr \big( \Pi_i(\mathbb{I}-\Pi_j ) \big) = 1$.
Since $\mathcal{A}$ contains at least two independent directions, $\bold{r}_a, \bold{r}_b$ determine a plane. Now rotate $A_a , A_b$ symmetrically in this plane until we find two orthogonal directions $A'_a$ and $A'_b$ that obey the following conditions: 
\begin{itemize}
\item $\vert  A'_a \vert = \vert A'_b \vert  = 1$ and $\langle A'_a , A'_b \rangle =0$ (unit orthogonal vectors).
\item $ \langle A_a , A'_a \rangle= \langle A_b ,  A'_b \rangle$ (symmetrical rotation).
\item $A_a \times A_b = A'_a \times A'_b$ (they remain in the same plane).
\item $\max_{A_i\in \mathcal{A}} \langle A'_a , A_i \rangle = \langle A'_a , A_a \rangle$, $\max_{A_i\in \mathcal{A}} \langle A'_b ,  A_i \rangle = \langle A'_b ,  A_b \rangle$ (closest directions that obey the above conditions).
\end{itemize}

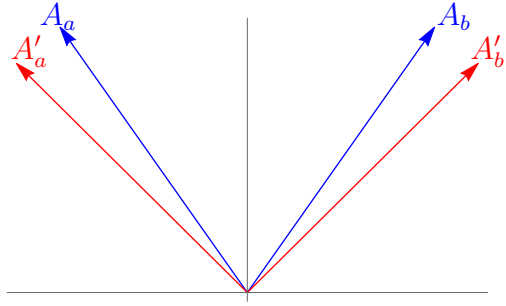
\begin{figure}
\centering
\begin{overpic}[percent,width=0.4\textwidth]{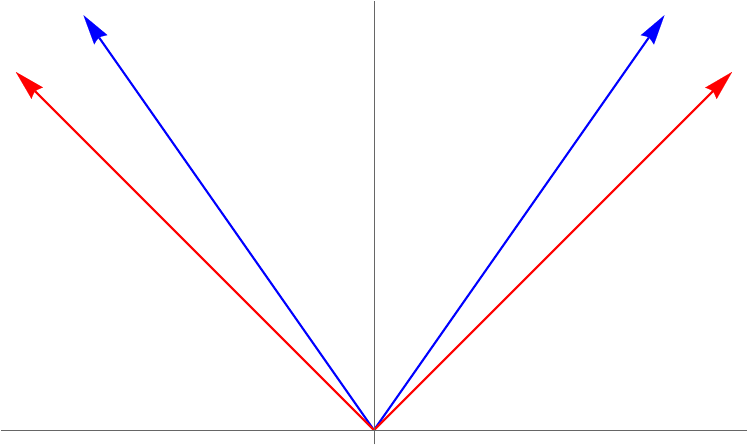}
\put(7,58){\color{blue} $A_a$}
\put(89,58){\color{blue} $A_b$}
\put(1,51){\color{red} $A'_a$}
\put(96,51){\color{red} $A'_b$}
\end{overpic}
\caption{Scheme for the choice of the vectors $A'_a$ and $A'_b$.}
\label{fig:scheme2}
\end{figure}

Using the directions $A'_a$ and $A'_b$ define the following environments $\rho_a , \rho_b$
\begin{align} 
\rho_a := \frac{\mathbb{I}}{2} + \frac{\Delta}{2}A'_a\cdot \sigma \quad \rho_b := \frac{I}{2} + \frac{\Delta}{2}A'_b \cdot \sigma .
\end{align}

Note that $\rho\geq 0$ and $\rho_b\geq 0$. In order to calculate the sub-optimality gaps~\eqref{suboptimality} for $\rho$ and $\rho_b$ we will use an orthonormal system of coordinates defined by the directions $\mathcal{C}_{a,b}= \left\lbrace A'_a , A'_b ,A'_a \times A'_b \right\rbrace$. Using the above conditions for $A'_a $ and $A'_b$ we have
\begin{align}\label{suboptimalityab}
\Delta^a_i = \Tr (\rho\Pi_i) = \Delta ( p - A_{ia} ), \quad \Delta^b_i= \Tr(\rho_b\Pi_b ) = \Delta ( p - A_{ib} ) \quad \text{for } i = 1,...,k,
\end{align}
where $p = \langle A_a , A'_a \rangle$ and $r_{ia} , r_{ib}$ are the first and second coordinates of $A_i \in \mathcal{A}$ in the coordinate system  $\mathcal{C}_{a,b}$ respectively. This last fact is because we defined $\mathcal{C}_{a,b}$ to be a orthonormal set of coordinates. Note that by the construction of $A'_a$ and $A'_b$ we have
\begin{align}
p = \cos \left( \frac{\pi}{4} - \frac{\theta}{2}\right),
\end{align}
where $\theta$ is defined via $\cos\theta = \langle A_a , A_b \rangle$. Let $\Pi_a,\Pi_b$ be the projectors associated to $\bold{r}_a,\bold{r}_b$, then using the trigonometric identity for $\cos\left(\frac{a}{2}\right)$ we have $\Tr(\Pi_a\Pi_b ) = \cos^2\left(\frac{\theta}{2} \right)$. Thus, using the trigonometric identity for $\cos(a+b)$ we have
\begin{align}\label{constantp}
p = \frac{1}{\sqrt{2}}\left( \sqrt{\Tr(\Pi_a \Pi_b}) + \sqrt{1-\Tr(\Pi_a \Pi_b})\right).
\end{align}
Now define the following subsets of indices of $[k]$,
\begin{align}
\mathcal{N}_a = \left\lbrace i \colon A^2_{ia} \geq \frac{1}{2} \text{ for } A_i \in \mathcal{A}_r\right\rbrace \quad \mathcal{N}_a^C = \left\lbrace i \colon A^2_{ia} < \frac{1}{2} \text{ for } A_i \in \mathcal{A}_r \right\rbrace.
\end{align}
Note that this sets are complementary, $\mathcal{N}_a \cap \mathcal{N}_a^C =\emptyset$ and $\mathcal{N}_a \cup \mathcal{N}_a^C = [k]$. Using the expression for the regret~\eqref{regretarms} and the sub-optimality gap~\eqref{suboptimalityab}, the regret for $\rho $ can be bounded as,
%
%\begin{equation} R_n ( \rho_a , \pi,\mathcal{A} ) = \Delta \sum_{i=1}^k ( p - r_{ia} ) \EX_{\rho_a, \pi} (T_i (n ) ) \geq \Delta \sum_{i\in \mathcal{N}_a^C} ( p - r_{ia} ) \EX_{\rho_a , \pi} (T_i (n ) )   
%\end{equation}
%
%\begin{equation} 
%\geq \Delta  \left( p - \frac{1}{\sqrt{2}} \right) \sum_{i\in   \mathcal{N}_a^C} \EX_{\rho_a , \pi} (T_i (n ) ), 
%\end{equation}
%
\begin{align} 
	\EX_{\rho_a , \pi} [ R_T (\mathcal{A}, \rho_a , \pi) ] &= \Delta \sum_{i=1}^k ( p - A_{ia} ) \EX_{\rho_a, \pi} [T_i ( T ) ] \\
	&\geq \Delta \sum_{i\in \mathcal{N}_a^C} ( p - A_{ia} ) \EX_{\rho_a , \pi} [T_i (T ) ]   \\
	&\geq \Delta  \left( p - \frac{1}{\sqrt{2}} \right) \sum_{i\in   \mathcal{N}_a^C} \EX_{\rho_a , \pi} [T_i ( T ) ], 
\end{align}
where we have used that for $i\in  \mathcal{N}_a^C$, we have $A_{ia} \leq \frac{1}{\sqrt{2}}$. Using that $ T = \sum_{i=1}^k  \EX_{\pi, \rho_a} (T_i ( T ) ),$
\begin{align}
 \EX_{\rho_a, \pi } [ R_T (\mathcal{A}, \rho_a , \pi) ] \geq \Delta  \left( p - \frac{1}{\sqrt{2}}\right) \EX_{ \rho_a , \pi} \left[ T - \sum_{i\in   \mathcal{N}_a} T_i (T) \right]. 
\end{align}
Applying Markov's inequality to the above expression we have,
\begin{align} 
\EX_{\rho_a,\pi} [R_T ( \mathcal{A},\rho_a , \pi ) ] \geq \frac{ T \Delta}{2} \left( p - \frac{1}{\sqrt{2}} \right) P_{ \rho_a , \pi  } \left(   T - \sum_{i\in   \mathcal{N}_a} T_i (T ) \geq \frac{T}{2} \right).   
\end{align}
Rearranging all the terms,
\begin{align}\label{regret1rho}
\EX_{\rho_a,\pi} [ R_T ( \mathcal{A},\rho_a , \pi ) ]  \geq \frac{T \Delta}{2} \left( p - \frac{1}{\sqrt{2}} \right) P_{\rho_a , \pi} \left(  \sum_{i\in   \mathcal{N}_a} T_i (T ) \leq \frac{T}{2} \right).
\end{align}
In order to bound the regret for $\rho_b$ note that if $i \in \mathcal{N}_a$ then $A^2_{ib} \leq \frac{1}{2}$ since $\| A_i \| = 1$. Using the same tricks as before,
\begin{align} 
	\EX [R_T (\mathcal{A}, \rho_b , \pi ) ]  &\geq \Delta \sum_{i\in \mathcal{N}_a} ( p - A_{ib} ) \EX_{\pi , \rho_b } [T_i ( T ) ]\\ 		&\geq \Delta  \left( p - \frac{1}{\sqrt{2}} \right) \sum_{i\in   \mathcal{N}_a} \EX_{\pi , \rho_b} [T_i ( T )]. 
\end{align}
Using again Markov's inequality,
\begin{align}\label{regret2rho}
\EX_{\rho_b , \pi} [ R_T (\mathcal{A}, \rho_b , \pi ) ] \geq \frac{ T \Delta}{2} \left( p - \frac{1}{\sqrt{2}} \right) P_{ \rho_b , \pi } \left(  \sum_{i\in   \mathcal{N}_a} T_i (T ) > \frac{T}{2} \right).
\end{align}
Thus, combining~\eqref{regret1rho} and \eqref{regret2rho},
\begin{align}
\EX_{\rho_a , \pi}[R_T ( \mathcal{A},\rho_a , \pi )]  + \EX_{\rho_b , \pi} [ R_T ( \mathcal{A}, \rho_b , \pi ) ] \geq   \nonumber
\end{align}
\begin{align}
\frac{T \Delta}{2} \left( p - \frac{1}{\sqrt{2}} \right) \left(   P_{ \rho_a , \pi } \left(  \sum_{i\in   \mathcal{N}_a} T_i ( T ) \leq \frac{T}{2} \right) +  P_{ \rho_b , \pi } \left(  \sum_{i\in   \mathcal{N}_a} T_i ( T ) > \frac{T}{2} \right) \right).
\end{align}
Using Lemma~\ref{pinsker}, we can bound the above expression as,
\begin{align}\label{sumregret}
\EX_{\rho_a , \pi}[R_T ( \mathcal{A},\rho_a , \pi )]  + \EX_{\rho_b , \pi} [ R_T ( \mathcal{A}, \rho_b , \pi ) ] \geq  \frac{ T \Delta}{4} \left( p - \frac{1}{\sqrt{2}}\right)\exp \left(- D(P_{\pi , \rho_a  } \| P_{\pi , \rho_b}) \right).
\end{align}
Using Lemma~\ref{divergence} combined with the data-processing inequality, we can bound the Kullback–Leibler divergence as,
\begin{align}
 D(P_{\pi , \rho_a  } \| P_{\pi , \rho_b}) &= \sum_{i=1}^k \EX_{\pi , \rho_a  } [T_i ( T ) ] D \big( P_{\rho_a}(\cdot | i) \big\| P_{\rho_b}(\cdot | i) \big) ) \\
 &\leq D ( \rho_a \| \rho_b ) \sum_{i=1}^k \EX_{\rho , \pi } [T_i ( T ) ] = T  D ( \rho_a \| \rho_b )  .
\end{align}
Using that the relative entropy $D ( \rho_a  \| \rho_b )$ is unitarily invariant and ${A}'_a,{A}'_b$ are orthogonal, we can use the computation of the relative entropy given by Lemma \ref{relativeentropy}. Thus we can lower bound~\eqref{sumregret} as,
\begin{align}
\EX_{\rho_a , \pi}[R_T ( \mathcal{A},\rho_a , \pi )]  + \EX_{\rho_b , \pi} [ R_T ( \mathcal{A}, \rho_b , \pi ) ] \geq \frac{T \Delta}{4} \left( p - \frac{1}{\sqrt{2}}\right)\exp \left(  - \frac{T\Delta}{2} \log \left( \frac{1+\Delta}{1-\Delta} \right) \right).
\end{align}
Choosing $\Delta = \frac{1}{2\sqrt{T}}$ we have the following upper bound,
\begin{align}
 \log \left( \frac{1+\Delta}{1-\Delta} \right) \leq 2 \log (3 ) \Delta \quad \text{for }\Delta \in [0,1/2],
\end{align}
where we have used that $\log \left( \frac{1+x}{1-x} \right) $ is convex and finite for $x\in [ 0,\frac{1}{2} ]$. Thus, choosing $\Delta = \frac{1}{2\sqrt{T}}$ we have
\begin{align}
 \EX_{\rho_a , \pi}[R_T ( \mathcal{A},\rho_a , \pi )]  + \EX_{\rho_b , \pi} [ R_T ( \mathcal{A}, \rho_b , \pi ) ] \geq  \frac{ p - \frac{1}{\sqrt{2}}}{8\cdot 3^{1/4}} \sqrt{T},
\end{align}
and the result follows using the expression for $p$ \eqref{constantp} and $15 > 8\cdot 3^{1/4}\sqrt{2}$. From~\eqref{constantp} is easy to check that $p\geq \frac{1}{\sqrt{2}}$, thus the above lower bound is positive.
\end{proof}

\section{Continuous lower bounds}

We now turn to the problem of establishing minimax regret lower bounds for the general bandit setting introduced in Section~\ref{sec:general_bandits}.

\subsection{Rank-1 projection measurements for mixed state environments}\label{sec:AllPureStates}

In the following theorem, we consider a bandit with a mixed state environment whose action set is the set of all rank-1 projections, i.e.,\ $\mc A=\mc S_d^*$ where $d\in\N$ is the dimension of the Hilbert space. When we play the action $|\fii\>\!\<\fii|\in \mathcal{S}_d^*$ on the environment state $\rho$, the probability of reward 1 is $\<\fii|\rho|\fii\>$ and the probability of reward 0 is $1-\<\fii|\rho|\fii\>$. %In this setting, we restrict to those policies $\pi$ such that the $\gamma_{\rho,\pi}$-essential supremum of $\<\fii|\rho|\fii\>|$ over all $|\fii\>\!\<\fii|\in\mc S_d^*$ (recall the definition of the measure $\gamma_{\rho,\pi}$ in Section~\ref{subsec:generalbandits})
%Denoting by $\lambda_{\rm max}(\rho)$ the highest eigenvalue of any environment state $\rho$, 
The sub-optimality gap is $\Delta_\fii:=\lambda_{\rm max}(\rho)-\<\fii|\rho|\fii\>$ for the environment state $\rho$ upon playing the action $|\fii\> \! \langle \fii|$.

\begin{theorem}\label{theor:AllPureStates}
Let $T,\,d\in\N$. For any policy $\pi$ and action set of observables $\mathcal{A}$ containing all rank-1 projections, i.e.,\ $\mc A=\mc S_d^*$, there exists an environment $\rho\in\mc S_d$ such that
\begin{align}
 \EX_{\rho , \pi} [ R_T(\mc A,\rho,\pi) ] \geq C \sqrt{T}
\end{align}
for some constant $C >0$.
\end{theorem}

\begin{proof}
Let us fix a policy $\pi$ and an orthonormal basis $\{|n\>\}_{n=0}^{d-1}$ for $\C^d$ and define $|\psi\>:=d^{-1/2}(|0\>+\cdots+|d-1\>)$ and the sets
\begin{align}
\mc N_1&:=\left\{|\eta\>\!\<\eta|\in\mc S_d^*\,\middle|\,|\<0|\eta\>|^2<\frac{3}{4}+\frac{1}{4d}\right\},\\
\mc N_2&:=\left\{|\eta\>\!\<\eta|\in\mc S_d^*\,\middle|\,|\<\psi|\eta\>|^2<\frac{3}{4}+\frac{1}{4d}\right\}.
\end{align}
Let us show that these sets are disjoint. First note that $P(|0\>\!\<0|,|\psi\>\!\<\psi|)=\sqrt{1-1/d}$ where $P=\sqrt{1-F^2}$ is the purified distance where, in turn, $F$ is the fidelity. Assume that $|\eta\>\!\<\eta|\in\mc N_1$, so that, using the triangle inequality for the purified distance,
\begin{align}
\sqrt{1-\frac{1}{d}} &= P(|0\>\!\<0|,|\psi\>\!\<\psi|)\leq P(|0\>\!\<0|,|\eta\>\!\<\eta|)+P(|\eta\>\!\<\eta|,|\psi\>\!\<\psi|)\\
&< \frac{1}{2}\sqrt{1-\frac{1}{d}}+P(|\eta\>\!\<\eta|,|\psi\>\!\<\psi|)
\end{align}
where we have used the definition of $\mc N_1$ in the final inequality. Thus, $P(|\eta\>\!\<\eta|,|\psi\>\!\<\psi|)>(1/2)\sqrt{1-1/d}$ which is easily seen to imply $|\eta\>\!\<\eta|\in\mc N_2^c$. Thus, $\mc N_1\cap\mc N_2=\emptyset$.

Let us now define 
\begin{align}
\rho:= \frac{1-\Delta}{d}\mathbb{I}+\Delta|0\>\!\<0| \, ,
\end{align}
where $\Delta\in[0,1]$ is a constant to be determined later. For this environment, the sub-optimality gap is
\begin{align}
\Delta_\eta=\Delta(1-|\<0|\eta\>|^2).
\end{align}
We may now evaluate the decomposition of the regret for general bandits~\eqref{eq:contRegr2}
\begin{align}
\EX_{\rho , \pi} [R_T(\mc A,\rho,\pi) ] &\geq \Delta\int_{\mc N_1^c}(1-|\<0|\eta\>|^2)\,d\gamma_{\rho,\pi}(|\eta\>\!\<\eta|)\geq\Delta\frac{d-1}{4d}\gamma_{\rho,\pi}(\mc N_1^c).
\end{align}
Let us define another state 
\begin{align}
\rho':= \frac{1-\Delta}{d}\mathbb{I}+\Delta|\psi\>\!\<\psi| \, ,
\end{align}
where $|\psi\>$ is the unit vector defined earlier. Similarly as above, we find
\begin{align}
\EX_{\rho' , \pi} [ R_T(\mc A,\rho',\pi) ] &\geq \Delta\int_{\mc N_2^c}(1-|\<\psi|\eta\>|^2)\,d\gamma_{\rho',\pi}(|\eta\>\!\<\eta|)\geq\Delta\frac{d-1}{4d}\gamma_{\rho',\pi}(\mc N_2^c)\\
&\geq \Delta\frac{d-1}{4d}\gamma_{\rho',\pi}(\mc N_1)
\end{align}
where the final inequality follows from $\mc N_1\subseteq\mc N_2^c$; recall that $\mc N_1$ and $\mc N_2$ are disjoint. Recalling that $(1/T)\gamma_{\rho,\pi}$ and $(1/T)\gamma_{\rho',\pi}$ are probability measures and using Lemma \ref{pinsker}, we now obtain
\begin{align}\label{eq:summaepayhtalo}
\EX_{\rho , \pi} [ R_T(\mc A,\rho,\pi) ]+ \EX_{\rho' , \pi} [R_T(\mc A,\rho',\pi) ] &\geq\Delta\frac{d-1}{4d}T\left(\frac{1}{T}\gamma_{\rho,\pi}(\mc N_1^c)+\frac{1}{T}\gamma_{\rho',\pi}(\mc N_1)\right) \\
&\geq\Delta\frac{d-1}{4d}T\exp{\left(-D\left(\frac{1}{T}\gamma_{\rho,\pi}\middle\|\frac{1}{T}\gamma_{\rho',\pi}\right)\right)}.
\end{align}
Recall that, e.g.,\ $\gamma_{\rho,\pi}=\sum_{t=1}^T\big(P_{\rho,\pi}^{(t)}(\cdot,0)+P_{\rho,\pi}^{(t)}(\cdot,1)\big)$. Defining, for all $t=1,\ldots,\,T$, the (measurable) function $f_t:(\mc A\times\{0,1\})^T\to\mc A$ through
\begin{align}
f_t\big((|\eta_s\>\!\<\eta_s|,x_s)_{s=1}^T\big)=|\eta_t\>\!\<\eta_t|,
\end{align}
we now see that $\gamma_{\rho,\pi}=\sum_{t=1}^T P_{\rho,\pi}\circ f_t^{-1}$ where we view $P_{\rho,\pi}$ as a measure. Using the joint convexity of the Kullback-Leibler divergence and the data processing inequality, we finally get
\begin{align}
D\left(\frac{1}{T}\gamma_{\rho,\pi}\middle\|\frac{1}{n}\gamma_{\rho',\pi}\right)
&\leq \frac{1}{T}\sum_{t=1}^T D(P_{\rho,\pi}\circ f_t^{-1}\|P_{\rho',\pi}\circ f_t^{-1})\\
&\leq \frac{1}{T}\sum_{t=1}^T D(P_{\rho,\pi}\|P_{\rho',\pi})=D(P_{\rho,\pi}\|P_{\rho',\pi})\\
&\leq T D(\rho\|\rho')
\end{align}
where the final inequality follows from Lemma \ref{lemma:D1/2ineq} (for $\alpha=1$). Combining this with~\eqref{eq:summaepayhtalo} gives us
\begin{align}
\EX_{\rho , \pi} [ R_T(\mc A,\rho,\pi) ]+ \EX_{\rho' , \pi} [R_T(\mc A,\rho',\pi) ] \geq\Delta\frac{d-1}{4d}T\exp{\big(-T D(\rho\|\rho')\big)}.
\end{align}
Whenever $\Delta\leq 1/2$, we may follow the end of the proof of Theorem \ref{generalower} and obtain $c>0$ such that $D(\rho\|\rho')\leq (c/2)\Delta^2$. Fixing $\Delta=1/\sqrt{T}$, we now have
\begin{align}
\max\{ \EX_{\rho , \pi} [ R_T(\mc A,\rho,\pi) ] , \EX_{\rho' , \pi} [R_T(\mc A,\rho',\pi) ] \} &\geq\frac{1}{2}\big(\EX_{\rho , \pi} [ R_T(\mc A,\rho,\pi) ]+ \EX_{\rho' , \pi} [R_T(\mc A,\rho',\pi) ] \big) \nonumber \\ 
&\geq\frac{d-1}{8d}\sqrt{T}e^{-c/2}.
\end{align}
\end{proof}

\subsection{PSMAQB lower bound}\label{sec:lower_bound}
% \JL{this part has to be adapted, it is just copy paste form mikhail notes}

%While the algorithm for PSMAQB presented above is inspired by classical bandit theory, the lower bound on the regret that we derive is essentially based on quantum information theory. The key insight here is that a policy for PSMAQB can be viewed as a sequence of state tomographies. The expected fidelity of these tomographies is linked to the regret. Hence, existing upper bounds on tomography fidelity also provide a lower bound for the expected regret of the policy. 

In the discrete case, we observed that the \(\sqrt{T}\) scaling of regret remains unchanged regardless of whether the environment is mixed or pure. The proof techniques were similar, as they relied on the Bretagnolle-Huber inequality (or its modified version in Lemma~\ref{pinsker2}). Given this, one might expect that the same behavior holds for general bandits and that the proof would closely follow that of Theorem~\ref{theor:AllPureStates}.  

As in the previous section, we consider an action set \(\mathcal{A} = \mathcal{S}_d^*\) and an environment represented by a pure state \(\ket{\psi} \! \bra{\psi} \in \mathcal{S}_d^*\). This is the PSMAQB setting that we defined in~\ref{def:PSMAQB}. However, applying the previous techniques in this setting results in only a constant lower bound. To illustrate this, consider the following example.

We first fix the reference state \(\ket{0}\) in some basis and define two possible environments
\begin{align}
\ket{\psi} := \ket{0}, \quad \ket{\psi'} := \cos (\Delta)\ket{0} + \sin (\Delta) \ket{1},
\end{align}
where \(\Delta \in [0,2\pi)\) is a constant. We then select a subset of the action set, \(\mathcal{N} \subset \mathcal{S}_d^*\), consisting of states that are closer to \(\ket{\psi}\) than to \(\ket{\psi'}\), formally defined as  
\begin{align}
\mathcal{N} := \left\lbrace \ket{v} \! \bra{v} : F (v,\psi ) \geq F ( v , \psi' ) \right\rbrace.
\end{align}
Following similar steps as in the proof of Theorem~\ref{theor:AllPureStates} and using Lemma~\ref{pinsker2} (with \(\alpha = 1/2\)), we obtain the bound  
\begin{align}
\mathbb{E}_{\psi , \pi} [R_T (\mathcal{A} , \psi , \pi )] + \mathbb{E}_{\psi' , \pi} [R_T (\mathcal{A} , \psi' , \pi )] \geq \frac{T\sin^2(\Delta)}{8} \exp \left( - T \frac{\sin^2(\Delta)}{1 - \sin^2 (\Delta )} \right).
\end{align}
Note that the bound $\alpha = 1/2$ is tighter than $\alpha = 1$. By choosing \(\sin^2 (\Delta ) \sim \frac{1}{T}\), this results in a constant lower bound. The reason for this lower bound stems from the fact that the sub-optimality gap for the environment \(\ket{\psi}\), given by  
\begin{align}
\Delta_{\eta} = 1 - |\langle \psi | \eta \rangle |^2,  
\end{align}
for elements in the complement set \(\mathcal{N}^C\), behaves as \(\Delta_{\eta} \sim \sin^2 (\Delta )\). This coincides with the divergence  
\begin{align}\label{eq:constant_d12}
D_{\frac{1}{2}} \left( \psi \| \psi' \right) \sim \sin^2 (\Delta ).
\end{align}
In all our previous results, the relationship between these two quantities was quadratic, which allowed us, by choosing appropriate parameters, to achieve the usual \(\sqrt{T}\) regret scaling.  

The natural question that follows is whether a more refined calculation can recover the standard \(\sqrt{T}\) lower bound, or if there is something fundamental in this setting that enables breaking the square-root barrier.

In order to establish a non-trivial lower bound in this setting, we use the key insight that a policy for PSMAQB can be understood as a sequence of state tomography procedures. This follows from the fact that the regret can be expressed as
\begin{align}\label{eq:regret_purestate_lowerbound}
    R_T (\mathcal{A} , \psi , \pi ) = \sum_{t=1}^T \left( 1 - \langle \psi | \Pi_t | \psi \rangle \right),
\end{align}
where \(\Pi_t = \ket{\psi_t} \! \bra{\psi_t} \in \mathcal{A} = \mathcal{S}_d^*\) is the action selected at time step \(t \in [T]\).  

Minimizing the regret means choosing \(\ket{\psi_t}\) so that it is close to \(\ket{\psi}\) in terms of infidelity distance. Since the expected fidelity of these tomography procedures is directly related to the regret, existing upper bounds on tomography fidelity also provide a lower bound for the expected regret of the policy. Our lower bound proof relies on the fact that individual actions of a strategy at time $t\in[T]$ can be viewed as quantum state tomography algorithms using $t$ copies of the state. A relation between the optimal fidelity achievable by these algorithms and the regret of the strategy allows us to convert the fidelity upper bound from~\cite{hayashi2005reexamination} to a regret lower bound. We use measure-theoretic tools to adapt the proof from~\cite{hayashi2005reexamination} to a more general case where the tomography can output arbitrary distribution of states. In the following we provide the necessary tools to connect a quantum state tomography protocol with our model and a regret lower bound.

\subsubsection{Average fidelity bound for pure state tomography}

In its most general form, a tomography procedure takes $n$ copies of an unknown state $\Pi = |\psi\rangle\!\langle \psi |\in \mathcal{S}_d^*$ and performs a joint measurement on the state $\Pi^{\otimes n}$. This is captured in the following definition. Let $(\mathcal{S}_d^*, \Sigma)$ be a $\sigma$-algebra. A \emph{tomography scheme} is a POVM $\mathcal{T}:\Sigma\to \End(\hil^{\otimes n})$ such that $\mathcal{T}(\mathcal{S}_d^*)=\Pi^+_n$, where $\Pi^+_n$ is the symmetrization operator on $\hil^{\otimes n}$. For any $\rho\in \End(\hil^{\otimes n})$, this POVM gives rise to a complex-valued measure
\begin{equation}
    P_{\mathcal{T},\rho}(A)=\Tr(\mathcal{T}(A)\rho)
\end{equation}
for $A\in\Sigma$. $P_{\mathcal{T},\rho}$ becomes a probability measure if $\rho$ satisfies $\rho\ge 0,\ \Pi^+_n\rho=\rho\Pi^+_n=\rho$, and $\Tr \rho=1$. Given $n$ copies of $\Pi$, the tomography scheme produces the distribution $P_{\mathcal{T},\Pi^{\otimes n}}$ of the predicted states. Note that $\Pi^{\otimes n}$ satisfies the properties above, so $P_{\mathcal{T},\Pi^{\otimes n}}$ is indeed a probability distribution. The fidelity of this distribution is given by
\begin{equation}
    F(\mathcal{T}, \Pi)=\int \Tr(\Pi \sigma) dP_{\mathcal{T},\Pi^{\otimes n}}(\sigma). \label{eq:avg_fidelity}
\end{equation}
Finally, the average fidelity of the tomography scheme is defined as
\begin{equation}
    F(\mathcal{T})=\int F(\mathcal{T}, |\psi\>\<\psi|)d\psi,
\end{equation}
where the integration is taken with respect to the normalized uniform measure over all pure states. In the following, $\int d\psi$ will always imply this measure. We will provide a lower bound on $F(\mathcal{T})$ in terms of $d$ and $n$, following the proof technique from~\cite{hayashi2005reexamination}. In~\cite{hayashi2005reexamination}, the proof is only presented for tomography schemes producing a finite number of predictions. For our definition, we will require more general measure-theoretic tools. Before we introduce the upper bound on the fidelity, we will need two auxiliary lemmas about the nature of the measure $P_{\mathcal{T},\rho}$. The proofs of these Lemmas are left to Appendix~\ref{ap:chap3}.

\begin{lemma}
  \label{lem:radon}
  Let $(\Omega, \Sigma)$ be a  $\sigma$-algebra, and let $O: \Sigma\to \End(\widetilde{\hil})$ be a POVM with values acting on a finite-dimensional Hilbert space $\widetilde{\hil}$ with $\operatorname{dim} \thil=\tdim$ s.t. $O(\Omega)\le\mathbb{I}$, where $\mathbb{I}$ is the identity operator. Further, let $P_{O,\sigma}: \Sigma\to\mathbb{C}$ be a complex-valued measure, defined for any $\sigma\in \End(\widetilde{\hil})$ as 
\begin{equation}
  P_{O,\sigma}(A)=\Tr[O(A)\sigma].
  \label{pdef_alt}
\end{equation}
Then, there exists a set of functions $\{f_\sigma\}$ indexed by $\sigma\in \End \widetilde{\mathcal{H}}$ that are linear w.r.t. $\sigma$ for all $\omega$ and that satisfy
\begin{equation}
  f_\sigma: \Omega\to \mathbb{C}\quad\text{s.t.}\quad \forall A\in\Sigma\ \ P_{O,\sigma}(A)=\int_A f_\sigma(\omega) dP_{O,\id}(\omega). \label{eq:radon_derivatives}
\end{equation}
  % Let $\sigma\in \End(H)$ for a finite-dimensional Hilbert space $H$, and let $P_\sigma$ be a complex-valued measure defined on $\tilde{\Sigma}$ by~\eqref{pdef}, where $\tilde{\Sigma}$ is the $\sigma$-algebra on $S(H)$ ``pushed'' by the tomography map $\Psi$ from the original $\Sigma$. Then, $P_\sigma$ is dominated by the finite nonnegative measure  $P_\mathbbm{1}$, also defined by~\eqref{pdef}, where $\mathbbm{1}\in S(H)$ is the identity operator. 
\end{lemma}

We purposefully formulated this Lemma with slightly more general objects than the ones used in the definition of tomography. That is, $\Omega$ does not need to be $\mathcal{S}_d^*$, and $\widetilde{\hil}$ does not need to be the n-th power $\hil^{\otimes n}$, although we will focus on this case.

Note that for $\sigma\ge 0$, the measure $P_{O,\sigma}$ is finite and nonnegative, but nonnegativity (and even real-valuedness) do not hold for a general $\sigma\in \End(\widetilde{\hil})$. 
% This lemma, applied to $\Omega=\mathcal{S}_d^*$ and $\widetilde{\hil}=\hil^{\otimes n}$, allows us to introduce $f_\sigma:\mathcal{S}_d^*\to \mathbb{C}$, which satisfies
% \begin{equation}
%   \forall A\in\Sigma\ \ P_{\mathcal{T},\sigma}(A)=\int_A f_\sigma(\omega) dP_{\mathcal{T},\mathbbm{1}}(\omega).
% \end{equation}
By our definition of $f_\sigma(\omega)$, it can be written as
\begin{equation}
    f_\sigma(\omega)=\Tr\left[K(\omega)\sigma\right],\quad\text{where}\ K(\omega)=\sum_{i,j=1}^{\text{dim}(\tilde{\mathcal{H}})}f_{|i\>\<j|}(\omega)|j\>\<i|. \label{eq:introk}
\end{equation}
As the following lemma demonstrates, $K(\omega)\ge 0$ for $P_{O,\id}$-almost every $\omega$:

\begin{lemma}
  \label{lem:magic}
  Let $(\Omega, \Sigma, \mu)$ be a measurable space and $V: \Omega\to \End(\thil)$ be a measurable operator-valued function with values acting on a finite-dimensional Hilbert space $\thil$ such that
\begin{equation}
  \forall A\in\Sigma\quad \int_A V(\omega)d\mu(\omega) \ge 0.
\end{equation}
Then, $V(\omega) \ge 0$ $\mu$-almost everywhere.
\end{lemma}

Now we can apply this analysis to the POVM corresponding to our tomography scheme, and get the desired upper bound on the fidelity.

\begin{theorem}\label{th:fidelity_upper_bound}
  For any tomography scheme $\mathcal{T}$ utilizing $n$ copies of the input state, the average fidelity is bounded by
\begin{equation}
  F(\mathcal{T})\le \frac{n+1}{n+d}.
\end{equation}
\end{theorem}

\begin{proof}
  We will introduce the density $K(\omega)$ from~\eqref{eq:introk} for our tomography scheme $\mathcal{T}$ and the corresponding measure $P_{\mathcal{T},\sigma}$. Lemma~\ref{lem:radon} allows us to introduce for any $\sigma\in\End(\hil^{\otimes n})$ the density $f_\sigma:\Omega\to\mathbb{C}$ s.t.
\begin{equation}
 \forall A\in\Sigma\ \ P_{\mathcal{T},\sigma}(A)=\int_A f_\sigma(\omega) dP_{\mathcal{T},\mathbbm{1}}(\omega).
\end{equation}
This density can be written as $f_\sigma(\omega)=\Tr\left(K(\omega)\sigma\right)$ 
for some $K(\omega)\in\End(\hil^{\otimes n})$. $K(\omega)$ can be considered as the operator-valued density of $\mathcal{T}$ w.r.t. $P_{\mathcal{T},\id}$:
\begin{equation}\label{eq:opdensity}
    \forall A\in\Sigma\quad \mathcal{T}(A)=\int_A K(\omega)dP_{\mathcal{T},\id}(\omega).
\end{equation}
Since $\mathcal{T}(A)\ge 0$, it follows by Lemma~\ref{lem:magic} that $K(\omega)\ge 0$ for $P_{\mathcal{T},\id}$-almost all $\omega$. Furthermore, as $\mathcal{T}(\mathcal{S}_d^*)=\Pi^+_n$, we have that for all $A\in\Sigma$,
$\mathcal{T}(A)\le \Pi^+_n$. Therefore, $\mathcal{T}(A)\Pi^+_n=\Pi^+_n \mathcal{T}(A)=\mathcal{T}(A)$. This means that $\tilde{K}(\omega)=\Pi^+_nK(\omega)\Pi^+_n$ would also satisfy~\eqref{eq:opdensity}. In the following, we will without loss of generality assume that
\begin{equation}
    K(\omega)=\Pi^+_nK(\omega)=K(\omega)\Pi^+_n.
\end{equation}

With these tools at hand, we are ready to adapt the proof from~\cite{hayashi2005reexamination} to the general case of POVM tomography schemes. We begin by rewriting the expression~\eqref{eq:avg_fidelity} for average fidelity:
\begin{align}
  F(\mathcal{T}) &= \int d\psi \int dP_{\mathcal{T},(\ket{\psi}\!\bra{\psi})^{\otimes n}}(\sigma) \Tr(\sigma\ket{\psi}\!\bra{\psi}) \\
                 &= \int d\psi \int dP_{\mathcal{T},\mathbb{I}}(\sigma)\Tr(\ket{\psi}\!\bra{\psi}\sigma)\Tr\left( K(\sigma)(\ket{\psi}\!\bra{\psi})^{\otimes n} \right). \label{eq:fidelity1}
\end{align}
Since fidelity is nonnegative and its average is bounded by 1, we can change the order of integration. Following~\cite{hayashi2005reexamination}, we introduce notation
\begin{equation}
  \sigma_n(k)=\mathbb{I}^{\otimes (k-1)}\otimes \sigma \otimes \mathbb{I}^{\otimes (n-k)}\in \hil^{\otimes n}.
\end{equation}
The product of traces in~\eqref{eq:fidelity1} can be rewritten in the following manner:
\begin{equation}
    F(\mathcal{T})=\int dP_{\mathcal{T},\id}(\sigma)\int d\psi \Tr\left( (K(\sigma)\otimes \mathbb{I})(\ket{\psi}\!\bra{\psi})^{\otimes (n+1)}\sigma_{n+1}(n+1) \right). \label{eq:fidelity2}
\end{equation}
We can now take the inner integral in closed form. As shown in~\cite[Eq.~(4)]{hayashi2005reexamination},
\begin{equation}
  \int d\psi (\ket{\psi}\!\bra{\psi})^{\otimes n}=\frac{\Pi^+_n}{D_n},
  \label{eq:avgpsi}
\end{equation}
where $D_n=\binom{n+d-1}{d}$. Another useful result in this paper is~\cite[Eq.~(8)]{hayashi2005reexamination}:
\begin{equation}
  \Tr_{n+1}\left( \Pi^+_{n+1}\sigma_{n+1}(n+1) \right)=\frac{1}{n+1}\Pi^+_n\left(\mathbb{I}+\sum_{k=1}^n\sigma_n(k)\right),
  \label{eq:trs}
\end{equation}
where $\Tr_{n+1}:\End(\hil^{\otimes (n+1)})\to\End(\hil^{\otimes n})$ is the partial trace on the $(n+1)$-st copy of the system. These expressions allow us to rewrite~\eqref{eq:fidelity2} as follows:
\begin{align}
  F(\mathcal{T}) &= \frac{1}{D_{n+1}}\int dP_{\mathcal{T},\id}(\sigma)\Tr\left( (K(\sigma)\otimes \mathbb{I})\Pi^+_{n+1} \sigma_{n+1}(n+1)  \right) \label{eq:avgnp1} \\
                 &= \frac{1}{D_{n+1}}\int dP_{\mathcal{T},\id}(\sigma)\Tr\left( K(\sigma)\Tr_{n+1}\left(\Pi^+_{n+1} \sigma_{n+1}(n+1)\right)\right)  \\
                 &= \frac{1}{(n+1)D_{n+1}}\int dP_{\mathcal{T},\id}(\sigma)\Tr\left( K(\sigma)\left(\mathbb{I}+\sum_{k=1}^n\sigma_n(k)  \right)\right).
\end{align}
Finally, $\sigma_n(k)\le\mathbb{I}$, so $\Tr(K(\sigma)\sigma_n(k))\le \Tr(K(\sigma))$, and we can bound the above as
\begin{equation}
F(\mathcal{T})\le \frac{1}{D_{n+1}}\int dP_{\mathcal{T},\id}(\sigma)\Tr\left( K(\sigma)\right) = \frac{\Tr \Pi^+_n}{D_{n+1}}=\frac{D_n}{D_{n+1}}=\frac{n+1}{n+d}.
\end{equation}
\end{proof}

\subsubsection{Bandit policy as a sequence of tomographies}

Given the fidelity bound in Theorem~\ref{th:fidelity_upper_bound} we are ready to apply it to our multi-armed quantum bandit problem. In the following theorem, we reduce the PSMAQB setting to a sequence of quantum state tomography protocols and apply the mentioned fidelity bound.

\begin{theorem}\label{th:state_tomography_lowerbound}
 Let $T,\,d\in\N$. For any policy $\pi$ and action set of observables $\mathcal{A}$ containing all rank-1 projections, i.e.,\ $\mc A=\mc S_d^*$, there exists an environment $\ket{\psi}\!\bra{\psi}\in\mc S^*_d$ such that
  \begin{align}
    \int d\psi \EX_{\psi,\pi} \left[R_T(\mathcal{A},\psi,\pi) \right] \geq (d-1)\log\left(\frac{T}{d+1} \right).
  \end{align}
\end{theorem}
\begin{proof}
  Given a policy $\pi$, we can introduce a POVM $E_t: (\Sigma\times \{0,1\})^{\times t}\to\End(\hil^{\otimes t})$ such that
  \begin{equation}
    P^t_{\ket{\psi}\!\bra{\psi}, \pi}(A_1,X_1,\dotsc,A_t,X_t)=\Tr\left( (\ket{\psi}\!\bra{\psi})^{\otimes t}E_t(A_1,X_1,\dotsc,A_t,X_t) \right),
  \end{equation}
where $P^t_{\ket{\psi}\!\bra{\psi},\pi}$ is the probability measure defined by~\eqref{eq:probdens_continous}, but only for actions and rewards until step $t$. The construction of this POVM is presented in the proof of Lemma~\ref{lemma:D1/2ineq}.
We will also define the coordinate mapping
\begin{equation}
  \Psi_t(\Pi_{1}, X_1,\dotsc,\Pi_{t},X_t)=\Pi_{t},
\end{equation}
where $\Pi_{i}\in\mathcal{A}$ are actions and $r_i\in\{0,1\}$ are rewards of the PSMAQB. Now we can for each step $t$ define a tomography scheme $\mathcal{T}_t=E_t\circ \Psi_t^{-1}$ as the pushforward POVM from $E_t$ to the space $(\mathcal{A},\Sigma)$. Informally, this tomography scheme takes $t$ copies of the state, runs the policy $\pi$ on them, and outputs the $t$-th action of the policy as the predicted state. For $A\in\Sigma$, we can rewrite the tomography's distribution on predictions as
\begin{equation}
    P_{\mathcal{T},(\ket{\psi}\!\bra{\psi})^{\otimes t}}(A)=\Tr\left(\mathcal{T}_t(A)(\ket{\psi}\!\bra{\psi})^{\otimes t}\right)=\Tr\left(E_t(\Psi_t^{-1}(A))(\ket{\psi}\!\bra{\psi})^{\otimes t}\right)=\left(P^t_{\ket{\psi}\!\bra{\psi},\pi}\circ \Psi^{-1}\right)(A).
\end{equation}
Then, the fidelity of $\mathcal{T}_t$ on the input $\ket{\psi}\!\bra{\psi}$ can be rewritten as
\begin{align}
    F(\mathcal{T}_t,\ket{\psi}\!\bra{\psi}) &=\int \<\psi|\rho|\psi\> dP_{\mathcal{T}_t,(\ket{\psi}\!\bra{\psi})^{\otimes t}}(\rho) \\
    &= \int \<\psi|\Psi_t(\Pi_1,X_1,\dotsc,\Pi_t,X_t)|\psi\> dP^t_{\ket{\psi}\!\bra{\psi},\pi}(\Pi_1,X_1,\dotsc,\Pi_t,X_t) \\
    &= \EX_{\psi,\pi}\left[\<\psi|\Pi_{t}|\psi\>\right].
\end{align}
Using the bound for average tomography fidelity on $\mathcal{T}_t$ from Theorem~\ref{th:fidelity_upper_bound}, we can now bound the average regret of $\pi$:
\begin{align}
  \int \EX_{\psi,\pi}\left[R_T( \mathcal{A} ,  \psi , \pi)\right]d\psi &= T-\sum_{t=1}^T\int \ \EX_{\psi,\pi}\left[\<\psi|\Pi_{t}|\psi\>\right]d\psi \label{eq:sumint}\\
  &= T-\sum_{t=1}^T F(\mathcal{T}_t) \label{eq:introt} \geq \sum_{t=1}^T 1 - \frac{t+1}{t+d} \\
  & = \sum_{t=1}^T \frac{d-1}{t+d} \geq (d-1)\log\left( \frac{T}{d+1} \right),
\end{align}
where the first equality follows from the form of the regret~\eqref{eq:regret_purestate_lowerbound} and the last inequality follows from bounding the sum by below with the integral of the function $f(t) = 1/(t+d)$.
\end{proof}

This result is averaged over all pure states, meaning that at least one state reaches the bound. Given the symmetries of the problem, we can intuitively expect that no particular subset of states is significantly harder to learn than others for any given policy.

Given this result, one may observe that since we employed a tomography-based approach without accounting for the exploration–exploitation constraint, the bound obtained is not tight. In particular, tomography guarantees at least a logarithmic contribution, as the infidelity scales as $1-F = \Omega(1/t)$ with $t$ copies, leading to a cumulative contribution of order $\log(T)$. However, as we will show in the next chapter this bound is indeed almost tight.

\chapter{Algorithms and regret upper bounds}\label{ch:upperbounds}

This chapter builds on the algorithms developed by the author in~\cite{Lumbreras2022multiarmedquantum, pmlr-v247-lumbreras24a, lumbreras24pure}. We begin by reviewing some of the most well-known algorithms based on upper confidence bounds, or $\textsf{UCB}$-type strategies. These methods can be readily applied to the MAQB settings introduced in the previous section and already yield a regret scaling of order $\sqrt{T}$. We then present a simple Explore-Then-Commit strategy for the PSMAQB setting, grounded in a quantum state estimation algorithm. After this warm-up, we move into the core technical contribution of the chapter: the design and theoretical analysis of algorithms for linear bandits with vanishing noise. To study these settings, we introduce a novel technique that tightly controls the eigenvalues of the design matrix, leading to new proofs of regret bounds in the linear stochastic bandit setting. These techniques will allow us to establish a polylogarithmic upper bound for the PSMAQB setting.

\section{Optimism in the face of uncertainty principle}\label{sec:ofu}

Two of the most widely used algorithms for solving bandit problems are UCB (Upper Confidence Bounds)~\cite{auer2002finite} and Thompson Sampling~\cite{thompson33,agrawal2013thompson}. In this thesis, we focus on UCB-type algorithms, which are based on the principle of optimism in the face of uncertainty (OFU). This principle is a natural and effective approach to designing algorithms for many stochastic online learning problems. It suggests that, when faced with uncertainty, the learner should act as if the most favorable plausible scenario is true, encouraging exploration of potentially optimal choices. In the context of bandit problems, this translates into selecting actions that balance exploration and exploitation by considering optimistic estimates of expected rewards.

Given a set of actions $\mathcal{A}$ and a set of possible environments $\Gamma$, the algorithm follows a simple yet effective strategy. At each time step $t \in [T]$, it uses past observations of rewards and actions to construct a confidence set $\mathcal{C}_{t-1} \subseteq \Gamma$ that contains the unknown environment with high probability. The next action is then selected based on the optimism in the face of uncertainty (OFU) principle:  
\begin{align}\label{eq:optimistic_principle}
    a_t = \argmax_{a\in\mathcal{A}} \max_{\theta'\in\mathcal{C}_{t-1}} \langle \theta' , a \rangle,
\end{align}
where $\langle \cdot , \cdot \rangle$ is a well-defined inner product between elements in the action set and the environment set $\Gamma$ (for linear bandits, we consider real vector spaces). The key challenge in applying the OFU principle lies in constructing confidence sets $\mathcal{C}_{t-1}$ that are small enough, so we choose actions close to the optimal one, while still containing the true environment with high probability.

The UCB-type algorithms often require certain assumptions on the reward distribution, such as finite support or bounded variance. The reason for this is that constructing confidence regions relies on concentration inequalities, which provide probabilistic guarantees on how much the observed rewards deviate from their expected values. In particular, we focus on rewards that are sub-Gaussian with a finite sub-Gaussian parameter. This assumption ensures that the reward distribution satisfies strong tail bounds, making it possible to derive meaningful confidence sets. To formalize this in the multi-armed quantum bandit setting, we state the following lemma, which will be used repeatedly throughout our analysis.

\begin{lemma}\label{lem:subgaussianrewards}
 Consider a general or discrete multi-armed quantum bandit problem $(\mathcal{A},\Gamma)$  with environment $\rho\in\Gamma$ such that, for any $O\in\mathcal{A}$, $\|O\| < \infty $. Then if $O_{A_t}\in\mathcal{A}$ is the observable selected at round $t\in[T]$,  the rewards are of the form
\begin{align}\label{eq:quantum_reward_lemma}
X_t = \Tr \left( \rho O_{A_t} \right) + \epsilon_t,
\end{align}
where the statistical noise $\epsilon_t$ is $\eta$-subgaussian given $X_1,O_{A_1},...,X_{t-1},O_{A_{t-1}},O_{A_t}$ with
\begin{align}
    \eta \leq \max_{O\in\mathcal{A}} \| O \| .
\end{align}
\end{lemma}

\begin{proof}
We define $\epsilon_t$ simply as
\begin{align}
    \epsilon_t := X_t -  \Tr \left(\rho O_{A_t} \right),
\end{align}
and the form~\eqref{eq:quantum_reward_lemma} follows. Given $\lambda_{\max}(O)$ and $\lambda_{\min}(O)$ we have that $\lambda_{\min}({O_t}) \leq X_t \leq \lambda_{\max}(O_t)$ and $\lambda_{\max}(O_t)-\lambda_{\min}(O_t)\leq 2\|O _{A_t} \|$. Then we can apply Hoeffding Lemma~\cite[Equation 4.16]{Hoeffding} and
   \begin{align}
       \EX [e^{\lambda (X_t-\Tr(\rho O_{A_t} ) } ] \leq \exp \left( \frac{\lambda^2 \| O_{A_t} \|^2}{2} \right) \leq \exp \left( \frac{\lambda^2(\max_{O\in\mathcal{A}} \| O \| )^2}{2} \right)
   \end{align}
   And the result follows from the definition of subgaussian parameter.
\end{proof}

\subsection{\textsf{LinUCB} algorithm for multi-armed quantum bandits}\label{sec:linucb}

In this section, we review one of the most well-known algorithms for linear stochastic bandits, called \textsf{LinUCB} (linear upper confidence bound) or \textsf{LinRel} (linear reinforcement learning). We will just call it \textsf{LinUCB}.  The first algorithm for linear stochastic bandits was studied in~\cite{abe1999associative}, and the first one based on the optimistic principle in~\cite{aurer} where the authors considered a finite action set. The version of the \textsf{LinUCB} algorithm that we introduce is based on~\cite{lin1,lin2,lin3}, which extended the setting to allow an infinite action set. This algorithm follows the optimism in the face of uncertainty principle to construct confidence sets and select actions accordingly. In the following, we present the algorithm and its main properties following these references before discussing its application to our model.

%The classical bandits are analogous to our quantum case with the main difference that instead of performing a measurement at each round we sample a reward from a set of probability distributions. In order to introduce the \textsf{LinUCB} algorithm we quickly review stochastic linear bandits. 

%Now we introduce the \textsf{LinUCB} algorithm and explain how it fits to our quantum model. 

Recall that in the classical linear model, the rewards are sampled as  
\begin{equation}
    X_t = \langle \theta, A_t \rangle + \epsilon_t,
\end{equation}  
where $\theta \in \mathbb{R}^d$ is the unknown parameter, $A_t \in \mathcal{A} \subseteq \mathbb{R}^d$ is the chosen action, and $\epsilon_t$ is conditionally sub-Gaussian noise. The \textsf{LinUCB} algorithm minimizes regret by constructing a confidence ellipsoid for the unknown parameter $\theta$ and selecting the action in $\mathcal{A}$ that maximizes the estimated reward within this confidence region.  

The confidence region is built using the least-squares estimator
\begin{align}\label{eq:least_squares_estimator_online}
    \hat{\theta}_t = \argmin_{\theta' \in \mathbb{R}^d} \left( \sum_{s=1}^t \left( X_s -  \langle \theta', A_s \rangle \right)^2 + \lambda \| \theta' \|_2^2 \right),
\end{align} 
where $X_s$ is the reward received at round $s$, $A_s$ is the action selected from $\mathcal{A}$, and $\lambda \geq 0$ is a regularization parameter ensuring the function has a unique minimum.  

Solving this minimization problem yields the closed-form estimator
\begin{align}\label{eq:lse}
    \hat{\theta}_t = V_t^{-1} \sum_{s=1}^t A_s X_s,
\end{align} 
where $V_t$, known as the \textit{design matrix}, is defined as  
\begin{align}\label{eq:design_matrix}
    V_t := \lambda \mathbb{I} + \sum_{s=1}^t A_s A_s^\mathsf{T}.
\end{align}
By construction, $V_t$ is positive definite and induces the norm $\| x \|^2_{V_t} = \langle x , V_t x\rangle$ for any $x \in \mathbb{R}^d$.  

Using the estimator $\hat{\theta}_t$, we can construct a confidence region $\mathcal{C}_t \subset \mathbb{R}^d$ at each round $t$. The following theorem provides the exact form of this confidence region.

\begin{theorem}[Theorem 20.5 in \cite{lattimore_banditalgorithm_book}]\label{th:confidence_region}
Let $\delta \in (0,1)$ and $(A_t)^\infty_{t=1}$ be the actions of linear stochastic bandit selected by some policy with corresponding rewards $(X_t)^\infty_{t=1}$ given  by $X_t = \langle \theta ,A_t \rangle + \epsilon_t$, where $\theta\in\mathbb{R}^d$ is the unknown parameter and $\epsilon_t$ is $\eta$-subgaussian. Let $L,m\in\mathbb{R}$ such that $\| A_t \|_2$ and $\| \theta \|_2 \leq m$. Then we can define the following confidence region 
\begin{align}\label{eq:confidence_region_book}
    \mathcal{C}_t := \lbrace \theta'\in\mathbb{R}^d : \| \theta' - \hat{\theta}_t \|^2_{V_t} \leq \beta_{t,\delta}  \rbrace,
\end{align}
where $\hat{\theta}_t$ is defined as in~\eqref{eq:lse}, $V_t$ in~\eqref{eq:design_matrix} for some $\lambda > 0$, and
\begin{align}\label{eq:beta_book}
\beta_{t,\delta} = \left( \eta\sqrt{2\log \frac{1}{\delta} + d\log\left(\frac{d\lambda + tL^2}{d \lambda } \right)}+ \eta\sqrt{\lambda}\right)^2. 
\end{align}
Then
\begin{align}\label{eq:prob_confidence}
    \mathrm{Pr}\left( \forall s \in [t]:  \theta \in \mathcal{C}_s \right) \geq 1 - \delta.
\end{align}
\end{theorem}

We note that the confidence region given in Theorem~\ref{th:confidence_region} is an ellipsoid centered at the least squares estimator and aligned with the principal axes of the eigenvectors of the design matrix. A sketch of this region for a particular case in $\mathbb{R}^2$ is shown in Figure~\ref{fig:confidence_region}. We will provide a proof of the confidence region $\mathcal{C}_t$ in Section~\ref{sec:subgaussian_parameter} for a modified version of the Lemma.  

Given the confidence region~\eqref{eq:confidence_region}, the \textsf{LinUCB} algorithm proceeds as follows: at each round $t\in \{1, \dots, T \}$, it constructs the confidence region $\mathcal{C}_t$, selects an optimistic estimate $\tilde{\theta}_t$, and chooses the action that maximizes the estimated reward:  
\begin{equation}\label{eq:action_selection_linUCB}  
    (A_t, \tilde{\theta}_t) = \argmax_{A\in\mathcal{A}, \theta'\in\mathcal{C}_{t-1}} \langle \theta', A \rangle.
\end{equation}  
The algorithm then selects $A_t$ and observes the corresponding reward. The vector $\tilde{\theta}_t$ can be interpreted as the optimistic estimate of the parameter used by the algorithm. The complete pseudo-code is provided in Algorithm~\ref{alg:LinUCB}.

\begin{figure}[h!]
\centering
\begin{overpic}[percent,width=0.65\textwidth]{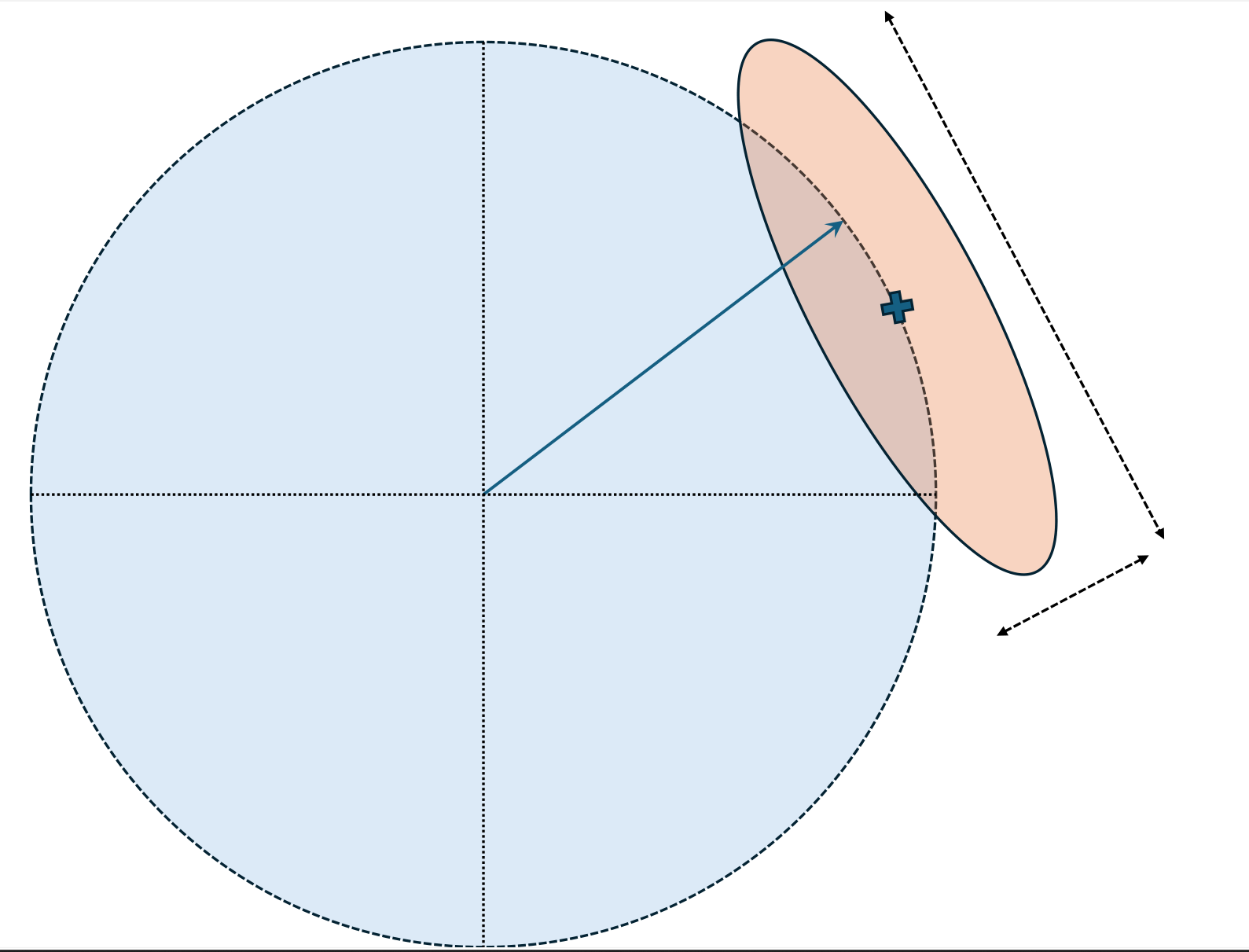}
\put(20,50){$\mathbb{S}^1$}
\put(50,43){$A_t$}
\put(74,49){$\hat{\theta}_{t-1}$}
\put(30,20){$V_t = \lambda I + \sum_{s=1}^t A_s A_s^{\text{T}}$}
\put(80,65){\rotatebox{-60}{$\sqrt{\frac{\beta_{t-1}}{\lambda_{\min}(V_{t-1})}}$}}
\put(80,18){\rotatebox{30}{$\sqrt{\frac{\beta_{t-1}}{\lambda_{\max}(V_{t-1})}}$}}
\end{overpic}
\label{fig:linucb_circle}
\caption{\label{fig:confidence_region} Sketch of a linear bandit with $\mathcal{A}= \mathbb{S}^1$, $\theta\in\mathbb{S}^1$, where $\hat{\theta}_{t-1}$ is the least squares~\eqref{eq:lse} centered around the confidence region defined in Theorem~\ref{th:confidence_region}. The orientation of the ellipsoid is defined by the eigenvectors of $V_t$ along its principal axis and the lengths are given by the eigenvalues.}
\end{figure}
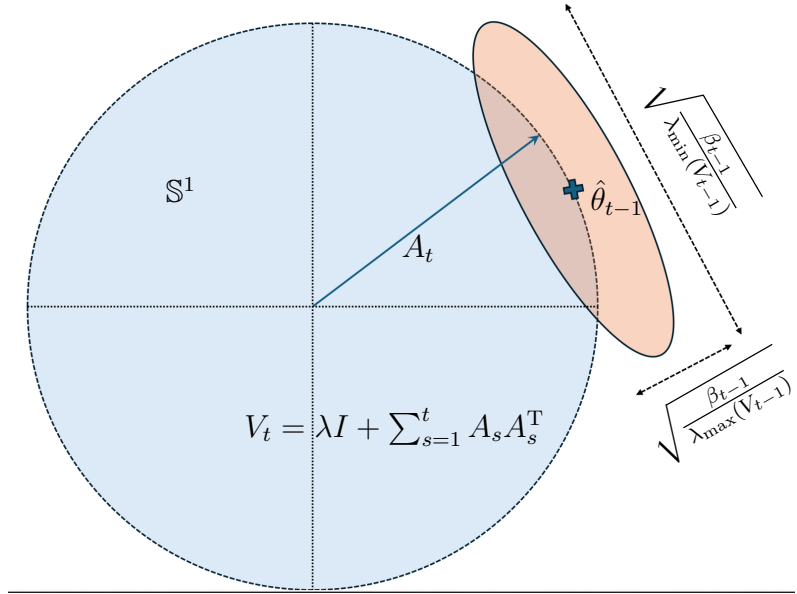

\begin{algorithm}[H]
	\caption{\textsf{LinUCB}} 
	\label{alg:LinUCB}
	\begin{algorithmic}[1]
		\For {$t=1,2,\ldots$}
			\State $(A_t, \tilde{\theta}_t) = \argmax_{A\in\mathcal{A},\theta'\in\mathcal{C}_{t-1}}\langle \theta, A \rangle$;
			\State Select $A_t$ and observe reward $X_t$;
			\State Update $\mathcal{C}_t$;
			\EndFor
	
	\end{algorithmic} 
\end{algorithm}

\subsubsection{Regret analysis}

With the tools introduced above, we are now ready to present the regret bound achieved by \textsf{LinUCB}. While the proof follows the same steps as in~\cite[Chapter 19]{lattimore_banditalgorithm_book}, we reproduce it here for completeness. Understanding this proof will help illustrate the limitations of standard approaches and justify why we later introduce a different algorithm that achieves better regret scaling in certain multi-armed quantum bandit settings. First of all we state the elliptical potential lemma, which is a standard tool for analyzing regrets in online learning problems especially for linear stochastic bandits~\cite{carpentier2020elliptical}.

\begin{lemma}[Elliptical potential lemma, Lemma 19.4 in \cite{lattimore_banditalgorithm_book}]\label{lem:elliptical_potential}
    Let $a_1,...,a_T\in\mathbb{R}^d$ be a sequence of vectors with $\| a_t \|_2 \leq L <\infty$ for all $t\in [T]$, $V_0\in\mathbb{R}^{d\times d}$ positive definite and $V_t = V_0 + \sum_{s\leq t} a_s a_s^\mathsf{T} $ . Then,
    \begin{align}
        \sum_{t=1}^T \min \lbrace 1 , \| a_t \|^2_{V^{-1}_{t-1}}  \rbrace \leq 2\log \left( \frac{\det V_T}{\det V_0} \right) \leq 2d\log\left(\frac{\Tr ( V_0 ) + TL^2}{d \det (V_0)^{1/d}} \right).
    \end{align}
\end{lemma}

The following regret bound is derived using the above lemma along with the probabilistic guarantees of $\mathcal{C}_t$ from Theorem~\ref{th:confidence_region}. We will use the normalization of the regret $\sup_{a,b\in\mathcal{A}}\langle \theta , a - b \rangle \leq 1$.

\begin{theorem}[Theorem 19.2 in \cite{lattimore_banditalgorithm_book}]\label{th:regret_linucb}
Let $T\in\mathbb{N}$ be the finite time horizon. Under the assumptions of Theorem \ref{th:confidence_region} for $\theta\in\mathbb{R}^d,\mathcal{A}\subseteq \mathbb{R}^d$ the regret~\eqref{eq:regret_linearbandits} of \textsf{LinUCB} satisfies with probability at least $1-\delta$
\begin{align}
R_T (\mathcal{A}, \theta,\pi ) \leq \sqrt{8dT\beta_{T,\delta}\log \left( \frac{d\lambda + TL^2}{d\lambda} \right)},
\end{align}
where $\beta_T$ is defined as in~\eqref{eq:beta_book}. In particular the expected regret of \textsf{LinUCB} with $\delta = 1/T$ is bounded by
\begin{align}
\EX_{\theta,\pi}[ R_T (\mathcal{A},\theta,\pi ) ]  \leq Cd\eta\sqrt{T}\log(T),
\end{align}
where $C>0$ is a suitably large universal constant, and $\eta$ is the subgaussian parameter of the rewards.
\end{theorem}

\begin{proof}
    Let $A^* = \argmax_{a\in\mathcal{A}} \langle \theta , a\rangle$ be the optimal action. We will prove the following for all $t\in[T]$ under the assumption that $\theta\in\mathcal{C}_t$ and apply that this holds with probability at least $1-\delta$ using Theorem~\ref{th:confidence_region}. Then using the action selection rule of the algorithm~\eqref{eq:action_selection_linUCB} and the assumption $\theta\in\mathcal{C}_t$ and  we have
    \begin{align}
        \langle \theta , A^* \rangle \leq \langle \tilde{\theta}_t , A_t \rangle ,
    \end{align}
    where $\tilde{\theta}_t$ is the optimistic parameter selected by the algorithm at time step $t$. Using Cauchy-Schwarz inequality, the above inequality and again the assumption $\theta\in\mathcal{C}_t$ we have
    \begin{align}
        \langle \theta , A^* - A_t \rangle \leq \|A_t \|_{V^{-1}_{t-1}} \|\tilde{\theta}_t - \theta\|_{V_{t-1}} \leq 2 \|A_t \|_{V^{-1}_{t-1}} \sqrt{\beta_{t,\delta}}.
    \end{align}
    Then using the normalization  $\sup_{a,b\in\mathcal{A}}\langle \theta , a - b \rangle \leq 1$ and $\beta_T \geq \max \lbrace 1 , \beta_t \rbrace$ we have
    \begin{align}
         \langle \theta , A^* - A_t \rangle \leq 2\sqrt{\beta_{T,\delta}} \min \lbrace 1 , \|A_t \|_{V^{-1}_{t-1}} \rbrace ,
    \end{align}
 Then finally applying Cauchy-Schwarz inequality into the regret expression we have
    \begin{align}
        R_T\left(\mathcal{A},\theta,\pi\right) &= \sum_{t=1}^T  \langle \theta , A^* - A_t \rangle \leq \sqrt{T\sum_{t=1}^T \langle \theta , A^* - A_t \rangle^2 } \\ 
        &\leq 2 \sqrt{T\beta_{T,\delta} \sum_{t=1}^T \min \lbrace 1 , \|A_t \|_{V^{-1}_{t-1}} \rbrace}.
    \end{align}
    The result follows applying the elliptical potential Lemma~\ref{lem:elliptical_potential}.
\end{proof}

%Now the question is, can we adapt this algorithm to our problem of stochastic quantum bandits?
Given the regret bound for \textsf{LinUCB} stated in Theorem~\ref{th:regret_linucb} we can use the connection between linear stochastic bandits and multi-armed quantum  bandits from Section~\ref{sec:connection_linearbandits} in order to get a minimax regret upper bound for the MAQB model. The following theorem shows how the \textsf{LinUCB} can be applied to the multi-armed quantum bandit problem and the scaling of the regret.

\begin{theorem}\label{teo:linucbregretmaqb}
Let $d,T\in\mathbb{N}$, and consider a general or discrete multi-armed quantum bandit described by $(\mathcal{A},\Gamma)$ where $\mathcal{A}\subseteq \mathcal{O}_d$ is the action set and $\Gamma\subseteq \mathcal{S}_d$ is the set of potential environments. Consider the normalization of the action set $\| O \| \leq 1$ for any $O\in\mathcal{A}$. Then the \textsf{LinUCB} algorithm~\ref{alg:LinUCB}, associated to policy $\pi$, can be applied to the general or discrete multi-armed bandit problem such that for any $\rho\in\Gamma$ the expected regret is bounded by
\begin{align}
\EX_{\rho,\pi} [R_T (\mathcal{A},\rho ,\pi  )] \leq C d^2 \sqrt{T} \log(T),
\end{align}
for some constant $C>0$ that depends on the action set.
\end{theorem}

\begin{proof}

Considering a set $\lbrace \sigma_i \rbrace_{i=1}^{d^2}$ of orthonormal Hermitian matrices, we can parametrize $\rho\in\mathcal{S}_d$ and any $O_a\in\mathcal{A}$ as $\rho = \sum_{i=1}^{d^2} \theta_i \sigma_i$, $ O_a = \sum_{i=1}^{d^2} A_{a,i} \sigma_i$ where $\theta_i = \Tr(\rho\sigma_i)$ and $A_{a,i} = \Tr (O_a\sigma_i)$. Note that $\theta,A_a\in\mathbb{R}^{d^2}$. We fix the normalization of the basis to $\Tr(\sigma_i \sigma_j) = 1$, thus since $\Tr(\rho^2)\leq 1$ we have $\| \theta \|_2 \leq 1$ and $\|A_a \|_2\leq \| O_a \|_2 = \Tr(O^2_a)$. We also have $\Tr(\rho O_a ) = \langle \theta, A_a \rangle$ and the rewards are of the form
\begin{align}
X_t = \langle \theta , A_t \rangle + \epsilon_t,
\end{align}
where $\epsilon_t$ is $\eta$-subgaussian with $\eta \leq \max_{O_a\in\mathcal{A}}\| O \| =1$ by Lemma~\ref{lem:subgaussianrewards}. Note that the regret can be rewritten as 
\begin{align}
    R_T (\mathcal{A},\rho ,\pi  ) = \sum_{t=1}^T \max_{O_A\in\mathcal{A}} \langle \theta , A \rangle  - \langle \theta , A_t \rangle .
\end{align}
Then we can use the regret bound of Theorem~\ref{th:confidence_region} and Theorem~\ref{th:regret_linucb} under the linear stochastic bandit given by the considered vector parametrization.

\end{proof}

\subsection{\textsf{UCB} and \textsf{Phased Elimination} algorithms for discrete multi-armed quantum bandits}\label{sec:UCB}

The previous algorithm can be applied to both discrete and general multi-armed quantum bandits. From Theorem~\ref{th:regret_linucb}, we see that the regret scales with the dimension of the Hilbert space. However, in cases where the action set is discrete and consists of only a few observables, it is natural to ask whether we can achieve a regret bound that depends on the number of observables instead. The key idea is that when the number of observables is small, we can estimate their mean rewards individually and apply the optimistic principle by treating each action independently, rather than relying on a high-dimensional parameter estimation.

This approach is precisely the foundation of one of the earliest multi-armed bandit algorithms, the \textit{Upper Confidence Bound} (\textsf{UCB}) algorithm, first introduced in~\cite{firstUCB}. We will see that under certain conditions on the number of actions and the dimension of the Hilbert space, \textsf{UCB} can provide a tighter regret bound than \textsf{LinUCB}. To understand this, we will first review the \textsf{UCB} algorithm and introduce a simplified version of the multi-armed bandit problem.

\subsubsection{Multi-armed stochastic bandit}

A \textit{multi-armed stochastic bandit} consists of a set of probability distributions (the environment) \( \nu = (P_a : a\in \mathcal{A}) \), where \( \mathcal{A} \) is the set of available actions with cardinality \( k = |\mathcal{A}| \). Each action \( i \in [k] \) has an associated mean reward \( \mu_i \). At each round \( t \in [n] \), a learner selects an action \( A_t \in \mathcal{A} \) and observes a reward \( X_t \sim P_{A_t} \). The expected regret is defined as  
\begin{align}
R_T ( \nu , \pi ) = \sum_{t=1}^T \max_{i\in[k]}\mu_i - \mathbb{E}_{\nu,\pi}[X_t],
\end{align}  
where \( \pi \) denotes the policy and $\mathbb{E}_{\nu,\pi}$ is the reward distribution induced by the interaction of the policy $\pi$ with the environment $\nu$. The \emph{sub-optimality gap} for each action \( a \in [k] \) is given by \( \Delta_a = \max_{i\in[k]}\mu_i - \mu_a \), which measures how much worse the action \( a \) is compared to the optimal one.  

This model can be seen as a special case of a linear bandit problem with the following structure:  
\begin{itemize}
    \item The action set is discrete, consisting of \( k \) elements, so that \( \mathcal{A} = \{ e_i \}_{i=1}^k \), where \( e_i \in \mathbb{R}^k \) are the standard basis vectors.  
    \item The unknown parameter is the reward vector \( \theta = (\mu_1, \dots, \mu_k) \in \mathbb{R}^k \).  
    \item The reward noise is given by \( \epsilon_t = X_t - \mu_{A_t} \), where \( X_t \sim P_{A_t} \).  
\end{itemize}

Thus, if we consider each action independently, a discrete multi-armed quantum bandit can also be viewed as a multi-armed stochastic bandit.

\subsubsection{UCB algorithm}\label{sec:ucb_algorithm_mab}

The core idea of the \textsf{UCB} algorithm is to estimate the mean reward of each action while ensuring that, with high probability, these estimates are optimistic. This approach follows the optimistic principle described in~\eqref{eq:optimistic_principle}, where the algorithm selects actions based on the upper confidence bound of their estimated rewards. Unlike the linear bandit setting, where the confidence region \( \mathcal{C}_t \) is an ellipsoid, here it is shaped along \( k \) independent directions. Each direction corresponds to an action, and its length is determined by an upper bound on the estimated mean reward.

The version of \textsf{UCB} that we describe is based on the following concentration bound for subgaussian random variables. If \( (Z_t)_{t=1}^T \) is a sequence of independent \( \eta \)-subgaussian random variables with mean \( \mu \) and empirical mean \( \hat{\mu} = \frac{1}{T} \sum_{t=1}^T Z_t \), then for any \( \delta \in (0,1) \),
\begin{align}
\Pr \left( \mu \geq \hat{\mu} + \sqrt{\frac{2\eta^2\log (1/\delta )}{T}} \right) \leq \delta .
\end{align}
This bound states that, with high probability, the true mean \( \mu \) is close to the empirical mean \( \hat{\mu} \), up to a confidence radius that decreases as more samples are collected.

Using this bound, the \textsf{UCB} algorithm estimates the reward of each action optimistically following~\eqref{eq:optimistic_principle}. Let \( T \) be the total number of rounds, and let \( (X_t)_{t=1}^T \) be the sequence of observed $\eta$-subgaussian rewards in a multi-armed stochastic bandit. For each action \( a \in [k] \), we define its empirical mean as
\begin{align}
\hat{\mu}_a (T) = \frac{1}{T_a(T)}\sum_{t=1}^T X_t \mathbbm{1}\lbrace A_t = a \rbrace,
\end{align}
where \( A_t \in [k] \) is the action selected at round \( t \), and \( T_a(T) = \sum_{t=1}^T\mathbbm{1}\lbrace A_t = a \rbrace \) represents the number of times action \( a \) has been chosen up to round \( T \). At each round \( t \), the \textsf{UCB} algorithm assigns an upper confidence estimate \( \tilde{\mu}_a(t) \) to each action \( a \in [k] \):
\begin{align}
\tilde{\mu}_a(t) =
\begin{cases}
\infty, & \text{if } T_a(t-1) = 0, \\
\hat{\mu}_a(t-1) + \sqrt{\frac{2\eta^2\log(1/\delta)}{T_a(t-1)}}, & \text{otherwise}.
\end{cases}
\end{align}
The assignment of \( \infty \) ensures that each action is selected at least once in the first \( k \) rounds. Then, the action selected at round \( t \) is the one with the highest upper confidence estimate:
\begin{align}
A_t = \argmax_{a\in[k]} \tilde{\mu}_a(t).
\end{align}
This strategy balances exploration and exploitation by favoring actions that have either shown high empirical rewards or have been played infrequently and have thus a large uncertainty. The pseudocode of the \textsf{UCB} algorithm is presented in Algorithm \ref{alg:UCB}.

\begin{algorithm}[H]
	\caption{UCB} 
	\label{alg:UCB}
	\begin{algorithmic}[1]
		
		\For {$t=1,2,\ldots,T$}
			\State Choose action $A_t = \argmax_{i\in [k]} \begin{cases}
\infty \quad \text{if}\quad  T_i(t-1)=0.\\
 \hat{\mu}_i(t-1) + \sqrt{\frac{2\eta^2\log(1/\delta)}{T_i(t-1)}}\quad \text{otherwise.}
\end{cases}$
			\State Observe reward $X_t$ and update confidence bounds;
			
			\EndFor
	
	\end{algorithmic} 
\end{algorithm}

The following result on the regret scaling of the \textsf{UCB} algorithm for discrete multi-armed quantum bandits follows from Lemma~\ref{lem:subgaussianrewards} and the regret analysis of \textsf{UCB} for stochastic bandits in~\cite[Theorem 7.2]{lattimore_banditalgorithm_book}. Applying \textsf{UCB} to a discrete multi-armed quantum bandit problem $(\mathcal{A},\Gamma)$ is straightforward: we fix an environment $\rho \in \Gamma$ and define the corresponding action set as
\begin{align}
    \nu = \left( P_{\rho}(\cdot | a) : O_a\in\mathcal{A} \right).
\end{align}
Then we can treat each action independently without need of estimating the unknown $\rho\in\Gamma$.

\begin{theorem}\label{th:UCB}
Let $d,T\in\mathbb{N}$, and consider a discrete multi-armed quantum bandit described by $(\mathcal{A},\Gamma)$, where $\mathcal{A} \subseteq \mathcal{O}_d$ is the action set and $\Gamma \subseteq \mathcal{S}_d$ is the set of potential environments. Consider the normalization of the action set $\| O \|\leq 1$ for any $O\in\mathcal{A}$. Let $k = |\mathcal{A}|$ be the number of available actions. Then, the \textsf{UCB} algorithm~\ref{alg:UCB}, following policy $\pi$, can be applied to this setting with $\delta = \frac{1}{T^2}$, ensuring that for any $\rho \in \Gamma$, the expected regret is bounded by
\begin{align}
\EX_{\rho,\pi}[R_T(\mathcal{A},\rho , \pi ) ]  \leq 8\sqrt{Tk\log(T)} + \sum_{a=1}^k \Delta_a,
\end{align}
where $\Delta_a = \max_{O\in\mathcal{A}}\Tr(\rho O ) - \Tr(\rho O_a )$ for $a \in [k]$.
\end{theorem}

Next, we compare the regret scaling of \textsf{UCB} with that of \textsf{LinUCB} for discrete multi-armed quantum bandits. Let $k$ be the number of observables and $d$ the dimension of the Hilbert space. Theorem~\ref{th:UCB} shows that when $k \leq d^4$, the regret scaling of \textsf{UCB} is more favorable than that of \textsf{LinUCB} (see Theorem~\eqref{teo:linucbregretmaqb}). However, it is important to note that when applying \textsf{UCB} to discrete multi-armed quantum bandits, the algorithm does not exploit the underlying linear structure of the rewards.

\subsubsection{Phased Elimination algorithm}

%One can still refine the dimensional dependence of the regret for \textsf{LinUCB} for discrete multi-armed quantum bandits. We mention that there is a variant of \textsf{LinUCB} (see \cite{banditalgorithm}[Chapter 22]) based on optimal designs that assumes a finite number of actions and can be applied to the discrete multi-armed quantum bandits. This variant is called \textsf{Phased Elimination} and it is stated in Algorithm \ref{alg:phased}. This algorithm requires the assumptions of Theorem \ref{teo:regretlinucb} and the action set $\mathcal{A}\subseteq \mathbb{R}^d$ to be finite. HERE The analysis of {Theorem 22.1} in \cite{banditalgorithm} along with the observations of the proof of Theorem \ref{teo:linucbregretmaqb} lead immediately to the following result.

One can further refine the dependence of the regret on the dimension for \textsf{LinUCB} in the discrete multi-armed quantum bandit setting. A variant of \textsf{LinUCB}, known as \textsf{Phased Elimination} (see \cite[Chapter 22]{lattimore_banditalgorithm_book}), leverages optimal design techniques and assumes a finite number of actions. This algorithm, stated in Algorithm~\ref{alg:phased}, can be applied to discrete multi-armed quantum bandits under the assumptions of Theorem~\ref{teo:linucbregretmaqb}, requiring the action set $\mathcal{A} \subseteq \mathbb{R}^d$ to be finite.  

The key idea behind phased elimination is to structure the learning process into successive phases, where in each phase, the algorithm maintains a set of promising actions and gradually eliminates suboptimal ones based on collected data. Instead of updating estimates after every round, the algorithm gathers observations over a phase and then uses an optimal design criterion (such as G-optimality~\cite[Chapter 21]{lattimore_banditalgorithm_book}) to determine the next set of actions to explore. 

The analysis of Theorem~\cite[Theorem 22.1]{lattimore_banditalgorithm_book}, together with the observations from the proof of Theorem~\ref{teo:linucbregretmaqb}, leads directly to the following result.

\begin{theorem}
Let $d,T\in\mathbb{N}$, and consider a discrete multi-armed quantum bandit described by $(\mathcal{A},\Gamma)$ where $\mathcal{A}\subseteq \mathcal{O}_d$ is the action set and $\Gamma\subseteq \mathcal{S}_d$ is the set of potential environments. Consider the normalization of the action set $\| O \|$ for any $O$ in $\mathcal{A}$. Let $k= | \mathcal{A} |$ be the number of available actions, Consider the normalization of the action set $\sup_{O_a,O_b\in\mathcal{A}}\Tr(\rho (O_a - O_b))\leq 1$ for any $\rho\in\Gamma$. Then, the \textsf{Phased Elimination} algorithm~\ref{alg:phased} associated to policy $\pi$, with $\delta= \frac{1}{T}$, can be applied to the discrete multi-armed bandit problem and for some universal constant $C>0$ and for all $\rho\in\mathcal{S}_d$ the expected regret is bounded by
\begin{align}
 \EX_{\rho,\pi} [R_T (\mathcal{A},\rho,\pi ) ] \leq C d\sqrt{T\log(Tk)}.
\end{align} 
\end{theorem}

Note that the scaling is better than \textsf{LinUCB} since the dimensional factor is $d$ instead of $d^2$ as long as the number of observables $k$ is not exponential in the dimension of the Hilbert space. 

\begin{algorithm}[H]
	\caption{\textsf{Phased Elimination}} 
	\label{alg:phased}
	\begin{algorithmic}[1]
		\State Set $l=1$ and $\mathcal{A}_1 = \mathcal{A}$.
		\State Let $t_l = t$ be the current time step $t$ and find $\pi_l : \mathcal{A}_l \rightarrow [0,1]$ with $\mathtt{support}(\pi_l) \leq \frac{d(d+1)}{2}$ that maximises 
		\[ \log\det \left( \sum_{a\in\mathcal{A}} \pi(a) a a^T \right) \; \;  \text{subject to } \sum_{a\in\mathcal{A}_l} \pi(a) = 1  \]
		\State Let $\epsilon_l = \frac{1}{2^l}$ and $T_l(a) = \left\lceil \frac{2d\pi_l (\bold{a})}{\epsilon^2_l} \log \left( \frac{kl(l+1)}{\delta}\right) \right\rceil $ and $T_l = \sum_{a\in\mathcal{A}_l} T_l (\bold{a}) $
		\State Choose each action $a\in\mathcal{A}_l$ exactly $T_l ({a})$ times
		\State Calculate the empirical estimate:
		\[ \hat{\theta}_l = V^{-1}_l \sum_{t=t_l}^{t_l+T_l} A_t X_t \;\;\text{with }  V_l = \sum_{a	\in\mathcal{A}_l} T_l(a)aa^T  \]
		\State Eliminate low rewarding arms:
			\[  \mathcal{A}_{l+1} = \lbrace a\in\mathcal{A}_l : \max_{b\in\mathcal{A}_l} \; \hat{\theta}_l \cdot (b-a) \leq 2\epsilon_l \rbrace \]
		\State $l=l+1$ and \textbf{Goto Step 2}
	
	\end{algorithmic} 
\end{algorithm}

\section{Explore--then--commit tomography for the PSMAQB}\label{sec:explorethencommit}

After reviewing well-known algorithms for linear stochastic bandits, we found that they achieve a regret scaling of $\sqrt{T}$ in the time horizon $T$. This matches the lower bounds established in Chapter~\ref{ch:lowerbounds}, except in the case of the PSMAQB setting where the action set consists of rank-1 projectors and the environment is a pure quantum state. Recall that in this setting, the regret is given by  
\begin{align}
    R_T(\mathcal{S}^*_d, \psi, \pi) = \sum_{t=1}^T 1 - \langle \psi | \Pi_{A_t} | \psi \rangle,
\end{align}
where $\rho = |\psi \rangle \! \langle\psi |$ is the unknown pure quantum state, and $\Pi_{A_t}$ is the rank-1 projector chosen from the action set $\mathcal{A} = \mathcal{S}^*_d$, which consists of all rank-1 projectors.  

From Theorem~\ref{th:state_tomography_lowerbound}, the minimax lower bound for this setting is $\Omega(\log T)$, which does not match the $O(\sqrt{T})$ regret of existing algorithms. This lower bound was derived by exploiting the connection between state tomography and this specific bandit setting. This raises the natural question of whether quantum state tomography techniques can also be used to improve the upper bound on regret.  

A straightforward approach is to apply an \textit{explore-then-commit} strategy. In this method, a fraction of the time horizon $T$ is dedicated solely to estimating the unknown quantum state. Once the estimation is sufficiently accurate, the algorithm commits to a single action for the remaining rounds based on that estimate. The algorithm that we propose is based on the \emph{projected least squares} (PLS) method given in \cite{guctua2020fast}, but any other single copy state tomography algorithm can be adapted (as long as we can map the outcome measurements to outcomes of rank-1 projectors). Now, we briefly review the PLS methond and explain how to apply it for the regret analysis.

%\begin{align}
%R_T(\mathcal{S}^*_d, \psi, \pi) = \frac{1}{4} \sum_ {t=1}^T \EX_{\rho , \pi} \big\| \rho - \Pi_{A_t} \big\|^2_1 \, ,
%\end{align}

\subsubsection{Least-squares quantum state tomography algorithm}

Let $\rho\in\mathcal{S}_d$ an unknown quantum state that we want to estimate, and $\left\lbrace M_i\right\rbrace_{i=1}^m$ the elements of a POVM used to measure $\rho$. Suppose we perform measurements over $T$ rounds, and let $t_i$ be the number of times we observe the outcome $i$ associated with the element $M_i$ from the POVM. Since the probabilities of each outcome $i$ are given by Born's rule $\Tr(\rho M_i)$, the least squares estimator for $\rho$ is defined as  
\begin{align}\label{leastsquareesti}
L_T = \argmin_{X\in \mathcal{O}_d} \sum_{i=1}^m \left( \frac{t_i}{T} - \Tr \left( M_i X \right)\right)^2.
\end{align}  
The estimator $L_T$ is not necessarily a physical state, so the next step is to project it onto the space of physical states
\begin{align}\label{projection}
\hat{\rho}_T = \argmin_{\tilde{\rho}\in \mathcal{S}_d} \| L_T - \tilde{\rho} \|_2,
\end{align}  
where $\| \cdot \|_2$ denotes the Frobenius norm. In \cite{guctua2020fast}, this estimator is analyzed for structured POVMs, Pauli observables, and Pauli basis measurements. In our case, we use Pauli observables. Fix $d=2^n$, for some $n\in\mathbb{N}$ and let $ \left\lbrace \sigma_i\right\rbrace_{i=1}^{d^2}$ be the set of all possible tensor products of the $2\times 2$ Pauli matrices. These matrices form a basis of the form \eqref{dparametrization}. For Pauli observables, the least squares estimator \eqref{leastsquareesti} takes the form  
\begin{align}
L_T = \frac{1}{d}\sum_{i=1}^{d^2} \left( \frac{t_i^+ - t_i^-}{T/d^2} \right)\sigma_i,
\end{align}  
where $t_i^{\pm}$ are the empirical frequencies associated with the two-outcome POVM $\Pi_i^{\pm} = \frac{1}{2}(\mathbb{I}\pm\sigma_i)$, and each observable is measured $T/d^2$ times. The convergence analysis in \cite{guctua2020fast} leads to the following theorem.  

\begin{theorem}[Theorem 1 in \cite{guctua2020fast}]\label{convergencePLS}
Let $\rho$ be a quantum state in a $d$-dimensional Hilbert space, and fix a number of samples $T\in\mathbb{N}$. Then, using Pauli measurements, the PLS estimator $\hat{\rho}_T$ \eqref{projection} satisfies  
\begin{align}
\Pr \left( \| \hat{\rho}_T - \rho   \|_1 \leq  \epsilon  \right) \geq 1 - e^{-\frac{T\epsilon^2}{43d^2 r^2}} \quad \text{for} \quad \epsilon \in [0,1],
\end{align}
where $r = \min \left\lbrace \mathrm{rank}(\rho ) , \mathrm{rank} ( \hat{\rho}_T ) \right\rbrace $.
\end{theorem}

\subsubsection{Explore-Then-Commit algorithm}

After reviewing the PLS estimation method we explain how we use it to construct an algorithm for our pure state problem and how to analyse the regret. Recall that we have considered an action set that contains all rank-1 projectors. For that reason we consider the algorithm for one qubit states ($d=2$) since Pauli observables for one qubit can be measured using rank-1 projectors. The pseudo-code for the algorithm can be found in Algorithm~\ref{alg:pure}. If $T$ is the number of rounds, fix the error $\epsilon^2 = \frac{172\log(T)}{\sqrt{T}}$ and assume $\frac{172\log(T)}{\sqrt{T}} \leq 1$ in order to have $\epsilon\in [0,1]$. The condition $\frac{172\log(T)}{\sqrt{T}} \leq 1$ is achieved by $ T \geq 8\cdot 10^6$. Then the algorithm (policy $\pi$) works as follows.

\begin{itemize}
\item During the first $\sqrt{T}$ rounds we perform the Pauli observables measurements, and calculate the estimator $\hat{\rho}_{\sqrt{T}}$~\eqref{projection} for $\rho$. In our setting this using the MAQB framework the two-outcome POVM $\Pi_i^{\pm} = \frac{1}{2}(\mathbb{I}\pm\sigma_i)$ simply corresponds to the rank-1 projector observable $\Pi^+_i$. The frequencies can be written as
\begin{align}
    t^+_i = \sum_{t=1}^{\sqrt{T}} \mathbbm{1}\lbrace X_t = 1 , A_t = \Pi^+_i \rbrace, \quad \Pi^+_i = \frac{\mathbb{I}}{2} + \frac{1}{2}\sigma_i .
\end{align}

Then using Theorem~\ref{convergencePLS} we have,
\begin{align}\label{probesti}
P_{\rho,\pi}\left( \| \hat{\rho}_{\sqrt{T}} - \rho   \|_1 \leq  \epsilon  \right) \geq 1 - \frac{1}{T},
\end{align}
where we have used $r = 1$ since $\mathrm{rank}(\rho ) = 1$ and $\epsilon^2 = \frac{172\log(T)}{\sqrt{T}}$. In order to ensure that the estimator is pure we project the estimator into the rank-1 subspace in a $\epsilon-$ball as follows,
\begin{align}\label{rank1estimator}
\hat{\rho} = \argmin_{\rho \in \mathcal{S}^*_2} \|  \hat{\rho}_{\sqrt{T}} - \rho  \|_1 \quad \text{such that} \quad \|  \hat{\rho}_{\sqrt{T}} - \rho   \|_1\leq \epsilon.
\end{align}
Recall that for the above equation there is at least one solution with probability greater than $1-1/T$ since $\rho$ is rank-1 by assumption. Note that using \eqref{probesti} we have,
\begin{align}\label{probesti2}
P_{\rho,\pi}\left( \| \hat{\rho} - \rho   \|_1 \leq  2\epsilon  \right) \geq 1 - \frac{1}{T}.
\end{align}
\item For the remaining rounds we perform the measurement using $\hat{\rho} = |\hat{\psi}\rangle \! \langle \hat{\psi} | = \hat{\Pi}$ as the rank-1 projector from the action set.
\end{itemize}

\begin{algorithm}[H]
	\caption{Bandit PLS} 
	\label{alg:pure}
	\begin{algorithmic}[1]
		\For {$i=1,2,\ldots,  d^2  $}
			\For {$j=1,2,\ldots , \lceil \sqrt{T}\rceil/d^2$}
				\State Measure $\rho$ with $\sigma_i$;
				\State Update outcomes $t_i^{\pm}$;
				\EndFor
			\EndFor
		\State Compute $L_T = \frac{1}{d}\sum_{i=1}^{d^2} \left( \frac{t_i^+ - t_i^-}{T/d^2} \right)\sigma_i$;
		\State Compute 	$\hat{\rho}_{\sqrt{T}} = \argmin_{\tilde{\rho}\in \mathcal{S}_d} \| L_T - \tilde{\rho} \|_2$;
		\State Fix $ \epsilon^2 = \frac{43d^2 \log(T)}{\sqrt{T}} $
		\State Compute $\hat{\rho} = \argmin_{\rho \in \mathcal{S}^*_d} \|  \hat{\rho}_{\sqrt{T}} - \rho   \|_1 \quad \text{such that} \quad \|  \hat{\rho}_{\sqrt{T}} - \rho   \|_1\leq \epsilon$
		\For {$t= \lceil \sqrt{T}   \rceil+1,...,n$}
			\State Measure $\rho$ with $\hat{\rho} = \hat{\Pi}$;
		\EndFor	
	\end{algorithmic} 
\end{algorithm}

\begin{theorem}
Let $T\in\mathbb{N}$ denote the time horizon, assume $\frac{172\log(T)}{\sqrt{T}} \leq 1$ $\left( T\geq 8\cdot 10^6\right)$. Consider the PSMAQB setting with $d=2$ (qubits). Then, for all environments $\ket{\psi}\! \bra{\psi}\in\mathcal{S}^*_2$ the \textsf{Bandit PLS} algorithm~\ref{alg:pure} associated to policy $\pi$ achieves the expected regret 
\begin{align}
\EX_{\psi,\pi} [R_T (\mathcal{S}^*_2,\psi,\pi ) ] = O\left( \sqrt{T}\log (T)\right).  
\end{align} 
\end{theorem}

\begin{proof}
First we can rewrite the regret in terms trace distance as,
\begin{align}
R_T (\mathcal{S}^*_2,\psi,\pi ):= \frac{1}{4} \sum_ {t=1}^n \big\| |\psi \rangle \! \langle\psi |  - \Pi_{A_t} \big\|^2_1 \, ,
\end{align}
and then we will take the expectation value. During the first $\sqrt{T}$ rounds the algorithm performs the Pauli measurements using rank-1 projectors and builds the estimator $\hat{\rho}$ using Equation~$\eqref{rank1estimator}$. Using Equation~\eqref{probesti2} we have the probabilistic error bound
\begin{align}\label{probreg}
P_{\psi,\pi}\left( \| |\psi \rangle \! \langle\psi | - \hat{\rho}  \|_1 \leq 2\epsilon \right)\geq 1-\frac{1}{T},
\end{align}
where $\epsilon^2 = \frac{172\log(T)}{\sqrt{T}}$. Then, with probability greater than $1-\frac{1}{T}$, the regret can be bounded as,
\begin{align}\label{Gbound}
R_T (\mathcal{S}^*_2,\psi,\pi ) \leq  \sqrt{T} + (T-\sqrt{T}) \frac{172\log(T)}{\sqrt{T}} = O(\sqrt{T}\log(T)),
\end{align}
where the first term comes from the trivial bound $ \frac12 \| |\psi \rangle \! \langle\psi | - \Pi_{A_t} \|_1 \leq 1 $ for the first $\sqrt{T}$ rounds and the second term from the bound of Equation~\eqref{probreg} for the remaining rounds where $\Pi_{A_t}$ is the estimator~\eqref{rank1estimator}. Let $C$ be the hidden constant in $O(\sqrt{T}\log(T))$, and $G$ the probabilistic event where $R_T (\mathcal{S}^*_2,\psi,\pi )\leq C\sqrt{T}\log(T)$ and $G^C$ its complement. Note that $P_{\psi,\pi}(G^C)\leq 1/T$. Then we can calculate the expected regret as,
\begin{align}
\EX_{\psi,\pi} [  R_T (\mathcal{S}^*_2,\psi,\pi )] = \EX_{\psi,\pi} [\mathbbm{1}\left\lbrace G \right\rbrace  R_T (\mathcal{S}^*_2,\psi,\pi )] + \EX_{\psi,\pi} [ \mathbbm{1}\left\lbrace G^C \right\rbrace  R_T (\mathcal{S}^*_2,\psi,\pi ) ]. 
\end{align}
For the first term we use the bound \eqref{Gbound} given by the event G. For the second term we use the trivial bound $R_T (\mathcal{S}^*_2,\psi,\pi )\leq T$ combined with $P_{\psi,\pi}(G_i^C ) \leq  1/T$. Thus,
\begin{align}
\EX_{\psi,\pi} [  R_T (\mathcal{S}^*_2,\psi,\pi )] \leq  C\sqrt{T}\log(T)  + TP_{\psi,\pi}(G_i^C) \leq C\sqrt{T}\log(T)  +1 = O ( \sqrt{T}\log(T)   ).
\end{align}
\end{proof}

\section{Beyond the Square Root Barrier}

From Section~\ref{sec:ofu}, we see that standard \textsf{UCB}-type algorithms achieve an \( \tilde{O}(\sqrt{T}) \) regret bound across all MAQB settings almost matching all minimax lower bounds studied in Chapter~\ref{ch:lowerbounds}. However, this remains suboptimal compared to the $ \Omega(\log T) $ lower bound we established for the PSMAQB setting. Since the \textsf{UCB} algorithms we used are general-purpose, they apply to a broad range of settings but do not exploit the specific structure of the quantum problem.  A key observation is that the environment is given by a pure state $ |\psi\rangle\!\langle\psi| \in \mathcal{S}^*_d $, which suggests the possibility of exploiting its lower-dimensional structure to design a more efficient algorithm. In Section~\ref{sec:explorethencommit}, we applied this idea to construct a simple explore-then-commit algorithm. While this approach still yields an \( O(\sqrt{T}) \) regret bound, it only involves a single round of adaptivity, and one might argue that a more adaptive algorithm could achieve a better scaling. The first question that we ask ourselves is why $\textsf{LinUCB}$ achieves the square root scaling. From the proof in Theorem~\ref{th:regret_linucb} we see that the square root behavior is introduced when we apply Cauchy-Schwarz as 
\begin{align}
    R_T (\mathcal{A}, \theta , \pi) = \sum_{t=1}^T  \langle \theta , A^* - A_t \rangle \leq \sqrt{T\sum_{t=1}^T \langle \theta , A^* - A_t \rangle^2 }.
\end{align}
Then the term $\sum_{t=1}^T \langle \theta , A^* - A_t \rangle^2 $ can be argued to behave logarithmically in $T$ using the optimistic principle and the elliptical potenital lemma~\ref{lem:elliptical_potential}. This procedure always introduces a $\sqrt{T}$ and for that reason we need to develop a new technique to upper bound the cumulative regret based on a tight control of the instantaneous regret  $\langle \theta , A^* - A_t \rangle$ at each time step $t\in[T]$. In order to design the algorithm we will take advantage of the underlying structure of the problem, specifically that the problem is constrained to the Bloch sphere. %The rest of this Chapter will be devoted to design an analyse an algorithm based on \textsf{LinUCB} that achieves a better regret scaling for the setting of pure state environments.

One of the key features of this problem is its underlying geometry. For qubits, we can represent any pure state projector \( |\psi \rangle \! \langle \psi | \in \mathcal{S}_2^* \) using the Pauli basis \( \sigma = (\sigma_x, \sigma_y, \sigma_z) \) as follows
\begin{align}\label{eq:PSMAQB_link_sphere}
    |\psi \rangle \! \langle \psi | = \frac{\mathbb{I}}{2} + \frac{1}{2} r \cdot \sigma,
\end{align}
where \( r \in \mathbb{S}^2 \) is a unit vector. This suggests that we can study the problem as a classical linear stochastic bandit with an action set \( \mathcal{A} = \mathbb{S}^{d-1} \), a \((d-1)\)-dimensional sphere, and an unknown parameter \( \theta \in \mathbb{S}^{d-1} \).  

Existing regret lower bounds for linear bandits, such as those in~\cite{lin2} and~\cite[Chapter 24, Theorem 24.2]{lattimore_banditalgorithm_book}, assume that \( \theta \) lies within the unit ball and give a $\Omega (\sqrt{T})$ scaling. These bounds do not apply here because the worst-case instances they construct involve unknown vectors far from the surface of the ball. However, when the noise is finite and \(\sigma\)-subgaussian, we adapt a lower bound for logistic bandits from~\cite{abeille2021instance} and establish a minimax regret lower bound of \( \Omega (\sigma d \sqrt{T}) \), as shown in Appendix~\ref{ap:lowerbound}.  We denote $\mathcal{N}(\mu,\sigma^2)$ a normal distribution of mean $\mu$ and variance $\sigma$.

\begin{theorem}\label{th:lowerbound_sphere}
For any policy $\pi$ of a stochastic linear bandit with action set $\mathcal{A} =  \mathbb{S}^{d-1} = \lbrace x\in\mathbb{R}^d: \| x \|_2 =1 \rbrace$, reward noise given by $\epsilon_t = \mathcal{N}(0,\sigma^2 )$ where $\sigma > 0$, then there exists an unknown parameter $\theta\in\mathbb{R}^d$ such that $\|\theta \|_2 = 1$ and
\begin{align}
    \EX_{\theta,\pi} [R_T (\mathbb{S}^{d-1},\theta,\pi)] \geq \frac{1}{100}{\sigma}d\sqrt{T}, 
\end{align}
for $T\geq \frac{1}{6400}d^2 {\sigma}^2$.
\end{theorem}

 The proof is a simplified proof of the lower bound found in~\cite{abeille2021instance} but we decided to not include it in the main text since the focus of this Chapter is on algorithms. This result shows that for linear bandits with a constant subgaussian parameter (or variance), the $\sqrt{T}$ regret scaling remains optimal, even when the problem is constrained to a sphere. One may try to now apply the above lower bound to the MAQB problem with pure state, however the bound does not apply. In order to see this we have to look more closely to the quantum reward model. Recall that the action is given by a rank-1 projector $\Pi_{A_t}\in\mathcal{S}_d^*$ and the reward is sampled by measuring the two outcome POVM $\lbrace \Pi_{A_t} , \mathbb{I}-\Pi_{A_t} \rbrace$ on the unknown state. The reward is given by
\begin{align}
    X_t = \langle \psi | \Pi_{A_t} | \psi \rangle + \epsilon_t ,
\end{align}
where $X_t \in \lbrace 0 , 1\rbrace$ is Bernoulli distributed and $|\psi\rangle$ is the environment. We note that the variance of the reward behaves as
\begin{align}\label{eq:vanishing_variance}
    \VX [\epsilon_t] = \langle \psi | \Pi_{A_t} | \psi \rangle ( 1 - \langle \psi | \Pi_{A_t} | \psi \rangle ) ,
\end{align}
which means that it is not constant and actually vanishes linearly as we select actions close to the optimal one. Although the lower bound in Theorem~\ref{th:lowerbound_sphere} assumes a different noise model, it already suggests that non-zero reward variance is necessary to achieve the typical $\sqrt{T}$ regret scaling. This becomes clear when considering the trivial case where the variance is always zero: in such a scenario, the environment could be learned using only $d^2$ actions, and the regret would no longer grow with $T$. While the subgaussian parameter also vanishes when actions are chosen close to the unknown environment, its behavior is more complicated. From~\cite{buldygin2013sub} we have that the exact $\sigma$-subgaussian parameter of $X_t$ is given by
\begin{align}\label{eq:subgaussian_parameter_bernoulli}
    \sigma^2 = \begin{cases}
    0 \quad \text{for} \quad \langle \psi | \Pi_{A_t} | \psi \rangle \in\lbrace 0,1 \rbrace \\
    \frac{1}{4} \quad \text{for} \quad \langle \psi | \Pi_{A_t} | \psi \rangle = \frac{1}{2}\\
    \frac{1-2\langle \psi | \Pi_{A_t} | \psi \rangle}{2\ln\left( \frac{1-\langle \psi | \Pi_{A_t} | \psi \rangle}{\langle \psi | \Pi_{A_t} | \psi \rangle} \right)}\quad \text{for} \quad \langle \psi | \Pi_{A_t} | \psi \rangle \neq 0,1,\frac{1}{2}.
    \end{cases}
\end{align}
In Figure~\ref{fig:variance_vs_subgaussian}, we observe that both the variance and the subgaussian parameter of $X_t$	vanish when selecting actions close to the environment. However, the key difference lies in their rates of decay: the variance decreases linearly, whereas the subgaussian parameter decays more slowly. This observation motivates our focus on the vanishing variance, rather than the subgaussian parameter, when designing the algorithm for the PSMAQB setting.

\begin{figure}
    \centering
    \includegraphics[width=0.9\linewidth]{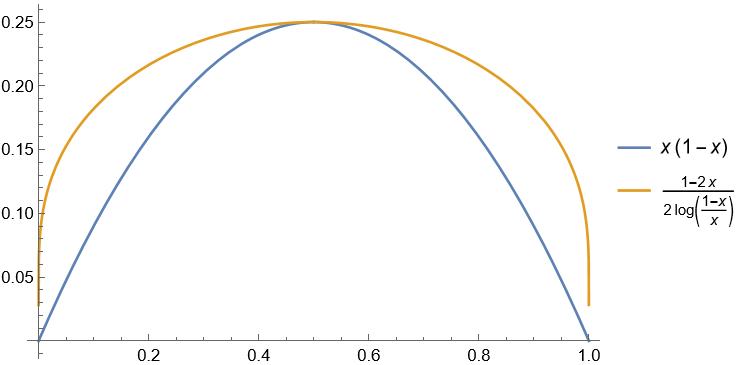}
    \caption{Plot of the variance and subgaussian parameter of a Bernoulli random variable $X\in\lbrace 0 , 1 \rbrace$ where $\mathrm{Pr}(X = 1 ) = x$.}
    \label{fig:variance_vs_subgaussian}
\end{figure}

We could also try to implement the extra structure given by~\eqref{eq:vanishing_variance} into the proof of Theorem~\ref{th:lowerbound_sphere}, but this leads to a constant lower bound. Specifically, we encounter the problem that these calculations rely on the computation of the $KL$ divergences $D_{KL} (\mathbb{P}_\theta , \mathbb{P}_{\theta '} ) $ for two ``close'' unknown parameters $\theta, \theta ' \in\mathbb{S}^d$ and it is not possible to give a non-trivial bound since their variance is arbitrarily close to $0$ when we select actions near the unknown parameters. We can easily see this fact for example if we consider Gaussian distributed rewards and use the divergence decomposition lemma that leads to
\begin{align}
    D_{KL}(\mathbb{P}_\theta , \mathbb{P}_{\theta '} ) = \EX_\theta \left[ \sum_{t=1}^T D_{KL}\left( \mathcal{N}(\langle \theta , a_t  \rangle, 1 - \langle \theta ,a_t \rangle^2 ) , \mathcal{\mathcal{N}}(\langle \theta' , a_t  \rangle, 1 - \langle \theta' ,a_t \rangle^2) \right) \right],
\end{align}
where $\mathcal{N}(\mu,\sigma^2)$ is a Gaussian probability distribution with mean $\mu\in\mathbb{R}$ and variance $\sigma^2 \geq 0$. The proof requires to upper bound of the above divergence but a simple computation shows that the divergence $ D_{KL}$ for the reward distributions cannot be upper bounded since the variances can go arbitrarily close to $0$ and we only get the trivial bound $D_{KL} \leq \infty$. Recall that in Section~\ref{sec:lower_bound} we also encountered the same issue of constant lower bound when trying to derive a minimmax lower bound for the PSMAQB setting using the techniques we developed for the other MAQB lower bounds. Specifically when computing $D_{\frac{1}{2}} \left( \psi \| \psi' \right)$  we found that the tradeoff with closeness of the environments and the time horizon $T$~\eqref{eq:constant_d12} leads to the constant bound. From these observations we conclude that the second structure that we want to exploit is the vanishing variance of the reward.

\section{Controlling the eigenvalues of the design matrix}\label{sec:designmatrix}

From the variance expression in~\eqref{eq:vanishing_variance} and the subgaussian parameter in~\eqref{eq:subgaussian_parameter_bernoulli}, we see that the statistical noise is not independent of the chosen actions and the environment. In fact, there is a well-defined dependence. This setting, where the noise depends on the actions, is known as \textit{heteroscedastic noise}, and was first studied for linear bandits in~\cite{kirschner2018information}. There, the authors considered a model in which the noise is \(\sigma(a_t)\)-subgaussian, where \(\sigma(a_t)\) is a \textit{link function} that maps actions \(a_t \in \mathcal{A}\) to real numbers, i.e., \(\sigma: \mathcal{A} \rightarrow \mathbb{R}\). The key difference in our model is that the noise also depends on the unknown environment.

One of the techniques introduced in~\cite{kirschner2018information} is to replace the online least-squares estimator~\eqref{eq:least_squares_estimator_online}—used, for instance, in the \textsf{LinUCB} algorithm—with a weighted version. The main idea is to bias the estimator towards actions with lower noise, using the inverse of the link function as a weight. In our case, we also will need such a weighted estimator, since the noise depends on the unknown environment. We note that the setting studied in~\cite{kirschner2018information} does not provide an improvement over the regret square root barrier since they focus on general class of linear bandits. However, our setting is more concrete, which will allow us to design a simpler strategy while still providing theoretical guarantees on the regret. Although~\cite{kirschner2018information} focuses on link functions that capture the subgaussian parameter, the same ideas can be applied when we have a link function for the variance.

Thus, we introduce an estimator for the noise subgaussian parameter (or variance) at time step $t \in [T]$ as
\begin{align}\label{eq:variance_estimator}
\hat{\sigma}^2_t:\mathcal{H}_{t-1}\times\mathcal{A} \rightarrow \mathbb{R}_{> 0},
\end{align}
where \(\mathcal{H}_{t-1} = (A_1,X_1,\dots,A_{t-1},X_{t-1})\) contains the history of past actions and rewards, and is independent of the unknown environment. For simplicity, we will often omit the dependence on past actions and rewards and simply write \(\hat{\sigma}^2_t(A)\) for \(A \in \mathcal{A}\). We will refer to~\ref{eq:variance_estimator} as the \textit{noise estimator}.
The weighted version of the linear least square estimator for linear bandits is defined with the following minimization problem
\begin{align}\label{eq:least_squares}
    \hat{\theta}^{\text{w}}_t := \argmin_{\theta' \in\mathbb{R}^d} \sum_{s=1}^t \frac{1}{\hat{\sigma}^2_s(A_s)}\left(X_s - \langle \theta', A_s\rangle \right)^2 + \lambda \| \theta' \|_2,
\end{align}
where $\lambda\in\mathbb{R}_{>0}$ is a regularization term. The analytical solution is 
\begin{align}\label{eq:estimator_weighted}
    \hat{\theta}^{\text{w}}_t  = V_t^{-1} (\lambda) \sum_{s=1}^{t} \frac{1}{\hat{\sigma}^2_s(A)} A_s X_s,
\end{align}
where 
\begin{align}\label{eq:design_matrix_weighted}
    V_t (\lambda ) = \lambda \mathbb{I} + \sum_{s=1}^{t} \frac{1}{\hat{\sigma}^2_s(A_s)} A_s A^\mathsf{T}_s,
\end{align}
is the \textit{weighted design matrix}.

\textbf{Note.} The above quantities have been defined in the context of a linear bandit problem rather than the PSMAQB setting. However, recall from~\eqref{eq:PSMAQB_link_sphere} that there exists a one-to-one correspondence between the action set and environment of the qubit PSMAQB problem and the unit sphere $\mathbb{S}^2$. Therefore, although we will focus on linear bandits in the following, we will later make the explicit connection to the PSMAQB problem.

In the next sections, we will introduce a method that allows us to control the eigenvalues of the design matrix at any time. Combined with the use of weighted estimators, this will enable us to break the square-root barrier. Before presenting the general case, we first illustrate the technique in a simplified setting: a linear bandit on the circle $\mathbb{S}^1$. This will help us see how the method works in practice. Later we will consider the general case.

\subsection{Controlling the eigenvalues of the design matrix in the circle}\label{sec:circle_eigenvalues}

We now have the tool we want to use to solve our problem, but how should we use it? As mentioned earlier, we aim to design strategies that follow the optimistic principle, which in the case of \textsf{LinUCB} typically takes the form given in~\eqref{eq:optimistic_principle}. As previously discussed, the action selection rule derived from this principle leads to a \(\sqrt{T}\) regret in the analysis. For this reason, we explore an alternative way to implement the optimistic principle. To illustrate our idea, we begin by studying the setting where the action set is the circle \(\mathcal{A} = \mathbb{S}^1\) and the unknown parameter also lies on the circle, i.e., \(\theta \in \mathbb{S}^1\). The optimistic principle tells us to select action accordingly to
\begin{align}\label{eq:at_selection_circle}
    A_t \in \argmax_{A\in\mathbb{S}^1} \max_{\theta'\in\mathcal{C}_{t-1}} \langle \theta' , A \rangle = \mathcal{C}_{t-1}\cap \mathbb{S}^1 ,
\end{align}
where $\mathcal{C}_{t-1}\subset \mathbb{R}^2$ is a confidence region centered in the weighted estimator~\eqref{eq:estimator_weighted} with the same form as \textsf{LinUCB} i.e 
\begin{align}
    \mathcal{C}_{t-1} = \lbrace \theta'\in\mathbb{R}^2 : \| \theta' - \hat{\theta}_{t-1} \|^2_{V_{t-1}} \leq \tilde{\beta}_{t-1}  \rbrace ,
\end{align}
where $\lbrace \tilde{\beta}_{t}\rbrace_{t=1}^\infty \subset \mathbb{R}$ is some sequence that defines a well-behaved confidence region (we will work out the specific details later). Using the geometry of the problem we have that for $A_t$ choosen as in~\eqref{eq:at_selection_circle} we can reduce the task of bounding the regret to a geometry problem since
\begin{align}\label{eq:distance_circle}
    \max_{A\in\mathbb{S}^1} \langle \theta , A \rangle - \langle \theta, A_t \rangle = 1 - \langle \theta, A_t \rangle = \frac{1}{2}\langle \theta - A_t , \theta - A_t \rangle = \frac{1}{2}\|\theta - A_t \|_2^2.
\end{align}
Moreover using that $\lambda_{\min} (V_t(\lambda)) \mathbb{I} \leq V_t(\lambda)$ and $a_t\in\mathcal{C}_{t-1}$ we can compute the distance~\eqref{eq:distance_circle} as
\begin{align}
 \frac{1}{2}\|\theta - A_t \|_2^2   \leq \frac{\tilde{\beta}_{t-1}}{\lambda_{\min} (V_{t-1}(\lambda))}.
\end{align}
Thus, to bound the individual terms in the regret, it is enough to understand the behavior of the minimum eigenvalue of the design matrix. For standard \textsf{LinUCB}, it is known that the regret scales as $\tilde{O}(\sqrt{T})$, which suggests that $\lambda_{\min}(V_t(\lambda)) \sim \sqrt{t}$. In fact, in~\cite{banerjee2023exploration}, the authors show (under mild assumptions) that any strategy minimizing optimally the regret for linear bandits with continuous and smooth action sets must satisfy $\lambda_{\min}(V_t(\lambda)) = \Omega(\sqrt{t})$ for constant noise models.

Since our goal is to go beyond the square-root scaling, this relation is not convenient, and we are motivated to look for a more general condition. We note that $\lambda_{\max}(V_t(\lambda))$ captures the direction that has been exploited the most, as it relates to the shortest principal axis of the ellipsoid $\mathcal{C}_t$. For the unweighted design matrix~\eqref{eq:design_matrix} with normalized actions $A_t$, we have $\Tr(V_t(\lambda)) = \lambda + t$, so $\lambda_{\max}(V_t(\lambda)) \sim t$, which suggests the relation 
\begin{align}\label{eq:ratio_eigenvalues} 
\lambda_{\min}(V_t(\lambda)) = \Omega\left( \sqrt{\lambda_{\max}(V_t(\lambda))} \right), 
\end{align} 
for any strategy over continuous smooth action sets. Can we prove this relation for the weighted version of \textsf{LinUCB}? It turns out that even in the unweighted case, this is still an open problem in~\cite{banerjee2023exploration}.

In our case, we design an alternative strategy tailored to our setting that ensures this ratio holds. The idea is to select actions by projecting the extremal points of the longest axis of the ellipsoid $\mathcal{C}_t$ onto $\mathbb{S}^1$, as illustrated in Figure~\ref{fig:scheme_actions}. This particular choice helps “stabilize” the ellipsoid and guarantees the ratio~\eqref{eq:ratio_eigenvalues}. Furthermore, we aim to define a general selection rule such that the ratio~\eqref{eq:ratio_eigenvalues} holds independently of both the algorithm’s randomness and the noise model. This is formalized in the following theorem

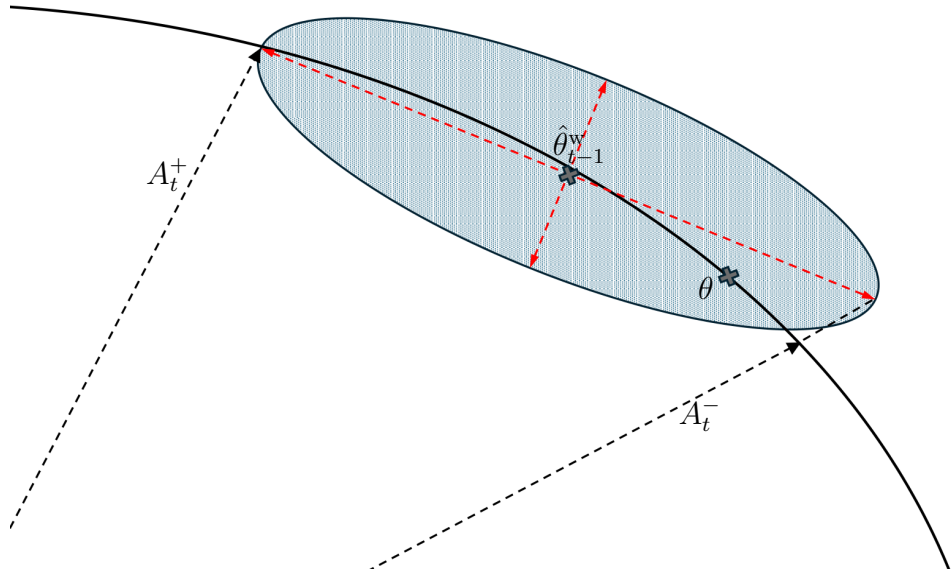
\begin{figure}
\centering

\begin{overpic}[percent,width=0.8\textwidth]{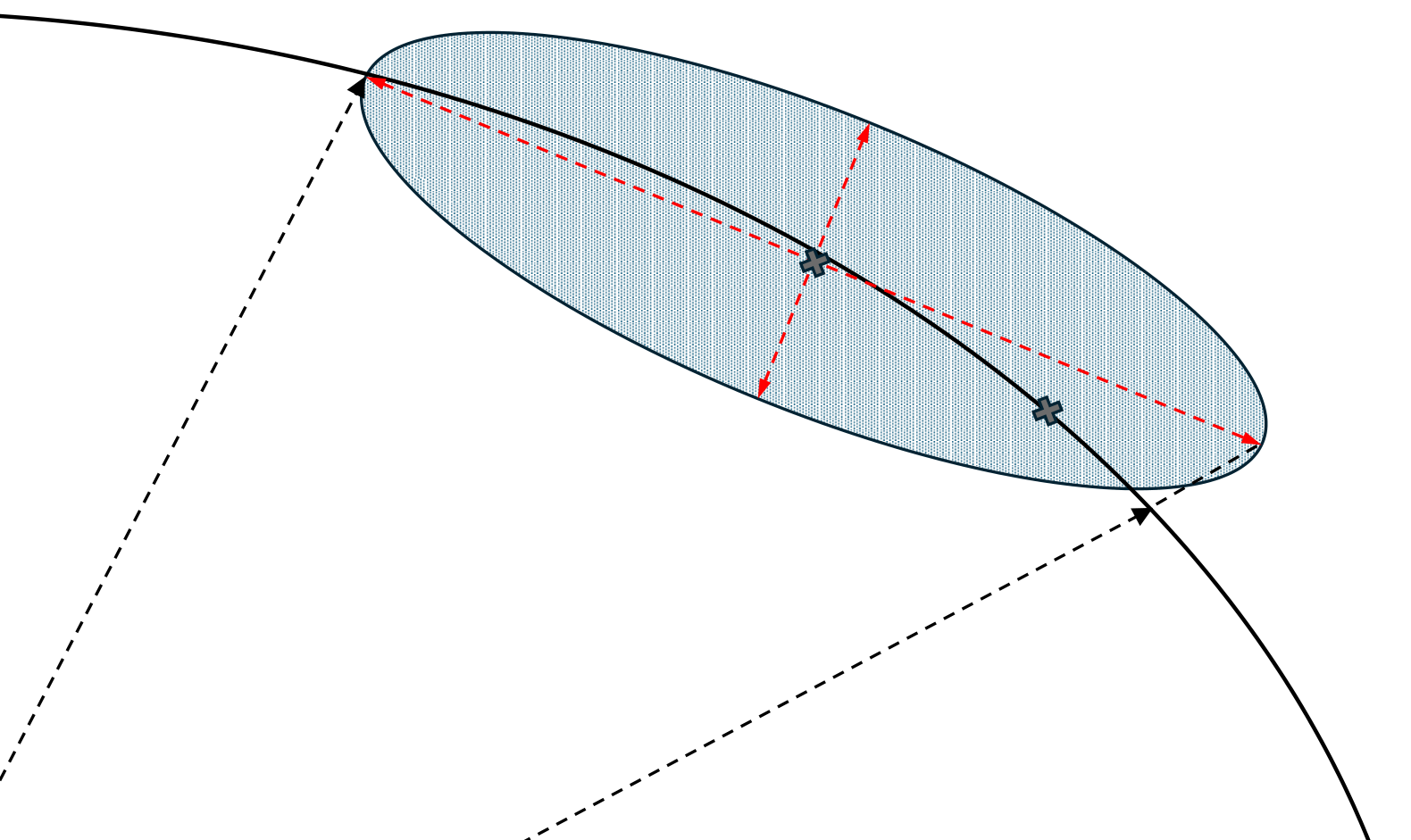}
\put(13,40){ $A^+_t$}
\put(68,15){ $A^-_t$}
\put(70,28){ $\theta$}
\put(55,43){ $\hat{\theta}^{\text{w}}_{t-1}$}
\end{overpic}
\caption{Scheme for the choice of actions $A_t^+,A_t^-$ as an altertative to \textsf{LinUCB}. The actions are selected as the projections of the extremal points across the largest axis of the confidence region centered around a weighted least squares estimator $\hat{\theta}^{\text{w}}_t$ of the unknown parameter $\theta$. This choice is sufficient to increase the minimum eigenvalue of $V_t$ such that the relation $\lambda_{\min}(V_t) = \Omega (\sqrt{\lambda_{\max}(V_t)})$ is satisfied. Moreover, the actions $A_t^+$ and $A_t^-$ are sufficiently close to $\theta$ to keep the regret small.}
\label{fig:scheme_actions}
\end{figure}

\begin{theorem}\label{th:circle_eigenvalues}
    Let $\lbrace c_t \rbrace_{t=0}^\infty \subset \mathbb{S}^1$ be a sequence of normalized vectors and $\omega: \text{P}^2_+ \rightarrow \mathbb{R}_{\geq 0}$ a function such that 
    \begin{align}
        \omega (X) \leq C\sqrt{\lambda_{\max} (X)},
    \end{align}
     for a constant $C > 0$ and any $X\in \text{P}^2_+$. Let $\lambda_0 \geq 4\sqrt{2}C+2$, and define a sequence of matrices
     $\lbrace V_t \rbrace_{t=0}^\infty \subset \mathbb{R}^{2\times 2}$ as
       \begin{align}\label{eq:vt_d2}
         V_0 := \lambda_0\mathbb{I}_{2\times2}, \quad      V_{t+1} := V_t + \omega ( V_t ) \left(a^+_{t+1}  (a^+_{t+1})^\mathsf{T}  +  a^-_{t+1}  (a^-_{t+1})^\mathsf{T} \right), 
       \end{align}
       where $a^+_{t+1},a^-_{t+1}\in\mathbb{S}^1$ are defined as
     \begin{align}
        a^+_{t+1} := \frac{c_t + \frac{1}{\sqrt{\lambda_{\min}(V_t)}}v_{t,\min}}{\sqrt{1+\frac{2 \langle c_t, v_{t,\min} \rangle }{\sqrt{\lambda_{\min}(V_t)}} + \frac{1}{\lambda_{\min}(V_t)}}}, \quad  a^-_{t+1} := \frac{c_t - \frac{1}{\sqrt{\lambda_{\min}(V_t)}}v_{t,\min}}{\sqrt{1-\frac{2 \langle c_t, v_{t,\min} \rangle }{\sqrt{\lambda_{\min}(V_t)}} + \frac{1}{\lambda_{\min}(V_t)}}}, 
        \end{align}
    and $v_{\min,t}$ is the normalized eigenvector corresponding to the minimum eigenvalue of $V_t$. Then we have
    \begin{align}
        \lambda_{\min}(V_t) \geq \sqrt{2\lambda_{\max}(V_t)} \quad \text{for all}\quad t\geq 0.
    \end{align}
  
\end{theorem}

\begin{proof}
 To simplify the notation in the proof we define
 \begin{align}
     \lambda_{\min,t} := \lambda_{\min}(V_t) , \quad \lambda_{\max,t} : = \lambda_{\max}(V_t),
 \end{align}
 with corresponding normalized eigenvectors $v_{t,\min},v_{t,\max}\in\mathbb{S}^2$. We will also need the following properties for positive semidefinite matrices $A,B\in \text{P}^d_+ $:
  
\begin{align}\label{eq:pertboundsAB}
       \lambda_{\min}( A+B) &\geq \lambda_{\min}( A ) + \lambda_{\min}( B ), \\
       \lambda_{\max}(A+B) &\leq \lambda_{\max}(A ) + \lambda_{\max}(B). \nonumber
\end{align}

 We are going to prove the result by induction. At time step $t=0$ we have $\lambda_{\min,t} = \lambda_{\max,t} = \lambda_0$ and the inequality $\lambda_0 \geq \sqrt{2\lambda_0}$ holds since $\lambda_0 \geq 2$.
 Then at time step $t\geq 1$ we can express $V_t$ as
\begin{align}
    V_t = \begin{pmatrix}
        \lambda_{t,\min} & 0 \\
        0 & \lambda_{t,\max}
    \end{pmatrix},
\end{align}
using the basis $\lbrace v_{t,\min}, v_{t,\max}\rbrace$. Then we assume that the inequality $\lambda_{t,\min} \geq \sqrt{2\lambda_{\max,t}}$ is satisfied and we want to check the same inequality at time step $t+1$. We express the vectors that update $V_{t+1}$~\eqref{eq:vt_d2} on the same basis as
\begin{align}\label{eq:atminmaxbasis}
    a^+_{t+1} = \begin{pmatrix}
        \langle a^+_{t+1}, v_{t,\min} \rangle  \\
         \langle a^+_{t+1} ,  v_{t,\max} \rangle
    \end{pmatrix} \quad  a^-_{t+1} = \begin{pmatrix}
          \langle a^-_{t+1} ,  v_{t,\min} \rangle  \\
          \langle a^-_{t+1} ,  v_{t,\max} \rangle
    \end{pmatrix},
\end{align}
and define the following quantities,
    \begin{align}
        x_t &:= \langle a^+_{t+1}, v_{t,\min} \rangle^2 + \langle a^-_{t+1} , v_{t,\min} \rangle^2, \\
        y_t &:= \langle a^+_{t+1} , v_{t,\max} \rangle^2 + \langle a^-_{t+1} ,  v_{t,\max} \rangle^2 , \\
        z_t &:= \langle a^+_{t+1} , v_{t,\min} \rangle  \langle a^+_{t+1} , v_{t,\max} \rangle + \langle a^-_{t+1} , v_{t,\min} \rangle  \langle a^-_{t+1} , v_{t,\max} \rangle  . 
    \end{align}
Using that $a^+_{t+1},a^-_{t+1}\in\mathbb{S}^1$  we get the following relation
\begin{align}\label{eq_relationxy}
    y_t = 2 - x_t .
\end{align}
Thus the perturbation at time step $t+1$ can be written as
\begin{align}
a^+_{t+1}  (a^+_{t+1})^\mathsf{T}  +  a^-_{t+1}  (a^-_{t+1})^\mathsf{T}  = \begin{pmatrix}
         x_t & z_t \\
         z_t & y_t
     \end{pmatrix},
\end{align}     
and $V_{t+1}$ as
\begin{align}
    V_{t+1} = \begin{pmatrix}
        \lambda_{t,\min} + \omega ( V_t ) x_t & \omega ( V_t )z_t \\
        \omega ( V_t )z_t & \lambda_{t,\max} + \omega ( V_t )(2-x_t)
    \end{pmatrix}.
\end{align}
In order to analyze the eigenvalue of $V_{t+1}$ we have to control the overlap $\langle v_{t,\min} , c_t \rangle$, so we define the following variable
    \begin{align}
        \alpha_t : = \langle v_{t,\min} , c_t \rangle \in [-1,1],
    \end{align}
and using that
\begin{align}
    \langle v_{t,\min} , a^+_{t+1} \rangle = \frac{\alpha_t + \frac{1}{\sqrt{\lmin}}}{\sqrt{1 +\frac{2}{\sqrt{\lmin}}\alpha_t + \frac{1}{\lmin}}}, \quad \langle v_{t,\min} , a^-_{t+1} \rangle = \frac{\alpha_t - \frac{1}{\sqrt{\lmin}}}{\sqrt{1-\frac{2}{\sqrt{\lmin}}\alpha_t + \frac{1}{\lmin}}}
\end{align}
we can express $x_t,z_t$ in terms of $\alpha_t$ as
\begin{align}
    x_t(\alpha_t) &= \frac{\big( \alpha_t + \frac{1}{\sqrt{\lmin}}\big)^2}{1+\frac{2}{\sqrt{\lmin}}\alpha_t + \frac{1}{\lmin}} + \frac{\big( \alpha_t - \frac{1}{\sqrt{\lmin}}\big)^2}{1-\frac{2}{\sqrt{\lmin}}\alpha_t + \frac{1}{\lmin}}, \\
    z_t(\alpha_t) &= \left( \frac{\alpha_t + \frac{1}{\sqrt{\lmin}}}{1+\frac{2}{\sqrt{\lmin}}\alpha_t + \frac{1}{\lmin}} + \frac{ \alpha_t - \frac{1}{\sqrt{\lmin}}}{1-\frac{2}{\sqrt{\lmin}}\alpha_t + \frac{1}{\lmin}}\right)\sqrt{1-\alpha_t^2}.
\end{align}
Then we can directly compute the minimum eigenvalue of $V_{t+1}$ as
\begin{align}\label{eq:exact_minimumeig}
    \lambda_{\min,t+1}(\alpha_t) =& \frac{\lambda_{\max,t}+\lmin}{2} + \omega ( V_t ) \\ & - \frac{1}{2}\sqrt{(\lambda_{\max,t} - \lambda_{\min,t}+2\omega ( V_t )(1-x_t(\alpha_t)) )^2 + 4\omega^2 ( V_t ) z_t^2(\alpha_t)}.
\end{align}
Now we define the difference of eigenvalues at time step $t$,
\begin{align}
    \phi_t := \lambda_{\max,t} - \lambda_{\min,t},
\end{align}
and analyze two different regimes for the induction to hold at time step $t+1$.

%\begin{center}
%\underline{\textbf{Case 1}: $\phi_t \geq 2\omega ( V_t )$.}
%\end{center}

\vspace{2mm}
{\textbf{Case 1}: $\phi_t \geq 2\omega ( V_t )$}
\vspace{2mm}

For this case we want to justify that $\lambda_{\min,t+1}(\alpha_t) \geq \lambda_{\min,t+1}(0)$ for $\alpha_t\in[-1,1]$, and later prove the induction using this lower bound. Defining the following
\begin{align}
    f(\alpha_t) &:= (\phi_t +2\omega ( V_t )(1-x(\alpha_t)) )^2 + 4\omega^2 ( V_t )z^2(\alpha_t) \\
    & = \phi^2_t + 4\omega ( V_t )\big( \phi_t ( 1 - x(\alpha_t)) + \omega ( V_t )(1-x(\alpha_t))^2 + \omega ( V_t )z^2 ( \alpha_t ) \big),
\end{align}
and comparing with the exact minimum eigenvalue~\eqref{eq:exact_minimumeig} we see that it suffices to prove that $f(\alpha_t)$ achieves a maximum at $\alpha_t=0$ in the range $\alpha_t\in [-1,1]$ in order to have $\lambda_{\min,t+1}(\alpha_t) \geq \lambda_{\min,t+1}(0)$ for all $\alpha_t\in[-1,1]$. A direct computation shows that
\begin{align}
    &f(\alpha_t ) = \phi_t^2 + 4\omega ( V_t )g(\alpha_t), \\
    &g(\alpha_t) :=\frac{2\big(\lmin-1\big)\lmin\phi_t \alpha_t^2 +(1-\lmin^2)\phi_t - (1-\lmin )^2\omega ( V_t )}{4\lmin \alpha_t^2 - \big( 1 + \lmin \big)^2 }.
\end{align}
We have that $g(\alpha_t)$ is of the form $g(x) = \frac{ax^2+b}{cx^2+d}$ which has an unique maximum at $x=0$ if $ad-bc < 0$. Then identifying the coefficients we have to check that
\begin{align}
  p:=  2(1-\lmin)\lmin (1+\lmin)^2 \phi_t - 4\lmin ( 1-\lmin^2)\phi_t + 4 \lmin (1-\lmin )^2 \omega ( V_t )< 0 .
\end{align}
Using that $\lambda_{\min,t}\geq 2$ and $\omega ( V_t ) \leq \frac{\phi_t}{2}$ we can bound the above as
\begin{align}
    p \leq \phi_t \left( 2(1-\lmin)\lmin (1+\lmin)^2  - 4\lmin ( 1-\lmin^2) + 2 \lmin (1-\lmin )^2 \right).
\end{align}
Finally summing all the terms we get,
\begin{align}
    p \leq -2\phi_t \left(\lambda_{\min,t}^2(1-\lambda_{\min,t})^2 \right) \leq 0 ,
\end{align}
where the inequality follows from $\phi_t \geq 0$. Thus, we conclude that 
\begin{align}\label{eq:minimumeig_inequality}
\lambda_{\min,t+1}(\alpha_t) \geq \lambda_{\min,t+1}(0).
\end{align}
Using that 
\begin{align}
x_t (0 ) = \frac{2}{1+\lambda_{\min,t}} , \quad z_t(0) = 0 ,
\end{align}
and $\phi_t \geq 2\omega (V_t ) $ we have
\begin{align}\label{eq:lambdamin0}
    \lambda_{\min,t+1}( 0 ) = \lambda_{\min,t} + \frac{2\omega ( V_t )}{1+\lambda_{\min,t}}.
\end{align}
Thus, using~\eqref{eq:minimumeig_inequality}\eqref{eq:lambdamin0} and 
\begin{align}
\Tr ( V_{t+1} ) = \lambda_{\min,t+1}(\alpha_t) + \lambda_{\max,t+1}(\alpha_t) = \lambda_{\max,t} + \lambda_{\min,t} + 2\omega ( V_t ),
\end{align}
we can upper bound the maximum eigenvalue as
\begin{align}\label{eq:maximumeig_bound}
    \lambda_{\max,t+1}(\alpha_t) &= \lambda_{\max,t} + \lambda_{\min ,t} - \lambda_{\min,t+1}(\alpha_t) + 2\omega ( V_t ) \\
    &\leq \lambda_{\max,t} + \lambda_{\min ,t} - \lambda_{\min,t+1}(0) + 2\omega ( V_t ) \\
     &= \lambda_{\max,t}+ 2\omega ( V_t )\frac{\lambda_{\min,t}}{1+\lambda_{\min,t}}.
\end{align}
Finally in order to check the induction step at $t+1$, $\lambda_{\min,t+1} \geq  \sqrt{2\lambda_{\max,t+1}}$ we can use the bounds~\eqref{eq:minimumeig_inequality}\eqref{eq:maximumeig_bound} and it suffices to check
\begin{align}
    \lambda_{\min,t} + \frac{2\omega ( V_t )}{1+\lambda_{\min,t}} \geq  \sqrt{2\left(\lambda_{\max,t} + 2\omega ( V_t )\frac{\lambda_{\min,t}}{1+\lambda_{\min,t}}\right)}.
\end{align}
Squaring both sides and rearranging we obtain
\begin{align}
    \underbrace{\lambda^2_{\min,t} - 2\lambda_{\max,t}}_{(i)} + \underbrace{\frac{4\omega^2 ( V_t )}{(1+\lambda_{\min,t})^2}}_{(ii)} \geq 0 ,
\end{align}
where $(i)$ is positive because we assume $\lambda_{\min,t} \geq \sqrt{2\lambda_{\max,t}}$ to hold at time step $t$ and  $(ii)$ is always positive. This concludes the case $\phi_t \geq 2 \omega ( V_t )$ induction.

%{\centering{\underline{\textbf{Case 2}: $\phi_t \leq 2\omega (V_t ) $}}\par} 
\vspace{2mm}
{\textbf{Case 2}: $\phi_t \leq 2\omega (V_t ) $}
\vspace{2mm}

From the definition of $\lambda_0$ we have
\begin{align}
    \lambda_{\min,t} &\geq \lambda_{0} \geq 4\sqrt{2}C+2.
    \end{align}
Multiplying both sides by $\lambda_{\min,t}$,    
    \begin{align}
    \lambda_{\min,t}^2 &\geq 4\sqrt{2}C\lambda_{\min,t}+2\lambda_{\min,t} \\
    &\geq 2\left( 4C\sqrt{\lambda_{\max,t}} + \lambda_{\min,t} \right) \hspace{20mm} \\
    &\geq 2\left( \lambda_{\min,t} + 4\omega (V_t ) \right) \\
    &\geq 2 \left( \lambda_{\min,t} +\phi_t+ 2\omega (V_t ) \right) \\
    & = 2\left( \lambda_{\max,t} + 2\omega (V_t ) \right) ,
\end{align}
where in the second inequality we used the induction hypothesis $\lambda_{\min,t}\geq \sqrt{2\lambda_{\max,t}}$, the third inequality $\omega(V_t) \leq C\sqrt{\lambda_{\max,t}}$, the fourth inequality $\phi_t\leq 2\omega(V_t )$ and the last equality the definition $\phi_t = \lambda_{\max,t} - \lambda_{\min,t}$.
Thus, we have
\begin{align}
    \lambda_{\min,t} \geq \sqrt{2(\lambda_{\max,t}+2\omega (V_t ))},
\end{align}
and the inequality at time step $t+1$, $\lambda_{\min,t} \geq \sqrt{2\lambda_{\max,t+1}}$ follows from the bounds
\begin{align}
    \lambda_{\min,t+1} &\geq \lambda_{\min,t}, \\
    \lambda_{\max,t+1} &\leq \lambda_{\max,t} + 2\omega (V_t ),
\end{align}
where we used~\eqref{eq:pertboundsAB} with $\lambda_{\min}(a_ta_t^\mathsf{T}) = \lambda_{\min}(b_tb_t^\mathsf{T}) = 0 $ and $\lambda_{\max}(a_ta_t^\mathsf{T}) = \lambda_{\max}(b_tb_t^\mathsf{T}) = 1$.
\end{proof}

\subsection{\textsf{LinUCB} on the circle without elliptical potential lemma}

Before presenting the full algorithm, we illustrate how the previously described action selection rule works and how it can already achieve sublinear regret. We focus on the circle example where the action set is $\mathcal{A} = \mathbb{S}^1$ and the unknown parameter is $\theta \in \mathbb{S}^1$, without making further assumptions on the noise model beyond a bounded subgaussian parameter. In Theorem~\ref{th:circle_eigenvalues} we will set the weights $\omega ( X) = 1 $. Thus, we will will not weight the design matrix and work with the usual least squares estimator~\eqref{eq:lse}. The algorithm operates in batches of 2 actions, with a total of $\tilde{T}$ batches, leading to a time horizon of $T = 2\tilde{T}$. At each batch $\tilde{t} \in [\tilde{T}]$, the algorithm selects the actions $A^+_{\tilde{t}}$ and $A^-_{\tilde{t}}$, defined as follows:
\begin{align}\label{eq:actions_circle}
    A^\pm_{\tilde{t}} := \frac{\tilde{A}^\pm_{\tilde{t}}}{\| \tilde{A}^\pm_{\tilde{t}} \|},   \quad \tilde{A}^\pm_{\tilde{t}} := \frac{\hat{\theta}_{\tilde{t}-1}}{\| \hat{\theta}_{\tilde{t}-1} \|_2} \pm \frac{1}{\sqrt{\lambda_{\min}(V_{\tilde{t}
    -1})}} v_{\tilde{t}-1}
\end{align}
where $v_{\tilde{t}-1}$ is the normalized eigenvector of minimum eigenvalue $\lambda_{\min}(V_{\tilde{t}-1})$ of the design matrix
    \begin{align}
        V_{\tilde{t}} = \mathbb{I} + \sum_{s=1}^{\tilde{t}} \left(A^+_{s}(A^+_{s})^{\mathsf{T}} + A^-_{s}(A^-_s )^\mathsf{T}   \right),
    \end{align}
and the least squares estimator is
\begin{align}
   \hat{\theta}_{\tilde{t}} =    V^{-1}_{\tilde{t}}\sum_{s=1}^{\tilde{t}} \left(X_s^+ A_s^+ + X_s^- A_s^- \right),
\end{align}
where $X_s^\pm$ are the rewards sampled for actions $ A^\pm_{s}$. The pseudocode can be found in Algorithm~\ref{alg:linucb_circle}. In order to prove the regret bound we are going to use the confidence region of \textsf{LinUCB} stated in Theorem~\ref{th:confidence_region}, which in our case will have the form
\begin{align}\label{eq:confidence_region_circle}
    \mathcal{C}_{\tilde{t}}  = \lbrace \theta^{'}\in\mathbb{R}^2 : \| \theta' - \hat{\theta}_{\tilde{t}} \|^2_{V_{\tilde{t}}} \leq \beta_{\tilde{t},\delta} \rbrace , 
\end{align}
with 
\begin{align}\label{eq:beta_circle}
    \beta_{\tilde{t},\delta} = \left(1+\sqrt{2\log\frac{1}{\delta} + 2 \log\left(1+\tilde{t} \right)} \right)^2 .
\end{align}

\begin{algorithm}
	\caption{\textsf{LinUCB} on $\mathbb{S}^1$} 
	\label{alg:linucb_circle}
    \begin{algorithmic}[1]
        
       \State Set initial design matrix $V_0 \gets \mathbb{I}_{2\times2}$ 
        
       \State Choose initial estimator $\hat{\theta}_0\in\mathbb{S}^1$ for $\theta$ at random 
        
        \For{$\tilde{t}=1,2,\cdots$}

              \State  Select actions $A_{\tilde{t}}^+$ and $ A_{\tilde{t}}^- $  according to Eq.~\eqref{eq:actions_circle}
              \State  Receive associated rewards $X^+_{\tilde{t}}$ and $X^-_{\tilde{t}}$
             
            \vspace{1mm}
            \textit{Update design matrix and LSE}
            \vspace{1mm}
            
           \State $  V_{\tilde{t}} \gets   V_{\tilde{t}-1}  +   A^+_{\tilde{t}}(A^+_{\tilde{t}})^{\mathsf{T}} + A^-_{\tilde{t}}(A^-_{\tilde{t}} )^\mathsf{T} $
            
          \State  $\widetilde{\theta}_t^\text{w} \gets V^{-1}_{\tilde{t}}\sum_{s=1}^{\tilde{t}} \left(X_s^+ A_s^+ + X_s^- A_s^- \right) $
          \EndFor
        \end{algorithmic}
\end{algorithm}

Before proving the regret scaling we state a Lemma that will quantify how far are the actions from the environment.

\begin{lemma}\label{lem:dist_action_centre}
    Given two normalized vectors $c,v\in\mathbb{S}^{d-1}$, a positive constant $\lambda > 1$ and the following vectors
    \begin{align}
        \widetilde{a}^{\pm} = c \pm \frac{1}{\sqrt{\lambda}}v, \quad {a}^\pm = \frac{\widetilde{a}^\pm}{\| \widetilde{a}^\pm \|_2}.
    \end{align}
    Then we have 
    \begin{align}
        \|{a}^\pm - c\|_2^2 \leq \frac{2}{\lambda}.
    \end{align}
\end{lemma}

\begin{proof}
We are going to give the proof for ${a}^-$. The one for ${a}^+$ follows from an identical calculation.

First, we can relate the distance to the inner product of the vectors as
    \begin{align}\label{eq:acdistance}
        \|{a}^- - c\|_2^2 = \langle {a}^- - c , {a}^- - c \rangle = 2 - 2\langle {a}^-, c \rangle.
    \end{align}
    Using the normalization factor is
    \begin{align}
        \| \widetilde{a}^- \|_2^2 =  1 - \frac{2}{\sqrt{\lambda}}\langle c, v \rangle + \frac{1}{\lambda},
    \end{align}
    then
    \begin{align}
      \langle {a}^-, c \rangle =   \frac{1-\frac{1}{\sqrt{\lambda}}\alpha}{\sqrt{1+\frac{1}{\lambda}-\frac{2}{\sqrt{\lambda}}\alpha}}, \quad \text{where} \quad \alpha: = \langle c ,v \rangle \in [-1,1].
    \end{align}
    In order to study the behavior of~\eqref{eq:acdistance} in terms of the overlap $\langle c , v \rangle$ we define
    \begin{align}\label{eq:def_distance_lambda}
        f(\alpha,\lambda) := 2 \left( 1 - \frac{1-\frac{1}{\sqrt{\lambda}}\alpha}{\sqrt{1+\frac{1}{\lambda}-\frac{2}{\sqrt{\lambda}}\alpha}} \right) 
    \end{align}
    and we are going to check the maximum in the range $\alpha\in [-1,1]$. Note that $f(\alpha,\lambda) =  \|{a}^- - c\|_2^2$.
    Taking the derivative respect to $\alpha$,
    \begin{align}
        \frac{\partial f(\alpha,\lambda)}{\partial \alpha} = 2\frac{\frac{1}{\sqrt{\lambda}}\sqrt{1+\frac{1}{\lambda}-\frac{2}{\sqrt{\lambda}}\alpha} - \left( 1 - \frac{1}{\sqrt\lambda}\alpha \right)\frac{1}{\sqrt{\lambda}\sqrt{1+\frac{1}{\lambda}-\frac{2}{\sqrt{\lambda}}\alpha}}}{1+\frac{1}{\lambda}-\frac{2}{\sqrt{\lambda}}\alpha}
    \end{align}
    Then we can find the extremal points of $f(\alpha)$ as
    \begin{align}
        \frac{\partial f(\alpha,\lambda)}{\partial \alpha} &= 0 \Leftrightarrow 1+\frac{1}{\lambda}-\frac{2}{\sqrt{\lambda}}\alpha = 1 - \frac{1}{\sqrt\lambda}\alpha  \\
        &\Leftrightarrow \alpha = \frac{1}{\sqrt{\lambda}}
    \end{align}
    Using that $\lambda > 1$ we have the following inequalites
    \begin{align}
        \frac{\partial f(\alpha,\lambda)}{\partial \alpha} \bigg\vert_{\alpha=1} = -\frac{2}{\lambda\left(1-\frac{1}{\sqrt\lambda}\right)^2} < 0, \quad \frac{\partial f(\alpha,\lambda)}{\partial \alpha}\bigg\vert_{\alpha=0} = \frac{2}{\sqrt{\lambda}\sqrt{1+\frac{1}{\lambda}}}\left( 1 - \frac{1}{1+\frac{1}{\lambda}} \right) > 0.
    \end{align}
    Thus, $\alpha = \frac{1}{\sqrt{\lambda}}$~\eqref{eq:def_distance_lambda} is a maximum and it is achieved at
    \begin{align}
        f\left( \frac{1}{\sqrt{\lambda}} ,\lambda \right) = 2\left( 1- \sqrt{1-\frac{1}{\lambda}} \right).
    \end{align}
    Using the original definition of $f(\alpha,\lambda)$ we can conclude that 
    \begin{align}
         \|{a}^- - c\|_2^2 \leq 2\left( 1- \sqrt{1-\frac{1}{\lambda}} \right) \leq \frac{2}{\lambda} ,
    \end{align}
    where the last inequality follows from $\lambda > 1$, and $1-x \leq \sqrt{1-x}$ for $0\leq x < 1$.
\end{proof}

This above Lemma actually holds for any dimension and we will also apply it when we study higher dimensional cases. We are now ready to prove a Theorem showing that selecting actions by projecting the extremal points of the principal axis of $\mathcal{C}_{\ttilde}$ onto the circle is enough to achieve the same regret bound as $\textsf{LinUCB}$. The initial part of the proof will also be useful for later Theorems when we upper bound the instantaneous regret in the full setting with vanishing noise.

\begin{theorem}\label{th:circle_regret}
    Let $T = 2\tilde{T}$ for $\tilde{T}\in\mathbb{N}$. Then if we apply Algorithm~\ref{alg:linucb_circle} to a linear bandit with action set $\mathcal{A} = \mathbb{S}^1$ and reward model $X_t = \langle \theta , A_t \rangle + \epsilon_t$ where $\epsilon_1$ is 1-subgaussian we have that with probability at least $1-\delta$
    \begin{align}
        R_T ( \mathbb{S}^1 , \theta , \pi ) \leq \frac{9\sqrt{2}}{4} \sqrt{T}\left(1+\sqrt{2\log\frac{1}{\delta} + 2 \log\left(1+2T \right)} \right)^2 .
    \end{align}
\end{theorem}
\begin{proof}
    First we note that since $\theta , A_t\in\mathbb{S}^1$ we can express the regret as
    \begin{align}
        R_T ( \mathbb{S}^1 , \theta , \pi ) &= \frac{1}{2}\sum_{t=1}^T \| \theta - A_t \|_2^2 \\
        & =  \frac{1}{2}\sum_{\tilde{t}=1}^{\tilde{T}} \left( \| \theta - A^+_{\tilde{t}} \|_2^2 + \| \theta - A^-_{\tilde{t}} \|_2^2 \right) .
    \end{align}
    For now, we will assume that $\theta\in\mathcal{C}_{\tilde{t}}$ holds for all $\tilde{t}\in[\tilde{T}]$  where $\mathcal{C}_{\tilde{t}}$ has the form~\eqref{eq:confidence_region_circle}.

    Then from triangle inequality we have
\begin{align}\label{eq:theta_at_bound_circle}
    \| \theta - A^\pm_{\tilde{t}} \|_2 &\leq \| \theta - \hat{\theta}_{\tilde{t}-1} \|_2 + \| \hat{\theta}_{\tilde{t}-1}  - \frac{\hat{\theta}_{\tilde{t}-1}}{\| \hat{\theta}_{\tilde{t}-1} \|_2} \|_2 + \| \frac{\hat{\theta}_{\tilde{t}-1}}{\| \hat{\theta}_{\tilde{t}-1} \|_2} -A^\pm_{t} \|_2 \nonumber \\
    &\leq \underbrace{\sqrt{\frac{\beta_{\tilde{t}-1,\delta}}{\lambda_{\min}(V_{\tilde{t}-1})}}}_{(i)} + \underbrace{\sqrt{\frac{\beta_{\tilde{t}-1,\delta}}{\lambda_{\min}(V_{\tilde{t}-1})}}}_{(ii)} + \underbrace{\sqrt{\frac{2}{\lambda_{\min}(V_{\tilde{t}-1})}}}_{(iii)} \nonumber
    \\ &\leq 3\sqrt{\frac{\beta_{\tilde{t}-1,\delta}}{\lambda_{\min}(V_{\tilde{t}-1})}}, 
\end{align}
where for the above bounds we use:
\begin{itemize}
    \item $(i).$ Since $\theta\in\mathcal{C}_{t-1}$, by definition of $\mathcal{C}_{t-1}$~\eqref{eq:confidence_region_circle} we have $\| \theta -  \tilde{\theta}^{\text{w}}_{t-1} \|_{V_{t-1}} \leq \sqrt{\beta_{t-1,\delta}}$ and then the inequality follows from $\lambda_{\min}(V_{\tilde{t}-1})\mathbb{I}_{2\times 2}\leq V_{\tilde{t}-1}$.
    
    \item  $(ii).$ Since $\theta\in\mathcal{C}_{\tilde{t}-1}$ and $\theta\in\mathbb{S}^1$ this implies that $\mathbb{B}^2_r (\hat{\theta}_{\tilde{t}-1} )\cap \mathbb{S}^1$ is non-empty with $r = \sqrt{\frac{\beta_{\tilde{t}-1,\delta}}{\lambda_{\min}(V_{\tilde{t}-1})}}$ (the longest axis of the ellipsoid). Then using that $\mathcal{C}_{\tilde{t}-1}\subseteq \mathbb{B}^2_r (\hat{\theta}_{\tilde{t}-1} )$ and $\frac{\hat{\theta}_{\tilde{t}-1}}{\| \hat{\theta}_{\tilde{t}-1} \|_2} $ is the normalized vector of $\hat{\theta}_{\tilde{t}-1}$ which is the center of $\mathcal{C}_{\tilde{t}-1}$  we get $\frac{\hat{\theta}_{\tilde{t}-1}}{\| \hat{\theta}_{\tilde{t}-1} \|_2} \in\mathbb{B}^2_r (\hat{\theta}_{\tilde{t}-1} )$.
    
    \item For $(iii).$ We apply Lemma~\ref{lem:dist_action_centre} using the expression of $A^\pm_t$~\eqref{eq:actions_circle} with $c= \frac{\hat{\theta}_{\tilde{t}-1}}{\| \hat{\theta}_{\tilde{t}-1} \|_2}$ and $\lambda= \lambda_{\min}(V_{\tilde{t}-1})$.
\end{itemize}

We sketch this bounds in Figure~\ref{fig:scheme_regret}.

\begin{figure}\label{fig:regret_distances}
\centering
\begin{overpic}[percent,width=0.6\textwidth]{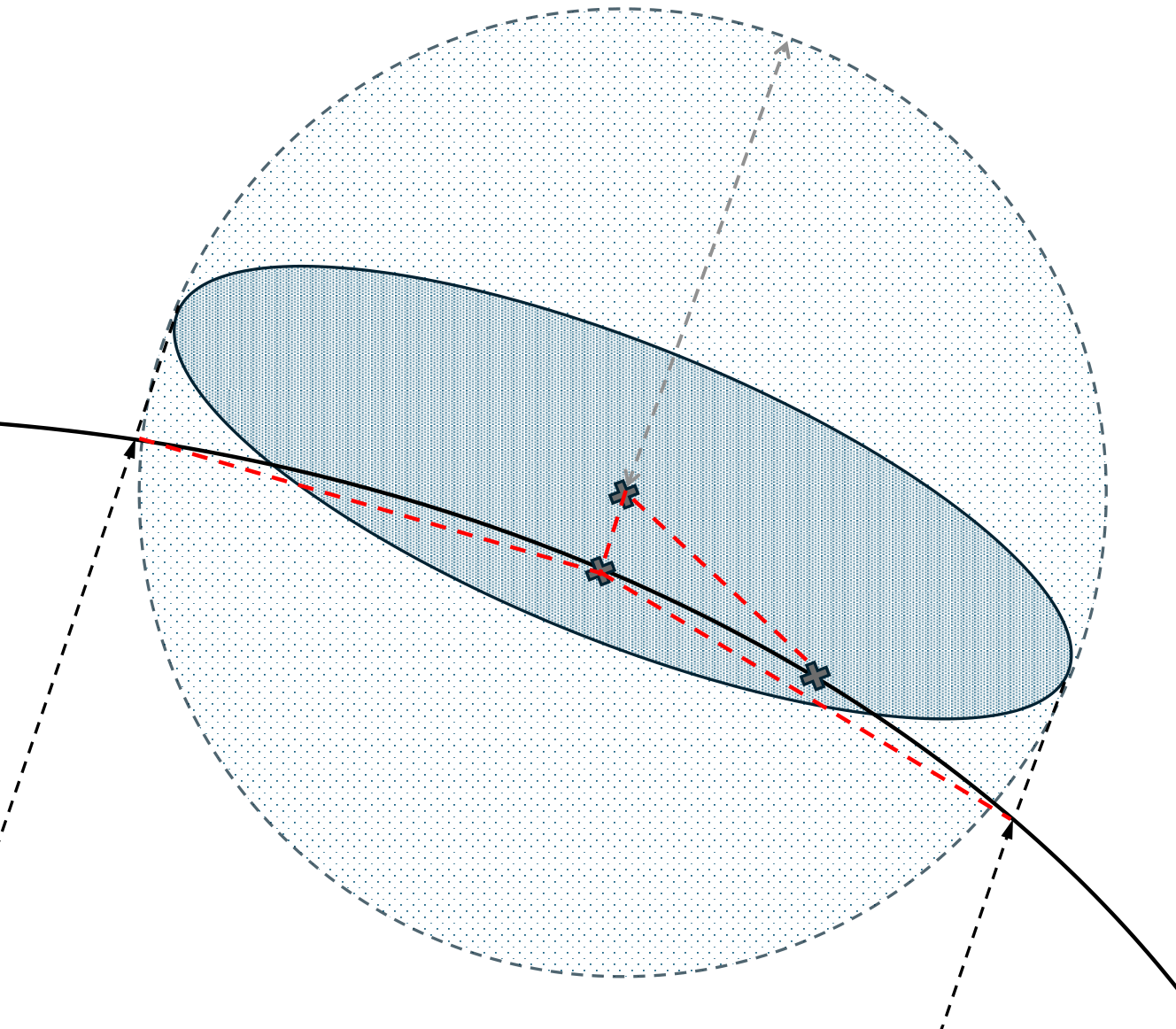}
\put(6,34){ $A^+_{\tilde{t}}$}
\put(78,12){ $A^-_{\tilde{t}}$}
\put(67,25){ $\theta$}
\put(54,46){ $\hat{\theta}_{\tilde{t}-1}$}
\put(46,32){ $\frac{\hat{{\theta}}_{\tilde{t}-1}}{\| {\hat{{\theta}}_{\tilde{t}-1} \|_2}}$ } 
\put(52,63){\rotatebox{70}{\small $r= \sqrt{\frac{\beta_{t-1,\delta}}{\lambda_{\min}(V_{\tilde{t}-1})}}$}}
\put(67,65){ $\mathbb{B}_r^2 (\hat{{\theta}}_{\tilde{t}-1}) $}
\put(35,55){ $\mathcal{C}_{\tilde{t}-1}$}
\end{overpic}
\caption{Sketch for the triangle inequality used to bound $ \| \theta - A^\pm_{\tilde{t}} \|_2 $ in~\eqref{eq:theta_at_bound_circle}. The red lines represent the distances $(i),(ii)$ and $(iii)$. Under the event $\theta\in\mathcal{C}_{\tilde{t}-1}$ we can use $\mathcal{C}_{\tilde{t}-1} \subseteq \mathbb{B}_r^2 (\hat{{\theta}}_{\tilde{t}-1}) $ with $r$ being the longest axis of the ellipsoid and bound all distances by the diameter $2r$.}
\label{fig:scheme_regret}
\end{figure}

\smallskip

To finalize the proof we need to study the eigenvalues of $V_{\tilde{t}}$. From the definition of $V_{\tilde{t}}$ we have
\begin{align}
    \Tr(V_{\tilde{t}} ) = 2 + 2\tilde{t} \Rightarrow \lambda_{\max}(V_{\tilde{t}} ) \geq \tilde{t},
\end{align}
where we used that $A^\pm_{\tilde{t}} \in\mathbb{S}^1$. Then using Theorem~\ref{th:circle_eigenvalues} we have
\begin{align}
    \lambda_{\min}(V_{\tilde{t}} ) \geq \sqrt{2\lambda_{\max}(V_{\tilde{t}})}\geq \sqrt{2\tilde{t}}
\end{align}
Thus combining the above bound with~\eqref{eq:theta_at_bound_circle} and $\beta_{\tilde{t},\delta}$ is an increasing sequence in $\tilde{t}$ leads to
\begin{align}
     R_T ( \mathbb{S}^1 , \theta , \pi ) & \leq \frac{9\beta_{\tilde{T},\delta}}{2} \sum_{\tilde{t}=1}^{\tilde{T}} \frac{1}{\sqrt{\tilde{t}}} \leq \frac{9\beta_{\tilde{T},\delta}}{2} \sqrt{\tilde{T}},
\end{align}
and the result follows using $T = 2 \tilde{T}$, and applying the probabilistic bound from Theorem~\ref{th:confidence_region}.

\end{proof}
\subsection{General case}

In the previous section, we saw how the action selection rule described in Section~\ref{sec:circle_eigenvalues} can be used to minimize the regret. However, so far this rule only works for $\mathbb{S}^1$. In this section, we generalize the same rule to work on any sphere $\mathbb{S}^{d-1}$. Instead of proving the ratio $\lambda_{\min}(V_t) = \Omega\left(\sqrt{\lambda_{\max}(V_t)}\right)$ by induction and directly computing the eigenvalues, we will use matrix perturbation bounds. Specifically, we will use a proof technique based on the following minimax characterization of eigenvalues for Hermitian matrices.

\begin{corollary}[ Corollary III.1.2 in~\cite{bhatia97} ]\label{cor:minimax_eigenvalues}
    Let $A \in \mathbb{C}^{d\times d}$ be a Hermitian matrix, then
    \begin{align}
        \lambda_{k} (A ) = \max_{\substack{\mathcal{M} \subset \mathbb{C}^d \\ \dim ( \mathcal{M} ) = d-k+1}} \min_{\substack{x\in\mathcal{M} \\ \| x\|_2 = 1}} \langle x , A x \rangle = \min_{\substack{\mathcal{M} \subset \mathbb{C}^d \\ \dim ( \mathcal{M} ) = k}} \max_{\substack{x\in\mathcal{M} \\ \| x\|_2 = 1}} \langle x , A x \rangle.
    \end{align}
    In particular, if $A \geq B$ then $\lambda_k (A) \geq \lambda_k (B) $.
\end{corollary}

The following Theorem is the generalization of Theorem~\ref{th:circle_eigenvalues}.

\begin{theorem}\label{th:main_eigenvalues}
    Let $d\geq 2$, $\lbrace c_t \rbrace_{t=0}^\infty \subset \mathbb{S}^{d-1}$ be a sequence of normalized vectors and $\omega: \textup{P}^d_+ \rightarrow \mathbb{R}_{\geq 0}$ a function such that 
    \begin{align}
        \omega (X) \leq C\sqrt{\lambda_{\max}(X)} \quad \forall X\in \textup{P}^d_+,
    \end{align}
     for some constant $C > 0$. Let $\lambda_0 \geq \max \big\lbrace 2,\sqrt{\frac{2}{3(d-1)}}2dC+\frac{2}{3(d-1)} \big\rbrace$, and define a sequence of matrices
     $\lbrace V_t \rbrace_{t=0}^\infty \subset \mathbb{R}^{d\times d}$ as
       \begin{align}\label{eq:vt_lemma}
         V_0 := \lambda_0\mathbb{I}_{d\times d}, \quad      V_{t+1} := V_t + \omega ( V_t ) \sum_{i=1}^{d-1}P_{t,i}, 
       \end{align}
       where 
       \begin{align}\label{eq:defP_a}
           P_{t,i} : = a^+_{t+1,i}(a^+_{t+1,i})^\mathsf{T}  +  a^-_{t+1,i} (a^-_{t+1,i})^\mathsf{T} , \quad
           a^\pm_{t+1,i} : = \frac{\tilde{a}^\pm_{t+1,i}}{\| \tilde{a}^\pm_{t+1,i}\|_2}, \quad \tilde{a}^\pm_{t+1,i} := c_t \pm \frac{1}{\sqrt{\lambda_{t,1}}} v_{t,i},
       \end{align}
       with $\lambda_{t,i} = \lambda_{i}(V_t)$ the eigenvalues of $V_t$ with corresponding normalized eigenvectors \\$v_{t,1},...,v_{t,d}\in\mathbb{S}^d$.
  Then we have
    \begin{align}\label{eq:eig_relation}
        \lambda_{\min}(V_t) \geq \sqrt{\frac{2}{3(d-1)}\lambda_{\max}(V_t)} \quad \text{for all}\quad t\geq 0.
    \end{align}
  
\end{theorem}

\begin{proof}
    We start giving an upper bound of the Euclidean norm of $\tilde{a}^\pm_{t+1,i} $~\eqref{eq:defP_a} with the following calculation
    \begin{align}
        \| \tilde{a}^\pm_{t+1,i} \|^2_2 &= 1 \pm \frac{2}{\sqrt{\lambda_{1,t}} }\langle c_t , v_{t,i} \rangle + \frac{1}{\lambda_{t,1}} \\
        & \leq 1 + \frac{2}{\sqrt{\lambda_{t,1}}} + \frac{1}{\lambda_{t,1}}  = \left( 1 + \frac{1}{\sqrt{\lambda_{t,1}}} \right)^2,
    \end{align}
    where we used that $c_t,v_{t,i}\in\mathbb{S}^d$. Thus using the definition of $P_{t,i}$~\eqref{eq:defP_a},
    \begin{align}
        P_{t,i} \geq \left( 1 + \frac{1}{\sqrt{\lambda_{t,1}}}  \right)^{-2} \big( \tilde{a}^+_{t+1,i}(\tilde{a}^+_{t+1,i})^\mathsf{T}  + \tilde{a}^-_{t+1,i} (\tilde{a}^-_{t+1,i})^\mathsf{T} \big) .
    \end{align}
    From the definition of $\tilde{a}^\pm_{t+1,i} $~\eqref{eq:defP_a} we have
    \begin{align}
    \tilde{a}^\pm_{t+1,i} (\tilde{a}^\pm_{t+1,i})^\mathsf{T} = c_t c_t^\mathsf{T} + \frac{1}{\lambda_{t,1}}v_{t,i} v^\mathsf{T}_{t,i} \pm \frac{1}{\sqrt{\lambda_{t,1}}} \left( c_t v_{t,i}^\mathsf{T} + v_{t,i}c_t^\mathsf{T} \right).
    \end{align}
    The above allow us to bound $P_{t,i}$ as 
    \begin{align}
        P_{t,i} &\geq 2 \left( 1 +\frac{1}{\sqrt{\lambda_{t,1}}} \right)^{-2}\left( c_t c^\mathsf{T}_t + \frac{1}{\lambda_{t,1}}v_{t,i} v^\mathsf{T}_{t,i} \right) \\
        &\geq \frac{2}{(1+\sqrt{\lambda_{t,1}})^2} v_{t,i}v^\mathsf{T}_{t,i},
    \end{align}
    where we used $c_t c^\mathsf{T}_t \geq 0$ and $\frac{1}{\lambda_{t,1}}\left( 1 +\frac{1}{\sqrt{\lambda_{t,1}}} \right)^{-2} = (1+\sqrt{\lambda_{t,i}})^{-2} $. Thus, from the abound bound and the definition of $V_t$~\eqref{eq:vt_lemma} we obtain the following matrix inequality
    \begin{align}\label{eq:vt+1_bound}
        V_{t+1} \geq V_t +  \frac{2 w(V_t)}{(1+\sqrt{\lambda_{t,1}})^2} \sum_{i=1}^{d-1} v_{t,i} v^\mathsf{T}_{t,i}.
    \end{align}

    We want to prove the result using induction. The case $t=0$ is immediately satisfied since $\lambda_{0,d} = \lambda_{0,1} \geq \frac{2}{3(d-1)}$.
    Now we will assume that 
    \begin{align}\label{eq:induction}
        \lambda_{t,1} \geq \sqrt{\frac{2}{3(d-1)}\lambda_{t,d}},
    \end{align}
    is satisfied and we want to prove the same inequality for $t+1$. We distinguish cases depending on the growth of the maximum eigenvalue of $V_t$.

    \begin{center}
\underline{\textbf{Case 1}: $\lambda_{t,d} \geq \lambda_{t,d-1} +  \frac{2 w(V_t)}{(1+\sqrt{\lambda_{t,1}})^2}$.}
\end{center}
%\textbf{Case 1:}  $\lambda_{t,d} \geq \lambda_{t,d-1} +  \frac{2 w(V_t)}{(1+\sqrt{\lambda_{t,1}})^2}$

Using the hypothesis $\lambda_{t,d} \geq \lambda_{t,d-1} +  \frac{2 w(V_t)}{(1+\sqrt{\lambda_{t,1}})^2}$, the fact that $V_t$ diagonalizes in the $\lbrace v_{t,i} \rbrace_{i=1}^d$ basis and the ordering  $\lambda_{t,1}\leq ....\leq \lambda_{t,d-1}\leq \lambda_{t,d}$ we have
\begin{align}
    \lambda_i \left( V_t +  \frac{2 w(V_t)}{(1+\sqrt{\lambda_{t,1}})^2} \sum_{i=1}^{d-1} v_{t,i} v^\mathsf{T}_{t,i} \right) &= \lambda_{t,i} + \frac{2 w(V_t)}{(1+\sqrt{\lambda_{t,1}})^2} \quad \text{for} \quad i= 1,...,d-1 \\
    \lambda_d \left(V_t +  \frac{2 w(V_t)}{(1+\sqrt{\lambda_{t,1}})^2} \sum_{i=1}^{d-1} v_{t,i} v^\mathsf{T}_{t,i} \right) &= \lambda_{t,d} .
\end{align}
Then combining with the mini-max principle for eigenvalues~\ref{cor:minimax_eigenvalues} and using that both sides of~\eqref{eq:vt+1_bound} are real and symmetric, we arrive at
\begin{align}\label{eq:lambdaminbound}
    \lambda_{i,t+1} \geq \lambda_{i,t} + \frac{2 w(V_t)}{(1+\sqrt{\lambda_{t,1}})^2} \quad \text{for} \quad i=1,...,d-1 .
\end{align}
From the expression for $V_{t+1}$~\eqref{eq:vt_lemma}, we deduce that
\begin{align}
    \Tr (V_{t+1}) = \big( \lambda_{t+1,d} + \sum_{i=1}^{d-1} \lambda_{t+1,i} \big) = \lambda_{t,d}+\sum_{i=1}^{d-1} \lambda_{t,i}+2(d-1)w(V_t).
\end{align}
Also from~\eqref{eq:lambdaminbound}
\begin{align}
    \Tr (V_{t+1} ) \geq \lambda_{t+1,d} + \sum_{i=1}^{d-1} \lambda_{t,i} + \frac{2(d-1) w(V_t)}{(1+\sqrt{\lambda_{t,1}})^2} . 
\end{align}
Combining the above we can bound the maximum eigenvalue as
\begin{align}\label{eq:lambda_upper}
    \lambda_{t+1,d} &\leq \lambda_{t,d}+\sum_{i=1}^{d-1} \lambda_{t,i}+2(d-1)w(V_t) - \left( \sum_{i=1}^{d-1} \lambda_{t,i} + \frac{2(d-1) w(V_t)}{(1+\sqrt{\lambda_{t,1}})^2} \right) \\
    &= \lambda_{d,t} + 2(d-1)w(V_t ) \frac{\lambda_{t,1}+2\sqrt{\lambda_{t,1}}}{(1+\sqrt{\lambda_{t,1}})^2}.
\end{align}
Recall that we want to check
\begin{align}
    \lambda_{t+1,1} \geq \sqrt{\frac{2}{3(d-1)}\lambda_{t+1,d}}.
\end{align}
Using~\eqref{eq:lambdaminbound} and~\eqref{eq:lambda_upper} we can square the above and see that it suffices to check
\begin{align}
    \left(\lambda_{1,t} + \frac{2 w(V_t)}{(1+\sqrt{\lambda_{t,1}})^2}\right)^2 \geq \frac{2}{3(d-1)} \lambda_{d,t} + \frac{4w(V_t )(\lambda_{t,1}+2\sqrt{ \lambda_{t,1}})}{3(1+\sqrt{\lambda_{t,1}})^2}.
\end{align}
Multiplying out the terms, this is equivalent to the condition
\begin{align}
    \underbrace{\lambda^2_{1,t} - \frac{2}{3(d-1)} \lambda_{d,t}}_{(i)} + \underbrace{\frac{4 w^2(V_t)}{(1+\sqrt{\lambda_{t,1}})^4}}_{(ii)} + \underbrace{\frac{4w(V_t)}{(1+\sqrt{\lambda_{t,1}})^2}(\lambda_{t,1} - \frac{1}{3}(\lambda_{t,1}+2\sqrt{\lambda_{t,1}})}_{(iii)} \geq 0.
\end{align}
It remains to observe that $(i)$ is positive by induction at time step $t$ (cf.~\eqref{eq:induction}), $(ii)$ is always positive and $(iii)$ is positive for $\lambda_{t,1} \geq 1$ and this is true since $\lambda_{t,1} \geq \lambda_0 \geq 2$, $\lambda_{t,1} $ is non-decreasing in $t$ and $f(x) = x-\sqrt{x}$ is positive for $x\geq 1$.
\begin{center}
\underline{\textbf{Case 2.1}: $ \lambda_{t,1} +  \frac{2 w(V_t)}{(1+\sqrt{\lambda_{t,1}})^2}<\lambda_{t,d} < \lambda_{t,d-1}+ \frac{2 w(V_t)}{(1+\sqrt{\lambda_{t,1}})^2}$.}
\end{center}

%\subsection*{\textbf{Case 2.1}: $ \lambda_{t,1} +  \frac{2 w(V_t)}{(1+\sqrt{\lambda_{t,1}})^2}<\lambda_{t,d} < \lambda_{t,d-1}+ \frac{2 w(V_t)}{(1+\sqrt{\lambda_{t,1}})^2}$}

First we find $k\in\mathbb{N}$, $1<k\leq d-1$ such that
\begin{align}
  \lambda_{t,k-1} +  \frac{2 w(V_t)}{(1+\sqrt{\lambda_{t,1}})^2} \leq  \lambda_{t,d} \leq \lambda_{t,k} + \frac{2 w(V_t)}{(1+\sqrt{\lambda_{t,1}})^2}.
\end{align}
Then using that $\lbrace v_{t,i} \rbrace_{i=1}^d$ are the eigenvectors of $V_t$ we have
\begin{align}
   \lambda_i \left( V_t +  \frac{2 w(V_t)}{(1+\sqrt{\lambda_{t,1}})^2} \sum_{i=1}^{d-1} v_{t,i} v^\mathsf{T}_{t,i} \right) &= \lambda_{t,i} + \frac{2 w(V_t)}{(1+\sqrt{\lambda_{t,1}})^2} \quad \text{for} \quad i= 1,...,k-1 \\
    \lambda_k \left(V_t +  \frac{2 w(V_t)}{(1+\sqrt{\lambda_{t,1}})^2} \sum_{i=1}^{d-1} v_{t,i} v^\mathsf{T}_{t,i} \right) &= \lambda_{t,d} \\
    \lambda_i \left( V_t +  \frac{2 w(V_t)}{(1+\sqrt{\lambda_{t,1}})^2} \sum_{i=1}^{d-1} v_{t,i} v^\mathsf{T}_{t,i} \right) &= \lambda_{t,i-1} + \frac{2 w(V_t)}{(1+\sqrt{\lambda_{t,1}})^2} \quad \text{for} \quad i= k+1,...,d .
\end{align}
Again using the mini-max principle for eigenvalues~\ref{cor:minimax_eigenvalues} and that both sides of~\eqref{eq:vt+1_bound} are real and symmetric
\begin{align}
    \lambda_{t+1,i} &\geq 
        \lambda_{t,i} + \frac{2 w(V_t)}{(1+\sqrt{\lambda_{t,1}})^2} \quad \text{if} \quad i \in \lbrace 1,...,k-1 \rbrace, \\
      \lambda_{t+1,i} &\geq    \lambda_{t,d} \quad \text{if} \quad i= k, \\
      \lambda_{t+1,i} &\geq    \lambda_{t,i-1} + \frac{2 w(V_t)}{(1+\sqrt{\lambda_{t,1}})^2} \quad \text{if} \quad i \in \lbrace k+1,...,d \rbrace.
\end{align}
Thus, using the above inequalities we can bound the trace of $V_{t+1}$ as
\begin{align}
    \Tr (V_{t+1} ) &= \sum_{i=1}^{k-1} \lambda_{t+1,i} + \lambda_{t+1,k} +\sum_{i=k+1}^{d-1} \lambda_{t+1,i} + \lambda_{t+1,d}  \\
    &\geq \left(\sum_{i=1}^{k-1} \lambda_{t,i} \right) + \frac{2(k-1) w(V_t)}{(1+\sqrt{\lambda_{t,1}})^2}  +  \lambda_{t,d} + \left(\sum_{i=k}^{d-2} \lambda_{t,i} \right) + \frac{2(d-k-1) w(V_t)}{(1+\sqrt{\lambda_{t,1}})^2}  + \lambda_{t+1,d}  \\
    & =  \lambda_{t+1,d} + \lambda_{t,d} + \sum_{i=1}^{d-2} \lambda_{t,i} + \frac{2(d-2) w(V_t)}{(1+\sqrt{\lambda_{t,1}})^2} \\
    & \geq \lambda_{t+1,d}  + \sum_{i=1}^{d-1} \lambda_{t,i} + \frac{2(d-2) w(V_t)}{(1+\sqrt{\lambda_{t,1}})^2},
\end{align}
where in the last bound we used simply $\lambda_{t,d} \geq \lambda_{t,d-1}$.
From the expression of $V_{t+1}$~\eqref{eq:vt_lemma}
\begin{align}
    \Tr (V_{t+1})  = \lambda_{t,d}+\sum_{i=1}^{d-1} \lambda_{t,i}+2(d-1)w(V_t),
\end{align}
and combining with the previous bound we obtain
\begin{align}
    \lambda_{t+1,d} &\leq \lambda_{d,t} + 2(d-1)w(V_t ) -\frac{2(d-2)w(V_t)}{(1+\sqrt{\lambda_{t,1}})^2} \\
    & = \lambda_{d,t} + \frac{2w(V_t )}{(1+\sqrt{\lambda_{t,1}})^2}\left( \lambda_{t,1}(d-1) + 2\sqrt{\lambda_{t,1}}(d-1) + 1 \right) \\
    & \leq \lambda_{d,t} + \frac{2(d-1)w(V_t )}{(1+\sqrt{\lambda_{t,1}})^2}\left( \lambda_{t,1} + 2\sqrt{\lambda_{t,1}} + 1 \right).
\end{align}
Again to check
\begin{align}
    \lambda_{t+1,1} \geq \sqrt{\frac{2}{3(d-1)}\lambda_{t+1,d}},
\end{align}
 we can square both sides and using the above bounds on $\lambda_{t+1,1}$ and $\lambda_{t+1,d}$ it suffices to check
\begin{align}
    \left(\lambda_{1,t} + \frac{2 w(V_t)}{(1+\sqrt{\lambda_{t,1}})^2}\right)^2 \geq \frac{2}{3(d-1)} \lambda_{d,t} + \frac{4w(V_t )(\lambda_{t,1}+2\sqrt{ \lambda_{t,1}}+1)}{3(1+\sqrt{\lambda_{t,1}})^2}.
\end{align}
Multiplying out the terms, this is equivalent to the condition
\begin{align}
    \underbrace{\lambda^2_{1,t} - \frac{2}{3(d-1)} \lambda_{d,t}}_{(i)} + \underbrace{\frac{4 w^2(V_t)}{(1+\sqrt{\lambda_{t,1}})^4}}_{(ii)} + \underbrace{\frac{4w(V_t)}{(1+\sqrt{\lambda_{t,1}})^2}(\lambda_{t,1} - \frac{1}{3}(\lambda_{t,1}+2\sqrt{\lambda_{t,1}}+1)}_{(iii)} \geq 0.
\end{align}
Finally, we observe that $(i)$ is positive by induction at time step $t$, $(ii)$ is always positive and $(iii)$ is positive since $\lambda_{t,1} \geq 2$, $\lambda_{t,1} \geq 2$ is non-decreasing and $f(x) = 2x-2\sqrt{x}-1$ is positive for $x \geq 2$.

\begin{center}
\underline{\textbf{Case 2.2}:  $\lambda_{t,d} \leq \lambda_{t,1} +  \frac{2 w(V_t)}{(1+\sqrt{\lambda_{t,1}})^2}$.}
\end{center}

%\subsection*{\textbf{Case 2.2}:  $\lambda_{t,d} \leq \lambda_{t,1} +  \frac{2 w(V_t)}{(1+\sqrt{\lambda_{t,1}})^2}$}

From the statement of the theorem we have 
\begin{align}
    \lambda_{t,1} \geq \lambda_0 \geq \sqrt{\frac{2}{3(d-1)}}2dC+\frac{2}{3(d-1)}.
\end{align}
Multiplying both sides by $\lambda_{t,1}$ we have
\begin{align}
     \lambda^2_{t,1} &\geq  \sqrt{\frac{2}{3(d-1)}}2dC\lambda_{t,1}+\frac{2}{3(d-1)}\lambda_{t,1} \\
     & \geq \frac{2}{3(d-1)}\left(2dC\sqrt{\lambda_{t,d}} + \lambda_{t,1} \right) \\
     & \geq \frac{2}{3(d-1)}\left(2dw(V_t) + \lambda_{t,1} \right) \\
     & = \frac{2}{3(d-1)}\left(2(d-1)w(V_t) + 2 w(V_t )+ \lambda_{t,1} \right) \\
     & \geq \frac{2}{3(d-1)}\left(2(d-1)w(V_t) + \lambda_{t,d} \right) ,
\end{align}
where the second inequality follows from the induction hypothesis $\lambda_{t,1}\geq \sqrt{\frac{2}{3(d-1)}\lambda_{t,d}}$, the third inequality from $w(V_t ) \leq C\sqrt{\lambda_{t,d}}$ and the fourth inequality from the assumption $\lambda_{t,d} \leq \lambda_{t,1} +  \frac{2 w(V_t)}{(1+\sqrt{\lambda_{t,1}})^2} \leq \lambda_{t,1}+2w(V_t)$. Thus taking the square root in both sides we have
\begin{align}
    \lambda_{t,1} \geq \sqrt{\frac{2}{3(d-1)}\left( \lambda_{t,d} + 2(d-1) w(V_t ) \right)},
\end{align}
and the induction at time step $t+1$,  $\lambda_{t+1,1}\geq \sqrt{\frac{2}{3(d-1)}\lambda_{t+1,d}}$ follows from the bounds
\begin{align}
    \lambda_{t+1,1} &\geq \lambda_{t,1}, \\
    \lambda_{t+1,d} & \leq \lambda_{t,d} + 2 w(V_t ) (d-1) ,
\end{align}
where we used the inequalities~\eqref{eq:pertboundsAB} and the definition of $V_t$~\eqref{eq:vt_lemma}.
\end{proof}

\section{Linear bandits with polylogarithmic minimax regret}

Inspired by the shot noise in the PSMAQB setting, we now study a class of noise-dependent linear stochastic bandits which, by combining the previously developed tools, will allow us to achieve polylogarithmic regret in the time horizon.

\subsection{Linear bandits with linearly vanishing subgaussian parameter}\label{sec:subgaussian_parameter}

In this Section, we combine the weighted least squares estimator defined in~\eqref{eq:estimator_weighted} with the action selection rule developed in Theorems~\ref{th:circle_eigenvalues} and~\ref{th:main_eigenvalues} to achieve polylogarithmic regret for a linear bandit with vanishing noise. Before tackling the variance condition~\eqref{eq:vanishing_variance} in the PSMAQB setting, we introduce a related model. We study a linear bandit with action set $\mathcal{A} = \mathbb{S}^{d-1}$, unknown parameter $\theta \in \mathbb{S}^{d-1}$, and reward model 
\begin{align}
X_t = \langle \theta, A_t \rangle + \epsilon_t,
\end{align}
where $\epsilon_t$ is conditionally $\eta_t$-subgaussian and satisfies
\begin{align}\label{eq:linear_vanishing_subgaussian}
    \eta^2_t \leq 1 - \langle \theta, A_t \rangle^2.
\end{align}
This noise model is inspired by the variance expression from the PSMAQB setting~\eqref{eq:vanishing_variance} in the qubit Bloch sphere, replacing the variance with the subgaussian parameter. Note that the noise decreases linearly as actions get closer to $\theta$, since $\eta_t \sim 1 - \langle \theta, A_t \rangle$. The extra factor $(1 + \langle \theta, A_t \rangle)$ is not needed to achieve polylogarithmic regret, but we keep it to better match the behavior of the quantum model. We also note that solving this model does not directly solve the PSMAQB setting, as the subgaussian parameter is strictly smaller than the variance. However, assumption simplifies the algorithm and helps illustrate how polylogarithmic regret can be achieved.

\subsubsection{Weighted confidence region}\label{sec:weighted_confidence_region}

Before proving a regret bound we need to derive a confidence region based on the weighted least squares estimator~\eqref{eq:estimator_weighted}. Our proof of is an adaptation from the one presented in~\cite{lattimore_banditalgorithm_book}[Chapter 20] to our setting. We refer to the reference for detailed computations. The proof is based on supermartingles properties specifically the following Theorem.

\begin{theorem}[Theorem 3.9 in~\cite{lattimore_banditalgorithm_book}]\label{th:supermartingale}
    Let $(X_t)_{t=0}^\infty$ be a supermartingale with $X_t \geq 0$ almost surely for all $t$. Then for any $\epsilon > 0$,
    \begin{align}
        \mathrm{Pr} \left( \sup_{t} X_t \geq \epsilon \right) \leq \frac{\EX \left[ X_0\right]}{\epsilon}.
    \end{align}
\end{theorem}

\begin{lemma}\label{lem:confidence_region_weighted}
Let $\delta \in (0,1)$ and $(A_t)^\infty_{t=1}$ be the actions selected by some policy with corresponding rewards $(X_t)^\infty_{t=1}$ given  by $X_t = \langle \theta ,A_t \rangle + \epsilon_t$, where $\theta\in\mathbb{S}^{d-1}$ is the unknown parameter and $\epsilon_t$ is $\eta_t$-subgaussian. Let $\hat{\sigma}^2_t$ be an estimator of the form~\eqref{eq:variance_estimator} and define the following event,
\begin{align}\label{eq:gt_event}
  G_t := \left\{ \big( (X_s, A_s)_{s=1}^{t-1}, A_t \big)  :  \sigma_s \leq \hat{\sigma}^2_s(A_1,X_1,\ldots,A_{s-1},X_{s-1},A_s)\ \forall {s \in [t]} \right\}.
\end{align}
Then we can define the following confidence region 
\begin{align}\label{eq:confidence_region}
    \mathcal{C}^{\textup{w}}_t := \lbrace \theta'\in\mathbb{R}^d : \| \theta' - \hat{\theta}^{\textup{w}}_t \|^2_{V_t (\lambda )} \leq \beta_{t,\delta}  \rbrace,
\end{align}
where $V_t(\lambda)$ is the weighted design matrix~\eqref{eq:design_matrix_weighted} and $\theta^{\textup{w}}_t$ the weighted least squares estimator~\eqref{eq:estimator_weighted} both defined with with $\hat{\sigma}^2_s$ and
\begin{align}\label{eq:beta}
\beta_{t,\delta} = \left( \sqrt{2\log \frac{1}{\delta} + \log\left(\frac{\det(V_t (\lambda ))}{\det (V_0 (\lambda ) )} \right)}+ \sqrt{\lambda}\right)^2. 
\end{align}
Then
\begin{align}\label{eq:prob_confidence_weighted}
    \mathrm{Pr}\left[ \forall s \in [t]:  \theta \in \mathcal{C}_s^{\textup{\textup{w}}} \big| G_t \right] \geq 1 - \delta.
\end{align}
\end{lemma}

\begin{proof} 
To simplify notation we use $\hat{\sigma}_s^2 = \hat{\sigma}_s^2 (A_s)$. From now we will condition all our calculations on the events $G_t$ since we want to prove~\eqref{eq:prob_confidence}. First define the following random quantity
\begin{align}
\widetilde{S}_t := \sum_{s=1}^t \frac{1}{\hat{\sigma}_s^2}\epsilon_s A_s,
\end{align}
and define the following process
\begin{align}
    \widetilde{M}_t (x) := \exp \left( \langle x , \widetilde{S}_t \rangle - \frac{1}{2} \| x \|^2_{{V}_t} \right),
\end{align}
for all $x\in\mathbb{R}^d$ and $V_t = V_t (0)$. We want to check that $\widetilde{M}_t(x)$ is a supermartingale, i.e $\EX \left[ \widetilde{M}_t (x) | \mathcal{H}_{t-1} \right] \leq \widetilde{M}_{t-1} (x)$ where $\mathcal{H}_{t-1}$ is the history for all information up to time step $t$ before reward $X_t$ is observed (also contains $A_t$). From a direct calculation, we have
\begin{align}\label{eq:supermartingale_check}
    \EX \left[ \widetilde{M}_t (x) | \mathcal{H}_{t-1} \right] = \widetilde{M}_{t-1}(x) \EX \left[ \exp \left( \frac{\epsilon_t}{\hat{\sigma}_t} \left\langle x ,\frac{A_t}{\hat{\sigma}_t}\right\rangle - \frac{1}{2}\| x \|_{\frac{1}{\hat{\sigma}^2_t} A_t A^\top_t} \right) \big| \mathcal{H}_{t-1}\right].
\end{align}
Then, using that in the definition of $\hat{\sigma}^2_t$~\eqref{eq:variance_estimator} is defined only using the information up to time step $t-1$, the subgaussian property, and that the event $G_t$ holds we have
\begin{align}
    \EX \left[ \exp \left( \frac{\epsilon_t}{\hat{\sigma}_t} \left\langle x ,\frac{A_t}{\hat{\sigma}_t}\right\rangle \right) \big| \mathcal{F}_{t-1}\right] \leq \exp \left(  \frac{1}{2}\left\langle x, \frac{A_t}{\hat{\sigma}_t} \right\rangle^2 \right) = \exp \left( \frac{1}{2}\| x \|_{\frac{1}{\hat{\sigma}^2_t} A_t A^\top_t}\right).
\end{align}
Inserting the above expression into~\eqref{eq:supermartingale_check} we immediately get $\EX \left[ \widetilde{M}_t (x) | \mathcal{F}_{t-1} \right] \leq \widetilde{M}_{t-1} (x)$. Using Lemma 20.3 in~\cite{lattimore_banditalgorithm_book}[Chapter 20] we have that
\begin{align}
    \Bar{M}_t := \int_{\mathbb{R}^d} \widetilde{M}_t (x) d h(x)
\end{align}
is a supermartingale where $h$ is a probability measure on $\mathbb{R}^d$. In particular we choose $h = \mathcal{N} ( 0 , H^{-1} )$ with $H = \lambda \mathbb{I}_{d\times d} \in \mathbb{R}^{d \times d}$ and we get
\begin{align}
\Bar{M}_t = \frac{1}{\sqrt{(2\pi)^d \det (H^{-1} )}} \int_{\mathbb{R}^d} \exp \left( \langle x , \widetilde{S}_t \rangle - \frac{1}{2} \| x\|^2_{V_t} - \frac{1}{2} \| x \|_{H} \right)dx .
\end{align}
The above quantity can be exactly computed as
\begin{align}\label{eq:expected_mt}
    \Bar{M}_t = \left( \frac{\det (V_0 )}{\det  ( V_t(\lambda) )}\right)^{\frac{1}{2}}\exp \left( \frac{1}{2} \| \widetilde{S}_t \|^2_{V^{-1}_t(\lambda )} \right),
\end{align}
where we used the solution for a Gaussian integral, we have completed the square inside the exponential term and we used that $V_t(\lambda ) = H + V_t$. Then we can use Theorem~\ref{th:supermartingale} in order to get
\begin{align}
    \mathrm{Pr} \left ( \sup_{t} \log \left( \Bar{M}_t\right) \geq \log \left( \frac{1}{\delta} \right) \right) = \mathrm{Pr} \left( \sup_t \Bar{M}_t \geq \frac{1}{\delta} \right) \leq \delta ,
\end{align}
where we used that by definition of $\widetilde{M}_{t-1}(x)$ we have $\Bar{M}_t \geq 0$ almost surely and $\Bar{M}_0 = 1$. Inserting~\eqref{eq:expected_mt} into the above equation we have
\begin{align}\label{eq:prob_wt}
    \mathrm{Pr}\left( \sup_{t} \| \widetilde{S}_t \|^2_{V^{-1}_t(\lambda )} \geq 2 \log \left( \frac{1}{\delta}\right) + \log \frac{\det  ( V_t(\lambda) )}{\det (V_0(\lambda))} \right) \leq \delta .
\end{align}
Finally using the expression for the weighted least squares estimator~\eqref{eq:estimator_weighted} we have
\begin{align}
    \| \hat{\theta}^{\text{wls}}_t - \theta \|_{V_t (\lambda )} \leq \| \widetilde{S}_t \|_{V^{-1}_t(\lambda )} + \sqrt{\lambda},
\end{align}
where we used triangle inequality and $\|\theta \|_2^2=1$. And the result follows by combining the above expression with~\eqref{eq:prob_wt} and conditioning under the event $G_t$.
\end{proof}

\subsubsection{\textsf{LinUCB} vanishing noise}

In this Section, we give the specific algorithm that minimizes the regret for the stochastic linear bandits with linear vanishing noise~\eqref{eq:linear_vanishing_subgaussian}. The algorithm is based on the principle of “optimism in the face of uncertainty" (OFU) or upper confidence bounds (UCB). We name the algorithm \textsf{LinUCB-VN}, where VN stands for vanishing noise. The \textsf{LinUCB-VN} algorithm is designed to keep the relation $\lambda_{\min} ( V_t ) = \Omega (\sqrt{\lambda_{\max}(V_t)} ) $ at each time step $t$ following the general action selection rule from Theorem~\ref{th:main_eigenvalues} combined with the weighted least squares estimator~\eqref{eq:estimator_weighted}.

\begin{figure}
    \centering
    \includegraphics[scale=0.5]{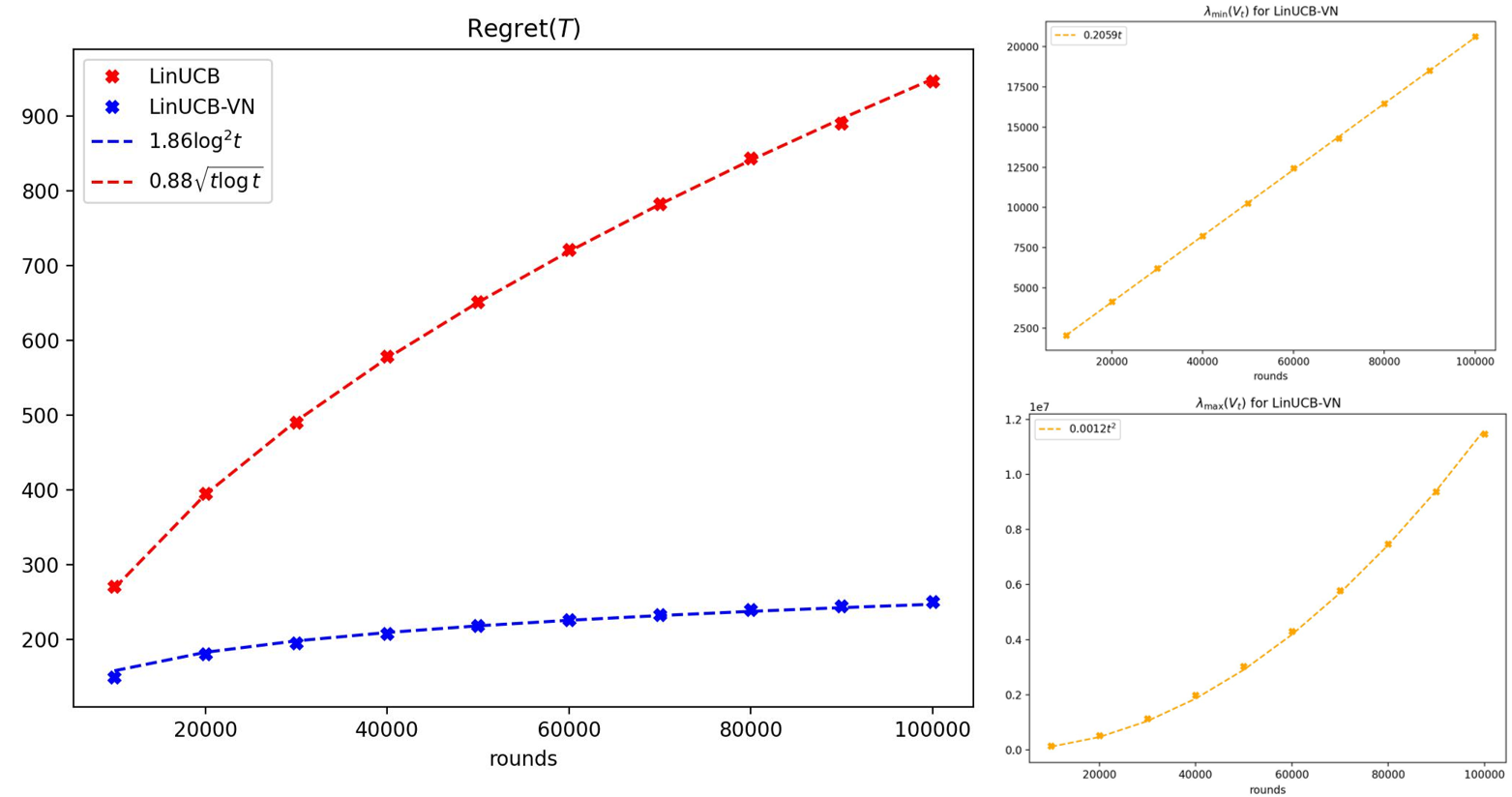}
    \caption{ \small We numerically test \textsf{LinUCB} and \textsf{LinUCB-VN} in a linear bandit with action set $\mathcal{A} = \mathbb{S}^1$ and reward model $X_t = \mathcal{N}(\langle \theta , A_t \rangle , 1 - \langle \theta , A_t \rangle^2 )$. Each point in the graphic is run independently and averaged over 100 instances for random environments $\theta\in\mathbb{S}^1$. \emph{Left plot:} Scaling of the regret for \textsf{LinUCB} algorithm and \textsf{LinUCB-VN}. We fit the functions $R(t) = 1.86\log^2 t$ for \textsf{LinUCB-VN} and $R(t) = 0.88\sqrt{t\log t}$ for \textsf{LinUCB}.
     \emph{Right plots:} Scaling of the maximum and minimum eigenvalue of the matrix $V_t$ for \textsf{LinUCB-VN}. The scaling shows the relation $\lambda_{\min}(V_t) = \Omega ( \sqrt{V_t(\lambda_{\max}}) )$. We fit the function  $\lambda_{\min}(V_t) = 0.2059t$ for the minimum eigenvalue and $ \lambda_{\max}(V_t) = 0.0012t^2$ for the maximum eigenvalue. The behavior $\lambda_{\min}(V_t) = \Theta ( t ) $ is the one that gives us the theoretical guarantee of polylogarithmic scaling of the regret.}
    \label{fig:numerical_results}
\end{figure}

Our algorithm works in batches of $2(d-1)$ actions for a total of $\tilde{T}\in\mathbb{N}$ batches or $T = 2(d-1)\tilde{T}$ rounds. At each batch $\tilde{t} \in [\tilde{T}]$ the algorithm selects the following actions, 
\begin{align}\label{eq:general_update}
        A^{\pm}_{\tilde{t},i} := \frac{\widetilde{A}^{\pm}_{\ttilde,i}}{\| \widetilde{A}^{\pm}_{\ttilde,i} \|_2},  
        \quad \textrm{where} \quad
        \widetilde{A}^{\pm}_{\ttilde,i} := \frac{ \hat{\theta}^{\text{w}}_{\tilde{t}}}{ \| \hat{\theta}^{\text{w}}_{\tilde{t}} \|_2}  \pm \frac{1}{\sqrt{\lambda_{\tilde{t}-1,1}}}v_{\tilde{t}-1,i } ,
\end{align}
for $i\in[d-1]$, where $v_{\tilde{t}-1,i}$ is the normalized eigenvector with eigenvalue $\lambda_{\tilde{t}-1,i} = \lambda_{i} ( V_{\tilde{t}-1}(\lambda ))$ of the weighted design matrix $V_{\tilde{t}} (\lambda )$~\eqref{eq:design_matrix_weighted} and $\hat{\theta}^{\text{w}}_{\tilde{t}}$ the weighted estimator~\eqref{eq:estimator_weighted}. The design matrix $V_{\tilde{t}}(\lambda)$ is updated at each batch $\tilde{t}$ as
\begin{align}\label{eq:vt_update}
    V_{\tilde{t}}(\lambda) := V_{\tilde{t}-1}(\lambda) + \omega ( V_{\tilde{t}-1}(\lambda))\sum_{i=1}^{d-1} \left(A^+_{\tilde{t},i}  (A^+_{\tilde{t},i})^\mathsf{T}  +  A^-_{\tilde{t},i}  (A^-_{\tilde{t},i})^\mathsf{T} \right),
\end{align}
and the weighted least squares estimator as
\begin{align}
    \hat{\theta}^{\text{w}}_{\tilde{t}} = V^{-1}_{\tilde{t}}(\lambda)\left( \sum_{s=1}^{\tilde{t}}\omega ( V_{s-1}(\lambda))\left(\sum_{i=1}^{d-1}\left(X^+_{s,i} A^+_{s,i} + X^-_{s,i} A^-_{s,i}   \right)\right)\right)
\end{align}
where $X^\pm_{s,i}$ is the reward sampled with $A^{\pm}_{s,i}$ and the weights (or noise estimator~\eqref{eq:variance_estimator}) are set to
\begin{align}\label{eq:omega}
    \omega ( V_{\tilde{t}-1} (\lambda) )  = \frac{\sqrt{\lambda_{\max}(V_{\tilde{t}-1}(\lambda))}}{12\sqrt{d-1}\beta_{\tilde{t}-1,\delta}}, \quad \hat{\sigma}^2_{\tilde{t}} ({A}^{\pm}_{\tilde{t},i} ) = \frac{1}{\omega ( V_{\tilde{t}-1}(\lambda) ) },
\end{align}
with $\beta_{\tilde{t},\delta}$ defined as in Lemma~\ref{lem:confidence_region_weighted} with input $V_{\tilde{t}}$~\eqref{eq:vt_update}. When is clear from the context we will denote $V_{\tilde{t}}(\lambda)$ simply as $V_{\tilde{t}}$. We state the pseudo-code of $\textsf{LinUCB-VN}$ below.

\begin{algorithm}
	\caption{\textsf{LinUCB-VN}} 
	\label{alg:weighted_linUCB}
    \begin{algorithmic}[1]
        \State Require: $\lambda_0\in\mathbb{R}_{>0}$, $\omega: \text{P}^d_+ \rightarrow \mathbb{R}_{\geq 0}$
        
         \State Set initial design matrix $V_0 \gets \lambda_0\mathbb{I}_{d\times d}$ 
        
        \State Choose initial estimator ${\theta}_0\in\mathbb{S}^{d-1}$ for $\theta$ at random 
        
        \For{$\tilde{t}=1,2,\cdots$}
            \vspace{1mm}
            \State \textit{Optimistic action selection}
            \vspace{1mm}
            
            \For{$i = 1,2,\cdots d-1$}   
                \State Select actions $A^+_{\ttilde,i}$ and $A^-_{\ttilde,i}$ according to Eq.~\eqref{eq:general_update}
                
               \State Receive associated rewards $X^+_{\ttilde,i}$ and $X^-_{\ttilde,i}$
            \EndFor
            
            \vspace{1mm}
            \textit{Update estimator of sub-gaussian noise for $a^+_{t,i}$}
            \vspace{1mm}
            
            $\hat{\sigma}^2_{\ttilde} \gets \frac{1}{\omega ( V_{\ttilde-1}(\lambda_0 ) ) }$ for $\ttilde\geq 2$ or $\hat{\sigma}^2_{\ttilde} \gets 1$ for $\ttilde=1$.

            \vspace{1mm}
          \State  \textit{Update design matrix and WLSE}
            \vspace{1mm}
            
         \State   $V_{\ttilde}(\lambda_0) \gets V_{t-1}(\lambda_0) + \frac{1}{\hat{\sigma}_t^2} \sum_{i=1}^{d-1} \left(A^+_{\ttilde,i}  (A^+_{\ttilde,i})^{\mathsf{T}}  +  A^-_{\ttilde,i}  (A^-_{\ttilde,i})^\mathsf{T} \right)$
            
           \State $\widetilde{\theta}_t^\text{w} \gets V_{\ttilde}^{-1} ( \lambda_0 )  \sum_{s = 1}^{\ttilde} \frac{1}{\hat{\sigma}_{\ttilde}^2} \sum_{i=1}^{d-1} (A^+_{s,i} X^+_{t,i} + A^-_{s,i} X^-_{s,i} ) $
        \EndFor
         \end{algorithmic}
\end{algorithm}

\subsubsection{Regret analysis for \textsf{LinUCB-VN}}

In the simplified proof of our action selection rule for the circle $\mathbb{S}^1$, one of the key steps was showing that $\lambda_{\max}(V_t)$ grows linearly with $t$. This followed from the fact that the actions were bounded and $\Tr(V_t) = O(t)$. However, for the weighted version~\eqref{eq:vt_update}, this no longer holds, and the analysis becomes a bit more involved. To understand the scaling, we will use a differential inequality involving $\lambda_{\max}(V_t)$, and the lemma below will help us solve it.

\begin{theorem}[\cite{MichelPetrovitch1901}]\label{th:dif_equation_ineq}
If $u$ satisfies the differential inequality $\frac{\text{d}u(t)}{\text{d}t} \lesseqgtr f(u(t),t)$, and $y$ is the solution to the ordinary differential equation (ODE) $\frac{\text{d}y(t)}{\text{d}t}=f(y(t),t)$ under the boundary condition $u(t_0)=y(t_0)$, then:
\begin{align}
    \forall t > t_0 , u(t) \lesseqgtr y(t) .
\end{align}
\end{theorem}

The Theorem below formalized the regret scaling for $\textsf{LinUCB-VN}$~\ref{alg:weighted_linUCB}.

\begin{theorem}\label{th:regret_bound_d2}
Let $d\geq 2$, $\delta\in ( 0,1)$ and $T = 2(d-1)\widetilde{T}$ for some $\widetilde{T}\in\mathbb{N}$. Let $\omega (X)$ defined as in~\eqref{eq:omega} using $\delta$ and $\lambda_0$ satisfy the constraints in Theorem~\ref{th:main_eigenvalues}. Then if we apply Algorithm~\ref{alg:weighted_linUCB} to a $d$ dimensional stochastic linear bandit with linear vanishing noise~\eqref{eq:linear_vanishing_subgaussian} with probability at least $1-\delta$ the regret satisfies
\begin{align}
     R_T ( \mathbb{S}^{d-1},\theta , \pi ) \leq 4(d-1) + \left(144d^2\beta^2_{T,\max}+24(d-1)^{\frac{3}{2}}\beta_{T,\max}\right)\log \left( \frac{T}{2(d-1)} \right),
\end{align}
    where
    \begin{align}
        \beta_{T,\max} : = \left(\lambda_0+\sqrt{2\log\frac{T}{\delta} + d\log\left(\frac{1}{144\lambda_0}T^2 + \frac{1}{6\sqrt{\lambda_0}}T +1\right)} \right)^2 .
    \end{align}
\end{theorem}

\begin{proof}
To simplify notation we use
\begin{align}
    \lambda_{\min,\ttilde} := \lambda_{\min}( V_{\ttilde}(\lambda_0 ) ), \quad  \lambda_{\max,\ttilde} := \lambda_{\max}( V_{\ttilde} (\lambda_0))
\end{align}
First, we fix
\begin{align}\label{eq:lambda_def}
    \lambda_0 = \max\left\lbrace2 , \sqrt{\frac{2}{3(d-1)}}\frac{d}{6\sqrt{d-1}} + \frac{2}{3(d-1)} \right\rbrace,
\end{align}
and later we will justify this choice.
Note that the regret can be written as
\begin{align}
      R_T ( \mathbb{S}^{d-1},\theta , \pi )  =   \frac{1}{2}\sum_{\ttilde = 1}^{\widetilde{T}} \sum_{i = 1}^{d-1} \left( \| \theta - A^+_{t,i} \|_2^2 + \| \theta - A^{-}_{t,i} \|_2^2  \right) ,
\end{align}
so to upper bound the regret, we need to quantify the distance $\| \theta - A^+_{\ttilde,i} \|_2^2$ between the unknown parameter and the actions we select at each batch $t$. The history up to batch $\ttilde \leq \widetilde{T}$ is defined as
\begin{align}
    \mathcal{H}_{\ttilde} :=  \big( X^+_{s,1}, A^+_{s,1},X^-_{s,1}, A^-_{s,1},...,X^+_{s,d-1}, A^+_{s,d-1},X^-_{s,d-1}, A^-_{s,d-1}\big)_{s=1}^{\ttilde} .
\end{align}
Lemma~\ref{lem:confidence_region_weighted} gives us this distance with a certain probability under the assumption that the event
    \begin{align}\label{eq:event_Gt}
     G_t := & \lbrace \big( \mathcal{H}_{t-1}, a^\pm_{s,i} \big)  :  
     \sigma^2_s(a^\pm_{s,i} \big)) \leq \hat{\sigma}^2_s(\mathcal{H}_{s-1},a^\pm_{s,i} )\ \forall {s \in [t]} \rbrace,
\end{align}
holds. 
Using the definition of the subgaussian parameter for the noise $\epsilon_t$~\eqref{eq:linear_vanishing_subgaussian} we can upper bound it for our choice of actions as
\begin{align}\label{eq:noise_bound}
    \eta^2_{t}(A^\pm_{t,i}) &\leq 1-\langle A^\pm_{t,i}, \theta  \rangle^2 = 1-(1-\frac{1}{2} \| \theta - A^\pm_{t,i} \|_2^2)^2 \nonumber \\
    &=  \| \theta - A^\pm_{t,i} \|_2^2 - \frac{1}{4} \| \theta - A^\pm_{t,i} \|_2^4 \leq \| \theta - A^\pm_{t,i} \|_2^2. 
\end{align}
Thus, we need to use an estimator of the form~\eqref{eq:variance_estimator} that upper bounds the distance $\| \theta - A^\pm_{t,i} \|_2^2$. This quantity depends on the unknown parameter $\theta$, thus we can not guarantee that $G_t$ holds with probability one at each batch $\ttilde$. The other event that we will need to hold to quantify the distance is the following 
\begin{align}\label{eq:event_Et}
    E_{\ttilde} := \lbrace \mathcal{H}_{\ttilde} : \forall s \in \left[ \ttilde \right], \theta \in \mathcal{C}_{s} \rbrace ,
\end{align}
which is guaranteed to hold with probability at least $1-\delta$ if $G_t$ holds by Lemma~\ref{lem:confidence_region_weighted}. We leave the probability computation that both events hold for the end of the proof.

\vspace{1.5 mm}

{\centering\underline{\textbf {Choosing $w(V_T)$ and instantaneous regret bound under $E_t$}}\par}

\vspace{1.5mm}
%\subsection*{Choosing $w(V_T)$ and instantaneous regret bound under $E_t$}

For now, we will assume that $E_t$ always holds to justify our choice of weights through a bound on $\| \theta - A^\pm_{t,i} \|_2$. This will also allow us to bound the instantaneous regret. Using the same geometrical argument as in Theorem~\ref{th:circle_regret} we have that 
\begin{align}\label{eq:theta_at_bound}
    \| \theta - A^\pm_{t,i} \|_2 \leq 3\sqrt{\frac{\beta_{\ttilde-1,\delta}}{\lambda_{\min,{\ttilde-1}}}},
\end{align}
where $\beta_{\ttilde,\delta}$ is defined as in Lemma~\ref{lem:confidence_region_weighted}.
 For $\ttilde =0$,
\begin{align}
    \beta_{0,\delta} = \left(\sqrt{\lambda_0} + \sqrt{2\log\left( \frac{1}{\delta}\right)} \right).
\end{align}
Note that at batch $\ttilde$ we use that $E_{\ttilde-1}$ holds instead of $E_{\ttilde}$ since we want to define $\hat{\sigma}_{\ttilde}$ such that only depends on past information up to batch $\\tilde{t}-1$. Thus, a choice that will guarantee the event $G_t$ to hold under the assumption that $E_{t-1}$ holds is $\hat{\sigma}^2_{\ttilde}  := \frac{9\beta_{\ttilde-1,\delta}}{\lambda_{\min,\ttilde-1}}$. However if we apply Theorem~\ref{th:main_eigenvalues} (later we will check that we are under the right assumptions to use it) to our particular update of $V_{\ttilde} $ with the choices of actions $A^+_{\ttilde,i},A^-_{\ttilde,i}$ we have that
\begin{align}\label{eq:minmaxsqrtrelation}
    \lambda_{\min,\ttilde} \geq \sqrt{\frac{2}{3(d-1)}\lambda_{\max,\ttilde}},
\end{align}
 It is important to note that the above bound is independent of the events $G_{\ttilde}$ and $E_{\ttilde}$ and it is a consequence only of our particular choice of actions~\eqref{eq:general_update}. Combining~\eqref{eq:minmaxsqrtrelation} and~\eqref{eq:theta_at_bound} 
\begin{align}
     \| \theta - A^\pm_{t,i} \|_2^2 \leq 9\frac{\beta_{\ttilde-1,\delta}}{\lambda_{\min , \ttilde-1}} \leq  \frac{9\sqrt{3}\sqrt{d-1}\beta_{\ttilde-1,\delta}}{\sqrt{2}\sqrt{\lambda_{\min,t-1}}}.
\end{align}
Thus, using that  $\frac{9\sqrt{3}}{\sqrt 2} \leq 12$ we see that with our definitions of estimator~\eqref{eq:estimator_weighted} and weights 
\begin{align}\label{eq:estimator_choice}
    \hat{\sigma}^2_t (A^\pm_{t,i} ) = \frac{12\sqrt{d-1}\beta_{\ttilde-1,\delta}}{\sqrt{\lambda_{\max,\ttilde-1}}}, \quad \omega (V_{\ttilde-1} ) = \frac{1}{\hat{\sigma}^2_{\ttilde} },
\end{align}
are well defined since only depends on the history $\mathcal{H}_{t-1}$, and 
\begin{align}\label{eq:upperbound_subgaussian}
    \sigma^2_{\ttilde} ( A^\pm_{\ttilde,i}) \leq \hat{\sigma}^2_{\ttilde} (A^\pm_{\ttilde,i} ) \quad \text{if} \quad \theta\in\mathcal{C}_{\ttilde-1}.
\end{align}

To bound the regret our technique uses upper and lower bounds on the scaling of $\Tr (V_t )$ as a function of the number of rounds. In the standard \textsf{LinUCB} technique if the actions are bounded ($\| A_t \|_2 = \Theta (1 )$) and this immediately gives the linear scaling $\Tr (V_t) = \Theta (t)$. Since we are updating $V_{\ttilde}$ using $\omega (V_{\ttilde} )$ we need to do some extra work, and we will see that $\Tr (V_{\ttilde} ) = \tilde{\Theta} (\ttilde^2 )$. While Theorem~\ref{th:circle_regret} only required an upper bound now we will need both upper and lower bound since we will also need to characterize $\beta_{\ttilde,\delta}$ that is defined through $V_{\ttilde}$ and in Theorem~\ref{th:circle_regret} only depends on $t$.

\vspace{1.5mm}

{\centering\underline{\textbf{Upper bound for $\Tr (V_t )$}}\par}

\vspace{1.5mm}
%\subsection*{Upper bound for $\Tr (V_t )$}

The reason why we need an upper bound for $\Tr (V_t )$ is because we will need an upper bound for $\beta_{\ttilde,\delta}$ that depends on $\Tr (V_{\ttilde} )$. A direct calculation shows
\begin{align}
    \Tr (V_{\ttilde} ) &= \lambda_0 + \sum_{s=1}^{\ttilde} 2(d-1)\omega (V_{s-1} ) \\
    &= \lambda_0 + \sum_{s=1}^{\ttilde} \frac{2(d-1)}{12\sqrt{d-1}\beta_{s-1,\delta}}\sqrt{\lambda_{\max , s-1}} \\
    &\leq \lambda_0 + \frac{\sqrt{d-1}}{6\beta_{0,\delta}}\sum_{s=0}^{t-1} \sqrt{\lambda_{\max , s}},
\end{align}
where we used $\beta_{0,\delta} \leq \beta_{\ttilde,\delta}$. Then using  $\lambda_{\max , t} \leq \Tr (V_{\ttilde})$ and $\beta_{0,\delta}\geq 1$ we have
\begin{align}\label{eq:lmaxinequality}
    \lambda_{\max,\ttilde} \leq \lambda_0+ \frac{\sqrt{d}}{6}\sum_{s=0}^{\ttilde-1} \sqrt{\lambda_{\max , s}}.
\end{align}
Now we want to extend $\lambda_{\max,\ttilde}$ to a monotonic function on the interval $[0, \widetilde{T} ]$. In order to do that we define the linear interpolation $g_1:[0, \widetilde{T} ] \rightarrow \mathbb{R}_{\geq 0}$ as
\begin{align}\label{eq:def_gx}
    g_1 (x) := (\ttilde+1-x)\lambda_{\max,t} + (x-\ttilde)\lambda_{\max,\ttilde+1} \quad \text{for} \quad x\in[\ttilde,\ttilde+1) \quad \text{and} \quad \ttilde\in\lbrace 0 ,..., \widetilde{T} \rbrace.
\end{align}
Then using~\eqref{eq:lmaxinequality} and $g_1(\ttilde) = \lambda_{\max,\ttilde}$ we have
\begin{align}
    g_1(\ttilde) \leq \lambda_0 + \frac{\sqrt{d}}{6}\sum_{s=0}^{t-1} \sqrt{g_1(s)} \leq \lambda_0 + \frac{\sqrt{d}}{6}\int_{0}^{\ttilde} \sqrt{g_1 (x)}\dd x \quad \text{for} \quad \ttilde\in \lbrace 0 ,..., \widetilde{T}  \rbrace.
\end{align}
Now we want to prove the above inequality but for $g_1 (x)$ with $x\in[0,\widetilde{T} ]$. From the above inequality and the definition of $g_1 (x)$~\eqref{eq:def_gx} we have that for $x\in[\ttilde,\ttilde+1)$ 
\begin{align}
    g_1 (x) &\leq \lambda_0 + \frac{\sqrt{d}}{6}\sum_{s=0}^{\ttilde-1} \sqrt{g_1(s)} + \frac{\sqrt{d}}{6}(x-\ttilde)\sqrt{g_1 (\ttilde)} \\
    &\leq \lambda_0 + \frac{\sqrt{d}}{6}\int_{0}^{\ttilde} \sqrt{g_1 (x)}\dd x + \frac{\sqrt{d}}{6}(x-\ttilde)\sqrt{g_1 (\ttilde)}.
\end{align}
Then we can use the above relation and 
\begin{align}
    \int_{\ttilde}^x \sqrt{g_1(s)}\dd s &= \int_0^{x-{\ttilde}} \sqrt{g_1(s+\ttilde)}\dd s = \int_0^{x-\ttilde} \sqrt{(1-s)g_1(\ttilde) + sg_1(\ttilde+1)} \dd s \\
    &\geq \int_{0}^{x-\ttilde}\sqrt{g_1 (\ttilde)} \geq (x-\ttilde)\sqrt{g_1(\ttilde)},
\end{align}
where the second equality follows from the definition of $g_1(x)$~\eqref{eq:def_gx} and the first inequality $g_1(\ttilde) =\lambda_{\max,\ttilde}\leq\lambda_{\max,\ttilde+1} = g_1(\ttilde+1)$. We conclude that
\begin{align}
    g_1 (x) \leq \lambda_0 +  \frac{\sqrt{d}}{6}\int_{0}^{x} \sqrt{g_1(s)}\dd s .
\end{align}
Using the fundamental theorem of calculus
\begin{align}\label{eq:gxupperbound}
    g_1 (x) \leq \lambda_0 + \frac{\sqrt{d}}{6}(G_1 (x)-G_1(0))\quad\text{where} \quad \frac{\dd G_1}{\dd x} = \sqrt{g_1 (x)}.
\end{align}
Fixing $ G(0) = 0$ and rearranging 
\begin{align}
    \frac{\dd G_1}{\dd x} \leq \sqrt{\lambda_0 +\frac{\sqrt{d}}{6} G_1(x) }.
\end{align}
Now solving for the equality with $G(0) = 0 $ and using Theorem~\ref{th:dif_equation_ineq} we arrive at
\begin{align}
    G_1(x) \leq  \frac{\sqrt{d}}{24}x^2 + \sqrt{\lambda_0}x.
\end{align}
Finally inserting the above into~\eqref{eq:gxupperbound} using that $g_1(\ttilde) = \lambda_{\max,\ttilde}$
\begin{align}
    \lambda_{\max,\ttilde} = g_1(\ttilde) \leq \frac{d}{144}{\ttilde}^2 + \frac{\sqrt{d\lambda_0}}{6}\ttilde + \lambda_0 ,
\end{align}
and this gives the following bound on $\Tr (V_{\ttilde})$
\begin{align}\label{eq:vtupperbound}
    \Tr ( V_{\ttilde} ) \leq d \lambda_{\max,\ttilde}  \leq d\left( \frac{d}{144}{\ttilde}^2 + \frac{\sqrt{d\lambda_0}}{6}\ttilde + \lambda_0 \right).
\end{align}

\vspace{1.5mm}

{\centering\underline{\textbf{Upper bound for $\beta_{\ttilde,\delta}$}}\par}

\vspace{1.5mm}
%\subsection*{Upper bound for $\beta_{t,\delta}$}

In order to give an upper bound for $\beta_{\ttilde,\delta}$~\eqref{eq:beta} we have to give an upper bound for $\det (V_{\ttilde})$. The inequality of arithmetic and geometric means gives
\begin{align}
    \det (V_{\ttilde} ) \leq \left(\frac{\Tr (V_{\ttilde} )}{d} \right)^d,
\end{align}
which combined with~\eqref{eq:vtupperbound} and~\eqref{eq:beta} gives
\begin{align}\label{eq:beta_upperbound}
    \beta_{\ttilde,\delta} \leq \left(\lambda_0+\sqrt{2\log\frac{1}{\delta} + d\log\left(\frac{d}{144\lambda_0}{\ttilde}^2 + \frac{\sqrt{d}}{\sqrt{\lambda_0}6}\ttilde +1\right)} \right)^2 .
\end{align}

\vspace{1.5mm}

{\centering  \underline{\textbf{Lower bound for $\lambda_{\max,\ttilde}$}}\par}

\vspace{1.5mm}
%\subsection*{Lower bound for $\lambda_{\max,t}$}

In order to give an upper bound for regret our main technique is to control the growth of $\lambda_{\max,{\ttilde}}$ through a lower bound for $\lambda_{\max,{\ttilde}}$. We will employ the same technique as in the previous part. For the lower bound we are only interested in the leading terms so  we start by lower bounding $\Tr ( V_{\ttilde} )$ as
\begin{align}
    \Tr (V_{\ttilde} ) &\geq \sum_{s= 2}^{\ttilde} 2(d-1)\omega (V_{s-1} ) \\
    &= \sum_{s=2}^{\ttilde} \frac{\sqrt{d-1}\sqrt{\lambda_{\max,s-1}}}{6\beta_{s-1,\delta}} \\
    & \geq \frac{\sqrt{d-1}}{6\beta_{\widetilde{T},\delta}} \sum_{s=1}^{\ttilde-1} \sqrt{\lambda_{\max,s}}. 
\end{align}
where we used $\beta_{\ttilde,\delta} \leq \beta_{\widetilde{T},\delta}$.
 Then since our problem is restricted to $\mathbb{R}^d$ we have that for the maximum eigenvalue of $V_{\ttilde}$,
\begin{align}
    \lambda_{\max,\ttilde} \geq \frac{\Tr (V_{\ttilde} )}{d}.
\end{align}
Thus combining with the above expressions we have
\begin{align}\label{eq:uplambdmax}
    \lambda_{\max,\ttilde} \geq \frac{\sqrt{d-1}}{6d\beta_{\widetilde{T},\delta}} \sum_{s= 1}^{\ttilde-1} \sqrt{\lambda_{\max,s}}.
\end{align}
Addding each side by $\frac{\sqrt{d-1}}{6d\beta_{\widetilde{T},\delta}}\sqrt{\lambda_{\max,\ttilde}}$ and using the fact that $\lambda_{\max,\ttilde} \geq \sqrt{\lambda_{\max,\ttilde}}$ since $\lambda_{\max,\ttilde}\geq 1$  we have 
\begin{align}\label{eq:lambdmaxinequality}
    \lambda_{\max,\ttilde} + \frac{\sqrt{d-1}}{6d\beta_{\widetilde{T},\delta}}\sqrt{\lambda_{\max,t}} &\geq \frac{\sqrt{d-1}}{6d\beta_{\widetilde{T},\delta}} \sum_{s= 1}^{\ttilde-1} \sqrt{\lambda_{\max,s}} + \frac{\sqrt{d-1}}{6d\beta_{\widetilde{T},\delta}} \sqrt{\lambda_{\max,t}},\\
    \left(1+ \frac{\sqrt{d-1}}{6d\beta_{\widetilde{T},\delta}} \right)\lambda_{\max,\ttilde} &\geq \lambda_{\max,\ttilde} + \frac{\sqrt{d-1}}{6d\beta_{\widetilde{T},\delta}}\sqrt{\lambda_{\max,\ttilde}} \geq \frac{\sqrt{d-1}}{6d\beta_{\widetilde{T},\delta}} \sum_{s= 1}^{\ttilde} \sqrt{\lambda_{\max,s}}, \\
     \lambda_{\max,\ttilde} &\geq \frac{1}{1+6\frac{d}{\sqrt{d-1}}\beta_{\widetilde{T},\delta}} \sum_{s= 1}^{\ttilde} \sqrt{\lambda_{\max,s}}.
\end{align}
As before, we want to extend $\lambda_{\max,\ttilde}$ to a monotonic function on the interval $[0, \widetilde{T}]$ and define the lower linear interpolation $g_2:[0, \widetilde{T}] \rightarrow \mathbb{R}_{\geq 0}$ as
\begin{align}\label{eq:def_fx}
    g_2(x) := (\ttilde+1-x)\lambda_{\max,\ttilde} + (x-\ttilde)\lambda_{\max,\ttilde+1} \quad \text{for} \quad x\in[\ttilde,\ttilde+1) \quad \text{and} \quad \ttilde\in \lbrace 0 , ..., \widetilde{T}\rbrace .
\end{align}
Then we have that 
\begin{align}
    g_2(\ttilde)  &\geq \frac{1}{1+6\frac{d}{\sqrt{d-1}}\beta_{\widetilde{T},\delta}} \sum_{s=1}^{\ttilde} \sqrt{g_2(s)}  \\
    &\geq \frac{1}{1+6\frac{d}{\sqrt{d-1}}\beta_{\widetilde{T},\delta}} \int_{0}^t \sqrt{g_2(x)} \dd x \quad \text{for} \quad \ttilde\in\lbrace 0 , ... , \widetilde{T} \rbrace,
\end{align}
since $g_2(\ttilde) = \lambda_{\max,\ttilde}$ for $\ttilde\in \lbrace 0 , .... , \widetilde{T} \rbrace$. Then we can check that the above inequality also holds for $x\in [0, \widetilde{T}]$. If $x\in [\ttilde,\ttilde+1)$ we have
\begin{align}\label{eq:fxbound}
    g_2(x) &\geq \frac{1}{1+6\frac{d}{\sqrt{d-1}}\beta_{\widetilde{T},\delta}} \left( \sum_{s=1}^{\ttilde} \sqrt{g_2 (s)} + (x-t)\sqrt{g_2 (\ttilde+1) }\right) \\
    & \geq \frac{1}{1+6\frac{d}{\sqrt{d-1}}\beta_{\widetilde{T},\delta}} \left(\int_0^{\ttilde}\sqrt{g_2 (x)}\dd x + (x-\ttilde)\sqrt{g_2 (\ttilde+1)}\right),
\end{align}
where the first inequality we applied the definition of $g_2(x)$~\eqref{eq:def_fx} and the inequality for $\lambda_{\max,\ttilde}$~\eqref{eq:uplambdmax}. Then,
\begin{align}
    \int_{\ttilde}^{x}\sqrt{g_2(s)}\dd s &= \int_{0}^{x-\ttilde}\sqrt{g_2(s+\ttilde)}\dd s = \int_{0}^{x-\ttilde} \sqrt{(1-s)g_2(\ttilde) + sg_2(\ttilde+1)}\dd s \\
    &\leq  \int_{0}^{x-\ttilde} \sqrt{g_2 (\ttilde+1)}\dd s \leq (x-t)\sqrt{g_2(\ttilde+1)},
\end{align}
where the second equality follows from the definition of $g_2 (x)$ and the first inequality from $g_2(\ttilde) = \lambda_{\max,\ttilde} \leq \lambda_{\max , t+1} \leq g_2 (\ttilde+1)$. Combining with~\eqref{eq:fxbound} we have
\begin{align}
    g_2(x) &\geq \frac{1}{1+6\frac{d}{\sqrt{d-1}}\beta_{\widetilde{T},\delta}} \int_{0}^x \sqrt{g_2 (s)}ds.
\end{align}
Using the fundamental theorem of calculus 
\begin{align}\label{eq:bound_lambdamax}
    g_2(x) \geq \frac{1}{1+6\frac{d}{\sqrt{d-1}}\beta_{\widetilde{T},\delta}} \left(G_2(x) - G_2(0) \right) \quad \text{where} \quad \frac{\dd G_2(x)}{\dd x} = \sqrt{g_2(x) }.
\end{align}
Fixing $ G_2 (0) = 0$ and rearranging 
\begin{align}
    \frac{\dd G_2(x)}{\dd x} \geq \sqrt{\frac{1}{1+6\frac{d}{\sqrt{d-1}}\beta_{\widetilde{T},\delta}}}\sqrt{G_2(x)}.
\end{align}
Now solving the equality with $G_2 (0) = 0$ and using Theorem~\ref{th:dif_equation_ineq}
\begin{align}
    G_2 (x) &\geq \frac{1}{4+24\frac{d}{\sqrt{d-1}}\beta_{\widetilde{T},\delta}}x^2 .
\end{align}
Using that $g_2(\ttilde) = \lambda_{\max ,\ttilde} $ and inserting the above result into~\eqref{eq:bound_lambdamax}
\begin{align}\label{eq:lambdamaxlowbound}
   \lambda_{\max,\ttilde} = g_2 (\ttilde) \geq \frac{1}{4(1+6\frac{d}{\sqrt{d-1}}\beta_{\widetilde{T},\delta})^2} {\ttilde}^2 .
\end{align}

\vspace{1.5mm}

{\centering\underline{\textbf{Bounding the regret}}\par}

\vspace{1.5mm}

%\subsection*{Bounding the regret}

%
From the weight update~\eqref{eq:omega} we have
\begin{align}
    \omega ( V_{\ttilde} ) = \frac{\sqrt{\lambda_{\max,\ttilde-1}}}{12\sqrt{d-1}\beta_{\ttilde,\delta}} \leq \frac{1}{12\sqrt{d-1}}\sqrt{\lambda_{\max,\ttilde-1}}, 
\end{align}
so to use Therorem~\ref{th:main_eigenvalues} we can choose $C = \frac{1}{12\sqrt{d-1}}$, our choice of $\lambda_0$~\eqref{eq:lambda_def}, actions~\eqref{eq:general_update} and the sequence $\lbrace \theta^\text{w}_t \rbrace_{t=0}^\infty \subset \mathbb{S}^d$. Then our update given by~\eqref{eq:vt_update} satisfies all conditions of Theorem~\ref{th:main_eigenvalues} and our choice of actions guarantees
\begin{align}\label{eq:eigenvalue_relation}
    \lambda_{\min,\ttilde} \geq \sqrt{\frac{2}{3(d-1)}\lambda_{\max,\ttilde}}.
\end{align}
Then from~\eqref{eq:beta_upperbound} we can upper bound $\beta_{t,\delta}$ as
\begin{align}
 \beta_{\ttilde,\delta}\leq \beta_{T,\max} \quad \text{where} \quad  \beta_{T,\max}  := \left(\lambda_0+\sqrt{2\log\frac{1}{\delta} + d\log\left(\frac{d}{144\lambda_0}T^2 + \frac{\sqrt{d}}{\sqrt{\lambda_0}6}T +1\right)} \right)^2 .
\end{align}
Finally, we can use the bound from~\eqref{eq:noise_bound}\eqref{eq:theta_at_bound} (under the assumptions that $E_t$ and $G_t$ hold) combined with~\eqref{eq:eigenvalue_relation} to arrive at
\begin{align}
    \| \theta - A^{\pm}_{\ttilde,i} \|_2^2 &\leq \frac{9\beta_{T,\max}}{\lambda_{\min,\ttilde-1}}\\
    &\leq \frac{12\sqrt{d-1}\beta_{T,\max}} {\sqrt{\lambda_{\max,\ttilde-1}}}\\
    &\leq \frac{24\sqrt{d-1}\beta_{T,\max}(1+6\frac{d}{\sqrt{d-1}}\beta_{T,\max})}{\ttilde-1} \\
    &= \frac{144d\beta^2_{T,\max}+24\sqrt{d-1}\beta_{T,\max}}{\ttilde-1}
\end{align}
where we used the lower bound for $\lambda_{\max,\ttilde}$~\eqref{eq:lambdamaxlowbound} and $\frac{9\sqrt{3}}{\sqrt 2} \leq 12$.
Then using the above result we can bound the regret as
\begin{align}
     R_T ( \mathbb{S}^{d-1},\theta , \pi ) &=   \frac{1}{2}\sum_{\ttilde= 1}^{\widetilde{T}} \sum_{i = 1}^{d-1} \left( \| \theta - A^+_{\ttilde,i} \|_2^2 + \| \theta - A^{-}_{\ttilde,i} \|_2^2  \right) \\
    &\leq 4(d-1) + \frac{1}{2}\sum_{\ttilde= 2}^{\widetilde{T}} \sum_{i = 1}^{d-1} \left( \| \theta - A^+_{\ttilde,i} \|_2^2 + \| \theta - A^{-}_{\ttilde,i} \|_2^2  \right) \\
    &\leq 4(d-1) + \left(144d^2\beta^2_{T,\max}+24(d-1)^{\frac{3}{2}}\beta_{T,\max}\right)\sum_{\ttilde = 2}^{\tilde{T}} \frac{1}{\ttilde-1}  \\
    & \leq 4(d-1) + \left(144d^2\beta^2_{T,\max}+24(d-1)^{\frac{3}{2}}\beta_{T,\max}\right)\log \left( \frac{T}{2(d-1)} \right).
\end{align}
Using that $\beta_{T,\max} = O (d\log(T))$ we have that 
\begin{align}
    \text{Regret}(T) = \widetilde{O}(d^{4}\log^3 (T)).
\end{align}

{\centering\underline{\textbf{Success probability analysis}}\par}

%\subsection*{Success probability analysis}

From the computations in~\eqref{eq:noise_bound}\eqref{eq:theta_at_bound}\eqref{eq:estimator_choice}\eqref{eq:upperbound_subgaussian} and
 our choice of $\sigma^2_t(a_t)$~\eqref{eq:omega} we have that 
\begin{align}
    \text{if} \quad \theta\in\mathcal{C}_{s-1} \Rightarrow \eta^2_s (A^{\pm}_{s,i} ) \leq \frac{12\sqrt{d-1}\beta_{t-1,\delta}}{\sqrt{\lambda_{\max,s-1}}} = \hat{\sigma}^2_s (A^{\pm}_{s,i}), 
\end{align}
thus,
\begin{align}
     \mathrm{Pr}( G_t ) \geq \mathrm{Pr}(E_{t-1} ) 
\end{align}
where we have used the definitions of the events $G_t$~\eqref{eq:event_Gt} and $E_t$~\eqref{eq:event_Et}.
%for our particular choice of actions $a^+_{s,i},a^-_{s,i}$.

The initial choice $\hat{\sigma}^2_1 = 1$ implies $\eta^2_1\leq \hat{\sigma}^2_1 $ and using Lemma~\ref{lem:confidence_region_weighted} we have
\begin{align}\label{eq:probg1}
     \mathrm{Pr}(G_1 ) & = 1 , \\
     \mathrm{Pr}(E_1 ) &  \geq 1-\delta .
\end{align}
Using Bayes theorem
\begin{align}
    \mathrm{Pr}(E_{\ttilde} )  = \frac{\mathrm{Pr}(E_{\ttilde} |G_{\ttilde})\mathrm{Pr}(G_{\ttilde})}{\mathrm{Pr}(G_{\ttilde} |E_{\ttilde})} = \mathrm{Pr}(E_{\ttilde} |G_{\ttilde})\mathrm{Pr}(G_{\ttilde}),
\end{align}
where in the last equality we used
\begin{align}
    \mathrm{Pr}(G_{\ttilde} |E_{\ttilde}) = 1 .
\end{align}
From Lemma~\ref{lem:confidence_region_weighted} we have
\begin{align}
    \mathrm{Pr} (E_{\ttilde} | G_{\ttilde} ) \geq 1 -\delta . 
\end{align}
Thus, applying recursively the above inequalities 
\begin{align}\label{eq:probgt}
     \mathrm{Pr}(G_2 ) & \geq \mathrm{Pr}(E_1 ) \geq 1-\delta,\\
     \mathrm{Pr}(E_2 ) & = \mathrm{Pr}(E_2 | G_2)\mathrm{Pr}(G_2)\geq (1- \delta)^2 \\
      &\vdots \nonumber \\
      \mathrm{Pr}(G_{\ttilde} ) & \geq \mathrm{Pr}(E_{\ttilde-1} )  \geq (1-\delta)^{\ttilde-1} \\
      \mathrm{Pr}(E_{\ttilde} ) & = \mathrm{Pr}(E_{\ttilde} | G_{\ttilde})\mathrm{Pr}(G_{\ttilde})\geq (1- \delta)^{\ttilde}     
\end{align}
Finally to bound the regret we used the assumption that $\theta \in\mathcal{C}_{\ttilde}$ for any $\ttilde\in \lbrace 1,...,\widetilde{T}-1\rbrace$. Thus we can bound the probability that the obtained bound holds as
\begin{align}
     \mathrm{Pr}&\left(  R_T ( \mathbb{S}^{d-1},\theta , \pi )\leq  4(d-1) + \left(144d^2\beta^2_{T,\max}+24(d-1)^{\frac{3}{2}}\beta_{T,\max}\right)\log\left( \frac{T}{2(d-1)} \right)\right) \\ 
    & \geq \mathrm{Pr}(E_{\widetilde{T}} \cap G_{\widetilde{T}}) = \mathrm{Pr}(G_{\widetilde{T}} )\mathrm{Pr}(E_{\widetilde{T}} | G_{\widetilde{T}} ). \\
    &\geq (1-\delta)^{\widetilde{T}-1}(1-\delta) = (1-\delta)^{\widetilde{T}}.
\end{align}
The result follows choosing $\delta = \frac{\delta '}{\widetilde{T}}$ for some $0 < \delta' < 1$, the inequality $\left( 1 -  \frac{\delta '}{\widetilde{T}} \right)^{\widetilde{T}} \geq 1-\delta'$ and $\frac{\sqrt{d}}{d-1} \leq 2$.
\end{proof}

\begin{corollary}\label{cor:expected_regret_linucbvn}
Under the same assumptions of Theorem~\ref{th:regret_bound_d2}, we can choose $\delta = \frac{1}{\widetilde{T}}$ and Algorithm~\ref{alg:weighted_linUCB} achieves
\begin{align}
    \EX_{\theta,\pi}[ R_T ( \mathbb{S}^{d-1},\theta , \pi )] \leq 8(d-1) + \left(144d^2\beta^2_{T,\max}+24(d-1)^{\frac{3}{2}}\beta_{T,\max}\right)\log\left( \frac{T}{2(d-1)} \right),
\end{align}
or 
\begin{align}
    \EX_{\theta,\pi} [ R_T ( \mathbb{S}^{d-1},\theta , \pi )] = O( d^{4}\log^3 (T) )  .
\end{align}

\end{corollary}

\begin{proof}   
    Defining the event
    \begin{align}
        R_T = \bigg\lbrace  \mathcal{H}_{\widetilde{T}}: 
         R_T ( \mathbb{S}^{d-1},\theta , \pi ) \leq 4(d-1) + \left(144d^2\beta^2_{T,\max}+24(d-1)^{\frac{3}{2}}\beta_{T,\max}\right)\log\left( \frac{T}{2(d-1)} \right) \bigg\rbrace
    \end{align}
    we have by Theorem~\ref{th:regret_bound_d2} that
    \begin{align}
        \mathrm{Pr}(R_T) &\geq 1 - \delta = 1 - \frac{1}{\widetilde{T}}, \\
        \mathrm{Pr}(R_T^C) &\leq \frac{1}{\widetilde{T}}.
    \end{align}
    Then
    \begin{align}
            \EX_{\theta,\pi}[ R_T ( \mathbb{S}^{d-1},\theta , \pi )] &= \EX_{\theta,\pi}\left[ R_T ( \mathbb{S}^{d-1},\theta , \pi ) \mathbbm{1}\lbrace R_T \rbrace\right] + \EX_{\theta,\pi}\left[ R_T ( \mathbb{S}^{d-1},\theta , \pi ) \mathbbm{1}\lbrace R_T^C \rbrace \right] \nonumber \\
            &\leq 4(d-1) + \left(144d^2\beta^2_{T,\max}+24(d-1)^{\frac{3}{2}}\beta_{T,\max}\right)\log\left( \frac{T}{2(d-1)} \right) \nonumber \\
            &+ 4(d-1)\widetilde{T} \mathrm{Pr}\left( R_T^C \right) \nonumber \\
            &\leq 8(d-1) + \left(144d^2\beta^2_{T,\max}+24(d-1)^{\frac{3}{2}}\beta_{T,\max}\right)\log\left( \frac{T}{2(d-1)} \right),
    \end{align}
    and the results follows using that for $\delta = \frac{1}{\widetilde{T}}$, $\beta_{T,\max} = O (d\log (T))$.   
\end{proof}

\subsection{Linear bandits with linearly vanishing variance}\label{sec:linearly_vanishin_variance}

We conclude the algorithms Chapter by studying the linear bandit model in which the PSMAQB setting naturally fits, allowing us to address the questions posed in the Introduction. Specifically, we relax the subgaussian condition~\eqref{eq:linear_vanishing_subgaussian} considered in Section~\ref{sec:subgaussian_parameter}, and study a linear bandit with action set $\mathcal{A} = \mathbb{S}^{d-1}$, unknown parameter $\theta \in \mathbb{S}^{d-1}$, and reward model $
X_t = \langle \theta, A_t \rangle + \epsilon_t$ where the noise $\epsilon_t$ satisfies
\begin{align}\label{eq:linearly_vanishing_variance}
     \EX[\epsilon_t \mid A_1,X_1,\dots,A_{t-1},X_{t-1},A_t] &= 0, \nonumber \\
    \mathrm{Var}[\epsilon_t \mid A_1,X_1,\dots,A_{t-1},X_{t-1},A_t ] &\leq 1 - \langle \theta, A_t \rangle^2.
\end{align}
This setting corresponds to the variance condition~\eqref{eq:vanishing_variance} in the Bloch sphere representation of the PSMAQB problem. Although we still aim to use the weighted estimators introduced under the subgaussian noise assumption~\eqref{eq:linear_vanishing_subgaussian}, the confidence region derived in Lemma~\ref{lem:confidence_region_weighted} does not directly extend to the case of weighted variance estimation. We will need an additional tool on top of the weighted estimators.

\subsubsection{Median of means for an online least squares estimator}\label{sec:MoMLSE}

We want to use the same idea as in the previous sections and, through a careful choice of actions, introduce weighted terms in the design matrix $1/\hat{\sigma}^2_s (a_s)$ that overestimate the reward variances and “boost” the confidence along directions $a_s$ close to $\theta$. However, this comes at a price: to obtain good concentration bounds for our estimator, we now have to deal with unbounded random variables and only finite variance. While the subgaussian property is strong enough to guarantee well-defined confidence regions, the variance condition is looser, and the random variables can have heavy tails with looser concentration. We address this issue using the recent ideas of median of means (MoM) for online least squares estimators introduced in~\cite{bandits_heavytail,heavy_tail_linear_noptimal,heavy_tail_linear_optimal}. The construction builds on the classical median of means method~\cite[Chapter 3]{lerasle2019lecture} for real-valued random variables with unbounded support and bounded variance, but requires non-trivial adaptation to the online linear least squares setting. As in the classical case, we use the $k$ independent estimators to construct a median of means estimator from which we can derive a confidence region with concentration bounds scaling as $1 - \exp(-k)$.

First we discuss the medians of means method for the online linear least squares estimator introduced in~\cite{heavy_tail_linear_optimal}. We are going to use this estimator later in order to design a strategy that minimizes the regret for the noise model~\eqref{eq:linearly_vanishing_variance}. In order to build the median of means online least squares estimator for linear bandits we need to sample $k$ independent rewards for each action. Specifically given an action set $\mathcal{A}\subset\mathbb{R}^d$, an unknown parameter $\theta\in\mathbb{R}^d$, at each time step $t$ we select an action $A_t\in\mathcal{A}$ and sample $k$ independent rewards using $A_t$ where the outcome rewards are distributed as
\begin{align}
    X_{t,i} = \langle \theta , A_t \rangle + \epsilon_{t,i} \quad \text{for }i\in[k],
\end{align}
for some noise such that $\EX [\epsilon_{t,i} | A_1,X_1,...,A_{t-1},X_{t-1}, A_t ] = 0$. We refer to $k$ as the number of subsamples per time step. Then at time step $t$ we define $k$ least squares estimators as
\begin{align}
    \hat{\theta}_{t,i} = V_t^{-1}(\lambda ) \sum_{s=1}^t X_{s,i} A_s \quad \text{for }i\in[k],
\end{align}
where $V_t$ is the design matrix defined as
\begin{align}
    V_t (\lambda ) = \lambda \mathbb{I} + \sum_{s=1}^t A_s A_s^{\mathsf{T}},
\end{align}
with $\lambda > 0$ being a parameter that ensures invertibility of $V_t$. We note that the design matrix is independent of $i$. Then the median of means for least squares estimator (MOMLSE) is defined as
    \begin{align}\label{eq:linearMOM}
        \hat{\theta}_{t}^{\text{\tiny MoM}} := \hat{\theta}_{t,k^*} \quad \text{where }k^* = \argmin_{j\in[k]} y_j ,
    \end{align}
where     
\begin{align}
        y_j = \text{median}\lbrace \|\hat{\theta}_{t,j} - \hat{\theta}_{t,i} \|_{V_t}: i\in [k]/j \rbrace \quad \text{for } j \in [k].
    \end{align}

Using the results in~\cite{heavy_tail_linear_optimal} we have that the above estimator has the following concentration property around the true estimator.
\begin{lemma}[Lemma 2 and 3 in~\cite{heavy_tail_linear_optimal}]\label{lem:concentration_mom}
Let $\hat{\theta}_t^{\textup{\tiny MoM}}$ be the MOMLSE defined in~\eqref{eq:linearMOM} with $k$ subsamples with  $\lbrace X_{s,i}\rbrace_{(s,i)\in [t]\times [k]} $ rewards and corresponding actions $\lbrace A_{s} \rbrace_{s\in [t]} $. Assume that the noise of all rewards has bounded variance, i.e $\EX \left[ \epsilon^2_{s,i} |  A_1,X_1,...,A_{t-1},X_{t-1}, A_t \right] \leq 1$ for all $s\in[t]$ and $i\in [ k ]$. Then we have
    \begin{align}
    \mathrm{Pr}\left( \|\theta - \hat{\theta}_t^{\text{\tiny MoM}} \|^2_{V_t} \leq 9\left(\sqrt{9d} + \lambda \| \theta \|_2 \right)^2 \right)\geq 1 - \exp \left( \frac{-k}{24}\right).
\end{align}

\end{lemma}

We will use a slight modification of the above result introducing the weights as in Section~\ref{sec:weighted_confidence_region} $,   \hat{\sigma}^2_t : \mathcal{H}_{t-1}\times A \rightarrow \mathbb{R}_{>0},$
where $\mathcal{H}_{t-1} = \lbrace X_{s,i}\rbrace_{(s,i)\in [t-1]\times [k]} \cup \lbrace A_{s} \rbrace_{s \in [t-1]}$ contains the past information of rewards and actions. As before for our purposes we will use only the information of the past actions and in order to simplify notation we will use $\hat{\sigma}^2_t (A)$ to denote an estimator of the variance for the reward associated action $A\in\mathcal{A}$ with the information collected up to time step $t-1$. Then to define the weighted version of the median of means we need to sample $k$ rewards for each action and define the following $k$ linear least squares estimators
\begin{align}\label{eq:weighted_lse}
    \hat{\theta}^{\text{w}}_{t,i} = V_t^{-1} \sum_{s=1}^t \frac{1}{\hat{\sigma}^2_s (A_s)} X_{s,i} A_s \quad \text{for }i\in[k],
\end{align}
with the weighted design matrix
\begin{align}\label{eq:weighted_design}
    V_t = \lambda \mathbb{I} + \sum_{s=1}^t \frac{1}{\hat{\sigma}^2_s (A_s)} A_s A_s^{\mathsf{T}}.
\end{align}
Then the weighted version of the median of means linear estimator is defined analogously to~\eqref{eq:linearMOM} with the corresponding weighted versions~\eqref{eq:weighted_lse}\eqref{eq:weighted_design} and we will denote it as $\hat{\theta}_t^{\text{\tiny wMOM}}$.
In our algorithm analysis we will use the following analogous concentration bound under the condition that the estimators $\hat{\sigma}^2_t$ overestimate the true variance.
\begin{corollary}\label{cor:concentration_wmom}
    Let $\hat{\theta}_t^{\textup{\tiny wMOM}}$ be the weighted version of the MOMLSE with $k$ subsamples, $\lbrace X_{s,i}\rbrace_{(s,i)\in [t]\times [k]} $ rewards with corresponding actions $\lbrace A_{s} \rbrace_{s\in [t]} $ and variance estimator $\hat{\sigma}^2_t$. Define the following event
    \begin{align}\label{eq:gt_event_mom}
        G_t : = \lbrace \big( \mathcal{H}_{t-1}, A_t \big) : \textup{\textrm{Var}}[\epsilon_{s,i}] \leq \hat{\sigma}^2 (A_s )\ \forall s,i\in [t]\times [k  ]\rbrace.
    \end{align}
    Then we have
    \begin{align}
    \mathrm{Pr}\left( \|\theta - \hat{\theta}_t^{\textup{\tiny wMOM}} \|^2_{V_t} \leq  \beta_{\textup{w}}\mid G_t \right)\geq 1 - \exp \left( \frac{-k}{24}\right),
    \end{align}
    where 
    \begin{align}\label{eq:beta_constant_mom}
        \beta_{\textup{w}} := 9\left(\sqrt{9d} + \lambda \| \theta \|_2 \right)^2 .
    \end{align}
\end{corollary}

\begin{proof}
    The result follows from applying Lemma~\ref{lem:concentration_mom} to the sequences of re-normalized rewards $\lbrace \frac{X_{s,i}}{\hat{\sigma}_s (A_s)}\rbrace_{(s,i)\in [t]\times [k]} $  and actions $\lbrace \frac{A_{s,i}}{\hat{\sigma}_s (a_s)} \rbrace_{s\in [t]} $. We only need to check that the sequence $\lbrace \frac{\epsilon_{s,i}}{\hat{\sigma}_s (A_s)}\rbrace_{(s,i)\in [t]\times [k]} $ has finite variance. Conditioning with the event $G_t$ and the fact that by definition $\hat{\sigma}^2_s (A_s)$ only depend on the past $s-1$ action and rewards we have that the re-normalized noise has bounded variance since
    \begin{align}
        \EX \left[ \left(\frac{\epsilon_{s,i}}{\hat{\sigma}_s (A_s)} \right)^2 \Bigg|\mathcal{H}_{s-1}, A_s\right] = \frac{1}{\hat{\sigma}^2_s (A_s)}\EX [ \epsilon^2_{s,i} | \mathcal{H}_{s-1}, A_s ] = \frac{\mathrm{Var} [\epsilon_{s,i}]}{\hat{\sigma}^2_s (A_s)}\leq 1.
    \end{align}
\end{proof}

\subsubsection{\textsf{LinUCB} variance vanishing noise algorithm}\label{sec:linucb_vvn}

The algorithm that we design for linear bandits with linearly variance vanishing noise~\eqref{eq:linearly_vanishing_variance} is \textsf{LinUCB-VVN} (\textsf{LinUCB} vanishing variance noise) stated in Algorithm~\ref{alg:linucb_vn_var}. This algorithm combines all the tools we developed in this Chapter including the action selection rule from Theorem~\ref{th:main_eigenvalues}, the weighted LSE from Section~\ref{sec:weighted_confidence_region} and the recently introduced median of means~\eqref{eq:linearMOM}. The algorithm runs in $\tilde{T}\in\mathbb{N}$ batches, for a total time horizon of $T = 2k(d-1)$ time steps where $k\in\mathbb{N}$ is a free parameter that controls the median of means. At each batch $\tilde{t}\in[\mathbb{T}]$ the algorithm plays the actions
\begin{align}\label{eq:action_general_update_mom}
    A^\pm_{\tilde{t},i} := \frac{\widetilde{A}^\pm_{\tilde{t},i}}{\| \widetilde{A}^\pm_{\tilde{t},i}\|_2}, \quad  \tilde{A}^\pm_{\tilde{t},i} = \frac{\hat{\theta}^{\text{wMoM}}_{\tilde{t}}}{\| \hat{\theta}^{\text{wMoM}}_{\tilde{t}} \|_2}  \pm \frac{1}{\sqrt{\lambda_{\min}(V_{\tilde{t}-1})}}v_{\tilde{t}-1,i}, \quad 
\end{align}
for $i\in [d-1 ]$, where $v_{\tilde{t}-1,i}$ is $i$-th eigenvector of the weighted design matrix $V_{\tilde{t}-1}$ and where for each action $A^\pm_{\tilde{t},i}$ we sample $k$ independent rewards $\lbrace X^\pm_{\tilde{t},i,j} \rbrace_{j=1}^k$ in order to build the weighted MOMLSE defined as in Section~\ref{sec:MoMLSE}. The design matrix $V_{\tilde{t}}$ is updated as
\begin{align}\label{eq:design_update_mom}
    V_{\tilde{t}} (\lambda ) = V_{\tilde{t}-1}(\lambda ) + \omega (V_{\tilde{t}-1} (\lambda ) ) \sum_{i=1}^{d-1} \left( A^+_{\tilde{t},i} (A^+_{\tilde{t},i})^{\mathsf{T}} + A^-_{\tilde{t},i} (A^-_{\tilde{t},i})^\mathsf{T} \right) ,
\end{align}
and the $k$ weighted least squares estimators are 
\begin{align}
    \hat{\theta}^{\text{w}}_{\tilde{t},j} =  V^{-1}_{\tilde{t}} (\lambda ) \left( \sum_{s=1}^{\tilde{t}} \omega (V_{s-1} )  \left( \sum_{i=1}^{d-1}(X^+_{s,i,j}A^+_{s,i} + X^-_{s,i,j}A^-_{s,i} \right) \right)
\end{align}
where the weights $\omega (V_{\tilde{t}-1} )$ and variance estimator are chosen as
\begin{align}\label{eq:weight_choice_mom}
    \omega (V_{\tilde{t}-1}) := \frac{\sqrt{\lambda_{\max}(V_{\tilde{t}-1})}}{12\sqrt{d-1}\beta_w} , \quad \hat{\sigma}^2_t (A^\pm_{\tilde{t},i}) := \frac{1}{\omega (V_{\tilde{t}-1})},
\end{align}
with $\beta_w$ defined as in~\eqref{eq:beta_constant_mom}. We note that the definition for $\hat{\sigma}^2_t (A^\pm_{\tilde{t},i})$ fulfills the definition of variance estimator~\eqref{eq:variance_estimator} stated in the previous section since it only depends on the past history $\mathcal{H}_{t-1}$.

\begin{algorithm}
    \begin{algorithmic}
    
	\caption{\textsf{LinUCB-VVN}} 
	\label{alg:linucb_vn_var}
 
       \State Require: $\lambda_0\in\mathbb{R}_{>0}$, $k\in\mathbb{N}$,  $\omega: \text{P}^d_+ \rightarrow \mathbb{R}_{\geq 0}$
        
      \State  Set initial design matrix $V_0 \gets \lambda_0\mathbb{I}_{d\times d}$ 
        
     \State   Choose initial estimator ${\theta}_0\in\mathbb{S}^d$ for $\theta$ at random 
        
        \For{$\tilde{t} = 1,2,\cdots$}
            \vspace{1mm}
           \State \textit{Optimistic action selection}
            \vspace{1mm}
            
            \For{$i = 1,2,\cdots d-1$}
               \State Select actions $A^+_{\tilde{t},i}$ and $A^-_{\tilde{t},i}$ according to Eq.~\eqref{eq:action_general_update_mom}
                
                \vspace{1mm}
               \State \textit{Sample $k$ independent rewards for each $a^\pm_{t,i}$}
                \vspace{1mm}
                
                \For{$j=1,...,k$}
                \State    Receive associated rewards $X^+_{\tilde{t},i,j}$ and $X^-_{\tilde{t},i,j}$
                    \EndFor
            
            \EndFor
            \vspace{1mm}
          \State  \textit{Update estimator of subgaussian noise for $A^+_{t,i}$}
            \vspace{1mm}
            
          \State  $\hat{\sigma}^2_{\tilde{t}} \gets \frac{1}{\omega ( V_{\tilde{t}-1}(\lambda_0 ) ) }$ for $\tilde{t}\geq 2$ or $\hat{\sigma}^2_{\tilde{t}} \gets 1$ for $\tilde{t}=1$.

            \vspace{1mm}
          \State  \textit{Update design matrix}
            \vspace{1mm}
            
        \State    $V_{\tilde{t}} \gets V_{\tilde{t}-1} + \frac{1}{\hat{\sigma}_{\tilde{t}}^2} \sum_{i=1}^{d-1} \left(A^+_{t,i}  (A^+_{t,i})^{\mathsf{T}}  +  A^-_{t,i}  (A^-_{t,i})^\mathsf{T} \right)$

            \vspace{1mm}
        \State    \textit{Update LSE for each subsample}
            \vspace{1mm}

            \For{$j=1,2,...,k$:}
              \State   $\hat{\theta}_{t,j}^\text{w} \gets V_{\tilde{t}}^{-1}  \sum_{s = 1}^{\tilde{t}} \frac{1}{\hat{\sigma}_{\tilde{t}}^2} \sum_{i=1}^{d-1} (A^+_{s,i} X^+_{s,i,j} + A^-_{s,i} X^-_{s,i,j} ) $
             \EndFor
          \State  Compute $\hat{\theta}_{\tilde{t}}^{\text{\tiny wMOM}}$ using $\lbrace \hat{\theta}_{\tilde{t},j}^\text{w} \rbrace_{j=1}^k$
        \EndFor
            \end{algorithmic}
\end{algorithm}

\subsubsection{Regret analysis for \textsf{LinUCB-VVN}}

\begin{theorem}\label{th:regret_scaling_variance}
Let $d\geq 2$, $k\in\mathbb{N}$, $\widetilde{T}\in\mathbb{N}$ the number of batches and $T=2(d-1)k\widetilde{T}$ the total time horizon. Let $\omega (X)$ defined as in~\eqref{eq:weight_choice_mom} using $\lambda_0$ satisfying the constraints in Theorem~\ref{th:main_eigenvalues}. Then if we apply Algorithm~\ref{alg:linucb_vn_var}($\lambda_0,k,\omega$) to a $d$ dimensional stochastic linear bandit with variance as in~\eqref{eq:linearly_vanishing_variance} with probability at least $(1-\exp (-k/24))^{\widetilde{T}}$ the regret satisfies
      \begin{align}
      R_T (\mathbb{S}^{d-1},\theta, \pi) \leq 4k(d-1)&+144d(d-1)k(\beta_w)^2\log\left(\frac{T}{2(d-1)k} \right) \nonumber \\
      &+24(d-1)^{\frac{3}{2}}k\beta_w\log\left(\frac{T}{2(d-1)k} \right),
    \end{align}
    and at each time step $t\in[T]$ with the same probability it can output an estimator $\hat{\theta}_{t}\in\mathbb{S}^{d-1}$ such that
    \begin{align}
        \|\theta - \hat{\theta}_{t}\|_2^2 \leq \frac{576d^2(\beta_w)^2 k+96d\sqrt{d-1}\beta_w k}{t},
    \end{align}
    with $\beta_w$ defined as in~\eqref{eq:beta_constant_mom}.
\end{theorem}

From the above Theorem we have that if we set $k = \lceil 24\log\left( \frac{\widetilde{T}}{\delta}\right) \rceil$ for some $\delta \in \left( 0 , 1 \right)$ then with probability at least $1 - \delta$ \textsf{LinUCB-VNN} achieves
\begin{align}
    \text{Regret}(T) = O\left( d^4\log^2 (T)\right), \quad  \|\theta - \hat{\theta}_t\|_2^2 = O\left( \frac{\log(T)}{t}\right).
\end{align}

\begin{proof}
The proof follows the lines of the proof in Theorem~\ref{th:regret_bound_d2} with few modifications. From the expression of the regret using $\theta,A^\pm_{\tilde{t},i}\in\mathbb{S}^{d-1}$ we have that
\begin{align}
       R_T (\mathbb{S}^{d-1},\theta, \pi) &= \frac{1}{2}\sum_{t=1}^{T} \|\theta - A_t \|_2^2   \\
    &= \frac{1}{2}\sum_{\tilde{t}=1}^{\tilde{T}}\sum_{i=1}^{d-1}\sum_{j=1}^k \left( \|\theta - A^{+}_{\tilde{t},i} \|_2^2 + \|\theta - A^{-}_{\tilde{t},i} \|_2^2 \right) \\
    &\leq  4k(d-1) + \frac{k}{2}\sum_{\tilde{t}=2}^{\tilde{T}}\sum_{i=1}^{d-1} \left( \|\theta - A^{+}_{\tilde{t},i} \|_2^2 + \|\theta - A^{-}_{\tilde{t},i} \|_2^2 \right) . 
\end{align}
Thus, it suffices to gives an upper bound between the distance of the unknown parameter $\theta$ and the actions $A^\pm_{\tilde{t},i}$ selected by the algorithm~\eqref{eq:action_general_update_mom}. We denote the batches $\tilde{t}\in [\widetilde{T} ]$. First we will do the computation assuming that the event
\begin{align}
    E_{\tilde{t}} := \lbrace \mathcal{H}_{\tilde{t}} : \forall s \in [\tilde{t}], \theta \in \mathcal{C}_s \rbrace , 
\end{align}
holds where $\mathcal{C}_s = \lbrace \theta' \in\mathbb{R}^d : \| \theta' - \hat{\theta}^{\text{wMOM}}_{s} \|^2_{V_s} \leq \beta_w \rbrace$. Here the history $\mathcal{H}_{\tilde{t}}$ is defined with the previous outcomes and actions of our algorithm i.e 
\begin{align}
    \mathcal{H}_{\tilde{t}} := \left(X^+_{s,i,j}, A^+_{s,i}, X^-_{s,i,j} , A^-_{s,i} \right)_{(s,i,j)\in[\tilde{t}]\times[d-1]\times [k]} 
\end{align}
Later we will quantify the probability that this event always hold. Using the definition of the actions $A^\pm_{\tilde{t},i}$~\eqref{eq:action_general_update_mom} and $\theta,\hat{\theta}^{\text{\tiny wMOM}}_{\tilde{t}}\in\mathbb{S}^{d-1}$ and the arguments of Theorem~\ref{th:circle_regret} we have that 
\begin{align}
    \|\theta - A^{\pm}_{\tilde{t},i} \|_2^2 \leq \frac{9\beta^\text{w}}{\lambda_{\min}(V_{\tilde{t}-1})}.
\end{align}
Then using that the design matrix $V_{\tilde{t}}$~\eqref{eq:design_update_mom} is updated as in Theorem~\ref{th:main_eigenvalues} and the choice of weights~\eqref{eq:weight_choice_mom} we fix 
\begin{align}\label{eq:lambda_condition}
    \lambda_0 \geq \max \left\lbrace 2, 2\sqrt{\frac{2}{3(d-1)}}\frac{d}{12\sqrt{d-1}\beta_w}+\frac{2}{3(d-1)} \right\rbrace
\end{align}
and we have that $\lambda_{\min}(V_{\tilde{t}})\geq \sqrt{\frac{2}{3(d-1)}\lambda_{\max}(V_{\tilde{t}})}$ applying Theorem~\ref{th:main_eigenvalues}. Inserting this into the above we have
\begin{align}\label{eq:dist_attheta_mom}
   \|\theta - A^{\pm}_{\tilde{t},i} \|_2^2 \leq \frac{12\sqrt{d-1}\beta^{\text{w}}}{\sqrt{\lambda_{\max}(V_{\tilde{t}})}}.
\end{align}
Thus, it remains to provide a lower bound on $\lambda_{\max}(V_{\tilde{t}})$. We note that in Theorem~\ref{th:regret_bound_d2} we also had to provide an upper bound but this was because the constant $\beta_w$ beta depended on the time step $t$. The update of the design matrix is the same as the algorithm analyzed in Theorem~\ref{th:regret_bound_d2}, thus we can reuse the same computations and arrive at
\begin{align}
    \lambda_{\max}(V_{\tilde{t}}) \geq \frac{\tilde{t}^2}{4(1+6\frac{d}{\sqrt{d-1}}\beta)^2}.
\end{align}
Now we can insert the above into~\eqref{eq:dist_attheta_mom} and we have
\begin{align}\label{eq:theta_at_distance_mom}
     \|\theta - A^{\pm}_{\tilde{t},i} \|_2^2 &\leq \frac{24\sqrt{d-1}\beta_w(1+6\frac{d}{\sqrt{d-1}}\beta_w)}{\tilde{t}-1} \\
     &= \frac{144d(\beta_w)^2+24\sqrt{d-1}\beta_w}{\tilde{t}-1}.
\end{align}
Thus, we can inserted the above bound into the regret expression and we have
\begin{align}
     R_T (\mathbb{S}^{d-1},\theta, \pi) & \leq 4k(d-1) + (144d(d-1)k(\beta_w)^2+24(d-1)^{\frac{3}{2}}k\beta_w) \sum_{\tilde{t}=2}^{\tilde{T}} \frac{1}{t-1}\\
    &\leq 4k(d-1)+144d(d-1)k(\beta_w)^2\log \widetilde{T} +24(d-1)^{\frac{3}{2}}k\beta_w\log\widetilde{T} \\
    &=  4k(d-1)+144d(d-1)k(\beta_w)^2\log\left(\frac{T}{2(d-1)k} \right) \\
    &+24(d-1)^{\frac{3}{2}}k(\beta_w)\log\left(\frac{T}{2(d-1)k} \right).
\end{align}
It remains to quantify the probability that the event $E_{\tilde{t}}$ holds. For that we will use the concentration bounds of the median of means for least squares estimator stated in Corollary~\ref{cor:concentration_wmom}. From the variance condition of our model~\eqref{eq:linearly_vanishing_variance} we have that for the rewards $X^\pm_{\tilde{t},i,j} = \langle \theta , A^\pm_{\tilde{t},i} \rangle + \epsilon^\pm_{\tilde{t},i,j}$ the variance of the noise satisfies
\begin{align}
   \mathrm{Var} [\epsilon^\pm_{\tilde{t},i,j} | \mathcal{H}_{\tilde{t}-1}, A^\pm_{\tilde{t},i}] \leq 1 - \langle \theta , A^\pm_{\tilde{t},i} \rangle^2 \leq 2(1-\langle \theta , A^\pm_{\tilde{t},i} )) = \| \theta - A^\pm_{\tilde{t},i} \|_2^2,
\end{align}
where we used $1+\langle \theta , A^\pm_{\tilde{t},i} \rangle \leq 2$.
Thus from our choice of weights~\eqref{eq:weight_choice_mom} and~\eqref{eq:theta_at_distance_mom} we have that
\begin{align}
    \text{if } \theta\in\mathcal{C}_{s-1} \Rightarrow \VX [\epsilon^\pm_{\tilde{t},i,j} | \mathcal{F}_{\tilde{t}-1}] \leq \hat{\sigma}_s^2 (A^\pm_{s,i}).
\end{align}
Then in order to apply Corollary~\ref{cor:concentration_wmom} we note that from the choice $\hat{\sigma}_s^2 (a^\pm_{1,i}) = 1$ the event $G_{\tilde{t}}$ at $\tilde{t}=1$ is always satisfied i.e $\mathrm{Pr}(G_1) = 1$. Then applying Bayes theorem, union bound over the events $G_1,E_1,...,G_{t-1},E_t$ and Corollary~\ref{cor:concentration_wmom} we have
\begin{align}
    \mathrm{Pr}(E_{\widetilde{T}} \cap G_{\widetilde{T}}) \geq \left(1 - \exp (-k/24) \right)^{\widetilde{T}} .
\end{align}   
This probability also quantifies the probability that~\eqref{eq:theta_at_distance_mom} holds since the only assumption we used is $\theta\in\mathcal{C}_{\tilde{t}-1}$. Then we can take simply one of the actions $A^\pm_{\tilde{t},i}$ as the estimator $\hat{\theta}_t$ and the result follows using the relabeling $t = 2(d-1)k\tilde{t}$ and the inequality $1/(\tilde{t} - 1) \leq 2/\tilde{t}$ for $\tilde{t}\geq 2$. The more detailed computation of the above probability is identical as the on in Theorem~\ref{th:regret_bound_d2}.
\end{proof}

In the previous Theorem we did not set a specific value for the parameter $k$ or the number of subsamples per action. We note that the regret scales linearly with $k$ but since the success probability scales exponentially with $k$ it will suffice to set $k\sim \log (T)$ such that in expectation we get the $\log^2 (T)$ behaviour. We formalize this in the following Corollary.

\begin{corollary}\label{cor:expected_regret_mom}
Under the same assumptions of Theorem~\ref{th:regret_scaling_variance} we can fix $k = \lceil 24\log(\widetilde{T}^2) \rceil$ and we have
\begin{align}
    \EX_{\theta,\pi} \left[ R_T (\mathbb{S}^{d-1},\theta, \pi) \right] \leq 344(d-1)\log\left(T \right)+ 
    \left( 3546d(d-1)(\beta_w)^2 +1152(d-1)^{\frac{3}{2}}\beta_w \right)\log^2\left(T \right)
\end{align}
and for $t\in [T ]$,
\begin{align}
    \EX_{\theta,\pi} \left[ \|\theta - \hat{\theta}_t \|_2^2 \right] \leq  \frac{27648d^2(\beta_w)^2 \log ( T)+4608d\sqrt{d-1}\beta_w \log ( T)}{t} + \frac{4(d-1)\log (T )}{T}.
\end{align}
Using that $\beta_w = O (d)$ gives
\begin{align}
    \EX_{\theta,\pi} \left[ R_T (\mathbb{S}^{d-1},\theta, \pi) \right] = O(d^4\log^2 (T) ), \quad  \EX \left[ \|\theta - \hat{\theta}_t \|_2^2 \right] = \tilde{O}\left( \frac{d^4}{t} \right).
\end{align}
\end{corollary}

\begin{proof}
The result of Theorem~\ref{th:regret_scaling_variance}  holds with probability at least $(1-\exp (-k/24))^{\widetilde{T}}$. Setting $k = \lceil 24\log(\widetilde{T}^2) \rceil $ gives 
    \begin{align}
        (1-\exp (-k/24))^{\widetilde{T}} \geq \left( 1 - \frac{1}{\widetilde{T}^2} \right)^{\widetilde{T}} \geq 1 - \frac{1}{\widetilde{T}}.
    \end{align}
Then given the event $R_T$ such that Algorithm~\ref{alg:linucb_vn_var} achieves the bounds given by Theorem~\ref{th:regret_scaling_variance} we have that the probability of failure is bounded by 
\begin{align}
    \mathrm{Pr}(R^C_T)  \leq \frac{1}{\widetilde{T}},
\end{align}
where we used $1 = \mathrm{Pr}(R_T)  + \mathrm{Pr}(R^C_T)  $. Then the expectation of the bad events can be bounded as
\begin{align}
    \EX \left[\text{Regret}(T)\mathbbm{1}\lbrace R^C_T \rbrace \right] &\leq  4(d-1)k\widetilde{T}\mathrm{Pr}(R^C_T) \leq  4(d-1)k \\
    \EX \left[ \| \theta -\hat{\theta}_t \|_2^2 \mathbbm{1}\lbrace R^C_T \rbrace \right] &\leq 4\mathrm{Pr}(R^C_T) \leq \frac{4}{\widetilde{T}} 
\end{align}
where we used $R_T (\mathbb{S}^{d-1},\theta, \pi) \leq 2T = 4(d-1)k\widetilde{T}$, $\| \theta - \hat{\theta}_t\|_2^2 \leq 4$. Finally the result follows inserting the value of $k = 24\log (\widetilde{T}^2) $ into the bounds of Theorem~\ref{th:regret_scaling_variance} and using $\widetilde{T} \leq T$.
\end{proof}

\textbf{Note.} While the above results apply for the qubit version of the PSMAQB setting, we will state the exact connection in Theorem~\ref{th:regret_PSMAQB} in the first section of the next Chapter~\ref{ch:applications} of applications.

\chapter{Applications}\label{ch:applications}

This chapter is based on the author's works~\cite{lumbreras24pure,bramachari24intelligence,lumbreras2025quantumwork}, which explore applications of the multi-armed quantum bandit framework.

We begin with an application of the PSMAQB algorithm~\ref{alg:linucb_vn_var} to quantum state tomography, where we relate regret to cumulative disturbance. We show that it is possible to learn pure quantum states with minimal disturbance to the samples.

Next, we consider a specific setting of learning with minimal disturbance: quantum state-agnostic work extraction protocols. In this task, we study protocols for extracting work when given oracle access to an unknown pure qubit state, with the goal of transferring energy to a battery system. We analyze both discrete and continuous battery settings, and show that using the PSMAQB algorithm~\ref{alg:linucb_vn_var}, we can achieve polylogarithmic cumulative dissipation. Here, dissipation refers to the loss of extractable energy due to not knowing the underlying state.

As a final application, we modify the MAQB setting to study a recommender system for quantum data. In this setting, in each round, a learner receives an observable (the context) and must choose one quantum state to measure from a finite set of unknown states (the actions). The goal is to maximize the reward in each round, defined as the outcome of the measurement. Based on this model, we formulate the low-energy quantum state recommendation problem, where the context is a Hamiltonian and the objective is to recommend the state with the lowest energy. For this task, we study two families of contexts: the Ising model and a generalized cluster model.

\section{Learning pure states without disturbing them}

We start with one of the most immediate applications of the PSMAQB setting, a new approach for quantum state tomography. Given sequential access to a finite number of samples of a quantum state, our goal is not only to accurately learn a classical description of the state but also to use measurements that disturb the copies as little as possible. Generally, these two goals are incompatible, and we are thus interested in tomography algorithms that find an optimal balance between them. 

Minimizing disturbance is important in many real-world scenarios where the samples that we use for tomography are in fact resources for another tasks\,---\, and thus we want to learn the state in a way that is as non-intrusive as possible, ensuring that the post-measurement states remain useful for their intended purpose. An example of this occurs in quantum key distribution, where tomography can be used to keep reference frames aligned during a run, but any disturbance due to tomographic measurements will induce bit errors in the correlations used to extract a secret key. Disturbance is also relevant for state-agnostic resource distillation, where resourceful states might be destroyed by tomographic measurements, but learning the unknown state is crucial since optimal extraction protocols generally depend on its description. 

In both cases, we encounter a fundamental trade-off between exploration (learning the state) and exploitation (using the samples for another purpose). This trade-off is exactly the one we find in a multi-armed bandit problem, and we are going to see how we can relate the disturbance to the samples to the figure of merit we studied during this thesis, the regret.

More formally, we consider a scenario where we have sequential access to $T$ samples in an unknown pure qubit quantum state $|\psi\rangle$ and at each round $t\in[T]$, we select a probe direction $|\psi_t\rangle$ and perform a measurement sampling a result $r_t\in\lbrace 0 ,1 \rbrace$ distributed according to Born's rule, i.e., $p_t = \Pr[R_t = 1] = |\langle \psi|\psi_t\rangle|^2$. Here, $r_t = 1$ indicates that the post-measurement state is $\psi_t$ and $r_t = 0$ indicates its collapse to the orthogonal complement, $\psi_t^c$. Specifically, the distribution of the post-selected state $\xi_t\in\mathcal{S}_d$ after the measurement is
\begin{align}\label{eq:post_measured_state_dist}
    \mathrm{Pr}\left(\xi_t | \psi_{t} \right) = \begin{cases}
      |\langle \psi | \psi_t \rangle |^2 \quad \text{if } \quad \xi_t =  |\psi_{t}\rangle \! \bra{\psi_{t}} \\
       1 - |\langle \psi | \psi_t \rangle |^2 \quad \text{if } \quad \xi_t = |\psi^c_{t}\rangle \! \bra{\psi^c_{t}}  \\
       0 \quad \text{otherwise} ,
    \end{cases} 
\end{align}
where 
\begin{align}
|\psi^c_{t}\rangle = \frac{\ket{\psi}- \langle  \psi_{t}| \psi \rangle \ket{\psi_{t}}}{\sqrt{1 - |\langle \psi | \psi_{t} \rangle |^2}} .
\end{align}

Then the disturbance is quantified by the cumulative infidelity between the unknown state and the post-measurement state $\rho_t = p_t \psi_t + (1-p_t) \psi_t^c$, i.e., and we define it as
\begin{align}\label{eq:disturbance}
  \text{Disturbance}(T) :=& \sum_{t=1}^T  1 - \bra{\psi}\! \rho_t\! \ket{\psi} \\
  =& \sum_{t=1}^T 2\, |\!\braket{\psi|\psi_t}\!|^2  \big( 1 - |\!\braket{\psi|\psi_t}\!|^2 \big) \,.
\end{align}
As can be seen from the second equation above, the disturbance is small when the probe direction is closely aligned with the unknown state. Hence, learning the state is, in fact, necessary to minimize disturbance. However, since we are interested in tomography, we also explicitly require that the algorithm outputs an estimate $\hat{\psi}_T$ with high fidelity to the unknown state after $T$ rounds, i.e., we aim to minimize the error
\begin{align}\label{eq:fidelity_intro_error}
    \text{Err}(T) := 1 - \big| \langle \psi |\hat{\psi}_T\rangle \big|^2 \,.
\end{align}

\subsection{Minimizing disturbance with the PSMAQB}

We note that the setting described in the previous section is captured by the PSMAQB framework, as it corresponds to a sequential quantum state tomography problem where the unknown state (the environment) is pure and the measurements are all rank-1 projectors. The only difference lies in the form of the figures of merit: the disturbance~\eqref{eq:disturbance} and the regret~\eqref{eq:regret_psmaqb}. However, these are technically equivalent, as we formalize in the following lemma.

\begin{lemma}
Let $\pi$ be a policy for a $d$-dimensional PSMAQB and let $T \in \mathbb{N}$ be the time horizon. Then, minimizing the cumulative regret $R_T(\mathbb{S}_d^*,\pi,\pi)$~\eqref{eq:regret_psmaqb} is equivalent to minimizing the disturbance $\textup{Disturbance}(T)$~\eqref{eq:disturbance}.
\end{lemma}

\begin{proof}
   Defining $\Pi_{a_t} = |\psi_t\rangle \! \langle \psi_t |$ we have that
    \begin{align}
         \text{Disturbance}(T)   = &2\sum_{t=1}^T \langle \psi |\Pi_{a_t} | \psi \rangle ( 1 - \langle \psi |\Pi_{a_t} | \psi \rangle ) \\
         &\leq 2\sum_{t=1}^T ( 1 - \langle \psi |\Pi_{a_t} | \psi \rangle )\\
         &= 2 R_T(\mathbb{S}_d^*,\pi,\psi)
    \end{align}
    where we used $\langle \psi |\Pi_{a_t} | \psi \rangle\leq 1$ . For the reverse inequality, we can use
    \begin{align}
         \text{Disturbance}(T)  \geq \sum_{t=1}^T \min \lbrace 1 - \langle \psi |\Pi_{a_t} | \psi \rangle , \langle \psi |\Pi_{a_t} | \psi \rangle\rbrace ,
    \end{align}
    combined with $ \min \lbrace 1 - \langle \psi |\Pi_{a_t} | \psi \rangle  , \langle \psi |\Pi_{a_t} | \psi \rangle  \rbrace =  1 - \langle \psi |\Pi_{a_t} | \psi \rangle  $ where we used that for $\langle \psi |\Pi_{a_t} | \psi \rangle \leq \frac{1}{2}$ we can just relabel the binary outcomes of the POVM $(\Pi_{a_t}, \mathbb{I}-\Pi_{a_t} )$. 
\end{proof}

With the above lemma, we can say that in fact the PSMAQB setting and the tomography setting are equivalent, and we can study policies that minimize regret since they will also minimize the disturbance~\eqref{eq:disturbance}. This is the main reason why we call this setting \textit{quantum state tomography with minimal regret}. In the following theorem, we have the specific instance of \textsf{LinUCB-VNN} (Algorithm~\ref{alg:linucb_vn_var}) for $d=3$ that applies to the qubit PSMAQB setting. We state the regret and the infidelity scaling of the estimator, since in quantum state tomography, we are also interested in minimizing the error~\eqref{eq:fidelity_intro_error}.

\begin{figure}
    \centering
    \begin{overpic}[percent,width=0.7\textwidth]{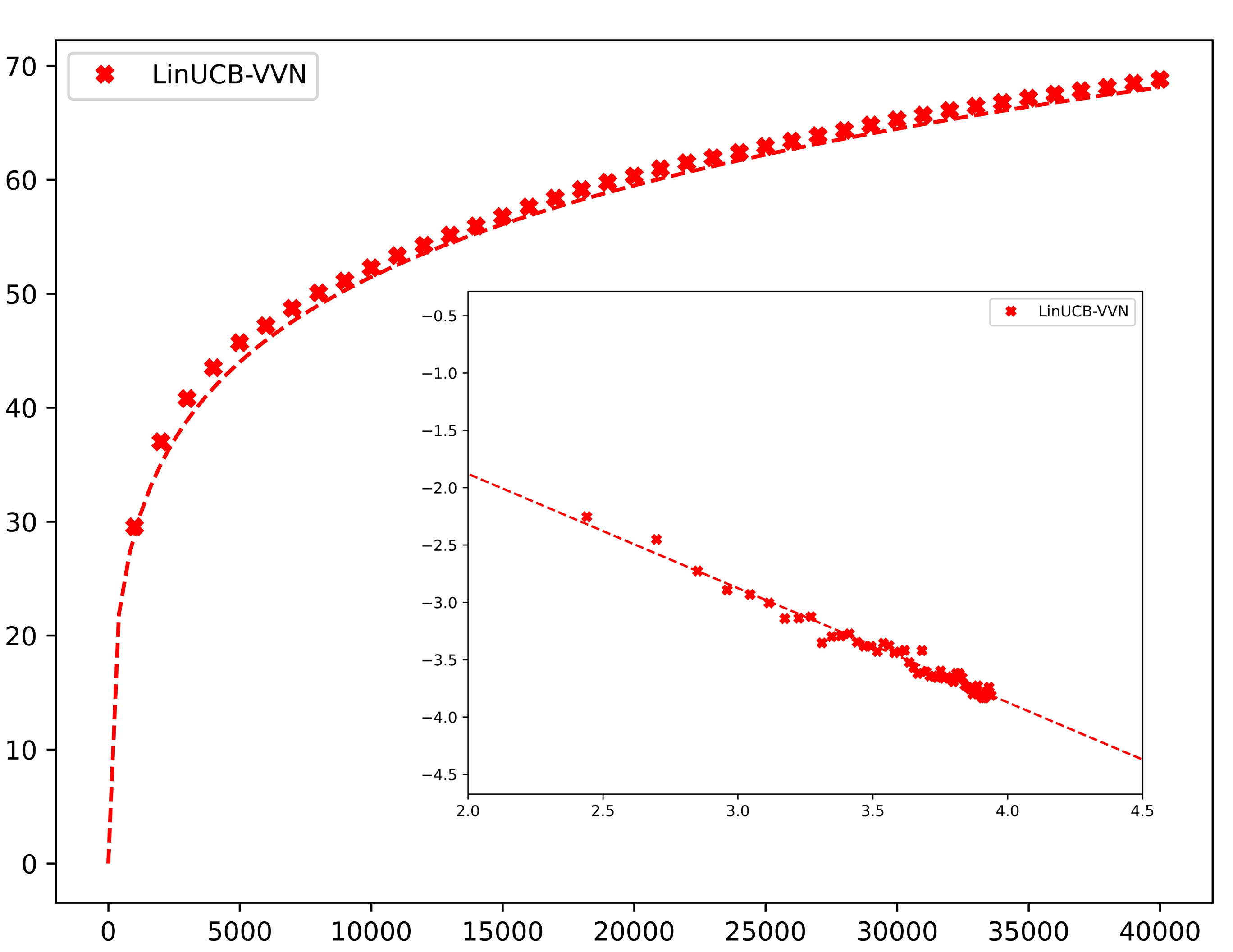}
    \put(-6,35){\rotatebox{90}{Regret($T$)}}
    \put(50,-2){$T$}
    \put(32,25){\rotatebox{90}{\tiny$\log\left( 1-F(\Pi,\Pi_t) \right)$}}
    \put(60,8){\tiny $\log\left( \frac{t}{\log (t)} \right)$}
    \end{overpic}
    \caption{Expected regret vs the number or rounds $T$ for the \textsf{LinUCB-VNN} algorithm. We run $T = 4\cdot 10^4$ rounds with $k = 10$ subsamples for the median of means construction. We use $100$ independents experiments and average over them. We obtain results for each round but only plot (red crosses) few for clarity of the figure. We fit the regression $\text{Regret}(T) = m_1\log^2 T + b_1$ with $m_1 = 3.2164 \pm 0.0009$ and $b_1 = 0.84 \pm 0.016$. In the inset plot we plot the expected infidelity  of the output estimator at each rounds $t\in [T]$ versus the number of rounds $t$. We take $\Pi_t = \Pi_{\theta^{\text{wMoM}}_t} $ as the estimator given by the median of means linear least squares estimator. We fit the regression $1- F(\Pi,\Pi_t) = b_2 \left(\frac{\log t}{t}\right)^{m_2} $ and we obtain $m_2 = -0.996 \pm 0.002$ $b_2 = 0.112\pm 0.007$. We note that the number of subsamples of the theoretical results is very conservative in comparison with the value we take for the simulations.}
    \label{fig:regret_PSMAQB_plot}
\end{figure}
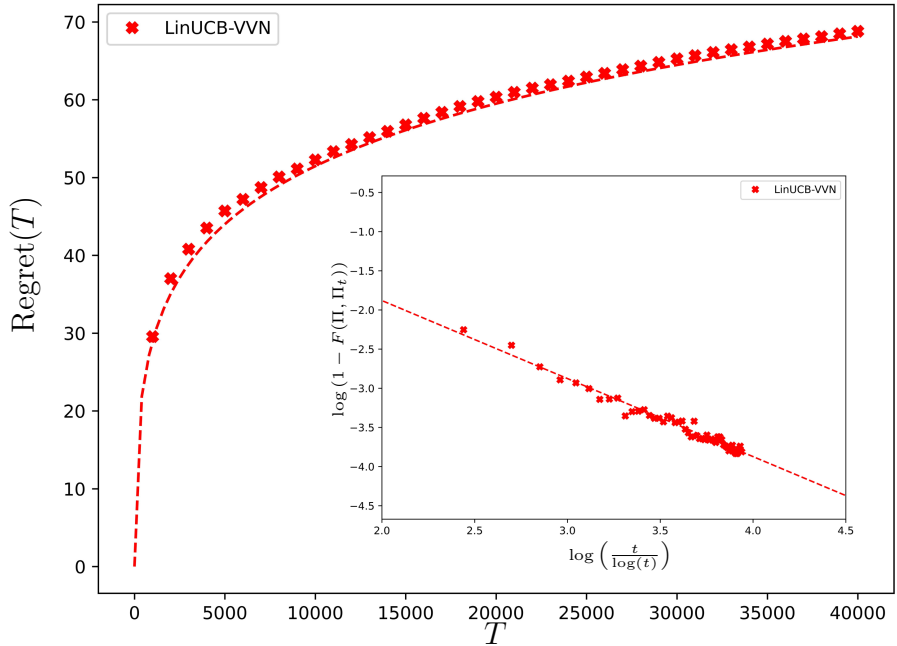

\begin{figure}
    \centering
    \begin{overpic}[percent,width=0.5\textwidth]{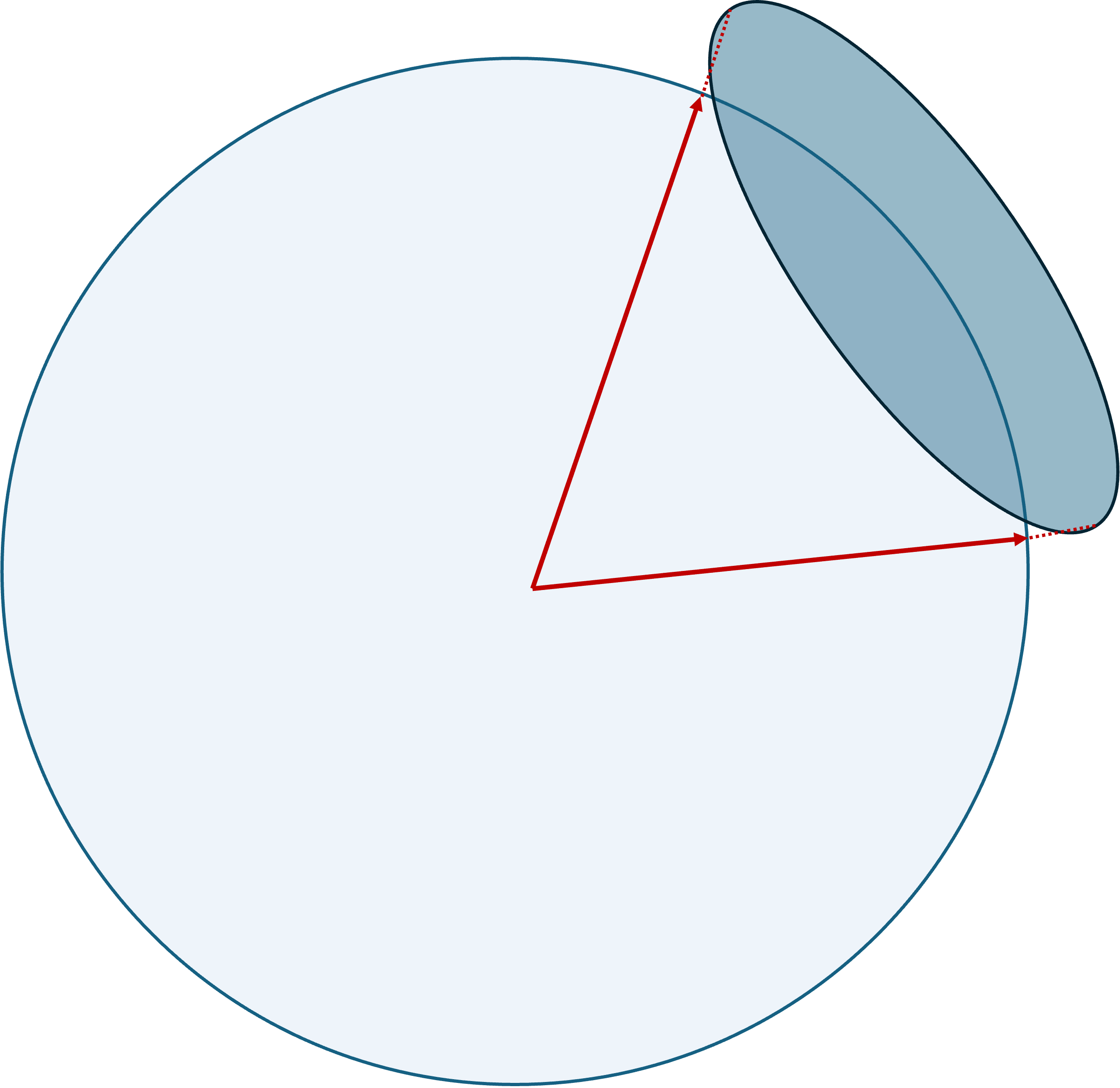}
    \put(50,70){\rotatebox{80}{$|\psi^+_{a_t}\rangle$}}
    \put(70,42){\rotatebox{0}{$|\psi^-_{a_t}\rangle$}}
    \put(75,85){\rotatebox{0}{$\mathcal{C}_t$}}
    \put(82,70){{$|\widehat{\psi}_t\rangle$}}
    \put(90,60){{$| \psi \rangle$}}
    
    \end{overpic}
    \caption{Sketch of Algorithm~\ref{alg:linucb_vn_var} that minimizes the disturbance. At each time step it computes an estimator $|\widehat{\psi}_t\rangle$ and builds a high-probability confidence region $\mathcal{C}_t$ (shaded region) around the unknown state $ | \psi \rangle $ on the Bloch sphere representation. Then uses the optimistic principle to output measurement directions $|\psi^\pm_{a_t} \rangle$ that are close the unknown state $\ket{\psi} $ projecting into the Bloch sphere the extreme points of the largest principal axis of $\mathcal{C}_t$. This particular choice allows optimal learning of $\ket{\psi}$ (exploration) and simultaneously minimizes the regret or disturbance (exploitation).}
    \label{fig:psmaqb_exploration_exploitation_chapterv}
\end{figure}

\begin{theorem}\label{th:regret_PSMAQB}
    Let $\widetilde{T}\in\mathbb{N}$ and fix the time horizon $T = \lceil 96\widetilde{T}\log (\widetilde{T}^2) \rceil $. Then given a \textup{PSMAQB} with qubit action set $\mathcal{A} = \mathcal{S}^*_2$ and qubit environment $|\psi\rangle \! \langle \psi |\in\mathcal{S}^*_2$, we can apply Algorithm~\ref{alg:linucb_vn_var} for $d=3$ and it achieves
    \begin{align}
        \EX_{\psi,\pi}\left[R_T(\mathbb{S}_2^*,\pi,\psi) \right] \leq  C_1\log\left(T \right)+ 
    C_2\log^2\left(T\right).
    \end{align}
    for some universal constants $C_1, C_2 \geq 0$. Also at each time step $t\in [T ]$, it outputs an estimator  $\hat{\Pi}_{t}\in\mathcal{S}^*_2$ of $\psi$ with infidelity scaling 
    \begin{align}
         \EX_{\psi,\pi} \left[ 1 -  \langle \psi | \hat{\Pi}_{t} | \psi \rangle  \right] \leq   \frac{C_3\log(T)}{t} ,
    \end{align}
    for some universal constant $C_3 \geq 0$.
\end{theorem}

\begin{proof}
    In order to apply Algorithm~\ref{alg:linucb_vn_var} to a PSMAQB, we set $d=3$ (dimension for a classical linear stochastic bandit) and the actions that we select will be given by $\Pi_{A^\pm_{t,i}}\in\mathbb{S}^*_2$ where $A^\pm_{t,i}$ are the Bloch vectors of $\Pi_{A^\pm_{t,i}}\in\mathbb{S}^*_2$ and are updated as in~\eqref{eq:action_general_update_mom}. Note that they are valid actions since $ A^{\pm}_{t,i} \in\mathbb{S}^2$ implies $\Pi_{A^\pm_{t,i}} \in\mathcal{S}^*_2$. In order to update the action in the right form, we need to renormalize the rewards $X_t\in\lbrace 0 ,1 \rbrace $ received by the algorithm as
    \begin{align}
        \tilde{X}_t = 2X_t - 1 ,
    \end{align}
    so we have the guarantee $\EX[\tilde{X}_t ] = \langle \theta ,A_t \rangle$, where $\theta\in\mathbb{S}^2$ is the Bloch vector of the unknown $|\psi\rangle$ and $A_t$ is the one of the measurements $\Pi_{A_t}$. Then we have that the variance of the statistical noise of $\tilde{X}_t $ has the following form
    \begin{align}
        \mathrm{Var}[\tilde{\epsilon_t}] = 1 - \langle \theta ,A_t \rangle^2 .
    \end{align}
    Thus, this model fits into the linear bandit with linearly vanishing variance noise model explained in Section~\ref{sec:linearly_vanishin_variance}, and we can apply the guarantees established in Theorem~\ref{th:regret_scaling_variance} and Corollary~\ref{cor:expected_regret_mom}.
    
    The algorithm is set with $k= \lceil 24\log (\widetilde{T}^2) \rceil $ batches for the median of means construction. We set $\lambda_0 = 2$, and using $\| \theta \|_2 =1 $
    we have that the constant $\beta$ given in~\eqref{eq:beta_constant_mom} has the value
    \begin{align}
        \beta = 9\left(3\sqrt{3}+2 \right)^2 = 279+108\sqrt{3}.
    \end{align}
    Then we can check that for $d=3$, the condition~\eqref{eq:lambda_condition} for the input parameter $\lambda_0$ for Theorem~\ref{th:regret_scaling_variance} to hold is satisfied since
    \begin{align}
        \lambda_0 = 2 \geq \max\left\lbrace 2 , \frac{1}{3}+\frac{1}{2\sqrt{6}(279+108\sqrt{3})} \right\rbrace = 2 .
    \end{align}
    In the above, we just substituted all numerical values. Then we are under the assumptions of Theorem~\ref{th:regret_scaling_variance} and Corollary~\ref{cor:expected_regret_mom}, and the result follows by applying both results with the relation of regrets between the classical and quantum model. Then we also have
    \begin{align}
        \| \theta - \hat{\theta}_t \|_2^2 = 4\left(1-F\left(\Pi_\theta , \Pi_{ \hat{\theta}_t} \right) \right),
    \end{align}
    where take the estimator $\hat{\theta}_t$ given in Theorem~\ref{th:regret_scaling_variance} for $d=3$ and we use it as the Bloch vector of $\Pi_{ \hat{\theta}_t}$. We also use the bound $\widetilde{T}\leq T$ and reabsorb all the constants into $C_1,C_2,C_3$.
\end{proof}

From the above theorem, we have the following remarks.

\textbf{Remark 1.} The constant dependence can be slightly improved taking the estimator for $\Pi_\theta$ as $\Pi_{\theta^{\text{wMoM}}_t}$ with $\hat{\theta}^{\text{wMoM}}_t$ defined in Section~\ref{sec:MoMLSE}.

\textbf{Remark 2.} The result of Theorem~\ref{th:regret_PSMAQB} also holds with high probability. In particular for the choice of batches $k = 24\log(\widetilde{T}^2)$ the guarantee holds with probability at least $1-\frac{1}{\widetilde{T}}$.

\textbf{Remark 3.} The same notion of disturbance can also be extended to mixed states recovering the general expression of the regret for general bandits \\ $R_T (\mathbb{S}_d^* , \pi , \rho ) = \EX_{\rho,\pi}\left[ \sum_{t=1}^T \lambda_{\max}(\rho) - \langle \psi_t | \rho |\psi_t\rangle \right]$. For more details, see~\cite[Lemma 3]{lumbreras24pure}.

\section{Quantum state-agnostic work extraction (almost) without dissipation}

As a specific task where learning with low disturbance is important, we consider work-extraction protocols from unknown quantum sources.

%In this section, we apply the multi-armed quantum bandit framework—specifically, the PSMAQB algorithm (Algorithm~\ref{alg:linucb_vn_var})—to the task of work extraction from unknown quantum systems.

Given sequential access to finite copies of an identical, unknown quantum system, what is the optimal strategy for extracting work from them and charging a battery? A natural strategy might involve first performing quantum state tomography to estimate the unknown state, followed by work extraction based on this estimate \cite{vsafranek2023work,watanabe2024black,watanabe2025universal}. Indeed, knowledge of the quantum state is essential for efficient work extraction; consider, for example, Szilard's engine, where the information about the system's configuration---such as the position of a particle---enables work to be extracted \cite{szilard1929entropieverminderung}. However, with only finitely many copies, any estimate of the true state has statistical uncertainty~\cite{ODonnell_2016,haah2016sample}, resulting in unavoidable heat dissipation during work extraction \cite{riechers2021initial}. Furthermore, quantum systems that are measured during state tomography are no longer available for work extraction and must therefore be accounted for as a form of dissipation.

This raises a natural question: can adaptive strategies, which simultaneously learn the unknown state and extract work, offer better performance?
The extraction of useful work from available resources has been a central problem in classical thermodynamics and continues to attract significant attention in the quantum domain. While numerous protocols have been proposed~ \cite{allahverdyan2004maximal,aaberg2013truly,brandao2013resource,skrzypczyk2014work,elouard2018efficient}, the majority assume that the agent has full knowledge of the quantum state---these are so-called \emph{state-aware} protocols. In practical scenarios, we may not always know how the state is prepared; in this case, how can an agent then design the protocol?

\begin{figure}
    \centering
    \includegraphics[width=0.75\linewidth]{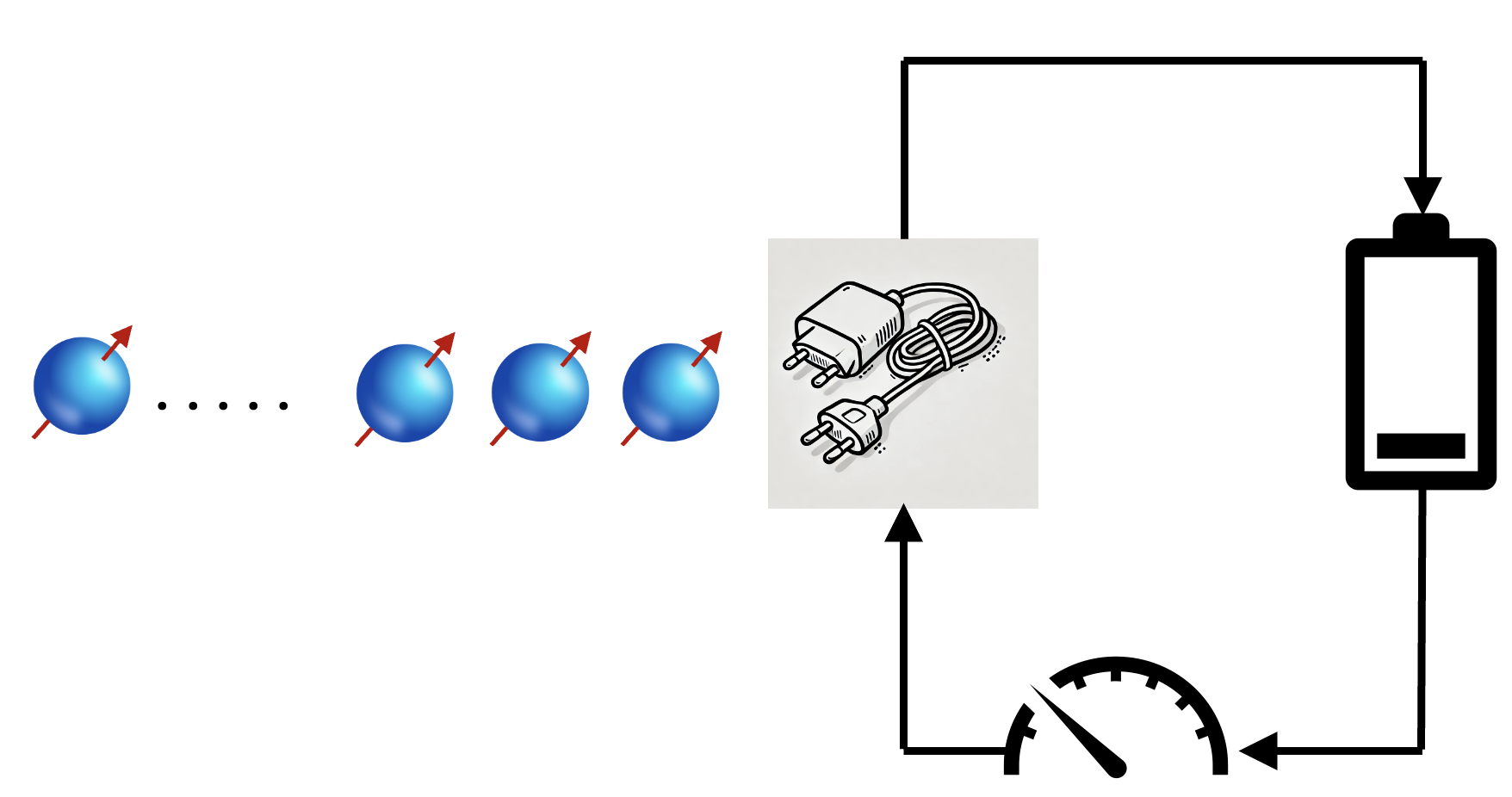}
    \caption{Sketch of the work extraction protocols considered in this section. On the left, a source provides single copies of an unknown pure qubit state. A charging protocol (represented as a cartoon charger) uses each copy to attempt charging a battery. Since the state is unknown, a measurement is performed during the process to extract partial information. This measurement outcome provides feedback that can be used to improve the charging strategy in future rounds.}
    \label{fig:charging_sketch}
\end{figure}

In this section, we investigate state-agnostic protocols for work extraction. Specifically, we are interested in how the cumulative dissipation scales with the number of available copies of an unknown quantum state. We show that adaptive strategies can simultaneously learn the unknown state and maximize the amount of work transferred to a battery system. Notably, such strategies can lead to an exponential improvement in cumulative dissipation: while tomography-based approaches yield a dissipation scaling of $\Omega(\sqrt{T})$ in the number of available copies $T$, we demonstrate that a scaling of $O(\operatorname{polylog}(T))$ is achievable for pure-state sources using suitable adaptive work extraction algorithms.

We consider a source of unknown pure qubit states used to charge batteries in two distinct models. First, we analyze a battery with discrete energy levels, modeled as a bosonic system where work is accumulated via a Jaynes–Cummings interaction~\cite{Jaynes_1963,Seah_2021,skrzypczyk2014work}. Then, we turn to a semiclassical thermodynamic model that includes a thermal reservoir, with the battery described as a classical weight storing energy by being lifted or dropped.

The main technical contribution is an upper bound on the dissipation in both models. We establish this by connecting each setting to the PSMAQB framework and relating the cumulative dissipation to regret. Using the PSMAQB algorithm (Algorithm~\ref{alg:linucb_vn_var}), which achieves polylogarithmic regret, we show that the cumulative dissipation in both battery models similarly scales polylogarithmically with the number of available copies $T$.

\subsection{Discrete battery}

In this setting, we model the \emph{battery} as a quantum harmonic oscillator with Hamiltonian
\begin{align}
    H_B = \omega a^\dagger a,
\end{align}
where $\omega > 0$ is the oscillator frequency, and $a^\dagger$, $a$ are the creation and annihilation operators, respectively, defined as
\begin{align}
    a = \sum_{n=1}^\infty \sqrt{n} \ket{n-1}\bra{n}, \quad a^\dagger = \sum_{n=0}^\infty \sqrt{n+1} \ket{n+1}\bra{n},
\end{align}
where $|n\rangle$ are the eigenvectors of $H_B$ with eigenvalue $\omega n$ for $n=0,1,2,....$.

We assume sequential access to $T$ copies of an unknown pure qubit state $\ket{\psi}$. The goal is to design a protocol that uses these unknown states to charge the battery by increasing its energy level over time.

At each time step $t$, the battery is assumed to be in an energy eigenstate $\ket{n_t}$ and the learner holds the state $|\psi\rangle\ket{n_t}$. The protocol proceeds as follows:
\begin{enumerate}
    \item \textbf{Direction Selection:} Choose a direction $\psi_t$ on the Bloch sphere. Apply a local Hamiltonian
    \begin{align}
        H_A = \omega \ket{\psi_t}\!\bra{\psi_t}
    \end{align}
    to the incoming qubit $\ket{\psi}$, effectively projecting the qubit into the chosen basis.

    \item \textbf{Interaction with the battery:} Use a Jaynes--Cummings-type interaction for a fixed time $\tilde{t}$:
    \begin{align}
        H_I = \frac{\Omega}{2} \left( a \otimes \ket{\psi_t}\!\bra{\psi_t^c} + a^\dagger \otimes \ket{\psi_t^c}\!\bra{\psi_t} \right),
    \end{align}
    where $\Omega > 0$ controls the interaction strength.
\end{enumerate}

After the interaction, a measurement is performed on the battery to determine whether it has successfully absorbed energy. If the interaction time is set to
\begin{align}
    \tilde{t} = \frac{\pi}{\Omega \sqrt{n_t + 1}},
\end{align}
then the battery may transition to one of the following energy levels:
\[
    \ket{n_t + 1}, \quad \ket{n_t}, \quad \text{or} \quad \ket{n_t - 1}.
\]
A transition to $\ket{n_t + 1}$ indicates a successful charging event, i.e., the battery has gained one quantum of energy. The task is to adaptively choose the direction $\psi_t$ at each round to maximize the number of successful charges over $T$ rounds.

We summarize this charging procedure in Algorithm~\ref{alg:jc_work_extraction}.

\begin{algorithm}
\begin{algorithmic}
	\caption{\textsf{Jaynes-Cummings work extraction}} 
	\label{alg:jc_work_extraction}

\For{$ t = 1,2,\cdots$}
       \State Get new copy of $|\psi\rangle$.
       \State Choose direction $\psi_t$ .
       \State Expose $|\psi\rangle$ to a field induced by Hamiltonian $H_A=\omega \ket{\psi_t}\!\bra{\psi_t}$. 
        
       \State  For a time $\tilde{t}=\pi\Omega^{-1}(n_t+1)^{-\frac{1}{2}}$, turn on the interaction between the system and the battery with Hamiltonian is $H_I = \frac{\Omega}{2}(a \otimes \ket{\psi_t}\!\bra{\psi_t^c} + a^\dagger \otimes \ket{\psi_t^c}\!\bra{\psi_t}) $ 

     \State    Measure battery system and get output $n_{t+1}$. 

        \EndFor
\end{algorithmic}        
\end{algorithm}

\subsubsection{Dissipation}

We now compute the expected energy transferred to the battery system when interacting with the unknown state using the chosen direction $\psi_t$. Based on this, we define the \emph{dissipation} at each time step $t$ as the difference between the maximum achievable energy input and the actual energy input obtained using $\psi_t$.

At each time step $t\in[T]$, the total Hamiltonian of the system is
\begin{align}
    H= \omega (\ket{\psi_t}\!\bra{\psi_t} + a^\dagger a) +\frac{\Omega}{2}(a \otimes \ket{\psi_t}\!\bra{\psi_t^c} + a^\dagger \otimes \ket{\psi_t^c}\!\bra{\psi_t}),
\end{align}
One can verify that the eigenstates of the total Hamiltonian are  
\begin{align}
    &\ket{0}\ket{\psi_t^c} \\
    &\ket{n,+}  = \frac{1}{\sqrt{2}}(\ket{n-1}\ket{\psi_t} + \ket{n}\ket{\psi_t^c}) \\
    &\ket{n,-}  = \frac{1}{\sqrt{2}}( \ket{n-1}\ket{\psi_t} - \ket{n}\ket{\psi_t^c}),
\end{align}
for $n=1,2,\ldots$ and the eigenvalues are respectively
\begin{align}
    E_0=0, \quad
    E_{n+}  =  n\omega + \frac{\Omega}{2} \sqrt{n}, \quad 
    E_{n-}  =  n\omega - \frac{\Omega}{2} \sqrt{n}.
\end{align}

The state at time step $t\in[T]$ is $\ket{n_t}\ket{\psi}$ and can be decomposed into $4$ eigenstates $\ket{n_t,\pm}$ and $\ket{n_t+1,\pm}$, i.e.,
\begin{align}
    \ket{n_t}\ket{\psi} = \frac{1}{\sqrt{2}}\langle \psi_t^c |\psi \rangle(\ket{n_t,+}-\ket{n_t,-}) + \frac{1}{\sqrt{2}} \langle \psi_t |\psi \rangle (\ket{n_t+1,+}+\ket{n_t+1,-}). 
\end{align}
These eigenstates gain phase factors during the time evolution. After time $t_k=\pi\Omega^{-1}(n_t+1)^{-\frac{1}{2}}$, the state evolves to, up to an irrelevant global phase, 
\begin{align}
    e^{-i H t_k}\ket{n_t}\ket{\psi} & = \frac{1}{\sqrt{2}}\langle \psi_t^c |\psi \rangle(e^{i\theta_k}\ket{n_t,+}-e^{-i\theta_k}\ket{n_t,-}) \\
    &+ \frac{1}{\sqrt{2}} e^{i\frac{\omega\pi}{\Omega\sqrt{n_t+1}} } \langle \psi_t |\psi \rangle(i\ket{n_t+1,+}-i\ket{n_t+1,-}) ,
\end{align}
which rearranged gives
\begin{align}\label{eq:evolution_jc}
  e^{-i H t_k}\ket{n_t}\ket{\psi} & =   \langle \psi_t^c |\psi \rangle( i \sin\theta_k \ket{n_t-1}\ket{\psi_t}+ \cos\theta_k \ket{n_t}\ket{\psi_t^c}) \\
    &+  \frac{1}{\sqrt{2}} i e^{i\frac{\omega\pi}{\Omega\sqrt{n_t+1}} } \langle \psi_t |\psi \rangle \ket{n_t+1}\ket{\psi_t^c}, 
\end{align}
where $\theta_t = \frac{\pi}{2}\sqrt{\frac{n_t}{n_t+1}}$. After the evolution we finally measure the battery energy. The measurement outcome is $n_{t+1}$. The probability distribution of $n_{t+1}$ can be computed from~\ref{eq:evolution_jc} is given by 
\begin{align}
    \Pr(n_{k+1}) = \begin{cases}
        |\langle \psi_t |\psi \rangle |^2, \quad & n_{t+1} = n_t+1,\\
        |\langle \psi_t^c |\psi \rangle |^2 \cos^2\theta_t, \quad & n_{k+1} = n_t,\\
        |\langle \psi_t^c |\psi \rangle |^2 \sin^2\theta_t, \quad& n_{t+1} = n_t-1. 
    \end{cases}
\end{align}
The extracted work is defined as $\Delta W_k =\omega( n_{t+1}-n_t) $. The expected extracted work is then given by 
\begin{align}
    \EX[\Delta W_t] = \omega( |\langle \psi_t|\psi \rangle |^2 - |\langle \psi_t^c |\psi \rangle |^2 \sin^2\theta_t )= \omega( |\langle \psi_t |\psi \rangle |^2 (1+\sin^2\theta_t) - \sin^2\theta_t), 
\end{align}
where we have used $|\langle \psi_t^c |\psi \rangle |^2 = 1 - |\langle \psi_t |\psi \rangle |^2 $.  Therefore, we can compute the maximal extracted work as
\begin{align}
    \max_{\ket{\psi_t}} \EX[\Delta W_t] = \omega ,
\end{align}
where the maximum is achieved if $\psi_t = \psi$. The dissipation in this round is thus given by the difference between the maximal and the actual expected extracted work
\begin{align}\label{eq:dissipation_jc}
    \dissipation^{\text{jc},t} := \max_{\ket{\psi_t}} \EX[\Delta W_t] -  \EX[\Delta W_t] = \omega (1+\sin^2\theta_t^2)(1- |\langle \psi_t|\psi\rangle |^2 ) \leq 2 \omega (1-|\langle \psi_t|\psi\rangle |^2). 
\end{align}

\subsubsection{Cumulative dissipation and regret}

Recall that the goal is to maximize the energy transferred to the battery system, or equivalently, to minimize the cumulative dissipation. Using the expression for the dissipation at each time step $t$~\eqref{eq:dissipation_jc}, we find that the cumulative dissipation over $T$ rounds is bounded by
\begin{align}
    \dissipation^{\text{jc}}(T) := \sum_{t=1}^{T} \dissipation^{\text{jc},t} \leq 2\omega \sum_{t=1}^{T} \big(1 - |\langle \psi_t | \psi \rangle|^2 \big).
\end{align}
It turns out that the right-hand side matches the regret expression in the PSMAQB setting. Moreover, given the measured battery energy level $n_t$ at each time step $t$, we can define a reward as follows:
\begin{align}\label{eq:reward_jc}
    X_t = \begin{cases}
        1, & \text{if } n_{t+1} = n_t + 1, \\
        0, & \text{otherwise}.
    \end{cases}
\end{align}
The probability distribution of this reward is
\begin{align}
    \Pr[X_t = x_t] = \begin{cases}
        |\langle \psi_t | \psi \rangle |^2, & \text{if } x_t = 1, \\
        1 - |\langle \psi_t | \psi \rangle |^2, & \text{if } x_t = 0,
    \end{cases}
\end{align}
which coincides with the reward distribution in the PSMAQB setting. Thus, we can adaptively select the directions $\psi_t$ using Algorithm~\ref{alg:linucb_vn_var} (on the Bloch sphere representation) and we arrive to the following result.

\begin{theorem}
Let $T \in \mathbb{N}$ be a finite time horizon. Then, there exists an update rule for the directions $\lbrace \psi_t \rbrace_{t=1}^T$ in the work extraction protocol described in Algorithm~\ref{alg:jc_work_extraction}, based on the past rewards $\{ x_s \}_{s=1}^{t-1}$ as defined in~\eqref{eq:reward_jc}, that guarantees, with probability at least $1 - \delta$,
\begin{align}
    \dissipation^{\text{jc}}(T) = O\left( \omega \log(T) \log\left( \frac{T}{\delta} \right) \right).
\end{align}
\end{theorem}

\textbf{Note. } From the explore-then-commit algorithm studied in Section~\ref{sec:explorethencommit}, we have that an approach that allocates a fraction of the copies $\psi$ to the learning and the remaining for extracting work at most can achieve $\dissipation^{\text{jc}}(T) =  O (\omega \sqrt{T} ) $.

\subsection{Continuous battery}

In this section, we consider the work extraction model studied in~\cite{skrzypczyk2014work,huang2023engines}, where the goal is to extract energy from a qubit system in contact with a thermal bath and transfer it into a semi-classical battery, represented by a weight. The battery stores or releases energy through the raising or lowering of this weight, corresponding to changes in its potential energy. In our setting, we assume that the qubit system is initially unknown. The three physical subsystems involved in the process are:

\begin{itemize}
    \item \textbf{System $A$ (Unknown Pure State Source):} A qubit in a pure state $\psi_A = \ket{\psi}\!\bra{\psi}$ with a degenerate Hamiltonian $H_A = \omega \mathbb{I}/2$ (same eigenvalues). This is the system from which free energy is to be extracted.
    
    \item \textbf{System $B$ (Battery):} A semi-classical weight described by a continuous-variable state $\varphi(x) \in L^2(\mathbb{R})$. The battery Hamiltonian is defined as  $H_B \varphi(x-x_0) = x\varphi(x-x_0)$, where $x$ represents the height of the weight.  We will assume $\varphi(x-x_0)$ to be a battery state whose energy is sharply centered at $x_0$.  The energy of the battery can be changed by translating the weight up by a certain height $x_0$, described by the translation operator $\Gamma^B_{x_0} \varphi(x) = \varphi(x-x_0)$. 

    \item \textbf{System $R$ (Thermal Reservoir):} A heat bath at a fixed inverse temperature $\beta$, modeled as a supply of qubit states $\gamma_\beta(\nu) = Z_R(\nu)^{-1} e^{-\beta H_R(\nu)}$, where the Hamiltonian is $H_R(\nu) = \nu \ket{1}\!\bra{1}$ and $Z_R(\nu) = \Tr(e^{-\beta H_R(\nu)})$ is the partition function. Here, $\ket{0}$ and $\ket{1}$ are energy eigenstates, and $\nu$ is the energy gap. This energy gap can be tuned.
\end{itemize}

We consider an agent with oracle access to the unknown system $A$ over a finite number of rounds $T\in\mathbb{N}$. The goal is to extract the full free energy of $A$ and store it in the battery $B$ through interactions with a thermal reservoir $R$. However, since the state $\psi$ of the system is unknown, the agent cannot extract work optimally from the outset. Instead, it must gradually improve its strategy by learning from each round, where the learning is done by measuring the battery. 

The work extraction protocol is defined by two key components: a policy that updates the agent’s guess $\psi_t$ of the unknown state $\psi$ at each round, and a sequence $\{ \epsilon_t \}_{t=1}^T$ that determines the accuracy of these guesses. Both quantities can be chosen adaptively and together, as we will compute later, they determine the dissipation at each round. We first present the general structure of the protocol and later discuss how the dissipation depends on these choices and how to optimize them for best performance. The protocol proceeds as follows at each round $t \in [T]$:

\begin{enumerate}
    \item The agent receives a copy of the unknown qubit state $\psi$.
    
    \item Based on the outcomes from previous rounds, the agent selects a direction $\psi_t$ on the Bloch sphere, sets an accuracy $\epsilon_t\in [0,1]$ and defines a basis $\lbrace \psi_t , \psi_t^c \rbrace$ for system $A$. This computation is done using the previously selected directions and measured battery energies $\lbrace \psi_s , \mu_s \rbrace_{s=1}^{t-1}$.
    \item We first implement a unitary on the system qubit in the form of 
    \begin{equation}
        U_t = \ket{0}\!\bra{\psi_t} + \ket{1}\!\bra{\psi_t^c},
    \end{equation}
    satisfying $[H_A,U_t]=0$, which tries to diagonalize the system qubit in the computational basis.
    \item The agent then performs a thermal operation by repeatedly appending a reservoir qubit $R$, applying an energy-conserving unitary on the combined system $A B R$ to transfer energy from $R$ to the battery $B$, and discarding $R$. Specifically, the agent will do the following for $M\in\mathbb{N}$ iterations (indexed by $\tau \in [M]$):
    \begin{itemize}
        \item Sets the energy gap of the reservoir qubit to $\nu (\tau,\epsilon_t )$ (depends on repetition $\tau$ and accuracy $\epsilon_t$, see~\eqref{eq:intro_gap_parametrization}) and gets a fresh qubit state $\gamma_\beta(\nu(\tau,\epsilon_t))$.
        \item Applies the following unitary 
        \begin{align}\label{eq:unitary_intro}
            V_{\psi_t , \tau} = \sum_{i,j} \ket{i}\!\bra{j}_A \otimes \ket{j}\!\bra{i}_R \otimes \Gamma^B_{(i-j)\nu(\tau,\epsilon_t)} ,
        \end{align}
       satisfying $[H_A+H_B+H_R,V_{\psi_t,\tau}]=0$ to $ABR$.
       %where the basis of $A$ is $\lbrace \psi_k , \psi_k^\perp \rbrace$.
       \item Discards the reservoir qubit.
    \end{itemize}
    
    \item After completing the $M$ steps, the agent measures the energy of the battery in its eigenbasis and records the energy $\mu_t$.
    
\end{enumerate}

Going back to the protocol, we note that the unitary $V_{\psi_t,\tau}$ in~\eqref{eq:unitary_intro} swaps the states of the unknown system $A$ and the thermal reservoir $R$. This operation induces an energy exchange determined by the energy gap $\nu (\tau,\epsilon_t )$ between the two systems. The energy difference is transferred to the battery through a translation operation $\Gamma^B_{(i-j)\nu (\tau,\epsilon_t )}$ to conserve the total energy of the combined system.

To maximize the extractable non-equilibrium free energy from system $A$, it is essential to minimize dissipation during the energy transfer to the battery. This can be achieved by ensuring that the process is \emph{quasi-static} --- that is, slow and nearly reversible. Under such conditions, the system qubit remains close to the thermal state of the reservoir's Hamiltonian, thereby suppressing heat flow during the interaction.  Conversely, if the process is carried out rapidly, i.e., the system state deviates from the reservoir's thermal state, the system will appear out of thermal equilibrium with the reservoir, resulting in heat exchange, which contributes to the entropy production during the protocol. The quasi-static limit can be approximated when $M\rightarrow \infty$, and in the next section, we are going to check that under this limit, at each round $t\in[T]$, the extracted work into the battery is
\begin{align}
\label{eq:avgwork}
 \EX[\Delta W_t] = \beta^{-1} \left[ D (\psi \| \mathbb{I}/2) - D (\psi \| \Delta_{2\epsilon_t}(\psi_t)) \right],
\end{align}
where $\Delta_{\epsilon}(\rho) = (1 - \epsilon)\rho + \epsilon \id/2$ denotes the depolarizing channel. The maximal expected work per round, $\beta^{-1} D(\psi \| \mathbb{I}/2)$, is achieved when the agent perfectly estimates the state, i.e., when $\psi_t = \psi$ and $\epsilon_t = 0$. Accordingly, we define the dissipation at round $t$ as
\begin{align}
\dissipation^{t} := \max_{\psi_t, \epsilon_t} \mathbb{E}[\Delta W_t] - \mathbb{E}[\Delta W_t] = \beta^{-1} D(\psi \| \Delta_{2\epsilon_t}(\psi_t)).
\end{align}
This expression highlights a key trade-off: to reduce dissipation, the agent must align $\psi_t$ with the true state $\psi$, but it cannot set $\epsilon_t$ too small unless the estimate is sufficiently accurate, as otherwise the divergence becomes large. The parameter $\epsilon_t$ thus plays a dual role, quantifying both the uncertainty in the estimate and its thermodynamic penalty. The agent's objective is to extract the maximum amount of work into the battery using the $N$ copies of the unknown system. Equivalently, the agent aims to minimize the cumulative dissipation over $N$ rounds, which is given by
\begin{align}\label{eq:dissipation_sc}
    \dissipation(T) = \beta^{-1} \sum_{t=1}^T \dissipation^{t}=\sum_{t=1}^T  D(\psi\| \Delta_{2\epsilon_t} (\psi_t)) .
\end{align}

\subsubsection{Dissipation in the quasi-static limit $M\rightarrow \infty$}

In order to simplify the notation, we are going to compute the dissipation in a single round of the protocol. Thus we fix a direction $\hat{\psi}$ and an accuracy $\epsilon\in[0,1]$,
and define the following quantities
\begin{align}\label{eq:simplification_thermal}
    p_0 := 1-\epsilon, \quad |\phi_0 \rangle = |\hat{\psi}\rangle, \\
    p_1 := \epsilon, \quad |\phi_1 \rangle = |\hat{\psi}^c\rangle .
\end{align}
We are going to use the following parametrization of the energy gap at repetition $\tau\in[M]$
\begin{align}\label{eq:intro_gap_parametrization}
    \nu (\tau,\epsilon ) = \beta^{-1} 
    \ln \left( \frac{1 - \frac{\tau}{2M}- \left(1 - \frac{\tau}{M} \right) \epsilon }{\frac{\tau}{2M}+\left(1 - \frac{\tau}{M} \right) \epsilon} \right).
\end{align}
For $\tau\in[M]$ and $i\in\lbrace 0, 1 \rbrace$, we also define the following quantities
\begin{align}
    p_{i,\tau} := p_{i} - (-1)^i \tau \delta p , \quad \delta p := \frac{1}{M}(p_{0} - \frac{1}{2}).
\end{align}
We state the work extraction for a single repetition in Algorithm~\ref{alg:sc_work_extraction}.

\begin{algorithm}
	\caption{\textsf{Thermal work extraction}} 
	\label{alg:sc_work_extraction}
    \begin{algorithmic}[1]
       \State Require: An unknown system state $\psi$, a direction $\hat{\psi}$, accuracy $\epsilon\in[0,1]$ a battery state $\varphi(x)$ with the battery energy $\mu$, a reservoir at inverse temperature $\beta$.

        \vspace{1mm}
        \textit{Unitary Rotation}
        \vspace{1mm}
        
       \State  Apply unitary $U= \ket{0}\!\bra{\hat{\psi}} + \ket{1}\!\bra{\hat{\psi}^c}$ to the unknown system $\psi$
        \vspace{1mm}
        
        \For{$\tau=1,2,\ldots,M$}

            \vspace{1mm}
            \textit{Prepare a fresh reservoir qubit and exchange it with the system}
            \vspace{1mm}
            
           \State Take a fresh thermal qubit $\gamma_\beta(\nu(\tau,\epsilon))=\frac{1}{Z_R(\nu(\tau,\epsilon))}e^{-\beta H_R(\nu(\tau,\epsilon))}$ with $\nu(\tau,\epsilon)$ defined as in~\eqref{eq:intro_gap_parametrization}
            
           \State Apply the swap unitary $V_{\rho^*,\tau}= \sum_{ij}\ket{i}\!\bra{j}_A\otimes\ket{j}\bra{i}_R \otimes \Gamma_{(i-j)\nu(\tau,\epsilon)}$ on the system $A,B,R$.
            
           \State Discard the reservoir qubit. 
        \EndFor
        
        \vspace{1mm}
        \textit{Measure the extracted work}
        \vspace{1mm}
            
        Measure the battery energy, obtain the battery energy $\mu'$, and compute the extracted work $\Delta W=\mu'-\mu$. 
        \end{algorithmic}
\end{algorithm}

In the first stage of the protocol, we rotate the unknown qubit via the unitary 
\begin{equation}
    U = \sum_i\ket{i}\!\bra{\phi_i}~.
\end{equation}
This operation attempts to diagonalize the system qubit in the computational basis.
We then interact the system with the battery. The state of the system together with the battery is 
\begin{equation}
\label{eq:intial_joint_state}
    \rho_{AB} = \sum_{ij}\bra{\phi_i}\psi\rangle\!\langle \psi\ket{\phi_j}\ket{i}\!\bra{j}_A\otimes \varphi(x)_B.
\end{equation}
Then the unitary~\eqref{eq:unitary_intro} swaps the system and the fresh qubit from the reservoir, extracts work into the battery due to the different energy gap between $\{\ket{i}_A\}_{i\in\lbrace 0,1 \rbrace}$ and $\{\ket{i}_R\}_{i\in\lbrace 0,1 \rbrace}$ and conserves energy of the system, the qubit from the reservoir and the battery. Finally, the qubit from the reservoir is discarded. At the end of each repetition $\tau$, the reduced state is 
\begin{align}
    \rho_{AB,\tau}  = \Tr_R\left(V_{\rho^*,\tau}\left(\rho_{AB,\tau-1}\otimes \gamma_\beta(\nu(\tau,\epsilon))\right) V_{\rho^*,\tau}^\dagger\right),
\end{align}
After the first repetition, we obtain 
\begin{align}
    \label{eq:joint_state_1}
    \rho_{AB,1} = \sum_{i} p_{i,1} \ket{i}\!\bra{i}_A \otimes \rho_{B,i,1},
\end{align}
where 
\begin{align}
    \label{eq:battery_state_1}
    \rho_{B,i,1} = \sum_j |\langle \phi_j| \psi \rangle|^2 \varphi(x-(i-j)\nu(\tau,\epsilon)), 
\end{align}
and after repetition $\tau$ where $\tau\geq 2$, we obtain
\begin{align}
    \label{eq:joint_state_tau}
    \rho_{AB,\tau} = \sum_{i} p_{i,\tau} \ket{i}\!\bra{i}_A \otimes \rho_{B,i,\tau}, 
\end{align}
where
\begin{align}
    \label{eq:battery_state_tau}
    \rho_{B,i,\tau} = \sum_j  p_{j,\tau-1} \Gamma_{(i-j)\nu(\tau,\epsilon)} \rho_{B,j,\tau-1} \Gamma_{(i-j)\nu(\tau,\epsilon)}^\dagger. 
\end{align}
%$\rho_B^{(i)} =\Gamma_{\eps_i}\rho_B\Gamma_{\eps_i}^\dagger$. 
From~\eqref{eq:joint_state_tau}, we observe that the reduced state of the system changes gradually, which resembles a quasi-static process in thermodynamics. This is the reason why we take the swap unitary in repetition.

We will need the following Lagrange mean value theorem and the first mean value theorem for definite integrals~\cite{strang2019calculus} as follows.

\begin{theorem}
    \label{thm:lag_mean_value_theorem}
    Let $f:[a, b] \to \mathbb{R}$ be a continuous function on the closed interval $[a,b]$ and differentiable on the open interval $(a,b)$. Then there exists $c\in (a, b)$ such that
    \begin{align}
        f(b)-f(a) = f'(c) (b-a). 
    \end{align}
\end{theorem}
\begin{theorem}
    \label{thm:int_mean_value_theorem}
    Let $f:[a, b] \to \mathbb{R}$ be a continuous function on the closed interval $[a,b]$. Then there exists $c\in (a, b)$ such that
    \begin{align}
        \int_a^b f(x) \textnormal{d} x = f(c) (b-a). 
    \end{align}
\end{theorem}

We will show the following theorem in the following subsection.
\begin{theorem}
    \label{thm:work_distribution}
    Let $\hat{\psi}$ be the input direction and $\epsilon \in [0,1]$ the accuracy, $\Delta W$ be the extracted work (which is a continuous random variable) and $M$ be the number of repetitions as in Algorithm~\ref{alg:sc_work_extraction}. It holds that: if the extraction protocol is operated on a state $\psi$ such that $\psi = \hat{\psi}$ or $\psi = \hat{\psi}^c$, then the expected extracted work $\EX[\Delta W]$ converges to a fixed value $w_i$ and the extracted work $\Delta W$ converges in probability to its expectation $\EX[\Delta W]$. To be precise, it means
        \begin{align}
            \lim_{M\to \infty} \EX[\Delta W] = w_i,
        \end{align}
        where
        \begin{align}
        \label{eq:thm3work_values}
            w_{0} &:= \beta^{-1} (D(\psi\|\id/2) + \ln \epsilon), \\
            w_{1} &:= \beta^{-1} (D(\psi^c \|\id/2) + \ln (1-\epsilon))
        \end{align} 
        and for any $\tilde{\epsilon}>0$
        \begin{align}
            \lim_{M\to\infty} \Pr[|\Delta W - \EX[\Delta W]|\geq \tilde{\epsilon} ] = 0.
        \end{align}
\end{theorem}

\begin{proof}
We consider the notation we defined in~\eqref{eq:simplification_thermal}. We assume that $p_{0} > \frac{1}{2}$ as we deal with an estimate for a pure state in Algorithm~\ref{alg:sc_work_extraction}, although a similar proof holds for other cases. According to~\eqref{eq:joint_state_1}, the state after the first repetition can be viewed as a classical state described as follows: the state after repetition $1$ is $\phi_{x_1}$ where $x_1$ is a random bit sampled from $\{0,1\}$ according to the probability distribution $(p_{0,1},p_{1,1})$; the extracted work after the first repetition conditioned on $x_1$ is $(x_1-i)\nu(1,\epsilon)$. According to~\eqref{eq:joint_state_tau}, the evolution in repetition $\tau$ where $\tau\geq 2$ can be viewed as a classical process described as follows: the state after repetition $\tau$ is $\phi_{x_\tau}$, where $x_{\tau}$ is a random bit sampled from $\{0,1\}$ according to the probability distribution $(p_{0,\tau}, p_{1,\tau})$; the extracted work in repetition $\tau$ conditioned on $x_{\tau-1}x_\tau$ and is $(x_\tau - x_{\tau-1}) \nu(\tau,\epsilon)$. Suppose that the random bits sampled during the above process is $x_1\ldots x_M$ after $M$ repetitions. The extracted work after repetition $M$ conditioned on $x_1\ldots x_M$ is
\begin{align}
\label{eq:B9}
    \Delta W &= (x_1-i)\nu(1,\epsilon) + \sum_{\tau=2}^M (x_\tau - x_{\tau-1})\nu(\tau,\epsilon) \\
    &= - i  \nu(1,\epsilon) 
    + \sum_{\tau=1}^{M-1} x_\tau (\nu(\tau,\epsilon) - \nu(\tau+1,\epsilon)) + x_M \nu(M,\epsilon).
\end{align}
We use the expression of the gap~\eqref{eq:intro_gap_parametrization} which can be expressed as
\begin{align}
\label{eq:nu_tau}
    \nu(\tau,\epsilon) = \beta^{-1} \ln \frac{p_{0} - \tau \delta p }{p_{1} + \tau \delta p} . 
\end{align}
where $\delta p = \frac{1}{M}(p_{0}-\frac{1}{2})$.
The expected extracted work is 
\begin{align}
    \EX[\Delta W] =   - i  \nu(1,\epsilon) + \sum_{\tau=1}^{M-1} \EX[x_\tau] (\nu(\tau,\epsilon) - \nu(\tau+1,\epsilon)) + \EX[x_M] \nu(M,\epsilon).
\end{align}
Notice that, from the definition of $x_\tau$, $\EX[x_\tau] = p_{1,\tau}=p_{1} + \tau \delta p$, 
we obtain
\begin{align}
\label{eq:B12}
    \EX[\Delta W] & =   - i \beta^{-1}\ln \frac{p_{0} - \delta p }{p_{1} + \delta p} + \beta^{-1} \sum_{\tau=1}^{M-1}  \left(  \ln \frac{p_{0} - \tau \delta p }{p_{1} + \tau \delta p} - \ln \frac{p_{0} - (\tau+1) \delta p }{p_{1} + (\tau+1) \delta p}\right) (p_{1}+\tau\delta p ). 
\end{align}
We now use the definition of $\nu(\tau,\epsilon)$ in~\eqref{eq:nu_tau} as well as the Lagrange mean value theorem in Theorem~\ref{thm:lag_mean_value_theorem} to obtain 
\begin{align}
    \beta (\nu(\tau,\epsilon) - \nu(\tau+1,\epsilon)) = \ln \frac{p_{0} - \tau \delta p }{p_{1} + \tau \delta p} -\ln \frac{p_{0} - (\tau+1) \delta p }{p_{1} + (\tau+1) \delta p} = \frac{1}{\xi_\tau(1-\xi_\tau)}  \delta p, 
\end{align}
for some $\xi_\tau\in [p_{0}-(\tau+1) \delta p, p_{0}-\tau\delta p ]$. 

Therefore,~\eqref{eq:B12} can be simplified to 
\begin{equation}
\begin{split}
    \label{eq:expected_extracted_work}
    \EX[\Delta W]  & =  - i  \beta^{-1} \left(\ln \frac{p_{0} }{p_{1}} -\frac{\delta p}{\xi_0 (1-\xi_0)}  \right)+ \beta^{-1} \sum_{\tau=1}^{M-1}   \frac{p_{1}+\tau\delta p }{\xi_\tau(1-\xi_\tau)}  \delta p   \\
    & =  - i  \beta^{-1} \ln \frac{p_{0} }{p_{1}} + \beta^{-1} \sum_{\tau=1}^{M} \frac{p_{1}+\tau\delta p }{\xi_\tau(1-\xi_\tau)}  \delta p \\
    & \quad + i  \beta^{-1} \frac{\delta p}{\xi_0 (1-\xi_0)}  - \beta^{-1} \frac{\delta p}{2\xi_M (1-\xi_M)}~.
\end{split}
\end{equation}

We will approximate the sum in the second line of~\eqref{eq:expected_extracted_work} with an integration, where the remainder is bounded due to the first mean value theorem for definite integrals as in Theorem~\ref{thm:int_mean_value_theorem}. Namely, 
\begin{align}
    \beta^{-1} \sum_{\tau=1}^{M}  \frac{p_{1}+\tau\delta p }{\xi_\tau(1-\xi_\tau)}  \delta p &=\beta^{-1} \int_{\frac{1}{2}}^{p_{0} } \frac{\text{d} p}{p} + \beta^{-1}R_1(\delta p) \\
    &= \beta^{-1}\ln p_{0}+ \beta^{-1}\ln(2) + \beta^{-1}R_1(\delta p),
\end{align}
where $R_1(\delta p)$ is the remainder given by 
\begin{align}
    R_1(\delta p) & = \sum_{\tau=1}^{M}  \frac{p_{1}+\tau\delta p }{\xi_\tau(1-\xi_\tau)}  \delta p -  \int_{\frac{1}{2}}^{p_{0}} \frac{\text{d}p}{p} = \sum_{\tau = 1}^{M} \left( \frac{p_{1}+\tau\delta p }{\xi_\tau(1-\xi_\tau)}  \delta p - \int_{p_{0}-\tau \delta p}^{p_{0}-(\tau-1)\delta p}\frac{\text{d} p}{p} \right) \\ 
    & = \sum_{\tau=1}^{M} \left(\frac{p_{1}+\tau\delta p }{\xi_\tau(1-\xi_\tau)} -  \frac{1}{\xi_\tau'} \right)\delta p =  \sum_{\tau=1}^{M} \frac{\xi_\tau(1-\xi_\tau) - \xi_\tau'(p_{1}+\tau\delta p)}{\xi_\tau'\xi_\tau(1-\xi_\tau)} \delta p, 
\end{align}
where from the first line to the second line, we have used the first mean value theorem for definite integrals~Theorem~\ref{thm:int_mean_value_theorem} that
\begin{align}
    \int_{p_{0}-\tau \delta p}^{p_{0}-(\tau-1)\delta p}\frac{\text{d} p}{p} = \frac{1}{\xi_\tau'} \delta p,
\end{align}
for some $\xi_\tau'\in [p_{0}-\tau \delta p, p_{0}-(\tau-1)\delta]$. Therefore, the remainder satisfies
\begin{align}
    |R_1(\delta p )| & \leq \sum_{\tau=1}^{M} \left|\frac{\xi_\tau(1-\xi_\tau) - \xi_\tau'(p_{1}+\tau\delta p)}{\xi_\tau'\xi_\tau(1-\xi_\tau)} \right|\delta p \leq \sum_{\tau=1}^{M} \left|\frac{p_{1}+\tau\delta p + \xi_\tau}{\xi_\tau'\xi_\tau(1-\xi_\tau)} \right| (\delta p)^2 \\
    &\leq \sum_{\tau=1}^{M}\frac{4}{p_{1}} (\delta p)^2 \leq  (2p_{0}-1)\frac{2}{p_{1}}\delta p. 
\end{align}
Therefore, the second line in~\eqref{eq:expected_extracted_work} is finite while the third line is infinitesimal, and we obtain
\begin{align}
    \EX[\Delta W] & =  - i  \beta^{-1} \ln \frac{p_{0} }{p_{1}} + \beta^{-1}\ln p_{0}+ \beta^{-1}\ln(2) + O(\delta p) \\
    & = -\beta^{-1} \Tr(\phi_{i}\ln\id/2)  - i  \beta^{-1} \ln \frac{p_{0} }{p_{1}} + \beta^{-1}\ln p_{0} + O(\delta p) \\
    & = -\beta^{-1} \Tr(\phi_{i}\ln\id/2) + \beta^{-1}\ln p_{i} + O(\delta p) \\
    & = \beta^{-1}\left[D\left(\phi_{i}\middle\|\id/2\right) + \ln p_{i} \right] + O(\delta p). 
\end{align}
Since $\delta p \propto\frac{1}{M}$, we then obtain that
\begin{align}
    \EX[\Delta W] =  \beta^{-1}\left[D\left(\phi_{i}\middle\|\id/2\right) + \ln p_{i} \right] + O\left(\frac{1}{M}\right).
\end{align}
Taking $M\to \infty$, we obtain
\begin{align}
    \lim_{M\to \infty}\EX[\Delta W] =  \beta^{-1}\left[D\left(\phi_{i}\middle\|\id/2\right) + \ln p_{i} \right].
\end{align}
Now we demonstrate the convergence of $\Delta W$ towards its expectation value, recall from~\eqref{eq:B9}, we have that 
\begin{align}
    \Delta W =  - i  \nu(1,\epsilon) + \sum_{\tau=1}^{M-1} x_\tau (\nu(\tau,\epsilon) - \nu(\tau+1,\epsilon)) + x_M \nu(M,\epsilon).  
\end{align}
By the Lagrange mean value theorem,  
\begin{align}
     \nu(\tau,\epsilon) - \nu(\tau+1,\epsilon) = \ln \frac{p_{0} - \tau \delta p }{p_{1} + \tau \delta p} -\ln \frac{p_{0} - (\tau+1) \delta p }{p_{1} + (\tau+1) \delta p} = \frac{\delta p}{\xi_\tau(1-\xi_\tau)} , 
\end{align}
for some $\xi_\tau \in [p_{0}-(\tau+1)\delta p ,p_{0}-\tau\delta p ] $ and $\tau=1,\ldots, (M-1)$ satisfying 
\begin{align}
    |\nu(\tau,\epsilon) - \nu(\tau+1,\epsilon)| \leq \frac{2}{p_{1}} \delta p.
\end{align}
We thus obtain that $x_\tau(\nu(\tau,\epsilon) - \nu(\tau+1,\epsilon))\in [0,\frac{2}{p_{1}} \delta p]$ for $\tau=1,\ldots, (M-1)$. Besides, $\nu(M,\epsilon)=0$. The convergence rate to the expectation, by the Hoeffding inequality, is given by 
\begin{align}
    \Pr[|\Delta W - \EX[\Delta W]|\geq \tilde{\epsilon} ] \leq 2 e^{-\frac{\tilde{\epsilon}^2}{\sum_{\tau=1}^{M-1} \left(\frac{2}{p_{1}} \delta p \right)^2}} \leq 2 e^{-\frac{p_{1}^2 \tilde{\epsilon}^2 M}{(2p_{0}-1)^2}}. 
\end{align}
Taking $M\to \infty$, we obtain
\begin{align}
    \lim_{M\to\infty}\Pr[|\Delta W - \EX[\Delta W]|\geq \tilde{\epsilon} ] =0. 
\end{align}
\end{proof}

Theorem~\ref{thm:work_distribution} shows that the extracted work is close to either $w_{0}$ or $w_{1}$, with probability close to $|\langle \psi | \hat{\psi
}\rangle|^2$ and $1-|\langle \psi | \hat{\psi
}\rangle|^2$, respectively. Therefore, measuring the extracted work $\Delta W$ from the state $\psi$ in Algorithm~\ref{alg:sc_work_extraction} is effectively measuring the state $\psi$ in the basis $\lbrace \hat{\psi},\hat{\psi}^c \rbrace$ up to an error probability exponentially vanishing with respect to the number of repetitions $M$ in Algorithm~\ref{alg:sc_work_extraction}. Then we have that the energy measurement of the battery is equivalent to the reward measurement in the PSMAQB setting, and we can define a reward $X\in\lbrace 0 ,1 \rbrace$ such that their correspondence is (without loss of generality assuming $w_{0}\geq w_{1}$)
\begin{align}
    X = \begin{cases}
        1,\quad & \Delta W \geq \frac{w_{0}+w_{1}}{2}, \\
        0,\quad & \Delta W \leq \frac{w_{0}+w_{1}}{2}. 
    \end{cases}
\end{align}
The distribution of the reward is 
\begin{align}
    \Pr[X=x] = \begin{cases}
        |\langle \psi | \hat{\psi}\rangle|^2+\epsilon_{\text{error}},\quad & x=1, \\
        1- |\langle \psi | \hat{\psi}\rangle|^2 - \epsilon_{\text{error}},\quad & x=0, 
    \end{cases}
\end{align}
where $|\epsilon_{\text{error}}| \leq O(e^{-CM})$ for some constant $C$. In the limit of $M\to\infty$, the correspondence reduces to 
\begin{align}\label{eq:reward_thermal}
    X = \begin{cases}
        1,\quad & \Delta W = w_{0}, \\
        0,\quad & \Delta W = w_{1}, 
    \end{cases}
\end{align}
and the distribution of the reward reduces to 
\begin{align}
    \Pr[X=x] = \begin{cases}
        |\langle \psi | \hat{\psi}\rangle|^2,\quad & x=1, \\
        1- |\langle \psi | \hat{\psi}\rangle|^2,\quad & x=0, 
    \end{cases}
\end{align}

\begin{theorem}\label{th:work_differnt_input}
\label{thm:6} 
Let $\hat{\psi}$, $\epsilon\in[0,1]$, and $w_i$ be the direction, accuracy of the value of work extracted, all defined in Theorem~\ref{thm:work_distribution}. It holds that, when applying the work extraction protocol defined in Algorithm~\ref{alg:sc_work_extraction} to an unknown state $\psi$, we have
\begin{enumerate}
    \item  The probability of measuring $\Delta W = w_i$ is given by 
        \begin{align}
        \label{eq:thm3work_probs}
            \Pr(\Delta W=w_0) &= |\langle \psi | \hat{\psi} \rangle |^2,\\
             \Pr(\Delta W=w_1) &= 1- |\langle \psi | \hat{\psi} \rangle |^2.
        \end{align}

    \item The expected extracted work is given by
        \begin{equation}
            \EX[\Delta W]=\beta^{-1}\left[D(\psi \|\id/2)-D(\psi\|\Delta_{2\epsilon}( \hat{\psi})\right]~,
        \end{equation}
        where $\Delta_{\alpha}(\rho ) = (1-\alpha)\rho + \alpha \frac{\mathbb{I}}{2}$ is the completely depolariznig channel.
\end{enumerate}
\end{theorem}
\begin{proof}
Recall that to simplify the notation, we use $|\phi_0\rangle = \ket{\hat{\psi}}$ and $|\phi_1 \rangle = \ket{\hat{\psi}^c}$. We first observe that the off-diagonal term $\bra{\phi_i} \psi \rangle \! \langle \psi\ket{\phi_j} \ket{i}\!\bra{j}_A\otimes \varphi(x)_B$ of the join state in~\eqref{eq:intial_joint_state} does not affect~\eqref{eq:joint_state_1}. Therefore, it is identical for the case where the input is $\psi$ and the case where the input is $\sum_i |\langle \phi_i | \psi \rangle |^2 \phi_i$. The latter case can be viewed as a probabilistic mixture of cases where the input is $\phi_{i}$ with probability $ |\langle \phi_i | \psi \rangle |^2 $. Therefore, Statement 1 in Theorem~\ref{thm:6} holds, i.e.,
\begin{equation}
    \Pr(\Delta W=w_i)=  |\langle \phi_i | \psi \rangle |^2 ~.
\end{equation}

Then using~\eqref{eq:thm3work_values} and~\eqref{eq:thm3work_probs}, the expected extracted work is given by 
\begin{align}
    \EX[\Delta W] & = \beta^{-1}\left(|\langle \psi | \hat{\psi} \rangle |^2( D(\hat{\psi} \| \mathbb{I}/2 )  + \ln \epsilon ) + (1-|\langle \psi | \hat{\psi} \rangle |^2)(D(\hat{\psi}^c \| \mathbb{I}/2)+ \ln (1-\epsilon) ) \right) .
\end{align}
Using $D(\hat{\psi} \| \mathbb{I}/2 ) = D(\hat{\psi}^c \| \mathbb{I}/2) $ we have
\begin{align}
     \EX[\Delta W] & = \beta^{-1}\left( D(\hat{\psi} \| \mathbb{I}/2 ) + |\langle \psi | \hat{\psi} \rangle |^2 \ln \epsilon + (1-|\langle \psi | \hat{\psi} \rangle |^2)\ln (1-\epsilon) \right).
\end{align}
Then we use that $\lbrace\hat{\psi},\hat{\psi}^c\rbrace$ determine an orthonormal basis and
\begin{align}
\EX[\Delta W] & = \beta^{-1}\left( D(\hat{\psi} \| \mathbb{I}/2 ) + \Tr(|\langle \psi | \hat{\psi} \rangle |^2 \hat{\psi}^c \ln \epsilon \hat{\psi}^c ) + \Tr ((1-|\langle \psi | \hat{\psi} \rangle |^2)\hat{\psi}\ln (1-\epsilon)\hat{\psi} )  \right) \\
&=  \beta^{-1}\left( D(\hat{\psi} \| \mathbb{I}/2 ) + \Tr \left(|\langle \psi | \hat{\psi} \rangle |^2 \hat{\psi}^c +(1-|\langle \psi | \hat{\psi} \rangle |^2)\hat{\psi} \ln \left((1-\epsilon)\hat{\psi}+\epsilon \hat{\psi}^c \right)\right)\right) \\
&=  \beta^{-1}\left( D(\hat{\psi} \| \mathbb{I}/2 ) + \Tr\left(\psi \ln \left((1-\epsilon)\hat{\psi}+\epsilon \hat{\psi}^c\right) \right) \right) \\
& = \beta^{-1}\left( D(\hat{\psi} \| \mathbb{I}/2 ) + \Tr\left(\psi \ln \Delta_{2\epsilon}(\hat{\psi}) \right) \right) \\
& = \beta^{-1}\left( D(\hat{\psi} \| \mathbb{I}/2 ) -D \left(\psi \| \Delta_{2\epsilon}(\hat{\psi}) \right) \right) .
\end{align}
\end{proof}

\subsubsection{Cumulative dissipation}

We finalize considering the original setting where we have oracle sequential access to an unknown pure qubit state $\psi$, and our goal is to extract the maximal amount of work into a battery system. To achieve this, we can use Algorithm~\ref{alg:linucb_vn_var} for the PSMAQB setting with the rewards~\eqref{eq:reward_thermal} to learn an approximate direction of the state, and then run Algorithm~\ref{alg:sc_work_extraction} to extract work based on the approximate input. Assuming sequential access to the unknown state over $N$ rounds, and using the expected extracted work from Theorem~\ref{th:work_differnt_input}, we define the dissipation at round $k \in [N]$ with respect to the optimal protocol as
\begin{align}
\dissipation^{t} := \beta^{-1} D(\psi \| \Delta_{2\epsilon_t}(\psi_t ) ),
\end{align}
where $\psi_t$ is the direction chosen by the bandit algorithm at time step $t\in[T]$ and $\epsilon_t\in[0,1]$ is the accuracy. Then the cumulative dissipation over all $T$ rounds is
\begin{align}\label{eq:apendix_dissipation_sc}
\dissipation(T) := \beta^{-1} \sum_{t=1}^T \dissipation^{t} = \beta^{-1} \sum_{t=1}^T D(\psi \| \Delta_{2\epsilon_t}(\psi_t )).
\end{align}

\begin{algorithm}[H]
	\caption{\textsf{Thermal work extraction}}
	\label{alg:cum_dissipation}
    \begin{algorithmic}[1]
      \State  Require: sequence $\lbrace \epsilon_k \rbrace_{t=1}^\infty$
          
	\For {$k=1,2,\ldots$}
           \State Receive unknown state $ | \psi \rangle$ and couple to battery state $\varphi(x-\mu_{t-1})$
           
          \State  Compute direction $|\psi_t \rangle$ with Algorithm~\ref{alg:linucb_vn_var} using $\lbrace \psi_s , X_s \rbrace_{s=1}^{t-1}$
           
           \State Extract work using Algorithm~\ref{alg:sc_work_extraction} with input $\psi_t$ and $\epsilon_t$ and get extracted work $\Delta W_t$ and energy $\mu_t$
           
          \State  Set reward $X_t = \lbrace 0 , 1 \rbrace$ according to~\eqref{eq:reward_thermal}
        \EndFor
        \end{algorithmic}
\end{algorithm}

To minimize the cumulative dissipation, we use Algorithm~\ref{alg:cum_dissipation}, which takes as input a sequence of accuracies $\lbrace \epsilon_t \rbrace_{t=1}^\infty$. At each round, the estimator uses $\psi_t$, the direction output by Algorithm~\ref{alg:linucb_vn_var}, and if $\epsilon_t$ is a good approximation of the infidelity between the true state $\psi$ and the estimate $\psi_t$, then the dissipation $\dissipation^t$ is controlled by $\epsilon_t$. This is formalized in the following theorem.
\begin{theorem}[Theorem 2.34 in~\cite{Flammia2024quantumchisquared}]\label{th:relative_entropy_fideity} 
Let $\rho$ and $\hat{\rho}$ be $d$-dimensional quantum states achieving infidelity $1 - F(\rho,\hat{\rho}) \leq \epsilon \leq \frac{1}{2}$. Then we have
\begin{align}
    D ( \rho \| \Delta_{2\epsilon} ( \hat{\rho}) ) \leq 16\epsilon \left( 2 + \ln \left( \frac{d}{2\epsilon} \right) \right) . 
\end{align}
\end{theorem}

Given the above bound, we can use the fidelity guarantee of Algorithm~\ref{alg:linucb_vn_var} in Theorem~\ref{th:regret_PSMAQB} to prove a bound on the cumulative dissipation.

\begin{theorem} Given a finite time horizon $T\in\mathbb{N}$ and $\delta\in (0,1 )$, there exists an explicit sequence of accuracies $\lbrace \epsilon_t \rbrace_{t=1}^\infty $ such that Algorithm~\ref{alg:cum_dissipation} achieves
\begin{align} 
\dissipation (T) = O\left(\beta^{-1}\log^2 ( T ) \log \left( \frac{T}{\delta} \right)  \right) .
\end{align} 
\end{theorem}

\begin{proof}
    We can choose
    \begin{align}
        \epsilon_t = \min \left\lbrace C \frac{\log\left(\frac{T}{\delta} \right)}{t} ,\frac{1}{2} \right\rbrace ,
    \end{align}
    where $C$ is the constant in Theorem~\ref{th:regret_PSMAQB} for the fidelity bound of the direction $\psi_t$. Then we can use Theorem~\ref{th:relative_entropy_fideity} combined with the fidelity guarantee of Theorem~\ref{th:regret_PSMAQB} to get that with  probability at least $1-\delta$, we have
    \begin{align}
        \dissipation (T) \leq \beta^{-1} \sum_{t=1}^{T} 16 \epsilon_t (2-\ln \epsilon_t ).
    \end{align}
The result follows by noting that, for a sufficiently large constant $t^*$, we have $\epsilon_t = C \frac{\ln\left(\frac{N}{\delta} \right)}{t}$ for all $t \geq t^*$. The dissipation incurred during the first $t^*$ rounds contributes a constant term. For the remaining rounds $t \geq t^*$, using $-\ln\epsilon_t\leq \log T$ for $t\leq T$ and summing the corresponding dissipation terms yields the claimed polylogarithmic scaling in the number of rounds $T$.
\end{proof}

\subsection{Cost of measurement and erasure: Landauer’s principle }

There is one final aspect to consider when analyzing dissipation in both of the previous models: the act of registering and erasing measurement outcomes by the learner may itself incur an energy cost. Indeed, both the measurement process and the subsequent erasure of the learner's memory are not thermodynamically free. 

This cost can be lower bounded by a quantity known as the \emph{QC-mutual information}, $I_{QC}$, which quantifies how much information the measurement extracts about the quantum system~\cite{Sagawa_2009}. Importantly, $I_{QC}$ is upper bounded by the entropy of the classical memory register. Thus, the total cost of measurement and erasure can be loosely lower bounded by the entropy of this memory register alone.

According to Landauer's principle~\cite{Sagawa_2009,Goold_2015,Tan_2022}, the minimum heat dissipation required to erase a memory register is lower bounded by the entropy change $\Delta S$ via
\begin{align}
    \beta Q \geq \Delta S.
\end{align}
Moreover, it is widely accepted that this bound can be saturated in the quasi-static limit, provided the state to be erased is known~\cite{riechers2021initial,Jun_2014,Miller_2020}.

In both of our work extraction models, we repeatedly perform measurements in the energy eigenbasis. Each such measurement requires memory erasure before the next outcome can be stored. When the source state $\ket{\psi}$ is known, the measurement outcome is deterministic and incurs no entropy cost. However, when $\ket{\psi}$ is unknown, outcomes are stochastic, leading to entropy increase and associated energy dissipation.

This additional contribution can be included in the total dissipation. Specifically, we define
\begin{align}
    \label{eq:regret_modified}
    \dissipation'(T) = \dissipation(T) + \beta^{-1} \sum_{t=1}^T \Delta S_t,
\end{align}
where $\Delta S_t$ is the entropy change of the memory register at round $t$.

Let $\alpha_t = 1 - |\langle \psi_t | \psi\rangle |^2$ be the infidelity at round $t$. $\Delta S_t$ is closely related to $\alpha_t$ in both models we consider. In the semi-classical battery model, the entropy change is upper bounded by 
\begin{align}
    \Delta S_t = -(1-\alpha_t) \ln (1-\alpha_t) - \alpha_t \ln \alpha_t \leq  \alpha_t - \alpha_t \ln \alpha_t, 
\end{align}
and the dissipation is upper bounded by 
\begin{align}
    \dissipation^{\text{sc},*}(T) = \dissipation^{\text{sc}}(T) + \beta^{-1} \sum_{t=1}^T (\alpha_t - \alpha_t \ln \alpha_t). 
\end{align}
In the Jaynes-Cummings battery model, the entropy change in round $t$ is upper bounded by 
\begin{align}
    \Delta S_t &= -(1-\alpha_t) \ln (1-\alpha_t) - \alpha_t \ln \alpha_t - \alpha_t (\cos^2\theta_t \ln \cos^2\theta_t + \sin^2\theta_t \ln \sin^2\theta_t ) \\
    &\leq 2\alpha_t - \alpha_t \ln \alpha_t, 
\end{align}
and the dissipation is upper bounded by 
\begin{align}
    \dissipation^{\text{jc},*}(T) = \dissipation^{\text{jc}}(T) + \beta^{-1} \sum_{t=1}^T (2\alpha_t - \alpha_t \ln \alpha_t). 
\end{align}

Then, the PSMAQB Algorithm~\ref{alg:linucb_vn_var}, we have from Theorem~\ref{th:regret_PSMAQB} that ensures a polylogarithmic cumulative infidelity with high probability. Now we focus on the case where the cumulative infidelity is upper bounded by $\ln\frac{T}{\delta}\ln T$ with high probability $1-\delta$, that is, for some constant $C$, 
\begin{align}
    \sum_{t=1}^T \alpha_t \leq C \ln \frac{T}{\delta}\ln T.
\end{align}
By taking the derivative $(-\alpha_t\ln\alpha_t)' = -  \ln\alpha_t - 1$, it is obvious that when $0<\alpha_t \leq e^{-1}$, $-\alpha_t\ln  \alpha_t$ increases with $\alpha_t$ while $e^{-1}\leq \alpha_t\leq 1$, $-\alpha_t\ln\alpha_t$ decreases with $\alpha_t$. Now we consider two cases: 

\textbf{Case 1: $C (\ln (T/\delta) \ln T)/T \leq e^{-1}$.} This happens when $T$ is large enough. As a result, $-\alpha_t\ln\alpha_t$ monotonically increases with respect to $\alpha_t$ when $\alpha_t\leq C  (\ln (T/\delta) \ln T)/T$. On the one hand, for $\alpha_t\leq  C(\ln (T/\delta) \ln T)/T $, 
\begin{align}
    -\alpha_t\ln\alpha_t\leq \frac{C\ln\frac{T}{\delta} \ln T}{T} \ln \frac{T}{C\ln\frac{T}{\delta} \ln T},
\end{align}
On the other hand, for $\alpha_t > C (\ln (T/\delta) \ln T)/T$, 
\begin{align}
    -\alpha_t\ln\alpha_t \leq  \alpha_t \ln\frac{T}{C\ln \frac{T}{\delta} \ln T}, 
\end{align}
Then, 
\begin{align}
    -\sum_{t=1}^{T} \alpha_t \ln\alpha_t & = - \sum_{\alpha_t:\alpha_t\leq \frac{C \ln \frac{T}{\delta} \ln T}{T}} \alpha_t \ln\alpha_t - \sum_{\alpha_t:\alpha_t >\frac{C \ln \frac{T}{\delta} \ln T}{T}}\alpha_t \ln\alpha_t \\
    & \leq \sum_{\alpha_t:\alpha_t\leq \frac{C \ln \frac{T}{\delta} \ln T}{T}} \frac{C \ln \frac{T}{\delta} \ln T}{T} \ln \frac{T}{C \ln \frac{T}{\delta} \ln T}+ \sum_{\alpha_t:\alpha_t>\frac{C \ln \frac{T}{\delta} \ln T}{T}} \alpha_t  \ln\frac{T}{C \ln \frac{T}{\delta} \ln T}  \\
    & \leq  C \ln \frac{T}{\delta} \ln T \ln \frac{T}{C \ln \frac{T}{\delta} \ln T} + C \ln \frac{T}{\delta} \ln T \ln \frac{T}{C \ln \frac{T}{\delta} \ln T} \\
    &\leq 2 C \ln \frac{T}{\delta} \ln T \ln \frac{T}{C \ln \frac{T}{\delta} \ln T}. 
\end{align}

\textbf{Case 2: $C (\ln (T/\delta) \ln T)/T > e^{-1}$.} This happens when $T$ is not very large. Therefore, $T/e\leq  C \ln (T/\delta) \ln T$ and certainly $T \geq 2$. Using the fact that $-\alpha_t\ln\alpha_t\leq e^{-1}$, we obtain
\begin{align}
    \sum_{t=1}^{T} \alpha_t \ln \frac{1}{\alpha_t} \leq \frac{T}{e} \leq C \ln \frac{T}{\delta} \ln T . 
\end{align}
Combining both cases, we conclude that 
\begin{align}
    - \sum_{t=1}^{T} \alpha_t \ln\alpha_t \leq O((\log T)^3). 
\end{align}
Substituting into~\eqref{eq:regret_modified}, we obtain the dissipation for the semi-classical model,  
\begin{align}
    \dissipation^{\text{sc},*}(T)  \leq  \dissipation^{\text{sc}}(T) + \beta^{-1} O((\log T)^3) = O((\log T)^3). 
\end{align}
and the Jaynes-Cummings battery model, 
\begin{align}
    \dissipation^{\text{jc},*}(T)  \leq  \dissipation^{\text{jc}}(T) + 2\beta^{-1} O((\log T)^3) = O((\log T)^3), 
\end{align}
Therefore, the regret still scales as $O(\mathrm{polylog}(T))$ even if we take the energy dissipation due to Landauer's principle into account.

As mentioned above, in order to achieve Landauer's bound, the erasure process must necessarily be done quasi-statically. This is not something that a sequential adaptive learner can accomplish since it needs to measure and decide on its action in real time. To circumvent this, an array of empty memory registers can be first prepared. Suppose the agent will be extracting work for $T$ steps, we first prepare a memory register $M_0$, consisting of $T$ empty registers.
\begin{equation}
    M_0 = \{\underbrace{0,0,\ldots,0}_T\}
\end{equation}
At every time step, rather than erasing the old memory that has the previous outcome remembered, the learner simply records the outcome of the measurement on the battery into a new empty register. In other words, after time step $t$, the memory register takes the form
\begin{equation}
    M_t = \{r_1,r_2,\ldots,r_t,0,\ldots,0\} \quad \text{for}\quad t\in \lbrace{1,2,\ldots T \rbrace}~.
\end{equation}
At the end of all the extractions, the memory $M_T$ can then be quasi-statically reset. As previously calculated, as $t\to T$, the distribution of $r_t$ becomes more peaked and the total entropy of the memory registers scales with $O(\log^3 T)$, likewise for the cost of measurement. The resultant dissipation therefore still scales with $O(\mathrm{polylog}(T))$.

\section{Recommender systems for quantum data}

The last application we study in this thesis is one of the natural settings where bandit algorithms can be used: \emph{recommender systems}. Specifically, we consider a model of a classical recommender system adapted to quantum data. To study this problem, however, we adopt a different model from the MAQB setting that has been the focus of previous chapters.

We are interested in analyzing a recommender system for quantum data modeled by a set of unknown quantum processes—referred to as the \textit{environment}—and a set of tasks that must be performed using these quantum processes, which we call the \textit{context set}. A learner interacts sequentially with the environment: at each round, they receive a task from the context set and must select the quantum process that best performs the task. For example, the environment could represent a set of noisy quantum computers, while the context set consists of different quantum algorithms. At each round, the learner is given a quantum algorithm to execute, and the goal is to recommend the best quantum computer for that specific task.

This model naturally shows the bandit exploration-exploitation tradeoff: the learner must explore different quantum computers to learn their performance (exploration) while also aiming to select the most suitable one for each task (exploitation). This trade-off is particularly relevant in practical scenarios where online decision-making is costly or performance-sensitive. For instance, in our example, the use of a quantum computer may incur a monetary cost, so the learner has an incentive to make near-optimal decisions at each step.

\begin{figure}
    \centering
    \includegraphics[scale=0.5]{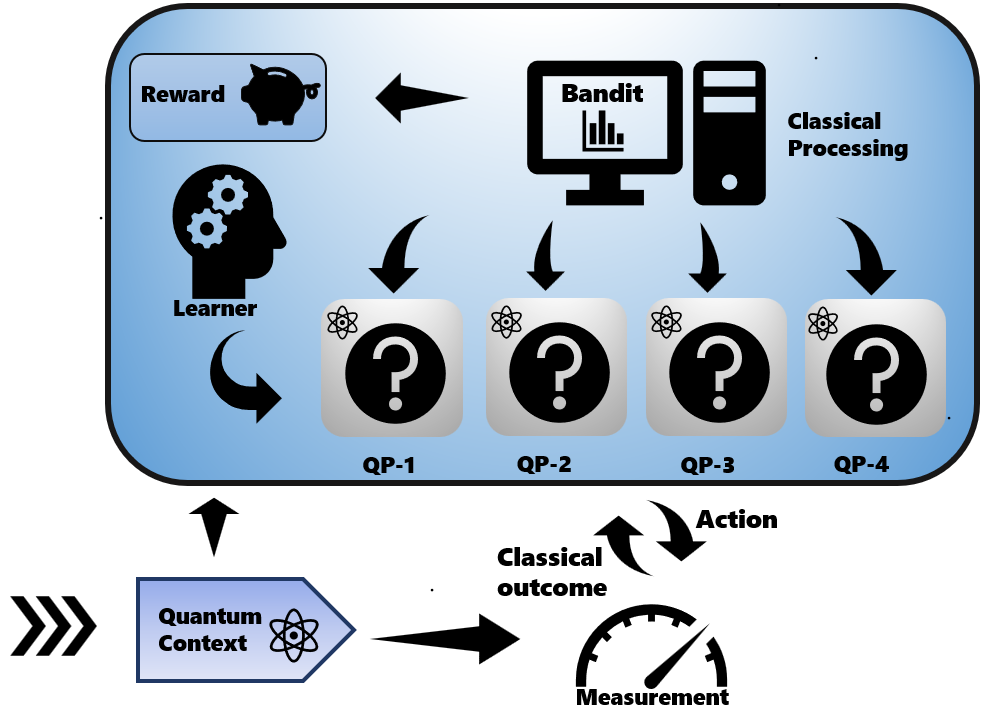}
    \caption{Sketch of a recommender system for quantum data. The learner receives sequential quantum contexts and feeds them into a classical processing system. The context is also used in the measurement system. The classical processor selects one of the quantum processes (which are unknown, except through measurements). The selected process is applied to the measurement system, and the outcome is fed back into the classical system and accumulated into the total reward.}
    \label{fig:QCB}
\end{figure}

In our formal model, the environment consists of a set of unknown quantum states, and the context set comprises a (possibly infinite) set of observables. At each round, the learner receives an observable and must choose one of the quantum states on which to perform a measurement, aiming to maximize the expected outcome. We refer to this model as the \textit{quantum contextual bandit} (QCB). It falls within the class of linear contextual bandits~\cite{abe2003reinforcement,aurer,chu} and serves as the framework for our study of recommender systems for quantum data.

The regret in the QCB setting is defined as the cumulative difference between the expected outcomes of the optimal and selected actions over all rounds. While the QCB is a restricted instance of the general linear contextual bandit problem, it still exhibits similar worst-case regret scaling (see Section~\ref{sec:lowerbound_qcb}).

The main objective of this part of the thesis is to demonstrate how the bandit framework can be applied to quantum recommender systems, rather than focusing solely on optimizing regret bounds in specific settings. 

As a concrete instance of the QCB framework, we study the \textit{low-energy quantum state recommendation problem}. In this problem, at each round, the learner receives a quantum Hamiltonian and must recommend from the environment the state with the lowest energy. This task relates closely to the ground state preparation problem, a key component in many NISQ algorithms~\cite{kishor_noisy2022}. Our model could potentially support the implementation of online recommendation strategies that guide the learner in selecting effective ansatz for energy minimization across multiple problem instances.

One advantage of using bandit algorithms for this purpose is that they do not require full reconstruction of the $d$-dimensional quantum state. Instead, they exploit the structure of the context set to focus only on the most relevant information for making recommendations. To achieve this, we combine a Gram-Schmidt procedure with classical linear bandit techniques, allowing our algorithm to maintain low-dimensional approximations of the quantum states without prior knowledge of the context set. We emphasize that our goal differs from the standard ground state finding problem: rather than identifying the exact ground state, we focus on recommending a \textit{good enough} state for the task at hand.

\subsection{The quantum contextual bandit model}

We start giving the formal definitions of quantum contextual bandit framework including the notions of policy and regret which are similar to the ones of the MAQB setting but adapted to this particular problem.

\begin{definition}[Quantum contextual bandit] Let $d\in\mathbb{N}$. A $d$-dimensional \textit{quantum contextual bandit} is given by a set of observables $\mathcal{C} = \lbrace O_c \rbrace_{c\in\Omega_{\mathcal{C}}}\subseteq \mathcal{O}_d$ that we call the \textit{context set}, $(\Omega_{\mathcal{C}},\Sigma_{\mathcal{C}})$ is a measurable space and $\Sigma_{\mathcal{C}}$ is a $\sigma$-algebra of subsets of $\Omega_{\mathcal{C}}$. The bandit is in an \textit{environment}, a finite set of quantum states  $\gamma = \lbrace \rho_1,\rho_2,\cdots , \rho_k \rbrace \subset \mathcal{S}_d $, that it is unknown. The quantum contextual bandit problem is characterized by the tuple $(\mathcal{C},\gamma)$.
\end{definition}

 Given the environment $\gamma$ such that $ |\gamma | = k$, we define the \textit{action set} $\mathcal{A} = \lbrace 1,...,k \rbrace$ as the set of indices that label the quantum states $\rho_i\in\gamma$ in the environment. For every observable $O_c \in\mathcal{C}$, the spectral decomposition is given by 
\begin{align}
O_c = \sum_{i=1}^{d_c} \lambda_{c,i}\Pi_{c,i}, 
\end{align}
where $\lambda_{c,i} \in \mathbb{R}$ denote the $d_c \leq d$ distinct eigenvalues of $O_c$ and $\Pi_{c,i}$ are the orthogonal projectors on the respective eigenspaces.
For each action $a\in\mathcal{A}$, we define the reward distribution with outcome $X \in \mathbb{R}$ as the conditional probability distribution associated with performing a measurement using $O_c$ on $\rho_a$ given by Born’s rule
\begin{align}\label{eq:reward_distribution}
Pr\left[ X = x | A = a , O= O_c \right] = P_{\rho_a}(x|a,c) = \begin{cases}
\Tr ( \rho_a \Pi_{c,i} ) \text{ if } x= \lambda_{c,i}, \\
0 \text{ else. }
\end{cases}
\end{align}

 With the above definitions, we can explain the learning process. The learner interacts sequentially with the QCB over $T$ rounds such that for every round $t\in [T]$:
 
 \begin{enumerate}
 \item The learner receives a context $O_{c_t} \in \mathcal{C}$ from some (possibly unknown) probability measure  $P_\mathcal{C}:\Sigma \rightarrow [0,1]$, over the set $\Omega_\mathcal{\mathcal{C}}$.
 
 \item Using the previous information of received contexts, actions played, and observed rewards, the learner chooses an action $A_t\in\mathcal{A}$.
 
 \item The learner uses the context $O_{c_t}$ and performs a measurement on the unknown quantum state $\rho_{A_t}$ and receives a reward $X_t$ sampled according to the probability distribution~\eqref{eq:reward_distribution}. 
 \end{enumerate}
 
 We use the index $c_t$ to denote the observable $O_{c_t}$ received at round $t\in[T]$. The strategy of the learner is given by a set of (conditional) probability distributions $\pi = \lbrace \pi_t \rbrace_{t\in\mathbb{N}}$ (policy) on the action index set $[k]$ of the form
\begin{align}
\pi_t (a_t|a_1,x_1,c_1,...,a_{t-1},x_{t-1},c_{t-1}, c_t ),
\end{align}
defined for all valid combinations of actions, rewards, and contexts \newline $(a_1,x_1,c_1,...,a_{t-1},x_{t-1},c_{t-1} )$ up to time $t-1$.
Then, if we run the policy $\pi$ on the environment $\gamma$ over $T\in\mathbb{N}$ rounds, we can define a joint probability distribution over
the set of actions, rewards, and contexts as
\begin{align}\label{eq:reward_prob}
P_{\gamma,\mathcal{C},\pi} (a_1,X_1,C_1,...,a_T,X_T,C_T)  =  \int_{C_T} \int_{X_T}\cdots\int_{C_1} \int_{X_1}  \prod_{t=1}^T \pi_t (a_t|a_1,x_1,c_1,...,a_{t-1},x_{t-1},c_{t-1} ) \times \nonumber \\
\times P_{\mathcal{C}} (dc_1) P_{\rho_{a_1}} (d x_1 |a_1,c_1 )\cdots P_{\mathcal{C}} (dc_T) P_{\rho_{a_T}} (d x_T |a_T,c_T ). 
\end{align}
Thus, the conditioned expected value of reward $X_t$ is given by
\begin{align}
\EX_{\gamma,\mathcal{C},\pi} [ X_{t} | A_t = a , O_{c_t} = O_c ] = \Tr ( \rho_a O_c ),
\end{align}
where $\EX_{\gamma,\mathcal{C},\pi} $ denotes the expectation value over the probability distribution~\eqref{eq:reward_prob}. The goal of the learner is to maximize its expected cumulative reward $\sum_{t=1}^T \EX_{\gamma,\mathcal{C},\pi} \left[ X_t \right]$ or, equivalently, minimize the \textit{cumulative regret} as
\begin{align}
R_T(\mathcal{\gamma,C},\pi ) = \sum_{t=1}^T  \max_{\rho_i\in\gamma}\Tr(\rho_iO_{c_t}) - X_t,
\end{align}
and the \textit{cumulative expected regret}
\begin{align}\label{eq:regretqcb}
\EX_{\gamma,\mathcal{C},\pi}[R_T(\mathcal{\gamma,C},\pi )] = \sum_{t=1}^T \EX_{\gamma,\mathcal{C},\pi}\left[ \max_{\rho_i\in\gamma}\Tr(\rho_iO_{c_t}) - X_t \right].
\end{align}
For a given action $a\in \mathcal{A}$ and context $O_c \in \mathcal{C}$ the \textit{sub-optimality gap}, is defined as
\begin{align}\label{eq:suboptimalitygap}
     \Delta_{a,O_c}=\max_{i\in \mathcal{A}}\Tr(\rho_{i}O_c) - \Tr(\rho_{a}O_c).
\end{align}

\textbf{Remark.} For the QCB, we will just focus on the expected regret, although similar bounds with probabilistic guarantees can be proven for the regret.

Note that the learner could try to learn the distribution of contexts $P_\mathcal{C}$, however, this will not make a difference in minimizing the regret. The strategy of the learner has to be able to learn the relevant part of the unknown states $\lbrace \rho_a \rbrace_{a=1}^{k}$ that depend on the context set and at the same time balance the tradeoff between exploration and exploitation.
We note that it is straightforward to generalize the above setting to continuous sets of contexts $\mathcal{C}$. To do that, we need a well-defined probability distribution $P_\mathcal{C} (O) dO$ over the context set $\mathcal{C}$.

\subsection{Lower bound}\label{sec:lowerbound_qcb}

In this section, we derive a lower bound for the cumulative expected regret by finding, for any strategy, a QCB that is hard to learn. Our regret lower bound proof for the QCB model relies on a reduction to a classical multi-armed stochastic bandit given in Theorem 5.1 in~\cite{nonstochasticbandit}. Recall that we reviewed the multi-armed stochastic bandit setting in Section~\ref{sec:ucb_algorithm_mab}.

  The bandit setting constructed in~Theorem 5.1 in~\cite{nonstochasticbandit} is defined with an environment $\nu = (P_a : a\in [k] )$ for $k\geq 2$ such that $P_a$ are Bernoulli distributions for all $a\in[k]$ with outcomes $\lbrace l_1,l_2 \rbrace$. Then we set the distributions as follows: we choose an index $i\in[k]$ uniformly at random and assign $P_i (R=l_1) = \frac{1+\Delta}{2}$ for some $\Delta >0$ and $P_a (R = l_1 ) = \frac{1}{2}$ for $a\neq i$.  Thus, there is a unique best action corresponding to $a=i$. Then choosing $\Delta = \small \epsilon\smash{\sqrt{\frac{k}{T}}} \normalsize$ for some small positive constant $\epsilon$, for $T\geq k$, the expected regret for any strategy will scale as 
\begin{align}\label{eq:indep_reg}
R_T (\nu,\pi)= \Omega (\sqrt{kT} ).
\end{align}

Using the above bound, we can construct the following lower bound.

\begin{theorem}
Consider a quantum contextual bandit with underlying dimension $d = 2^n$ and $T\in\mathbb{N}$, context size $c\geq 1$ and $k\geq 2$ actions. Then, for any strategy $\pi$, there exists a context set $\mathcal{C}$, $|\mathcal{C}|=c$, a probability distribution over the context set $\mathcal{C}$, $P_\mathcal{C}$ and an environment $\gamma \subset \mathcal{S}_d$ such that for the QCB defined by $(\mathcal{C},\gamma )$, the expected cumulative regret will scale as
\begin{align}\label{eq:qcblowerbound}
\EX_{\gamma,\mathcal{C},\pi}[R_T(\mathcal{\gamma,C},\pi )] = \Omega \left(\sqrt{kT}\cdot\min \left\lbrace d,\sqrt{c} \right\rbrace\right),
\end{align}
for $T\geq  k\min \lbrace c,d^2 \rbrace$.
\end{theorem}

\begin{proof}
We use a similar proof technique to~\cite{abe2003reinforcement,chu} in order to analyze the regret by dividing the problem into subsets of independent rounds. We start dividing the $T$ rounds in $c' = \min \lbrace c,d^2-1 \rbrace$ groups of $T' = \lfloor \frac{T}{c'} \rfloor$ elements. We say that time step $t$ belongs to group $s$ if $\lfloor \frac{t}{T'} \rfloor = s$. We construct a context set $\mathcal{C}$ by picking a set of $c'$ distinct Pauli observables (which is possible since the maximum number of independent Pauli observables is $d^2-1 \geq c'$), so $ \mathcal{C} = \lbrace \sigma_i \rbrace_{i=1}^{c'}$. Recall that a Pauli observable is a $n$-fold tensor product of the 2 × 2 Pauli matrices; thus, the reward will be a binary outcome $r_t\in\lbrace -1,1\rbrace$. Then, the context distribution works as follows: at each group $s$ of rounds, the learner will receive a different context $\sigma_s\in\mathcal{C}$, so at group $s$ the learner only receives $\sigma_s$. 

We want to build an environment such that for each group of rounds $s\in[c' ]$, all probability distributions are uniform except one that is slightly perturbed. We associate each Pauli observable $\sigma_i$ to one unique action $a\in [k]$, and we do this association uniformly at random (each action can be associated with more than 1 Pauli observable). Then each action $a\in[k]$ will have $\lbrace{ \sigma_{a,1},...,\sigma_{a,n_a} \rbrace}$ associated Pauli observables and we can construct the following environment $\gamma = \lbrace \rho_a \rbrace_{a=1}^{k}$ where
\begin{equation}
\rho_a = \frac{I}{d} + \sum_{j=1}^{n_a} \frac{\Delta}{d}\sigma_{a,j},
\end{equation}
$n_a \in \left\lbrace 0,1,...,d^2 -1 \right\rbrace$, $\sum_{a=1}^k n_a = c'$ and $\Delta$ is some positive constant. For every group $s\in[c' ]$, the learner will receive a fixed context $\sigma_s \in \mathcal{C}$ and there will be a unique action $a'$ with $P_{\rho_a'}(1 | A_t =a',s) =\frac{1}{2}+\frac{\Delta}{2}$ (probability of obtaining $+1$) and the rest $a\neq a'$ will have $P_{\rho_a}(1 | a,s) = \frac{1}{2}$ (uniform distributions). Thus, using that the contexts are independent ($\Tr (\sigma_i\sigma_{j}= 0)$ for $i\neq j$), we can apply~\eqref{eq:indep_reg} independently to every group $s$ and we obtain a regret lower bound $ \Omega ( \sqrt{T'k} ) =\Omega ( \smash{\sqrt{\frac{Tk}{c'}}} ) $. Note that in order to apply~\eqref{eq:indep_reg}, we need $T' \geq k$ or equivalently $T\geq c'k$. Thus, summing all the $c' $ groups we obtain the total expected regret scales as,
\begin{align}
\EX_{\gamma,\mathcal{C},\pi}[R_T(\mathcal{\gamma,C},\pi )] = \Omega \left( c' \sqrt{\frac{Tk}{c'}} \right) = \Omega \left( \sqrt{kT}\cdot\min \left\lbrace d,\sqrt{c} \right\rbrace \right).
\end{align}
\end{proof}

\subsection{Algorithm}\label{sec:algorithm}

In this section, we review the linear contextual model of multi-armed stochastic bandits and one of the main classical strategies that can be used to minimize regret in this model and also in the QCB model.

\subsubsection{Linear disjoint single context bandits and QCB}\label{subsec: linear disjoint}
The classical setting that matches our problem is commonly referred to as linear contextual bandits~\cite{chu}, although it has received other names depending on the specific setting, such as linear disjoint model~\cite{contextual_new_article_recommendation} or associative bandits~\cite{aurer}. The setting that we are interested in uses discrete action sets, and optimal algorithms are based on upper confidence bounds (\textsf{UCB}). We use the contextual linear disjoint bandit model from~\cite{contextual_new_article_recommendation} where each action $a\in [k]$ has an associated unknown parameter $\theta_a\in\mathbb{R}^d$ and at each round $t$, the learner receives a context vector $c_{t,a}\in \mathbb{R}^d$ for each action. Then after selecting an action $a\in [k]$, the sampled reward is
\begin{align}\label{eq:contextualproblem}
 X_t = \langle \theta_a , c_{t,a} \rangle + \eta_t,
\end{align} 
where $\eta_t$ is some bounded subgaussian noise.

In order to map the above setting to the $d$-dimensional QCB model $(\gamma,\mathcal{C})$, it suffices to consider a vector parametrization, similarly done for the MAQB studied in Section~\ref{sec:connection_linearbandits}. We choose a set $\lbrace \sigma_i \rbrace_{i=1}^{d^2}$ of independent Hermitian matrices and parametrize any $\rho_a \in \gamma$ and $O_l \in\mathcal{C}$ as
\begin{align}\label{eq:qcbvecparametrization}
    \rho_a = \sum_{i = 1}^{d^2}\theta_{a,i} \sigma_i , \quad O_l = \sum_{i=1}^{d^2} c_{l,i} \sigma_i ,
\end{align}
where $\theta_{a,i} = \Tr (\rho_a \sigma_i ) $ and $c_{l,i} = \Tr (O_c \sigma_i) $ and we define the vectors $\theta_a = (\theta_{a,i})_{i=1}^{d^2} \in \mathbb{R}^{d^2}$ and  $c_l  = (c_{a,i})_{i=1}^{d^2}  \in \mathbb{R}^{d^2}$. Then we note that for the QCB model, the rewards will be given by~\eqref{eq:contextualproblem} with the restriction that since we only receive one observable at each round, then the context vector is constant among all actions. Thus, in our model, the rewards have the following expression
\begin{align}\label{eq:singlecontextbandit}
    X_t = \langle \theta_a , c_{t} \rangle + \eta_t.
\end{align}
We denote this classical model as \textit{linear disjoint single context bandits}. In order to make clear when the classical real vectors parametrize an action $\rho_a\in\gamma$ or context $O_l \in \mathcal{C}$, in~\eqref{eq:qcbvecparametrization} we will use the notation $\theta_{\rho_a}$ and $c_{O_l}$  respect to the standard Pauli basis.

\subsubsection{Contextual linear Upper Confidence Bound algorithm}

Now we discuss the main strategy for the linear disjoint single context model~\eqref{eq:singlecontextbandit} that is a variant of the \textsf{LinUCB} (linear upper confidence bound) algorithm~\cite{aurer,chu,lin1,lin2,lin3} studied in Section~\ref{sec:linucb}. We describe the procedure for selecting an action, and we leave for the next section a complete description of the algorithm for the QCB setting. In order to differentiate this algorithm from the studied \textsf{LinUCB}, we named it \textsf{CLinUCB}.

At each time step $t$, given the previous rewards $X_1,...,X_{t-1} \in \mathbb{R}$, selected actions $a_1,...,a_{t-1} \in [k]$ and observed contexts $c_1 ,..., c_t \in \mathbb{R}^d$, the \textsf{CLinUCB} algorithm builds the \textit{regularized least squares estimator} for each unknown parameter ${\theta}_a$ that have the following expression
\begin{align}
\tilde{{\theta}}_{t,a}  = V_{t,a}^{-1} \sum_{s=1}^{t-1} X_s c_{s} \mathbbm{1} \lbrace a_t = a \rbrace,
\end{align}
where $V_{t,a} = \mathbb{I} + \sum_{s=1}^{t-1} c_{s} c^\top_{s} \mathbbm{1} \lbrace a_t = a \rbrace$.
Then \textsf{CLinUCB} selects the following action according to
\begin{align}
a_{t+1} = \argmax_{a\in [k]} \langle \tilde{{\theta}}_{t,a}  , {c}_{t} \rangle + \alpha \sqrt{  \langle c^\top_{t} , {V^{-1}_{t}} {c}_t \rangle } ,
\end{align}
where $\alpha > 0$ is a constant that controls the width of the confidence region on the direction of $c_{t,a}$. The idea behind this selection is to use an overestimate of the unknown expected value using an upper confidence bound. This is the principle behind \textsf{UCB}~\cite{firstUCB}, which is the main algorithm that gives rise to this class of optimistic strategies.
The value of the constant $\alpha$ is chosen depending on the structure of the action set. In the next section, we will discuss the appropriate choice of $\alpha$ for our setting.
We note that this strategy can also be applied in an adversarial approach where the context is chosen by an adversary instead of sampled from some probability distribution. 

The above procedure is shown to be sufficient for practical applications~\cite{contextual_new_article_recommendation}, but the algorithms achieving the optimal regret bound are \textsf{SupLinRel}~\cite{aurer} and \textsf{BaseLinUCB}~\cite{chu}. They use a phase elimination technique that consists of each round playing only with actions that are highly rewarding, but still the main subroutine for selecting the actions is \textsf{LinUCB}. This technique is not the most practical for applications, but it was introduced in order to derive rigorous regret upper bounds. For these strategies, if we apply it to a $d$-dimensional QCB bandit $(\gamma,\mathcal{C})$ they achieve the almost optimal regret bound of
\begin{align}\label{eq:sublinear}
\EX_{\gamma,\mathcal{C},\pi}[R_T(\mathcal{\gamma,C},\pi )] = O \left( d\sqrt{kT\log^3(T^2\log (T))} \right).
\end{align}
The above bound comes from~\cite{chu} and it is adapted to our setting using the vector parametrization~\eqref{eq:qcbvecparametrization}. Our model uses a different unknown parameter ${\theta}_a \in \mathbb{R}^d$ for each action $a\in [k]$. This model can be easily adapted to settings where they assume only one unknown parameter shared by all actions if we enlarge the vector space and define ${\theta} = (\theta_1,..., \theta_k ) \in\mathbb{R}^{dk}$ as the unknown parameter.
Their regret analysis works under the normalization assumptions $\| {\theta} \|_2 \leq 1$, $\| c_t \|_2 \leq 1$ and the choice of $\alpha = \sqrt{\frac{1}{2} \ln (2T^2 k ) } $.  We note that it matches our lower bound~\eqref{eq:qcblowerbound} except for the logarithmic terms.
While we deal with arbitrarily large qubit systems, because of the choice of the families of Hamiltonians, the dimension the algorithm works in (which we define as $d_\text{eff}$) is much smaller than the dimension of the entire Hilbert space of these multi-qubit systems. This is discussed in detail in the next section. The previous argument regarding the upper and lower bounds holds also for $d_{\text{eff}}$. It is important to note that this effective dimension is less than the entire Hilbert space in these chosen families, but not in general sets of Hamiltonians.

\subsection{Low energy quantum state recommender system}\label{sec:lowenergy}

In this section, we describe how the QCB framework can be adapted for a recommender system for low-energy quantum states.
We consider a setting where the learner is given optimization problems in an online fashion and is able to encode these problems into Hamiltonians and also has access to a set of unknown preparations of (mixed) quantum states that they want to use in order to solve these optimization problems. The task is broken into several rounds; at every round, they receive an optimization problem and are required to choose the state that they will use for that problem. As a recommendation rule, we use the state with the lowest energy with respect to the Hamiltonian where the optimization problem is encoded. We denote this problem as the \textit{low energy quantum state recommendation problem}. We note that our model focuses on the recommendation following the mentioned rule. After selecting the state the learner will use it for the optimization problem (for example the initial ansatz state of a variational quantum eigensolver), but that is a separate task.
When a learner chooses an action, they must perform an energy measurement using the given Hamiltonian on the state corresponding to the chosen action. Then the measurement outcomes are used to model rewards, and their objective is to maximize the expected cumulative reward, i.e, the expectation on the sum of the measurement outcomes over all the rounds played.  Any Hamiltonian can be written as a linear combination of Pauli observables, and we use this property to perform measurements. Now by measuring each of these Pauli observables (since these measurements are conceivable, \cite{peruzzo2014variational}), and taking the appropriately weighted sum of the measurement outcomes, we can simulate such a measurement. The QCB framework naturally lends itself to this model, where the Hamiltonians are the contexts, and the set of states that can be prepared reliably serve as the actions.

In this section, we study some important families of Hamiltonians --- specifically, the Ising model and a generalized cluster model from~\cite{verresen2017one}, which are linear combinations of Pauli observables with nearest-neighbor interactions and for $n$ qubits can be written as 
\begin{align}\label{eq:ising}
&H_\text{ising}(h)=\sum_{i=1}^{n}(\sigma_{z,i} \sigma_{z, i+1}+h\sigma_{x,i}),\\
 \label{eq:cluster}   &H_\text{cluster}(j_1,j_2)=\sum_{i=1}^{n} (\sigma_{z,i}-j_i\sigma_{x,i} \sigma_{x, i+1}-j_2 \sigma_{x,i-1}\sigma_{z,i} \sigma_{x,i+1}),
\end{align}
where $h,j_1,j_2 \in \mathbb{R}$ and $\sigma_{l,i}$ is the $l\in{x,y,z}$ Pauli matrix acting in the $i\in[n]$ site. In the Ising model, $h$ corresponds to the external magnetic field. Specifically, we consider QCB with the following context sets

\begin{align}\label{eq:contexts}
    \mathcal{C}_{\text{Ising}} = \left\lbrace H_\text{ising}(h) : h\in \mathbb{R} \right\rbrace, \quad  \mathcal{C}_{\text{cluster}} = \left\lbrace  H_\text{cluster}(j_1,j_2) : j_1,j_2\in \mathbb{R} \right\rbrace .
\end{align}
   
Important families of Hamiltonians like the models discussed above show translation-invariance and are spanned by Pauli observables showing nearest-neighbor interactions, and as a result, span a low-dimensional subspace. We illustrate the scheme described above through the example of the Ising model contexts. The Pauli observables that need to be measured are $\lbrace \sigma_{x,i}\rbrace_{ i \in [n]}$ and $\lbrace \sigma_{z,i} \sigma_{z, i+1} \rbrace_{i \in [n]}$ . These observables have 2 possible measurement outcomes, -1 and 1, and by the reward distribution of a Pauli observable $M$ given by Born's rule~\eqref{eq:reward_distribution} on a quantum state $\rho$, the reward can be modeled as
\begin{align}
    X_{M,\rho}= 2\text{Bern}\left(\frac{\Tr(M\rho)+1}{2}\right)-1 , 
\end{align}
where $\text{Bern}(x) \in \lbrace 0 , 1 \rbrace$ is a random variable with Bernoulli distribution with mean $x\in [0,1]$. By performing such a measurement for all the Pauli observables and adding the rewards, the reward for $\mathcal{C}_\text{Ising}$ is
\begin{align}
    X_{\text{Ising}}=-h\sum_{i \in [n]}X_{\sigma_{x,i},\rho} -\sum_{i \in [n]}X_{\sigma_{z,i} \sigma_{z, i+1}',\rho},
\end{align}
where we took the negative of the sum of the measurements because we are interested in a recommender system for the lowest energy state. A similar formulation applies to the QCB with generalized cluster Hamiltonian contexts.\\
In the rest of this section, we illustrate a modified \textsf{CLinUCB} algorithm for the QCB setting. Then we implement this recommender system where the contexts are Hamiltonians belonging to the Ising and generalized cluster models~\eqref{eq:contexts}, and demonstrate our numerical analysis of the performance of the algorithm by studying the expected regret. We also demonstrate that, depending on the action set, the algorithm is able to approximately identify the phases of the context Hamiltonians. 

\subsubsection{Gram-Schmidt method}\label{subsec:gram}

Similarly to the tasks of shadow tomography~\cite{shadow_aaronson}, and classical shadows~\cite{huang2020predicting}, we do not need to reconstruct the full quantum states since the algorithm only has to predict the trace between the contexts and the unknown quantum states. Thus, the \textsf{LinUCB} algorithm has only to store the relevant part of the estimators for this computation. As the measurement statistics depend only on the coefficient corresponding to the Pauli observables spanning the observables in the context set, only those Pauli observables in the expansion of the estimators are relevant. This means that our algorithm can operate in a space with a smaller dimension than the entire space spanned by $n$ qubits, which has a dimension that is exponential in the number of qubits. As a side note, we would like to remark that one could hope to apply classical shadows to our problem since both tasks seem very similar. However, this is not the case since we need an online method that performs the best recommendation at each time step $t$.

In order to exploit this property to improve the space complexity of the \textsf{CLinUCB} algorithm, we use the Gram-Schmidt procedure in the following way. At any round, a basis for the vector parameterizations (as shown in~\eqref{eq:qcbvecparametrization}) of all the previously received contexts is stored. If the incoming vector parameterization of the context is not spanned by this basis, the component of the vector orthogonal to the space spanned by this set is found by a Gram-Schmidt orthonormalization-like process, and this component is added to the set, after normalization. Therefore, at any round, there will be a list of orthonormal vectors that span the subspace of all the vector parameterizations of the contexts  received so far, and the size of the list will be equal to the dimension of the subspace, which we call \textit{effective dimension}, i.e, \\
\begin{align}\label{eq:d_eff}
    d_{\text{eff},t} :=\text{dim}(\{O_{c_s} \in \mathcal{C}\}: s \in [t]).
\end{align}
From now on, we will omit the subscript for the time step $t$ and simply denote the effective dimension as $d_\text{eff}$. Instead of feeding the context vectors directly, for any incoming context vector, we construct $d_\text{eff}$-dimensional vectors, whose $i^{th}$ term is the inner product of the context vector and the $i^{th}$ basis vector. In case the incoming vector is not spanned by the basis, we first update the list by a Gram-Schmidt procedure (which will result in an addition of another orthonormal vector to the list, and an increase in $d_\text{eff}$ by 1), and then construct a $d_\text{eff}$-dimensional vector as described before.
 This vector is fed to the \textsf{CLinUCB} algorithm. The Gram-Schmidt procedure is stated in Algorithm~\ref{alg:Gram-Schmidt} and the modified \textsf{LinUCB} algorithm is stated explicitly in Algorithm~\ref{alg:LinUCBgram}.
The efficiency of this method is well illustrated in the case where all the contexts are local Hamiltonians. As an example, we discuss the case of generalised cluster Hamiltonians. Note that the space complexity of the standard QCB framework is $O(kd^2)$, where $k$ is the number of actions, and $d$ is the dimension of the vector parameterizations of the contexts. In the standard \textsf{CLinUCB} technique, the context vectors $c_t,t \in [T]$ would be $4^n$-dimensional, where $n$ is the number of qubits the Hamiltonian acts on, in which case the space complexity of the algorithm is $O(k4^{2n})$. 
In our studies, the contexts are Ising Hamiltonians and a generalised cluster Hamiltonian~\eqref{eq:contexts} with $d_\text{eff}\leq 2$ and $d_\text{eff}\leq 3$, respectively. Since the vectors fed into the modified \textsf{LinUCB} are $d_\text{eff}$-dimensional, the space complexity is $O(kd_\text{eff}^2)$, i.e, $O(4k)$ and $O(9k)$, respectively.

\begin{algorithm}[H]
	\caption{Gram-Schmidt Algorithm ($\text{Gram}(c,{V_{a}},b_a,\text{CBasis})$)}
	\label{alg:Gram-Schmidt}
	\begin{algorithmic}
	    \State Input $[c,\{V_{a}\}_{a \in \mathcal{A}}, \{b_a\}_{a \in \mathcal{A}},\text{CBasis}]$
		    \For {v in CBasis}  
    		    \State $c \leftarrow c - \langle v, c\rangle v$
    		    \State $v_\text{ct} = v_\text{ct} \oplus \langle v , c \rangle$
    	       \EndFor
            \If {$c \ne 0$}
                \State $v_\text{ct} = v_\text{ct} \oplus \|c\|_2$
                \State Add $c/\|c\|_2$ to CBasis
    		    \For {$a = 1, 2, \ldots, K$}
    		        \State Set $V_a = V_a \oplus I_{1}$, $b_a = b_a \oplus 0_{1}$
    		    \EndFor 
            \EndIf
            \State Return $\left[c', \{V_{a}\}_{a \in \mathcal{A}}, \{b_a\}_{a \in \mathcal{A}}, \text{CBasis}\right]$	
	\end{algorithmic} 
\end{algorithm}

\begin{algorithm}[H]
	\caption{CLinUCB with Gram-Schmidt} 
	\label{alg:LinUCBgram}
	\begin{algorithmic}[1]
        \State Input $\alpha \in \mathbb{R}$ 
        \State Set CBasis = $\left[ \,\, \right]$
        \State Set $V_a = 1,\, b_a = 0, \, \forall a \in \mathcal{A}$
		\For {$t = 1, 2, \ldots$}
            \State $\left[c'_{O_t}, \{V_a\}_{a \in \mathcal{A}}, \{b_a\}_{a \in \mathcal{A}}, \text{CBasis}\right] \leftarrow \text{Gram}(c_{O_t}, \{V_a\}_{a \in \mathcal{A}}, \{b_a\}_{a \in \mathcal{A}}, \text{CBasis})$
            \For {$a \in \mathcal{A}$}
                \State $\tilde{\theta}_{\rho_a} \leftarrow V_a^{-1} b_a$
                \State $p_{t,a} \leftarrow \tilde{\theta}_{\rho_a} c'_{O_t} + \alpha \sqrt{c_{O_t}^{'\top} V_a^{-1} c'_{O_t}}$
            \EndFor
    		\State Choose action $a_t = \arg\max_{a \in \mathcal{A}} p_{t,a}$
    		\State Measure state $\rho_{a_t}$ with $O_{c_t}$ and observe reward $X_{O_t}$
    		\State Set $V_{a_t} \leftarrow V_{a_t} + c'_{O_t} c_{O_t}^{'\top}$
    		\State Set $b_{a_t} \leftarrow b_{a_t} + X_{O_t} c'_{O_t}$
		\EndFor
	\end{algorithmic} 
\end{algorithm}

\subsubsection{Phase classifier}

In order to implement the numerical simulations, we need to choose the environments for the QCB with context sets $\mathcal{C}_\text{ising}$ and $\mathcal{C}_\text{cluster}$. Elements of both context sets are parameterized by tunable parameters. We study the performance of the recommender system by choosing a context probability distribution that is uniform (over a chosen finite interval on the real line) on these parameters.  Then we chose the actions 
 as ground states of Hamiltonians that corresponded to the limiting cases (in terms of the parameters) of these models. In order to study the performance of our strategy apart from the expected regret~\eqref{eq:regretqcb}, we want to observe how the actions are chosen. For every action, we maintained a set, which contained all the Hamiltonians for which that action was chosen. We observed that almost all the elements in each of these sets belonged to the same phase of the Hamiltonian models. 
 
 In order to study the performance of the algorithm in this respect, we define the \textit{classifier regret} as
\begin{align}
   \text{ClassifierRegret}_{T} (\mathcal{\gamma,C},\pi )=\sum_{t=1}^{T}\mathbbm{1}\left[a_t \neq a_\text{optimal,t} \right],
\end{align}
where $a_\text{optimal,t}=\argmax_{a \in [k]}\tr{O_{t}\rho_a}$, and $O_t \in \mathcal{C}$ is the context observable received in $t^\text{th}$ round.
Note that the above classifier regret is not guaranteed to be sublinear like expected regret is ~\eqref{eq:sublinear} for the \textsf{LinUCB} strategy. This can be understood intuitively: consider a scenario where the bandit picks an action with a small sub-optimality gap~\eqref{eq:suboptimalitygap}; then the linear regret will increase by a very small amount, the classifier regret will increase by one unit, as all misclassifications have equal contribution to regret. These, however, are theoretic worst-case scenarios, and this classifier regret is useful to study the performance of the algorithm in practice in our settings.

\subsubsection{Numerical simulations}\label{sec:simulations}

 Before we move into the specific cases, we note the importance of the choice of $\alpha$ in Algorithm~\ref{alg:LinUCBgram}. While the theoretical analysis of the \textsf{LinUCB} algorithm depends on the choice of $\alpha$, in practice, one can tune this value to observe a better performance. We primarily use  the $\alpha$ described in~\cite{lattimore_banditalgorithm_book} (Chapter 19) given by 
\begin{align}
    \alpha_t= m + \sqrt{2\log\left( 
\frac{1}{\delta} \right)+d\log\left(1+\frac{tL^2}{d} \right)}.
\end{align}
Here, $L$ and $m$ are upper bounds on the 2-norm of the action vectors and unknown parameter, respectively,  $d$ is the dimension, and $\delta$ is once more a probability of failure.

Finally, while we study the performance of our algorithm in our simulations with estimates of expected regret and expected classifier regret, it is important to note that in an experimental setup, the learner will only be able to measure the cumulative reward at every round. However, since these are simulations, we are able to study the regret as well, as they are standard metrics to gauge the performance of the algorithms. In the next subsection, we discuss our simulations of the QCB bandit $(\gamma,\mathcal{C}_\text{cluster})$ model and the QCB bandit $(\gamma,\mathcal{C}_\text{Ising})$.

We study the performance of the recommender system for the QCB bandit $(\gamma,\mathcal{C}_\text{cluster})$, where the generalised cluster Hamiltonians  ~\cite{verresen2017one}, act on 10 qubits and 100 qubits respectively. We observe that the performance of the algorithm is not affected by the number of qubits, as the effective dimension of the context set remains unchanged,i.e, $d_\text{eff}=3$. We study the expected regret and classifier regret for these two cases, and illustrate the system's performance in finding the phases of the generalised cluster Hamiltonians. This model was also studied in \cite{caro2021generalization}, where they designed a quantum convolutional neural network to classify quantum states across phase transitions. We chose 5 actions corresponding to approximate ground states of Hamiltonians that are the limiting cases of the generalised cluster model; i.e, generalised cluster Hamiltonians with parameters $j_1,j_2$ in ~\eqref{eq:cluster}, $j_1,j_2\rightarrow \lbrace 0,0 \rbrace,\lbrace 0,\infty \rbrace,\lbrace \infty,0 \rbrace,\lbrace 0,-\infty \rbrace$and $\lbrace 0,-\infty \rbrace$. Note that these methods of approximating ground states are only for simulation purposes.

Initially, a steep growth in regret is observed, followed by a sudden slower pace. Looking closely, in the plot~\ref{fig:clusterregret}, we find that the regret indeed continues to grow, albeit at a slower pace. This was explained by observing that the sub-optimality gap of the second-best action is quite small in comparison to the sub-optimality gaps of the rest of the actions. 
Initially, the \textsf{LinUCB} algorithm does not have enough information about the unknown parameters and has to play all actions, resulting in an exploration phase. However, at some point, the bandit recognizes the “bad" actions, and plays either the best action or the action with a small sub-optimality gap most of the time - this is when the bandit has begun to balance exploration and exploitation. This is illustrated by observing the growth of the regret before and after the first 50 rounds in the insets of Fig. \ref{fig:clusterregret}.
\begin{figure}
    \centering
    \includegraphics[scale=.5]{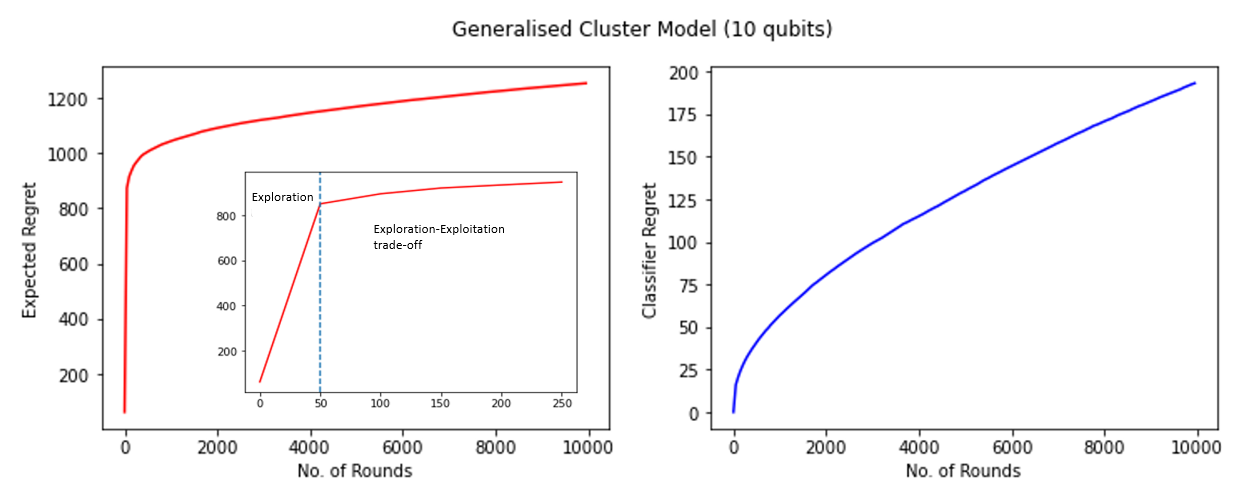}
    \includegraphics[scale=.6]{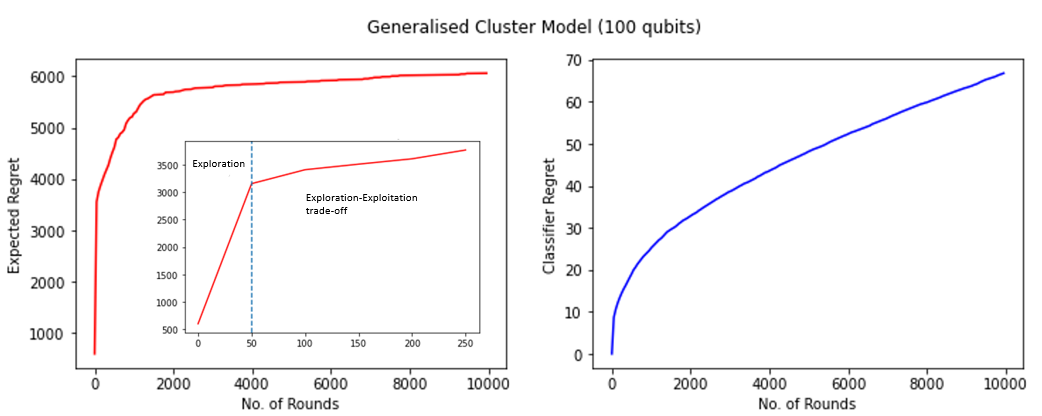}
    \caption{Plots for Regret and Classifier regret for QCB bandit $(\gamma,\mathcal{C})$, where the Hamiltonians in $\mathcal{C}$ are a specific form of generalized cluster models acting on 10 and 100 qubits, respectively. The performance is not very different since $d_\text{eff}=3$ ~\eqref{eq:d_eff} for both cases. The action set is chosen to be approximately. ground states of some generalized cluster Hamiltonians.}
    \label{fig:clusterregret}
\end{figure}
At the beginning of this subsection, we mentioned that the recommendation system picks the same action for context Hamiltonians belonging to the same phase.
 We illustrate this in Figure \ref{fig:phaseCluster}. In the scatter plot, when a context generalized cluster Hamiltonian is received, a dot is plotted with the x-axis and y-axis coordinates corresponding to its parameters $j_1,j_2$, respectively. Depending on the action picked by the algorithm, we associate a color with the dot. The resultant plot is similar (but not exact) to the phase diagram of the generalized cluster Hamiltonian. We stress that the algorithm that we are using for the below plots is the same as the one for Figure~\ref{fig:clusterregret}, and we are just plotting a visual representation in terms of the phases. We do not expect that this method can compete with other phase classifiers where the algorithms were designed for such tasks as the one in~\cite{Huang2021ProvablyEM}.
\begin{figure}
    \centering
    \includegraphics[scale=.5]{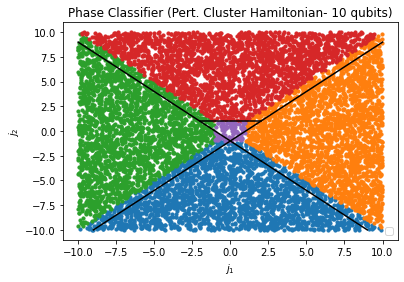}
    \includegraphics[scale=.5]{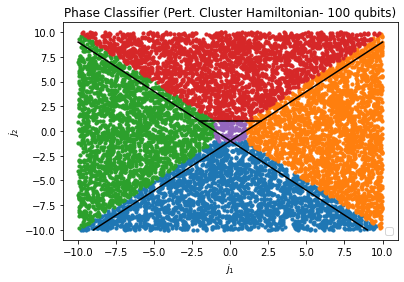}
    \caption{These plots illustrate how the recommender system identifies the phases of the generalized cluster Hamiltonian. The x and y-axis represent the coupling coefficients of the generalized cluster Hamiltonian received as context. Like the Ising Model simulations, we associate a color to each action. For any context, $H_\text{cluster}(j_1,j_2)$  corresponding to any of the T rounds, one of these actions is picked by the algorithm. We plot the corresponding colored dot (blue for the ground state of $H_\text{cluster}(-\infty,0)$, orange for $H_\text{cluster}(0,\infty)$, red for $H_\text{cluster}(\infty,0)$, green for $H_\text{cluster}(0,-\infty)$ and purple for $H_\text{cluster}(0,0)$) at the appropriate coordinates, for rounds that follow after the bandit has “learned" the actions, i.e, the growth in regret has slowed down. The lines in black represent the actual phase diagram.} 
    \label{fig:phaseCluster}
\end{figure}

We also study the performance of the recommender system for the QCB bandit $(\gamma,\mathcal{C}_\text{ising})$, where the Ising Hamiltonians act on 10 qubits and 100 qubits, respectively. We observe that the performance of the algorithm is not affected by the number of qubits, as the effective dimension of the context set remains unchanged. We study the expected regret and classifier regret for these two cases and illustrate the system's performance in finding the phases of the Ising Model.  The action set corresponds to the ground state of the 3 limiting cases of the Ising Model, i.e., Ising Hamiltonians for parameter $h$ in~\eqref{eq:ising}, $h=0,h\rightarrow -\infty$ and $h \rightarrow \infty$ .
\begin{figure}
    \centering
    \includegraphics[scale=.6]{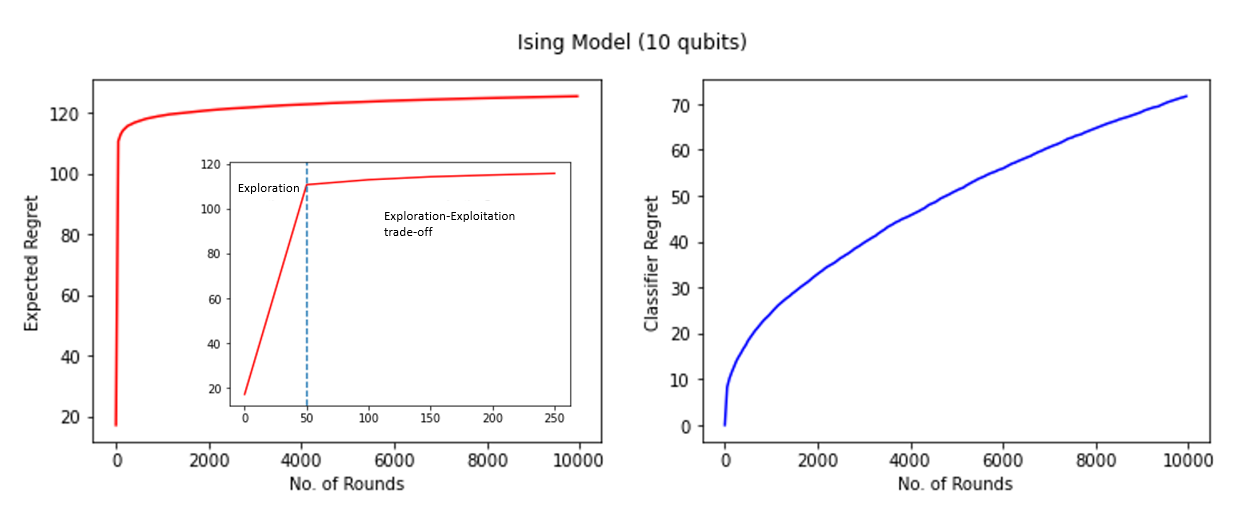}
    \includegraphics[scale=.6]{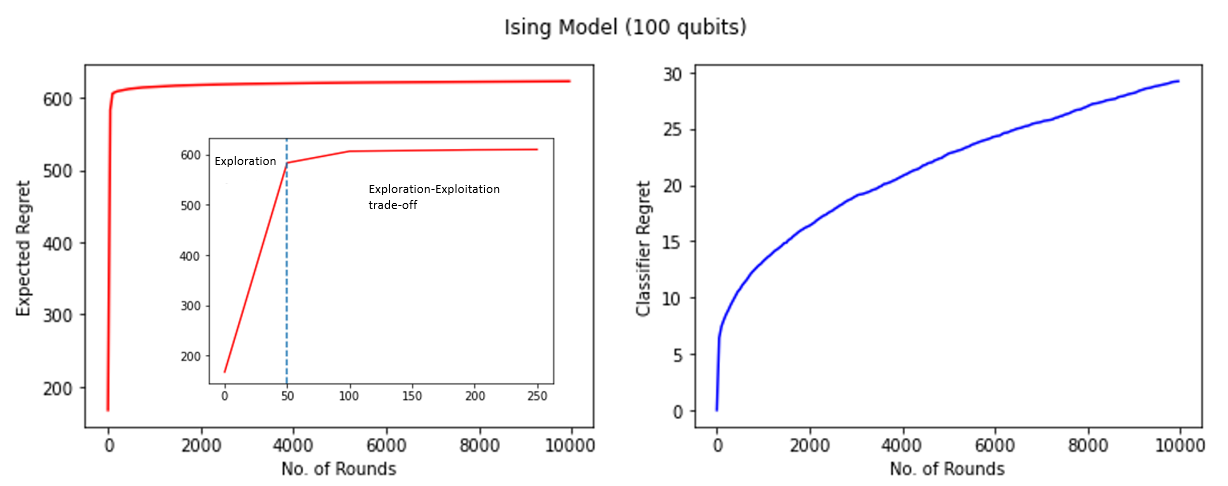}
    \caption{Plots for Regret and Classifier regret for QCB bandit $(\gamma,\mathcal{C})$, where the Hamiltonians in $\mathcal{C}$ are Ising Hamiltonians acting on 10 and 100 qubits, respectively. The performance is not very different since $d_\text{eff}=2$ ~\eqref{eq:d_eff} for both cases. The action set is chosen to be approximately the ground states of some Ising Hamiltonians.}
    \label{fig:iusingregret}
\end{figure}

Once more, like the generalized cluster Model simulations, we observe a distinct “exploration" stage, followed by an “exploration-exploitation trade-off" stage, which we again plot separately. This is illustrated by observing the growth of the regret before and after the first 50 rounds, shown in the insets of Fig. \ref{fig:iusingregret}.

At the beginning of this subsection, we mentioned that the recommendation system picks the same action for context Hamiltonians belonging to the same phase.
 We illustrate this in Figure \ref{fig:phaseIsing}. In the scatter plot, when a context Ising Hamiltonian is received, a dot is plotted with the x-axis coordinate corresponding to its parameter. Depending on the action picked by the algorithm, we associate a color with the dot. The resultant plot is very similar to that of the phase diagram of an Ising model. The Ising model is known to have phase transitions at $h=-1,1$, resulting in 3 phases, $(-\infty,-1],[-1,1]$ and $[1,\infty)$, and we observe that different actions were played in each of these ranges.
\begin{figure}[H]
    \centering
    \includegraphics[scale=.5]{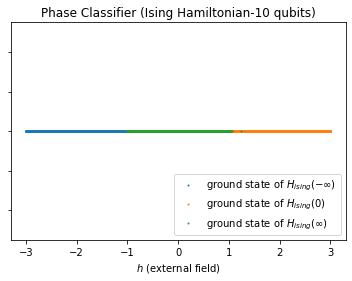}
    \includegraphics[scale=.5]{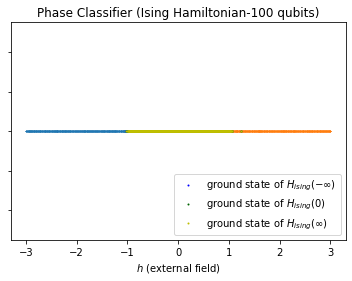}
    \caption{These plots illustrate how the recommender system identifies the phases of the Ising Hamiltonian. The x-axis represents the external field coefficient of the Ising Hamiltonian received as context. The blue, green, or yellow mark indicates that the algorithm plays the $1^\text{st}$,$2^\text{nd}$, or $3^\text{rd}$ action. We plot the corresponding colored dot at the appropriate coordinates, for rounds that follow after the bandit has “learned" the actions, i.e, the growth in regret has slowed down.}
    \label{fig:phaseIsing}
\end{figure}

\chapter{Outlook and open problems}\label{ch:open_problems}

In this thesis, we have proposed a new quantum learning framework that captures the exploration-exploitation tradeoff in reinforcement learning. Specifically, we introduced the multi-armed quantum bandit, a model that belongs to the class of multi-armed stochastic bandits, but where reward sampling is inspired by the shot noise of quantum mechanics, as described by Born's rule.

We began by establishing fundamental bounds on the scaling of regret across different environments and action sets, highlighting the difficulty of the corresponding learning tasks. For the case of mixed-state environments and discrete action sets, our upper and lower bounds match. In the case of general bandits with action sets containing all rank-1 projectors, our bounds match in terms of the number of rounds $T$, but not in the dimension $d$.

Given that learning a density matrix requires estimating $d^2$ parameters~\cite{gross2010quantum,haah2016sample}, and that the regret for continuous action sets scales accordingly, we propose the following open problem. This is motivated by the classical setting, where linear bandits exhibit regret scaling as $R_T = \Theta(d\sqrt{T})$.

\begin{itemize}
\item \textit{\textbf{Open Problem 1.} Can we tighten the lower bound for general bandits~\ref{def:general_maqb} with mixed-state environments and continuous action sets to $R_T = \Omega(d^2\sqrt{T})$? Moreover, if we know that the rank of the environment is $r\leq d$, can we have a rank-dependent lower bound $R_T = \Omega(rd\sqrt{T})$?}
\end{itemize}

Interestingly, for general bandits with pure-state environments, we show that in the qubit case the regret scales as $R_T = \Theta(\mathrm{polylog}(T))$. From the perspective of quantum state tomography, the algorithm for the PSMAQB problem introduces new techniques for the adaptive setting, including the use of the median-of-means online least squares estimator and the principle of optimism in the face of uncertainty. We expect these techniques to be applicable in other quantum learning scenarios that demand adaptiveness, especially when quantum states act as resources and must be minimally disturbed during the learning process.

Our algorithm achieves polylogarithmic regret, representing an exponential improvement over all previously known algorithms for quantum tomography, which could only reach similar fidelities by incurring linear regret. At a fundamental level, our algorithm goes beyond traditional tomography methods by showing that it is sufficient to project near the state to learn it optimally. From a classical bandit perspective, it is notable that the problem of learning pure quantum states provides the first non-trivial example of a linear bandit with a continuous action set that achieves polylogarithmic regret. This model builds a novel connection between quantum state tomography and linear stochastic bandits.

However, while we derived a dimension-dependent lower bound in Section~\ref{sec:lower_bound}, the algorithm we employed (Algorithm~\ref{alg:linucb_vn_var}) only works for qubits. This is because we formulated the algorithm over the unit sphere action set and then mapped it to the Bloch sphere—a correspondence that does not hold in higher dimensions. This leads to another open problem:

\begin{itemize}
\item \textit{\textbf{Open Problem 2.} Can we find an explicit algorithm for the PSMAQB problem in $d$ dimensions such that the regret scales as $R_T = O(d \mathrm{polylog}(T))$?}
\end{itemize}

Note that we conjecture the linear scaling in $d$ (as opposed to the $d^2$ scaling in Open Problem 1) because pure states have rank $r=1$.

We also leave as an open question of whether other quantum state tomography algorithms—particularly those designed to minimize disturbance, such as those based on weak measurements—can achieve sublinear regret, and in particular, polylogarithmic regret. We believe that adaptiveness is a key ingredient in any such algorithm.

\begin{itemize}
\item \textit{\textbf{Open Problem 3.} Are there other quantum state tomography algorithms that achieve $R_T = O(\mathrm{polylog}(T))$?}
\end{itemize}

In order to study the PSMAQB problem, we first examined a classical toy model of linear stochastic bandits with {vanishing noise}, i.e., a noise parameter that decreases linearly with the noise as we select actions close to the unknown environment. This led us to develop new proof techniques based on tightly controlling the eigenvalues of the {design matrix}. This technique allows to control the regret at any time, in contrast with the well-established regret bounds based on the elliptical potential lemma.   In~\cite{banerjee2023exploration}, the authors showed (under mild assumptions) that any strategy minimizing the regret in linear bandits with continuous and smooth action sets must satisfy $\lambda_{\min}(V_t) = \Omega(\sqrt{t})$ for {constant} noise models. However, this result does not apply to our setting due to the non-constant nature of the noise.

Our strategy instead achieves the relation $\lambda_{\min}(V_t) = \Omega(\sqrt{\lambda_{\max}(V_t)})$ (see Theorem~\ref{th:main_eigenvalues}), which is a more general, noise-independent condition. Based on this, we propose the following conjecture.

\begin{conjecture}[Informal]
Consider a linear stochastic bandit with $\mathcal{A} \subset \mathbb{R}^d$ a smooth, continuous action set and an arbitrary bounded noise model. Then, any strategy that minimizes the regret must satisfy  
\[
\lambda_{\min}(V_t) = \Omega(\sqrt{\lambda_{\max}(V_t)}).
\]
\end{conjecture}

Interestingly, during the writing of this thesis, the work~\cite{zhan2025regularizedonlinenewtonmethod} extended our setting to more general action sets, specifically in the context of convex bandits. In their framework, the reward mean is no longer given by an inner product but by a convex function applied to the action.

However, we currently lack matching lower bounds for linear bandits with non-constant noise, which leads to the following open problem.

\begin{itemize}
    \item \textbf{Open Problem 5.} \textit{Can we extend minimax lower bound techniques for linear stochastic bandits to the non-constant noise setting?}
\end{itemize}

In order to solve this problem, one has to go beyond the lower bound proofs based on analysing the KL divergences of the probability distributions of close environments.

To explore potential applications of the MAQB framework, we considered scenarios where minimizing disturbance is essential. One such application is {quantum work extraction} from unknown quantum sources. We examined both discrete and continuous battery models and we are the first to frame {sequential work extraction} as an instance of the exploration-exploitation tradeoff, quantifying cumulative dissipation in the finite-copy regime.

Even without full knowledge of the source state, our results show that it is possible to extract work {near-optimally}, with cumulative disturbance scaling polylogarithmically. However, our analysis focused only on {free energy} as a resource, which motivates the following problem.

\begin{itemize}
    \item \textbf{Open Problem 6.} \textit{Can the MAQB framework be extended to extract other quantum resources—such as entanglement, coherence, or magic—within their respective resource theories?}
\end{itemize}

Finally, we investigated a novel application of MAQB to {recommender systems for quantum data}. We introduced a rigorous bandit-based framework for quantum data recommendation, implemented using the theory of linear contextual bandits. Our analysis shows that upper and lower bounds on performance match asymptotically, and simulations confirm the efficiency of our algorithm. We further demonstrated that such systems can be used to identify {phases of Hamiltonians}.

Our current model assumes that the expected reward is a linear function of the context and the unknown state. While this is suitable for tasks like recommending low-energy quantum states (where the reward is derived from a measurement outcome), more complex tasks may require {non-linear reward functions}. Such models, known as {structured bandits}, have been studied in the classical setting~\cite{lattimore2014bounded,combes2017minimal,russo2013eluder}, and extending them to the quantum case is a natural next step.

More generally, we modeled the environment as a set of unknown, stationary quantum states. However, in realistic quantum technologies, environments may evolve dynamically due to Hamiltonian evolution or external noise. In such scenarios, recommender systems might be necessary to efficiently select quantum resources or protocols under resource constraints. This leads to our final open problem.

\begin{itemize}
    \item \textbf{Open Problem 7.} \textit{Can we go beyond the QCB setting and design recommender systems for more general quantum environments, with dynamic behaviors and practical relevance?}
\end{itemize}

\bibliography{thesis_qbandits}
\bibliographystyle{ultimate}

\appendix

\chapter{Auxiliary Lemmas of Chapter~\ref{ch:lowerbounds}}\label{ap:chap3}

\section{Proof of Lemma~\ref{lem:radon}}
\begin{proof}
Let $\{\ket{i}\}_{i=1}^{\tdim}$ be a basis of $\widetilde{\mathcal{H}}$ We will first show that $P_{O,\sigma}$ is dominated by $P_{O,\id}$ for all $\sigma$. Indeed, let $A\in \Sigma$. Assume that $P_{O,\mathbb{I}}(A)=0$. This gives us 
\begin{equation}
  \Tr [O(A)\mathbbm{1}]=\Tr [O(A)]=0,
\end{equation} 
and, because $O(A) \ge 0$, we also have $O(A)=0$.
Therefore,
\begin{equation}
  P_{O,\sigma}(A)=\Tr[O(A)\sigma]=0.
\end{equation}
Hence, for any $\ket{i},\ket{j}$ from the basis we can introduce the Radon-Nikodym derivatives $f_{\ket{i}\bra{j}}$, which will satisfy~\eqref{eq:radon_derivatives}. Then, for any $\sigma\in \End \widetilde{\mathcal{H}}$ we can define
\begin{equation}
    f_\sigma(\omega)=\sum_{i,j=1}^{\tdim} \bra{i}\sigma\ket{j}f_{\ket{i}\bra{j}}(\omega).
\end{equation}
These $f_\sigma$ are linear in $\sigma$ by definition. A direct calculation shows that they also satisfy~\eqref{eq:radon_derivatives}.
\end{proof}

\section{Proof of Lemma~\ref{lem:magic} }

\begin{proof}
  Let $\ket{\psi}\in \thil$ and define
\begin{equation}
  g_{\psi}(\omega) = \bra{\psi}V(\omega)\ket{\psi}.
\end{equation}
By the given condition, for any $A\in\Sigma$
\begin{equation}
  \int_A g_{\psi}(\omega)d\mu(\omega)= \bra{\psi}\int_A V(\omega)d\mu(\omega) \ket{\psi} \ge 0.
\end{equation}
It follows that $g_{\psi}(\omega)\ge 0$ $\mu$-almost everywhere. Let 
\begin{equation}
  Z_\psi=\{\omega\in \Omega \text{ s.t. } g_\psi(\omega) <0\}
\end{equation}
We have shown that $\mu(Z_\psi)=0$. Next, since  $\thil$ is finite-dimensional, it is separable. Therefore, there exists a countable set $\{\ket{\psi_k}\}_k$ dense in $\thil$. Let
\begin{equation}
  Z=\bigcup_k Z_{\psi_k}.
\end{equation}
We have that $\mu(Z)=0$. Finally, let $\omega\in \Omega \setminus Z$ and $\ket{\psi}\in \thil$. Because $\{\ket{\psi_k}\}$ is dense in $\thil$, there exists a sequence $\{\ket{\psi_{k_i}}\}$ converging to $\ket{\psi}$. Then,
\begin{equation}
  0 \le \bra{\psi_{k_i}}V(\omega)\ket{\psi_{k_i}} \xrightarrow{i \to \infty } \bra{\psi}V(\omega)\ket{\psi}.
\end{equation}
Overall, we get that
\begin{equation}
\forall \omega \in \Omega \setminus Z,\ \ket{\psi}\in \thil\quad \bra{\psi}V(\omega)\ket{\psi} \ge 0.
\end{equation}
Together with $\mu(Z)=0$, this gives the desired result.
\end{proof}

\chapter{Proof of Theorem~\ref{th:lowerbound_sphere}}\label{ap:lower_bound_noise}

\section{Minimax lower bound for linear bandit $\mathcal{A} = \mathbb{S}^{d-1}, \theta\in\mathbb{S}^{d-1}$ and constant noise}\label{ap:lowerbound}

In this section, we prove that any small perturbation in the reward noise of a linear stochastic bandit on the sphere (action set and environment) yields a $\Omega(\sqrt{T})$ lower bound on the regret. We study a minimax lower bound for the following reward model
\begin{align}\label{eq:noisy_classical_quantum_reward2}
    X_t =  \langle \theta, A_t \rangle +  \mathcal{N}\left(0, \sigma^2 \right) 
\end{align}
where $A_t,\theta\in\mathbb{S}^{d-1}$ and $\sigma >0$. For this setting using that $\max_{A\in\mathbb{S}^{d-1}} \langle \theta , A\rangle =1$ we have that the regret is given by
\begin{align}\label{eq:logisticregret}
    \text{Regret} (T) = \sum_{t=1}^T 1 - \langle \theta , A_t \rangle.
\end{align}
Our lower bound proof is an adaptation of the lower bound given in~\cite{abeille2021instance} that was introduced for logistic bandits and this provides a lower bound for linear bandits with $\mathcal{A} =  \mathbb{S}^{d-1},\theta\in \mathbb{S}^{d-1}$. For completeness, in  Lemma~\ref{lem:regret_logistic_bandits} we reproduce the main steps and we note that in our setting some parts simplify.  In order to state the result we define $\lbrace e_i \rbrace_{i=1}^d $ the standard basis in $\mathbb{R}^d$ and the flip operator. Given $\theta\in\mathbb{R}$, $i\in [d]$ the flip operator is defined as
\begin{align}
    \text{Flip}_i (\theta ) := (\theta_1,\theta_2,...,-\theta_i,...,\theta_d ).
\end{align}
\begin{lemma}\label{lem:regret_logistic_bandits}
Given a stochastic linear bandit with action set $\mathcal{A} =  \mathbb{S}^{d-1} = \lbrace x\in\mathbb{R}^d : \| x \|_2 =1 \rbrace$, a policy $\pi$, a reference parameter $\theta_* = \|\theta_* \|_2e_1\in\mathbb{R}^d$ and the set of parameters
\begin{align}
    \Xi = \left\lbrace \theta_* + \epsilon\sum_{i=2}^d v_ie_i , \quad v_i \in \lbrace -1 , 1 \rbrace \right\rbrace,
\end{align}
where $0< \epsilon\leq \frac{1}{d-1}$.
Then for any policy $\pi$ the average of the expected regret over the set can be lower bounded as
\begin{align}
   \frac{1}{|\Xi |} \sum_{\theta\in\Xi} \EX_{\theta,\pi}\left[ R_T ( \mathbb{S}^{d-1},\theta,\pi ) \right] \geq  T \epsilon^2  \left(\frac{d}{8} - \frac{\sqrt{d}}{4}\sqrt{\frac{1}{|\Xi |}\sum_{\theta\in\Xi}\sum_{i=2}^d D \left( \mathbb{P}_{\theta,\pi} , \mathbb{P}_{\textup{Flip}_i (\theta),\pi} \right)} \right),
\end{align}
where $\mathbb{P}_{\theta,\pi},\mathbb{P}_{\textup{Flip}_i (\theta),\pi}$ are the probability measures of actions and rewards obtained by the interaction of the policy $\pi$ with the environments $\theta$,$\textup{Flip}_i (\theta)$ respectively. Moreover, $\epsilon$ can be fixed such that for all $\theta\in\Xi$ then $\theta\in\mathbb{S}^{d-1}$.
\end{lemma}

\begin{proof}
First, we note that for any $\theta\in\Xi$ they have constant norm since
\begin{align}
    \| \theta \| = \sqrt{\| \theta_* \|^2 + (d-1)\epsilon^2}, \quad \text{for all }\theta\in\Xi .
\end{align}
Thus, we fix 
\begin{align}
    \| \theta_* \| = \sqrt{1-(d-1)\epsilon^2}
\end{align}
and we have that 
\begin{align}
    \text{if} \quad \theta \in \Xi \Rightarrow \theta \in\mathcal{S}_{d-1}.
\end{align}
Note that when choosing $\epsilon$ we will have the following restriction
\begin{align}
    \epsilon^2 \leq \frac{1}{d-1}.
\end{align}
From the expression of the regret we have
\begin{align}
    \EX_{\theta,\pi}\left[ R_T ( \mathbb{S}^{d-1},\theta,\pi ) \right] &= \frac{1}{2}\EX_{\theta,\pi}\left[ \sum_{t=1}^T \| \theta - A_t \|^2 \right] = \frac{1}{2}\EX_{\theta,\pi}\left[ \sum_{i=1}^d \sum_{t=1}^T \left[ \theta - A_t \right]_i^2 \right] \\ &\geq \frac{1}{2}\EX_{\theta,\pi}\left[ \sum_{i=1}^d \sum_{t=1}^T \left[ \theta - A_t \right]_i^2 \mathbbm{1}\lbrace B_i(\theta) \rbrace \right] = (\text{a})
\end{align}
where the event $B_i(\theta)$ is defined as
\begin{align}
    B_i(\theta) := \left\lbrace \left[\theta -  \frac{\theta_*}{\| \theta_* \|} \right]_i \left[ \sum_{t=1}^T\frac{\theta_*}{\| \theta_* \|} - A_t\right]_i \geq 0 \right\rbrace,
\end{align}
for $i=1,2,...,d$. Then we want to compare with the reference environment $\theta_*$, and we can introduce it as
\begin{align}
    (\text{a}) &= \frac{1}{2}\EX_{\theta,\pi}\left[ \sum_{i=1}^d \sum_{t=1}^T \left[ \theta - \frac{\theta_*}{\|\theta_*\|}+ \frac{\theta_*}{\|\theta_*\|} - A_t \right]_i^2 \mathbbm{1}\lbrace B_i(\theta) \rbrace \right] \\ &\geq \frac{1}{2}\EX_{\theta,\pi}\left[ \sum_{i=1}^d \sum_{t=1}^T \left[ \theta - \frac{\theta_*}{\|\theta_*\|} \right]_i^2 \mathbbm{1}\lbrace B_i(\theta) \rbrace \right] = \frac{T}{2}\EX_{\theta,\pi}\left[ \sum_{i=1}^d \left[ \theta - \frac{\theta_*}{\|\theta_*\|} \right]_i^2 \mathbbm{1}\lbrace B_i(\theta) \rbrace \right] 
\end{align}
where the last inequality follows from the identity $(x+y)^2 = x^2+2xy+y^2$ with $x = \theta - \frac{\theta_*}{\| \theta_* \|} $, $y = \frac{\theta_*}{\|\theta_*\|} - A_t $ and the definition of the event $B_i(\theta)$. Using the above computation and that $\theta\in\Xi$,
\begin{align}\label{eq:regret_with_prob}
      \EX_{\theta,\pi}\left[ R_T ( \mathbb{S}^{d-1},\theta,\pi ) \right] &\geq  \frac{T}{2} \EX_{\theta,\pi}\left[ \sum_{i=2}^d \left[ \theta - \frac{\theta_*}{\|\theta_*\|} \right]_i^2 \mathbbm{1}\lbrace B_i(\theta) \rbrace \right] \\
    &\geq \frac{T\epsilon^2}{2} \sum_{i=2}^d  \EX_{\theta,\pi}\left[ \mathbbm{1}\lbrace B_i(\theta) \rbrace \right] = \frac{T\epsilon^2}{2} \sum_{i=2}^d \mathbb{P}_{\theta,\pi} \left( B_i(\theta) \right). 
\end{align}
Now we want to apply the averaging hammer technique that was introduced in~\cite[Chapter 24, Theorem 24.1]{lattimore_banditalgorithm_book}. Let $i\in \lbrace 2,...,d \rbrace$, fix $\theta\in\Xi$ then using that $[ \theta_* ]_i = 0$ it is easy to check that
\begin{align}
    B_i (\text{Flip}_i (\theta )) = B^C_i (\theta).
\end{align}
Then using the definition of the total variation distance and Pinsker inequality we have
\begin{align}
    \mathbb{P}_{\text{Flip}_i(\theta ),\pi} (B_i (\text{Flip}_i(\theta) ) &\geq \mathbb{P}_{\theta,\pi} ( B_i (\text{Flip}_i ( \theta )  ) ) - D_{TV} (\mathbb{P}_{\theta,\pi} ,\mathbb{P}_{\text{Flip}_i(\theta )}) \\
    &\geq \mathbb{P}_{\theta,\pi} ( A^C_i (\theta)) - \sqrt{\frac{1}{2}D (\mathbb{P}_{\theta,\pi} ,\mathbb{P}_{\text{Flip}_i(\theta )}),\pi}. 
\end{align}
Then applying the above results we have
\begin{align}
  \frac{1}{|\Xi |} \sum_{\theta\in\Xi} \sum_{i=2}^d \mathbb{P}_{\theta,\pi} \left( B_i(\theta) \right) &\geq \frac{1}{2|\Xi |}\sum_{i=2}^d \sum_{\theta\in\Xi} \left( \mathbb{P}_{\theta}(B_i (\theta ) )  +\mathbb{P}_{\text{Flip}_i(\theta )} (B_i (\text{Flip}_i(\theta) )\right) \\
  & \geq \frac{1}{2|\Xi |}\sum_{i=2}^d \sum_{\theta\in\Xi} 1 - \sqrt{\frac{1}{2}D(\mathbb{P}_{\theta,\pi} ,\mathbb{P}_{\text{Flip}_i(\theta )})}.
\end{align}
Using Jensen inequality, Cauchy-Schwartz, and the fact that $d\geq1$ we have
\begin{align}\label{eq:sum_prob}
    \frac{1}{|\Xi |} \sum_{\theta\in\Xi} \sum_{i=2}^d \mathbb{P}_{\theta,\pi} \left( B_i(\theta) \right) &\geq \frac{d}{4} - \frac{\sqrt{d}}{2}\sqrt{\frac{1}{|\Xi |}\sum_{\theta\in\Xi}\sum_{i=2}^d D \left( \mathbb{P}_{\theta,\pi} , \mathbb{P}_{\text{Flip}_i (\theta),\pi} \right)}
\end{align}
And the result follows from combining~\eqref{eq:regret_with_prob} and~\eqref{eq:sum_prob}.
\end{proof}

Before proving the main theorem we need a formula for the KL divergence between two normal distributions. Given $\mathcal{N}(\mu_1,\sigma^2_1)$ and $\mathcal{N}(\mu_2,\sigma^2_2)$ it follows from a direct calculation that
\begin{align}\label{eq:kl_gaussians}
    D \left(\mathcal{N}(\mu_1,\sigma^2_1),\mathcal{N}(\mu_2,\sigma^2_2) \right) = \frac{1}{2}\left( \log\left( \frac{\sigma^2_2}{\sigma^2_1}\right)+ \frac{\sigma^2_1}{\sigma^2_2}-1 \right) + \frac{(\mu_1-\mu_2)^2}{\sigma_2^2}.
\end{align}
With these results, we are ready to prove the main theorem.
\begin{theorem}
Given a stochastic linear bandit with action set $\mathcal{A} =  \mathbb{S}^{d-1} = \lbrace x\in\mathbb{R}^d: \| x \|_2 =1 \rbrace$ and reward model given by~\eqref{eq:noisy_classical_quantum_reward2}, then there exists an unknown parameter $\theta\in\mathbb{R}^d$ such that $\|\theta \|_2 = 1$ and for any policy $\pi$
\begin{align}
    \EX_{\theta,\pi}\left[ R_T ( \mathbb{S}^{d-1},\theta,\pi ) \right]  \geq \frac{1}{100}{\sigma}d\sqrt{T}, 
\end{align}
for $T\geq \frac{1}{6400}d^2{\sigma}^2$.
\end{theorem}
\begin{proof}
First we define
\begin{align}
    \mu_{t,\theta} := \langle \theta , A_t \rangle,
\end{align}
and suppose that for any $\theta \in \Xi$,
\begin{align}\label{eq:regret_hypothesis}
   \EX_{\theta,\pi}\left[ R_T ( \mathbb{S}^{d-1},\theta,\pi ) \right] \leq C d \sigma \sqrt{T},
\end{align}
for some constant $C >0$ that without loss of generality we set to $C=1$. We see from Lemma~\ref{lem:regret_logistic_bandits} that we need to compute $D \left( \mathbb{P}_{\theta,\pi} , \mathbb{P}_{\text{Flip}_i (\theta),\pi} \right)$. From de divergence decomposition Lemma in~\cite{lattimore_banditalgorithm_book}[Chapter 24, Theorem 24.1] we have
\begin{align}
    D \left( \mathbb{P}_{\theta,\pi} , \mathbb{P}_{\text{Flip}_i (\theta),\pi} \right) = \EX_\theta \left[\sum_{t=1}^T D\left( \mathcal{N}(\mu_{t,\theta},\sigma^2),\mathcal{N}({\mu}_{t,\text{Flip}_i(\theta)},\sigma^2) \right) \right].
\end{align}
Using~\eqref{eq:kl_gaussians} we have
\begin{align}
   &D\left( \mathcal{N}(\mu_{t,\theta},\sigma^2),\mathcal{N}({\mu}_{t,\text{Flip}_i(\theta)},\sigma^2) \right) = \frac{\langle A_t,\theta - \text{Flip}_i(\theta)\rangle ^2}{\sigma^2}.
\end{align}
Thus, we have that 
\begin{align}
   D\left( \mathcal{N}(\mu_{t,\theta},\sigma^2),\mathcal{N}({\mu}_{t,\text{Flip}_i(\theta)},\sigma^2) \right) \leq \frac{4\epsilon^2}{\sigma^2} \left[A_t \right]^2_i .
\end{align}
Inserting the above into $D\left( \mathbb{P}_{\theta,\pi} , \mathbb{P}_{\text{Flip}_i (\theta),\pi} \right)$ we have
\begin{align}
   &\sum_{i=2}^d D_{KL} \left( \mathbb{P}_\theta , \mathbb{P}_{\text{Flip}_i (\theta),\pi} \right) \leq \frac{4\epsilon^2}{\sigma^2} \sum_{i=2}^d \EX_{\theta,\pi}\left[ \sum_{t=1}^T \left[\theta - A_t +\theta \right]^2_i \right]  \\ 
   &\leq \frac{8\epsilon^2}{\sigma^2} \EX_{\theta,\pi} \left[\sum_{t=1}^T \sum_{i=2}^d \left[\theta - A_t \right]^2_i + \left[ \theta \right]^2_i \right] 
   \leq \frac{8\epsilon^2}{\sigma^2} \EX_{\theta,\pi} \left[\sum_{t=1}^T \sum_{i=1}^d \left[\theta - A_t \right]^2_i + \sum_{t=1}^T\sum_{i=2}^d\left[ \theta \right]^2_i \right] \\
   &= \frac{8\epsilon^2}{\sigma^2}\left(2\EX_{\theta,\pi}\left[ R_T ( \mathbb{S}^{d-1},\theta,\pi ) \right]+T(d-1)\epsilon^2\right) \leq \frac{16\epsilon^2\sigma d\sqrt{T}}{\sigma^2} + \frac{8d\epsilon^4 T }{\sigma^2}
\end{align}
where we have used that $d\geq 1$, $\theta\in\Xi$, and~\eqref{eq:regret_hypothesis}. Thus, inserting the above into the result of Lemma~\ref{lem:regret_logistic_bandits},
\begin{align}
    \frac{1}{|\Xi |} \sum_{\theta\in\Xi} \EX_{\theta,\pi}\left[ R_T ( \mathbb{S}^{d-1},\theta,\pi ) \right]\geq  dT \epsilon^2  \left(\frac{1}{8} - \frac{1}{4}\sqrt{\frac{16\sigma\epsilon^2 \sqrt{T}+ 8\epsilon^4 T}{\sigma^2}} \right) 
\end{align}
Finally, choosing $\epsilon^2 = \frac{\sigma}{80\sqrt{T}}$ we have
\begin{align}
    \frac{1}{|\Xi |} \sum_{\theta\in\Xi} \EX_{\theta,\pi}\left[ R_T ( \mathbb{S}^{d-1},\theta,\pi ) \right]&\geq \sigma d\sqrt{T}\left( \frac{1}{8} - \frac{1}{4}\sqrt{\frac{16}{80}+\frac{8}{6400}} \right) \\
    &\geq \frac{1}{100}\sigma d\sqrt{T}.
\end{align}
We note that in order to have $\epsilon^2 \leq \frac{1}{d-1}$ we need $T\geq \frac{1}{6400}d^2\sigma^2$. Note that we proved the result under the hypothesis that~\eqref{eq:regret_hypothesis} holds. If~\eqref{eq:regret_hypothesis} does not hold the result follows trivially. 
\end{proof}

\end{document}